\documentclass{article}
\usepackage[sectionbib]{natbib}
\usepackage{array,epsfig,fancyheadings,rotating}
\usepackage[]{hyperref}  
\usepackage{sectsty, secdot}

\usepackage{graphicx} 
\usepackage{bm}
\usepackage{amsmath}
\usepackage{mathtools}
\usepackage{algorithm}
\usepackage{xcolor}
\usepackage{subfigure}
\usepackage{geometry}
\usepackage{color}
\usepackage{float}
\usepackage{amssymb}
\usepackage{algpseudocode}
\usepackage{amsthm}
\usepackage{verbatim}
\usepackage{booktabs}
\usepackage{url} 

\newtheorem{property}{Property}

\newtheorem{mydefinition}{Definition}
\newtheorem{myproposition}{Proposition}

\theoremstyle{remark}                  

\newcommand{\bw}{\bm{w}} 
\newcommand{\btheta}{\bm{\theta}} 
 
\newcommand{\bz}{\bm{z}} 
\newcommand{\iidsim}{\overset{iid}{\sim}}
\newcommand{\indsim}{\overset{ind}{\sim}}

\newcommand{\Cat}{\text{Cat}} 

\providecommand{\keywords}[1]
{
  \small	
  \textbf{\textit{Keywords:}} #1
}

\title{The generalized underlap coefficient with an application in clustering}
\author{Zhaoxi Zhang$^{1}$, Vanda In\'acio$^{1}$, Sara Wade$^{1}$\\ 
$^{1}$School of Mathematics, University of Edinburgh, Scotland, UK}
\date{}

\begin{document}
\maketitle
\footnotetext[1]{Zhaoxi Zhang, School of Mathematics, University of Edinburgh, Scotland, UK (Z.Zhang-156@sms.ed.ac.uk). Vanda In\'acio, School of Mathematics, University of Edinburgh, Scotland, UK (Vanda.Inacio@ed.ac.uk). Sara Wade, School of Mathematics, University of Edinburgh, Scotland, UK (sara.wade@ed.ac.uk).}

\maketitle

\begin{abstract}
Quantifying distributional separation across groups is fundamental in statistical learning and scientific discovery, yet most classical discrepancy measures are tailored to two-group comparisons. We generalize the underlap coefficient (UNL), a multi-group separation measure, to multivariate settings. We study its relationship with Bayes risk and mutual information, and further interpret the UNL as a measure of dependence between group labels and variables of interest. We propose an efficient importance sampling estimator of the UNL that can be combined with flexible density estimation methods. A key application is the assessment of partition-covariate dependence in clustering, where the UNL provides an interpretable measure of whether latent group structure can be explained by specific covariates.  The methodology is illustrated on two real-world datasets.
\end{abstract}


\section{Introduction}

Statistical analysis often requires quantifying distributional discrepancy across groups. This is essential in many tasks, such as quantifying differences across experimental conditions, evaluating classification performance, and assessing 
whether identified clusters are meaningful. 
In many settings, it is necessary to go beyond simple mean comparisons and account for heterogeneity in the entire shape, spread, and tail behavior of the distributions. 

For the two-group setting, a rich literature provides probability metrics and divergences, including total variation distance, Hellinger distance, and Kullback-Leibler (KL) divergence  \citep[see, e.g.,][for a review]{gibbs2002choosing}. Such measures play a central role in machine learning and statistics: for instance, variational inference minimizes a KL divergence to the posterior \citep{jordan1999introduction}, and classical GANs optimize an objective closely related to Jensen-Shannon divergence \citep{goodfellow2020generative}.



The overlap coefficient (OVL), defined as the common area under two univariate densities, is directly linked to total variation distance \citep{schmid2006nonparametric}. It was introduced in \cite{Weitzman1970} and has since been used as a non-directional alternative to ROC-based summaries for biomarker accuracy \citep[e.g.,][]{inacio2022bayesian}. The underlap coefficient (UNL) was recently proposed as a multi-group generalization of OVL and studied primarily in the three-group setting as a summary measure of biomarker discriminatory ability \citep{zhang2025underlap}.





Although the UNL was originally introduced for continuous univariate variables, many applications involve multivariate variables or distributions defined on non-Lebesgue supports. Our first contribution is to formulate the UNL with respect to a general $\sigma$-finite dominating measure, yielding a unified definition for multivariate continuous, discrete, and mixed-type variables. We then establish its theoretical properties and connections to total variation distance and Bayes risk. For continuous variables, the resulting definition coincides with the generalized $L^1$ distance among $K$ densities proposed by \cite{Pham-Gia01022008}. We also interpret the UNL as a measure of statistical dependence between a group label and variables of interest and compare it with mutual information, a widely used general purpose dependence measure.

Additionally, we develop an 
importance sampling algorithm to estimate the UNL. While \cite{zhang2025underlap} provides a grid-based numerical integration approach for estimating the UNL in the univariate case, such approaches become computationally prohibitive in higher dimensions. The proposed importance sampling estimator provides a scalable alternative and can be paired with Bayesian or frequentist density estimation procedures.

We demonstrate the utility of UNL in assessing partition-covariate dependence in cluster analysis. In many applications, the group structure is not observed and a common strategy is to discover latent homogeneous groups within a heterogeneous population via clustering. 
Once a partition has been obtained, an important question is whether the inferred cluster labels can be explained by additional covariates (e.g., genotypes, socioeconomic status, or environmental factors). The UNL provides a direct way to quantify this dependence by measuring the separation of covariate distributions across clusters. 

Recent advances in clustering allow covariate information to be incorporated, through formulations such as the mixture of experts \citep[MOE,][]{mu2025comprehensive}. When constructing the MOE architecture, a popular simplification assumes \textit{single weights}, under which cluster weights do not vary with covariates, implying that the inferred partition is independent of covariates. Violating this assumption can can substantially degrade predictive performance in some settings \citep{wade2025bayesian}. In such cases, the UNL provides a valuable complement to posterior predictive checks, offering a scalar diagnostic that can flag poor predictive performance of single-weight mixture models, and provide insight into where model adjustments are needed, particularly regarding whether mixture weights should depend on specific covariates.



In summary, our main contributions are: (i) a unified measure-theoretic formulation of the UNL for multivariate continuous, discrete, and mixed-type variables (Section~\ref{sec:UNL}); (ii) properties of the UNL and connections and comparisons to the total variation distance, Bayes risk, and mutual information (Section~\ref{sec:comparison}-\ref{sec:unl_proo}) (iii) a scalable importance sampling estimator (Section~\ref{sec:is}); and (iv) an application of the UNL to quantifying partition–covariate dependence in clustering (Section~\ref{sec:clustering}), illustrated using (i) a breast cancer genomic dataset and (ii) a dataset examining the influence of the pesticide DDT on the gestational age at delivery (Section~\ref{sec:app}).

\section{The generalized underlap coefficient}\label{sec:UNL}

The UNL was introduced in \cite{zhang2025underlap}  as a measure of discriminatory ability for continuous univariate biomarkers. More generally,  UNL can be viewed  as a separation metric for arbitrary collections of distributions. This interpretation extends naturally to multivariate settings and to continuous, discrete, and mixed-type variables, as formalized in Definitions \ref{def:UNL-continuous}–\ref{def:UNL-unified}.


\begin{mydefinition}[UNL for continuous variables]
\label{def:UNL-continuous}
For $K$ groups, let $X_1,\dots,X_K\in\mathrm{R}^{p}$ have probability densities
$f_1,\dots,f_K$.  
The (continuous) underlap coefficient is
\[
  \mathrm{UNL}(f_1,\dots,f_K)
  =
  \int_{\mathrm{R}^{p}}
      \max_{1\le k\le K} f_k(x)\;dx.
\]
\end{mydefinition}

\begin{mydefinition}[UNL for discrete variables]
\label{def:UNL-discrete}
For each group $k=1,\dots,K$, let $X_k=(X_{k1},\dots,X_{kp})$ take values in the finite or countably infinite product space $S = S_1\times \cdots \times S_p$, where $S_j$ is the state space of the $j$th categorical variable.
Define:
\[
p_k(x_1,\dots,x_p)
=
\Pr\{X_{k1}=x_1,\dots,X_{kp}=x_p\},
\qquad (x_1,\dots,x_p)\in S.
\]
The (discrete) underlap coefficient is
\[
  \mathrm{UNL}(p_1,\dots,p_K)
  =
  \sum_{(x_1,\dots,x_p)\in\mathcal{S}}
    \max_{1\le k\le K} p_k(x_1,\dots,x_p).
\]
\end{mydefinition}

\begin{mydefinition}[UNL for mixed type variables]
\label{def:UNL-mixed}
For each group $k=1,\dots,K$, $X_k=(X_k^{c},X_k^{d})$ with
$X_k^{c}\in\mathrm{R}^{p}$ and
$X_k^{d}\in S$, let the joint probability density/mass function be
\(
  f_k(x^{c},x^{d})
  = f_k^{c}(x^{c}| x^{d})\,
    p_k^{d}(x^{d}).
\)
The (mixed) underlap coefficient is
\[
  \mathrm{UNL}(f_1,\dots,f_K)
  =\sum_{x^d\in S}
     \int_{\mathrm{R}^{p}}
       \max_{1\le k\le K}
         f_k(x^{c},x^d)dx^{c}.
\]
\end{mydefinition}

\begin{mydefinition}[Measure theoretic formulation of UNL]
\label{def:UNL-unified}
Let $(\mathcal{X},\mathcal{B},\nu)$ be a $\sigma$–finite measure space and let
$P_1,\dots,P_K$ be probability measures absolutely continuous with
respect to $\nu$, with Radon-Nikodym derivatives
\(f_k(x)=dP_k/d\nu(x)\).
The (general) underlap coefficient is 
\[
  \mathrm{UNL}(f_1,\dots,f_K)
  =
  \int_{\mathcal{X}}
      \max_{1\le k\le K} f_k(x)\;d\nu(x).
\]
Choosing $\nu$ as (i) the Lebesgue measure,
(ii) a counting measure, or
(iii) the product of the Lebesgue measure and a counting measure
recovers Definitions \ref{def:UNL-continuous}-–\ref{def:UNL-mixed},
respectively.
\end{mydefinition}

The UNL is bounded between one and the number of groups $K$ and can be interpreted as the ``effective" number of distinguishable distributions among the $K$ groups. A value of $K$ indicates complete separation of $X$ across groups, signifying the presence of $K$ distinct ``effective" distributions without any overlap across groups. Conversely, a value of one suggests that only one ``effective" distribution exists, as all $K$ groups share one common distribution together. Intermediate values between one and $K$ correspond to partial separation, with higher values indicating a greater degree of separation.

\subsection{Comparing the underlap coefficient with prevalence-weighted measures} \label{sec:comparison}

The underlap coefficient is a prevalence-free measure of group separation by its definition. In this subsection, we compare the UNL with prevalence-weighted measures, namely, the Bayes risk as a measure of classification performance and mutual information as a dependence measure. We highlight the differences and discuss their practical implications.

For simplicity, we carry out the comparison under the assumption that $X$ is continuous (i.e. $X \in \mathrm{R}^p$). Let the group label $Z$ take values in $\{1,\dots,K\}$, and suppose $(Z,X)$ has joint density $p_{Z,X}(k,x)=\pi_k f_k(x)$,
where $\pi_k = \mathrm{P}(Z=k)$ and $f_k$ is the conditional density of $X$ given $Z=k$.

\subsubsection{Connection with Bayes risk}
The underlap coefficient is closely connected to the classical Bayes risk in multi-class classification \citep{devroye2013probabilistic}. The Bayes classifier is
\[
g^*(x)=\arg\max_{1\le k\le K} \pi_k f_k(x),
\]
and the Bayes risk is 
\[
\mathrm{BR}^*_{\pi}
=
1-
\int_{\mathrm{R}^p} \max_{1\le k\le K} \pi_k f_k(x)\,dx.
\]
In the special case of assuming equal group prevalence, \(\pi_k=1/K\), this reduces to
\[
\mathrm{BR}^*
=
1-\frac{1}{K}
\int_{\mathrm{R}^p}\max_{1\le k\le K} f_k(x)\,dx
=
1-\frac{1}{K}\mathrm{UNL}(f_1,\ldots,f_K).
\]
Equivalently,
\[
\mathrm{UNL}(f_1,\ldots,f_K)=K(1-\mathrm{BR}^*).
\]
Thus, under the assumption of equal group prevalences,  \(1-\operatorname{UNL}/K\) is the classical Bayes risk.

This connection provides a classification-based interpretation of the UNL. At the same time, the UNL differs from the Bayes risk in its treatment of group prevalence. The Bayes risk depends on the group prevalence, whereas the UNL is prevalence-free. 
This extends the perspective of Bayes risk from evaluating the performance limit of classification rules to quantifying the geometric separation among the group-conditional distributions. This difference is particularly important in settings where rare but well-separated groups are scientifically meaningful: such groups may contribute little to the prevalence-weighted Bayes risk, but they are represented equally in the UNL. Besides, estimating the Bayes risk requires specifying or estimating the target group prevalences, and therefore requires the pooled sample, or some external source of prevalence information, to be representative of the target population. By contrast, estimating the UNL only requires that the observations within each group are representative of their corresponding group-specific populations, which is a more relaxed assumption. 


\begin{myproposition}[Bayes risk prior-shift bound]
\label{prior-shift_bound}
Let \(\pi=(\pi_1,\ldots,\pi_K)\) and \(\rho=(1/K,\ldots,1/K)\) be the two group prevalence vectors, and let the associated Bayes risk be denoted by $\mathrm{BR}^*_{\pi}$ and $\mathrm{BR}^*_{\rho}$, respectively. Then
$
|\mathrm{BR}^*_{\pi}-\mathrm{BR}^*_{\rho}|
\le
\|\pi-\rho\|_1,
$
or equivalently, 
\[
\left|
\mathrm{BR}^*_{\pi}
-
\left\{
1-\frac1K\operatorname{UNL}(f_1,\ldots,f_K)
\right\}
\right|
\le
\left\|
\pi-\frac1K\mathbf 1
\right\|_1,
\]
where \(\mathbf 1=(1,\ldots,1)^\top\). Thus the difference between the
Bayes risk under group prevalences \(\pi\) and the balanced-prior Bayes
risk associated with the UNL is controlled by the size of the prior shift away from the uniform prior.
\end{myproposition}

The proof of Proposition \ref{prior-shift_bound} is provided in the Appendices. The prior-shift bound clarifies the relationship between the UNL and the classical Bayes risk. 
When the groups are close to balanced, \(1-\mathrm{UNL}/K\) provides a close surrogate for the usual Bayes risk. When the groups are highly imbalanced, the two quantities may differ substantially, especially when there are rare but well-separated groups.

\subsubsection{Comparing with mutual information}
Mutual information (MI) is a classical information–theoretic measure of statistical dependence between two random variables which quantifies how much uncertainty about one variable is reduced by observing the other. It has been applied across many scientific domains as a general-purpose dependence measure, and is extensively used as criteria in a large family of feature selection methods \citep[see e.g.][]{vergara2014review}. 




The mutual information $I(Z;X)$ is the Kullback–Leibler divergence between the joint law and the product of its marginals:
\[
  I(Z;X)
  =
  \sum_{k=1}^{K}\int_{\mathrm{R}^p}
    p_{Z,X}(k,x)
    \log\frac{p_{Z,X}(k,x)}{\pi_k\,p_X(x)}dx,
  \qquad
  p_X(x)=\sum_{k=1}^{K}\pi_k\,f_k(x).
\]
Equivalently,
$I(Z;X)=\mathrm{E}_{(Z,X)}\Bigl[\log p_{Z,X}(Z,X)/(p_Z(Z)p_X(X))\Bigr]$, measured in nats when the natural logarithm is used (or in bits when
$\log_2$ is used).

Mutual information quantifies how far the joint distribution $p_{Z,X}$ deviates, on average, from the product of its marginals $p_Z p_X$. The quantity $I(Z;X)$ can be interpreted as the amount of information that $Z$ and $X$ share, accessible symmetrically from $Z$ to $X$ and from $X$ to $Z$, without privileging either direction. By contrast, the UNL is asymmetric in its treatment of $Z$ and $X$: it is defined through the group-conditional densities of $X\mid Z=k$, so their roles are not interchangeable. 



The normalized quantity \(\mathrm{MI}_Z = I(Z;X)/H(Z)\) also quantifies the dependence of \(Z\) on \(X\), by measuring the proportion of uncertainty in \(Z\) that is removed (or “explained”) by observing \(X\). Nonetheless, \(\mathrm{MI}_Z\) and UNL target different aspects of the dependence: \(\mathrm{MI}_Z\) reflects information flow, whereas the UNL reflects geometric separability of the distributions across groups. 

\subsubsection{Illustrative examples}

We use simple three-class Gaussian examples to illustrate how the UNL behaves relative to the two prevalence-weighted criteria: the Bayes risk and the normalized mutual information \(\mathrm{MI}_Z=I(Z;X)/H(Z)\). Consider three groups with
\[
X \mid Z=k \sim
\mathrm N(\mu_k,\,1), \qquad k \in \{1,2,3\}.
\]

We first consider a symmetric separation setting in Scenario A, where \(\mu_2=0\), \(\mu_1=-D\), and \(\mu_3=D\) for \(D\in[0,6]\). This setting yields the curves of UNL, \(\mathrm{MI}_Z\), and Bayes risk in the top row of Figure \ref{UNLMIbayes_plots}. We then consider Scenario B in an asymmetric separation setting, where only one rare group moves away from two nearly indistinguishable common groups. We fix $\mu_1=-0.1$, $\mu_2=0$, and vary \(\mu_3=D\) for \(D\in[0,6]\). 
This yields the curves of UNL, \(\mathrm{MI}_Z\), and Bayes risk in the bottom row of Figure \ref{UNLMIbayes_plots}. 
In both scenarios, we compare (i) a balanced prevalence scenario,
\[
\Pr(Z=1)=\Pr(Z=2)=\Pr(Z=3)=1/3,
\]
and (ii) a highly imbalanced scenario,
\begin{align*}
&\Pr(Z=1)=\Pr(Z=3)=0.01,\qquad \Pr(Z=2)=0.98, \quad \text{for Scenario A;}\\
&\Pr(Z=1)=\Pr(Z=2)=0.495,\qquad \Pr(Z=3)=0.01, \quad \text{for Scenario B.}
\end{align*}

From Figure \ref{UNLMIbayes_plots}, we see that \(\mathrm{UNL}\) and \(\mathrm{MI}_Z\) behave qualitatively similarly, whereas the Bayes risk moves in the opposite direction. When the three groups are identical, \(\mathrm{UNL}=1\), \(\mathrm{MI}_Z=0\), and, under balanced prevalence, \(\mathrm{BR}^*=2/3\). When the three groups are perfectly separated, \(\mathrm{UNL}=3\), \(\mathrm{MI}_Z=1\), and $\mathrm{BR}^*=0$, indicating that \(Z\) can be determined exactly from \(X\).

In Scenario A, as \(D\) increases, \(\mathrm{MI}_Z\) and UNL follow a similar trend, although the relationship between them is nonlinear. This nonlinearity becomes more pronounced as the group prevalence becomes imbalanced. In contrast, when the group prevalence is extremely imbalanced, the change  in the Bayes risk is negligible as the separation increases, whereas the UNL increases from \(1\) to \(3\). This illustrates the prevalence-weighted nature of the Bayes risk: a well-separated but rare group contributes little to the overall classification error, while the UNL still reflects its geometric separation from the other groups.

It is worth noting that, in the symmetric separation example, the rare groups can still have a substantial effect on \(\mathrm{MI}_Z\). This is because \(\mathrm{MI}_Z=I(Z;X)/H(Z)\) measures the proportion of uncertainty in the group label \(Z\) that is removed by observing \(X\). The influence of group prevalences on $\mathrm{MI}_Z$ is not linear. This contrasts with the Bayes risk, where rare groups are weighted linearly by prevalence, thus correctly classifying the two rare groups can only reduce the overall error in proportion to their small prevalences.


The difference between the UNL and \(\mathrm{MI}_Z\) becomes more pronounced in Scenario B in the asymmetric partial separation setting. As shown in the bottom row of Figure \ref{UNLMIbayes_plots}, when the third group is rare, the increase in \(\mathrm{MI}_Z\) and the decrease in the Bayes risk as \(\mu_3\) moves away from \(\mu_1\) and \(\mu_2\) are almost negligible. In most empirical contexts, such values of \(\mathrm{MI}_Z\) would be considered numerically small and difficult to distinguish from noise. By contrast, the UNL increases from approximately \(1.04\) to \(2.04\), reflecting the emergence of an additional geometrically separated group.

These examples illustrate the main distinction between the UNL and the two prevalence-weighted criteria. The Bayes risk measures the prevalence-weighted performance limit of classification rules, while mutual information measures the prevalence-weighted information shared between \(Z\) and \(X\). Rare groups therefore have limited influence on the Bayes risk, since their contribution to the overall classification error is weighted linearly by their prevalence. Their effect on mutual information is more nuanced. In the symmetric separation example, all groups eventually become distinguishable, so observing \(X\) can explain most of the uncertainty in \(Z\), and \(\mathrm{MI}_Z\) can increase substantially even under class imbalance. In the asymmetric partial separation example, however, only the rare group separates, while the two common groups remain nearly indistinguishable. Observing \(X\) therefore does little to resolve the uncertainty between the common labels, and the increase in \(\mathrm{MI}_Z\) remains small. In contrast, the UNL always treats all group-specific distributions equally and directly reflects the geometric separation of rare groups from the others, regardless of their prevalence.

This prevalence–free property of the UNL is particularly important in biomedical and financial applications where rare but distinct groups (e.g., rare toxicity phenotypes, rare cell states, defaulters) are of primary scientific interest. Moreover, as discussed above, the UNL may be preferable to Bayes risk or mutual information for evaluating discriminatory ability or dependence when the pooled sample proportions, or external prevalence information, may not reliably represent the target population prevalences.




\begin{figure}[!t]
\centering
\subfigure{
\centering
\includegraphics[width=0.38\textwidth]{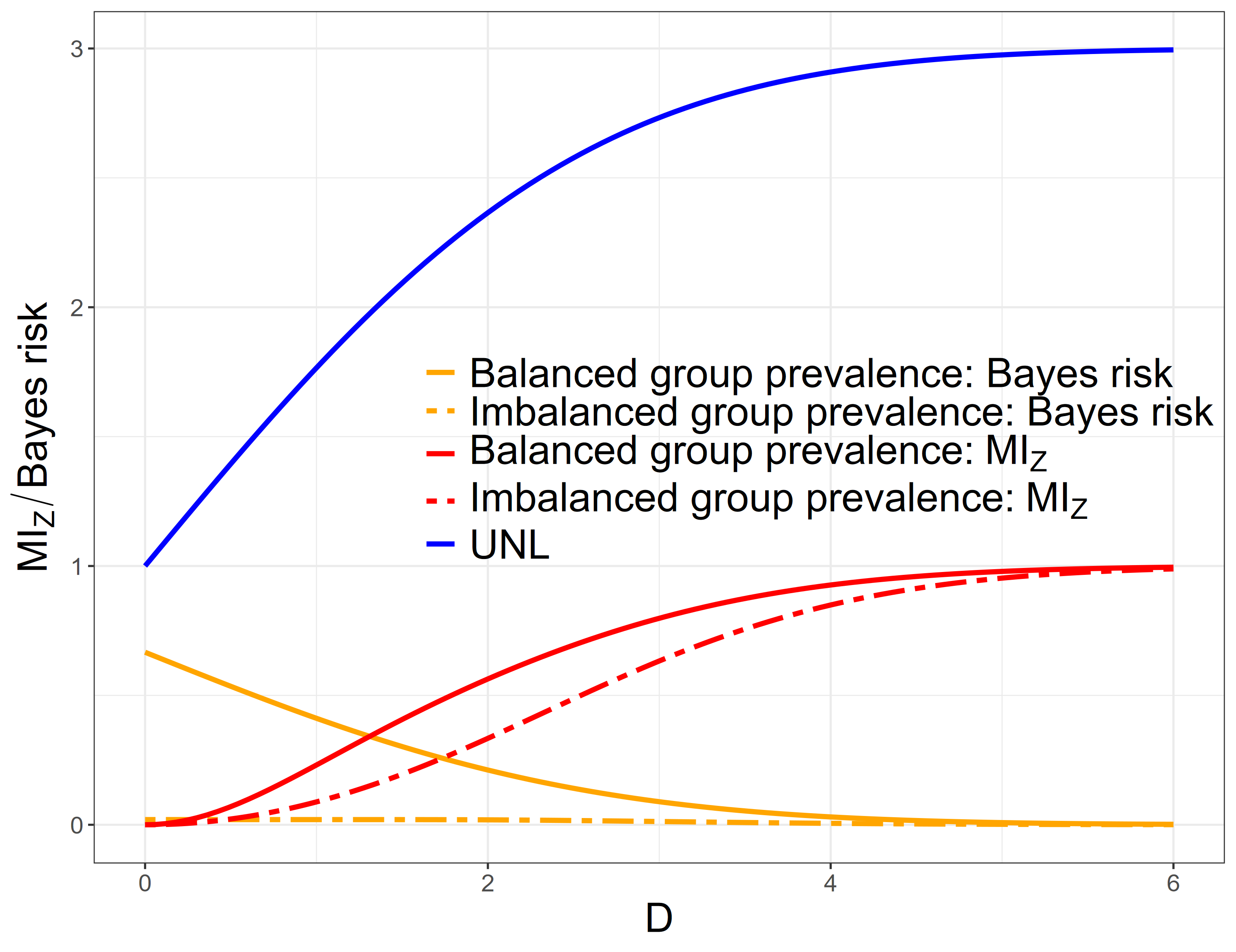}
}
\subfigure{
\centering
\includegraphics[width=0.38\textwidth]{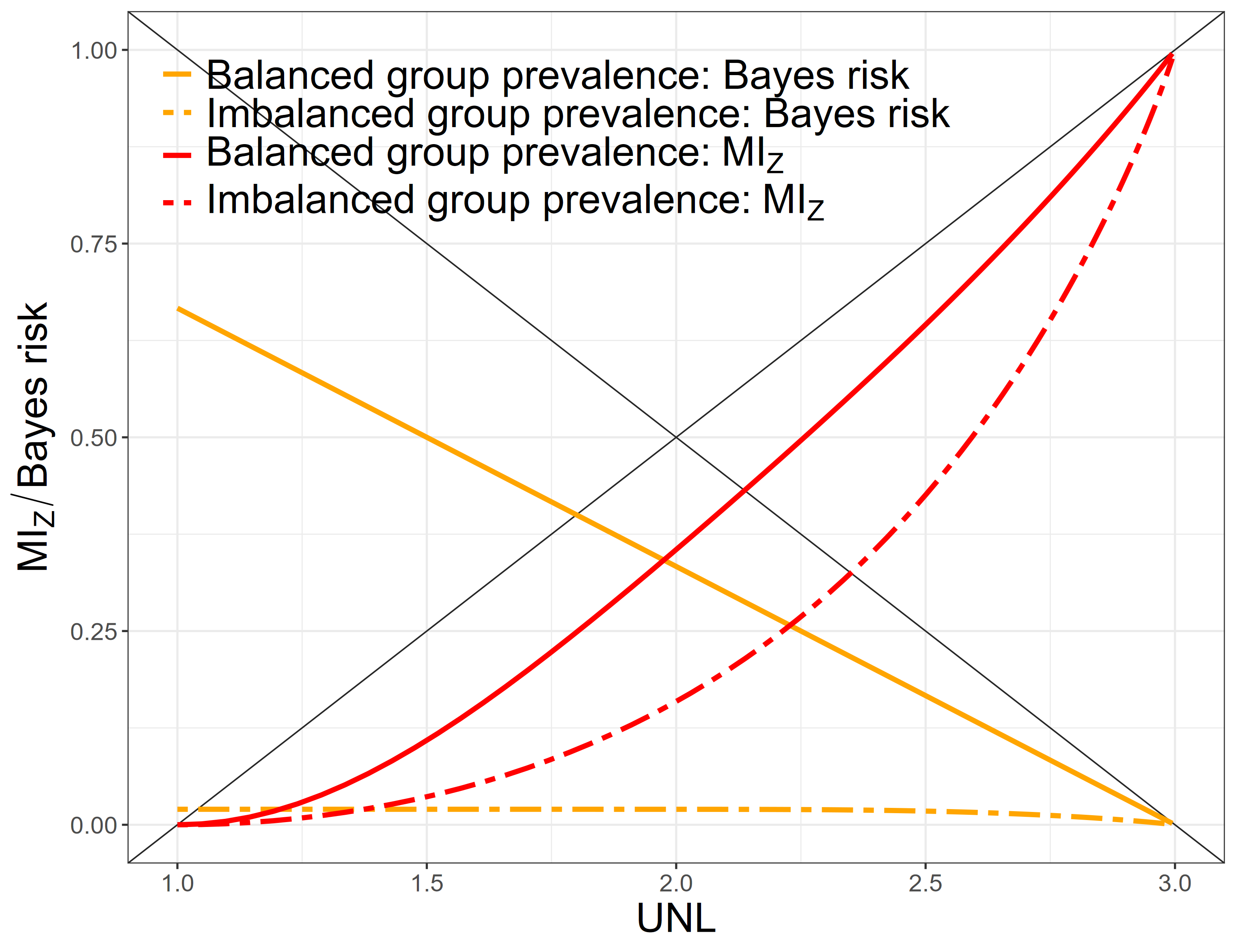}
}
\\
\centering
\subfigure{
\centering
\includegraphics[width=0.38\textwidth]{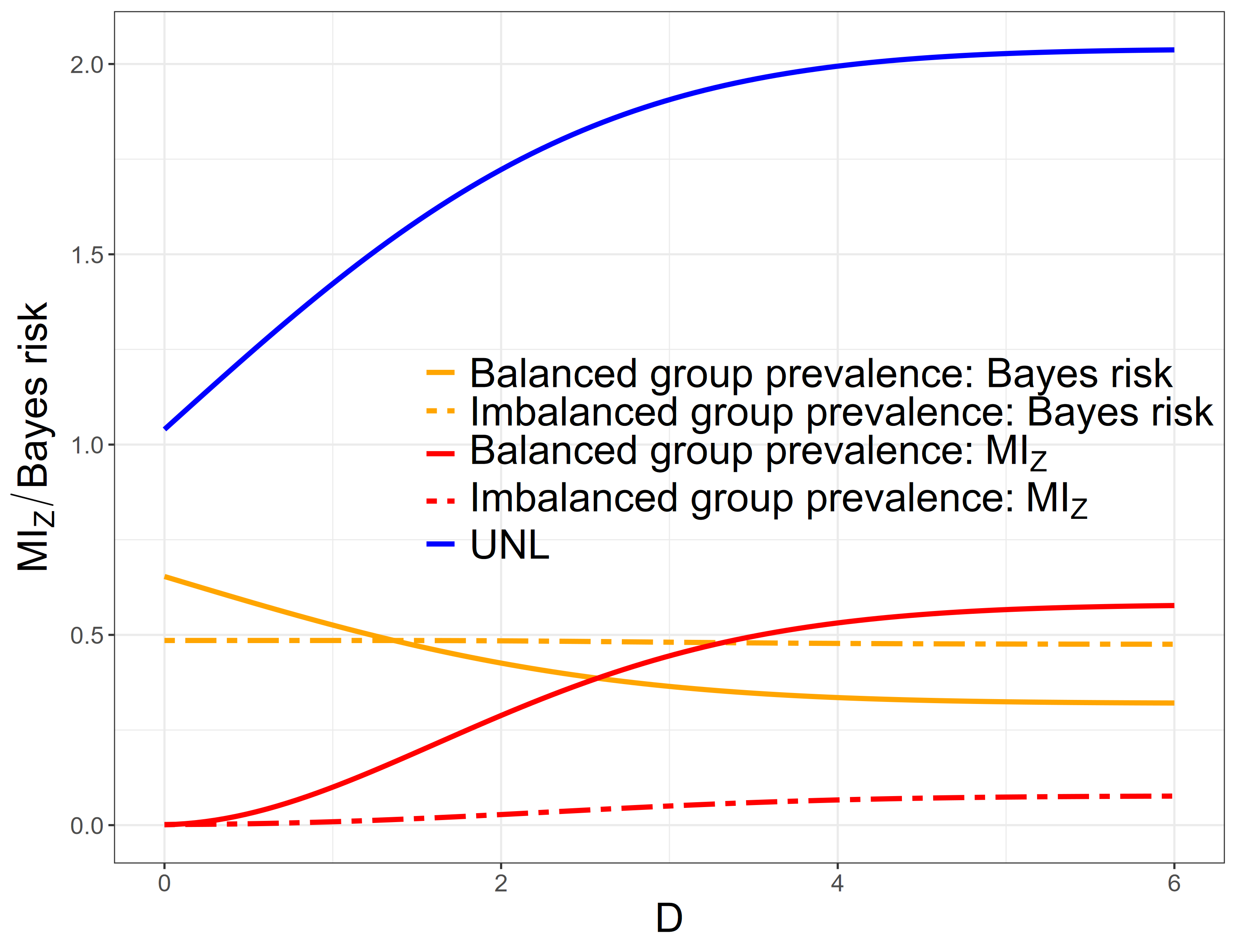}
}
\subfigure{
\centering
\includegraphics[width=0.38\textwidth]{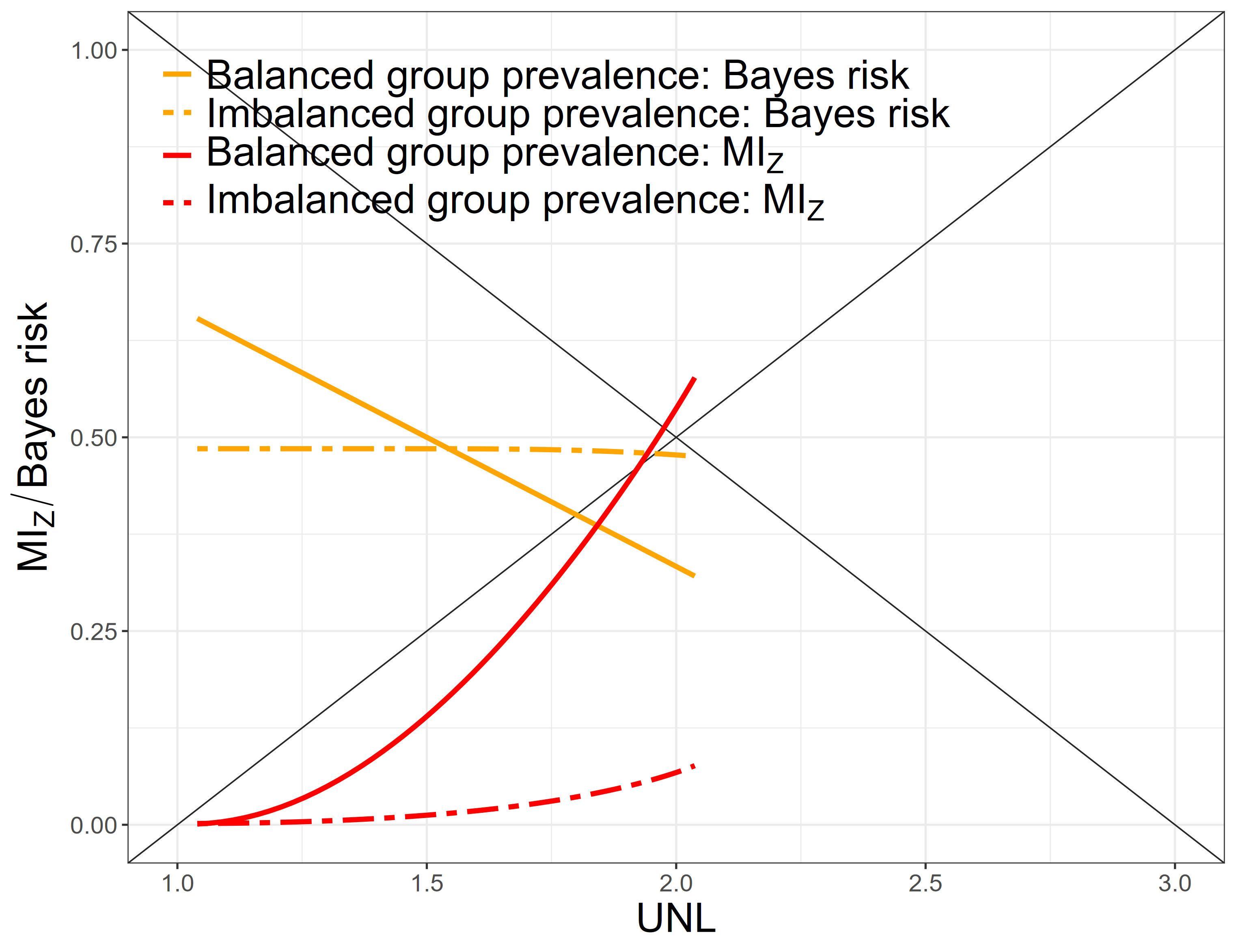}
}

\centering
\caption{Curves of UNL, \(\mathrm{MI}_Z\) and Bayes risk in the three-class Gaussian example, where \(Y \mid Z=k \sim \mathrm{N}(\mu_k,1)\) and $\mu_1=-D, \mu_2=0, \mu_3=D$ (top row), and  where \(Y \mid Z=k \sim \mathrm{N}(\mu_k,1)\) and $\mu_1=-0.1, \mu_2=0, \mu_3=D$ (bottom row).}
\label{UNLMIbayes_plots}
\end{figure}

\subsection{Connecting the underlap coefficient with total variation}
\label{subsec:underlap_tv}

The total variation distance is a canonical metric quantifying the discrepancy between two probability measures. Many other statistical distances and divergence metrics (such as Hellinger distance and various f-divergences) admit comparison inequalities with total variation or can be controlled in terms of it \citep{gibbs2002choosing}.

For two probability measures \(P_1\) and \(P_2\) on \((\mathcal{X}, \mathcal{B})\), the total variation distance is
$
\mathrm{TV}(P_1, P_2) = \sup_{A \in \mathcal{B}} \bigl|P_1(A) - P_2(A)\bigr|
$,
and can be interpreted as the largest possible difference in probabilities that the two measures assign to the same measurable event.

\begin{myproposition}[UNL's relationship with total variation distance when $K=2$]
\label{prop:UNL_TV_K2}
Suppose $P_1$ and $P_2$ are probability measures absolutely continuous with respect to $\nu$, with Radon-Nikodym derivatives $f_1$ and $f_2$, then $\text{UNL}(f_1,f_2)=1+\text{TV}(P_1,P_2)$.
\end{myproposition}


More generally, define a vector-valued measure $\mu$ as a vector consisting $K$ probability measures:
$\mu:\mathcal B\;\longrightarrow\; \mathrm R^{K}$, $\mu(A)=\bigl(P_1(A),\dots,P_K(A)\bigr)$. The corresponding total variation norm of $\mu$ is:
$
\lVert\mu\rVert_{\mathrm{TV}}
=
\sup_{\pi}
\sum_{A\in\pi}
\lVert\mu(A)\rVert
$,
where the supremum is taken over all finite measurable partitions \(\pi\) of \(\mathcal{X}\). This construction extends the concept of total variation distance between two probability measures to total variation norm of a specific specification of a vector measure \citep{dinculeanu2014vector}. If adopting the $\ell^{\infty}$ norm 
$\lVert x\rVert_{\infty}=\max_{1\le k\le K}|x_k|$, the total variation norm is then: 
\[
\lVert\mu\rVert_{\mathrm{TV},\infty}
\;=\;
\sup_{\pi}\;
\sum_{A\in\pi}\,
\lVert\mu(A)\rVert_{\infty}
=
\sup_{\pi}\;
\sum_{A\in\pi}\,
\max_{1\le k\le K} P_k(A).
\]

\begin{myproposition}[UNL equals total variation norm of a vector-valued measure consisted of K probability measures]
\label{prop:UNL_equals_TV}
The underlap of the densities $f_1,\dots,f_K$ satisfies
\[
\mathrm{UNL}(f_1,\dots,f_K)
=
\int_{\mathcal X}\max_{1\le k\le K}f_k(x)d\nu(x)
=
\lVert\mu\rVert_{\mathrm{TV},\infty}.
\]
\end{myproposition}


\subsection{Properties of the underlap coefficient}\label{sec:unl_proo}

We present three properties of the underlap coefficient: marginal monotonicity, transformation invariance, and monotonicity under dimension reduction. Since the UNL is a linear transformation of the balanced-prior Bayes risk, these results are consistent with the Bayes risk. We nevertheless state them directly in terms of the UNL for completeness and their implications for geometric separation across groups.

\begin{property}[Marginal monotonicity of UNL for continuous variables \citep{devroye2013probabilistic}]
\label{marginal_monotonicity_UNL_property_continuous}
Let $I\subseteq \{1,2,\dots,p\}$ index a subset of variables and $I^c$ denote its complement. For each class $k=1,\cdots,K$, $X_k = (X_{kI},X_{kI^c})$ is partitioned into two sets, taking values in $R^{|I|}$ and $R^{p-|I|}$, respectively. 
Define the marginal probability density function of $X_{kI}$ as:
\[
g_k(x_{I})
= \int_{R^{p-|I|}}
f_k(x)d x_{I^c},\qquad k=1,\cdots,K. 
\]
The UNL has the property:
$
\mathrm{UNL}(f_1,\dots,f_K)
\ge
\mathrm{UNL}(g_1,\dots,g_K).
$
\end{property}

The proof of Property \ref{marginal_monotonicity_UNL_property_continuous} is provided in the Appendices, along with the proofs of all other properties and propositions. Analogous results of the UNL for discrete and mixed type variables are also established there. This marginal monotonicity of the UNL confirms that it behaves in an intuitive way as a measure of separation across groups: including additional variables can only maintain or increase the degree of separation across groups and can never reduce it. Moreover, the property shows that, when the UNL is interpreted more broadly as an index of statistical dependence between the group label and the explanatory variable, it is information monotone: augmenting the variable set cannot decrease the measured dependence. In high–dimensional settings, where the UNL of the full joint density is difficult to estimate precisely, the value computed on a lower–dimensional subset of variables therefore provides a conservative yet informative lower bound benchmark. Moreover, the difference $\mathrm{UNL}(f_1,\dots,f_K)-\mathrm{UNL}(g_1,\dots,g_K)$ quantifies the marginal contribution of the omitted variables and can serve as a ranking criterion for variable selection procedures.


\begin{property}[Transformation invariance of UNL for continuous variables \citep{theisen2021evaluating}]
\label{prop:UNL-continuous}
Let $X_1,\dots,X_k\in\mathrm{R}^p$ have probability density functions $f_1,\dots,f_k$. Consider a continuously differentiable, invertible map \(\psi:\mathrm{R}^p\to\mathrm{R}^p\) with everywhere–positive Jacobian determinant
\(|J\psi(x)|>0\).  
The UNL satisfies
\[
  \text{UNL}(f_1,\dots,f_k)
  =
  \text{UNL}(f_1^{\psi},\dots,f_k^{\psi}),
\]
where 
$
  f_i^{\psi}(u)
  =
  f_i\bigl(\psi^{-1}(u)\bigr)/|J\psi^{-1}(u)|
$ 
is the density of $U_i=\psi(X_i)$.
\end{property}


This property extends the univariate result of \cite{zhang2025underlap} to multivariate settings with an arbitrary finite number of groups and applies generally to discrete, continuous, and mixed type variables. It should be noted that transforming a continuous response may have an advantage in modeling the densities in some cases, but transforming discrete responses usually does not have the same advantages for modeling purposes. Also, the invariance property is usually not preserved by the estimators we propose.

If the transformation is not invertible, for example, when performing dimension reduction on continuous variables (e.g., PCA), the UNL is monotone: projecting onto a lower-dimensional subspace can only decrease (or leave unchanged) UNL. This supports the common belief that dimension reduction typically entails information loss. 

\begin{property}[Monotonicity of UNL for continuous variables under linear dimension reduction \citep{devroye2013probabilistic_chapter2}]
\label{UNL-linear-reduction}
Let $A\in\mathrm{R}^{q\times p}$ have full row rank $q\le p$. For each $k=1,\dots,K$, let $f_k$ be the probability density
of $X_k$ on $\mathrm{R}^p$, and let $g_k$ be the probability density
of $AX_k$ on $\mathrm{R}^q$. Then
\[
  \mathrm{UNL}(f_1,\dots,f_K)\ge\mathrm{UNL}(g_1,\dots,g_K).
\]
\end{property}

The widely accepted manifold hypothesis posits that many high-dimensional datasets lie on a low-dimensional manifold embedded in a high-dimensional space \citep{cayton2005algorithms}. This has motivated a broad range of dimension-reduction methods. In practice, however, the choice of methods (and the target dimensions) is often subjective and depends strongly on the scientific objective. For group structured data, many commonly used dimension-reduction techniques are unsupervised and therefore do not explicitly incorporate group label information. If the primary goal is to preserve group separation as much as possible in the reduced representation, the UNL offers a natural criterion: it can be used to guide the choice of the number of reduced dimensions and to compare competing dimension reduction methods by assessing how well the resulting embedding retains distributional separability across groups (see e.g., the comparison between PCA and PCA+t-SNE in the adult planarian cell genomic example in \cite{whiteley2022statistical}).

\subsection{Estimation of UNL by importance sampling}\label{sec:is}
Estimating the UNL is essentially evaluating the non–overlapping probability mass across the group–conditional distributions. We propose a procedure based on importance sampling to estimate UNL. We favour importance sampling over other Monte Carlo approaches such as Markov chain Monte Carlo (MCMC) because the samples produced by an importance sampling scheme are independent, in contrast to the serially correlated draws generated by MCMC. This independence makes the computation amenable to parallel implementation, thereby improving computational efficiency.

The importance sampling procedure for estimating UNL is outlined in Algorithm \ref{alg:importance_sampling}. We describe the setting in which the densities of each group are obtained in a Bayesian framework, but the same strategy can be applied with alternative density estimators: one should simply plug the corresponding estimated densities into the importance sampling scheme.


\begin{algorithm}[!t]
\caption{Importance sampling procedure for estimating UNL for continuous variables}\label{alg:importance_sampling}
{\small
\setlength{\baselineskip}{0.5\baselineskip}
\begin{algorithmic}
    \State \textbf{Step 1: Density estimation:}
    \State Estimate the density of $X$ for every group, i.e., estimate $f_k(x)$, for $k \in \{1,\cdots,K\}$. In the Bayesian context, posterior realizations of the density estimate for group $k$ would be denoted as $\widehat{f}_k^{(s)}(x)$ for $s =1,2,\cdots,S $, where $S$ represents the number of samples (e.g., iterations of the MCMC sampler after burn-in, draws in predictive modelling \citep{fong2023martingale}).

    \For{$s=1,2,\ldots,S$}
    \State \textbf{Step 2: Draw samples from the proposal distribution $q^{(s)}(x)$:}
    \State Generate $M$ iid samples $x_{i}$ from $q^{(s)}(x)$.
    \For{$i=1,2,\ldots,M$}
        \State Compute the importance weights: 
        $
        w_i^{(s)} = \max_{1\le k\le K} \widehat{f}_k^{(s)}(x_i)/{q^{(s)}(x_i)}.
        $
    \EndFor
    
    \State \textbf{Step 3: Calculate the sample average of weights:}
    \State Compute the importance sampling estimate of the UNL:
    \[
    \widehat{\widehat{\text{UNL}}}^{(s)}_M = \frac{1}{M}\sum_{i=1}^M w_i^{(s)}.
    \]
    \EndFor
    \State Intervals of $\{\widehat{\widehat{\text{UNL}}}^{(s)}_M\}_{s=1}^S$ reflect uncertainty arising from modelling the densities given the observed data in each group.
\end{algorithmic}
}
\end{algorithm}

The accuracy and efficiency of importance sampling hinge on a wise choice of proposal density $q(x)$. It is crucial that $q(x)>0$ wherever $\max_{1\le k\le K} \widehat{f}_k(x)$ is nonzero to ensure that the ratio $\max_{1\le k\le K} \widehat{f}_k(x)/q(x)$ is well-defined. Ideally, $q(x)$ should be chosen such that the variance of the estimator is minimized. Here, we set the proposal distribution as the mixture of the $K$ group densities: \(q(x) = \sum_{k=1}^{K} \pi_k \, \widehat{f}_k(x)\).  
A convenient choice is the equally weighted mixture $q(x)=1/K\sum_{k=1}^{K} \widehat{f}_k(x),$ obtained by setting \(\pi_k=1/K\).





This equally-weighted mixture proposal has two attractive properties: (i) support matching: because \(q(x)\) is positive wherever any \(\widehat{f}_k(x)\) is positive, hence the importance ratio is well defined; (ii) bounded weights: for every \(x_i\),
      \[
        1
        \le
        \frac{\max_{1\le k\le K} \widehat{f}_k(x_i)}{q(x_i)}
        =\frac{\max_{1\le k\le K}\widehat{f}_k(x_i)}
               {1/K\sum_{k=1}^{K} \widehat{f}_k(x_i)}
        \le K,
      \]
      so the weights can never explode.      
It is worth noting that Algorithm \ref{alg:importance_sampling}, combined with the equal-weight mixture proposal relates to the 
multiple importance sampling scheme described in \cite{10.1214/18-STS668}.

The following properties give complementary consistency guarantees of the estimator, incorporating both the Monte Carlo error and the density-estimation error.

\begin{property}[Error bound from density estimation and plug-in consistency of the UNL]
\label{dense_unl_error}
Let \(f_1,\ldots,f_K\) be densities with respect to \(\nu\). For sample size \(n\), let
\(\widehat f_{1,n},\ldots,\widehat f_{K,n}\) be estimated densities with respect
to the same dominating measure \(\nu\). 
Then
\[
\left|
\operatorname{UNL}(\widehat f_{1,n},\ldots,\widehat f_{K,n})
-
\operatorname{UNL}(f_1,\ldots,f_K)
\right|
\le
\sum_{k=1}^K
\|\widehat f_{k,n}-f_k\|_1,
\]
where
\[
\|\widehat f_{k,n}-f_k\|_1
=
\int_{\mathcal X}|\widehat f_{k,n}(x)-f_k(x)|\,d\nu(x).
\]
Consequently, if $\widehat f_{k,n}$ is a consistent estimator of $f_k$ for $k=1,\cdots,K$: $\|\widehat f_{k,n}-f_k\|_1\xrightarrow{p}0$,
then $\mathrm{UNL}(\widehat f_{1,n},\ldots,\widehat f_{K,n})$ is a consistent estimator of $\mathrm{UNL}(f_1,\ldots,f_K)$: $\mathrm{UNL}(\widehat f_{1,n},\ldots,\widehat f_{K,n})
\xrightarrow{p}
\mathrm{UNL}(f_1,\ldots,f_K).$

\end{property}

\begin{property}[Monte Carlo error bound and consistency of the importance-sampling estimator]
\label{imp_unl_error}
Let \(f_{1},\ldots, f_{K}\) be densities with respect to \(\nu\). Based on the equal weighted mixture proposal, $q(x)=1/K\sum_{k=1}^K  f_{k}(x)$, let
\(X_1,\ldots,X_M\) be iid samples from \(q\); the importance sampling estimator of UNL is:
\[\widehat{\operatorname{UNL}}_M
=
\frac1M
\sum_{i=1}^M
\frac{
\max_{1\le k\le K} f_{k}( X_i)
}{
\ q( X_i)
}.\] 
\(\widehat{\operatorname{UNL}}_M\) is an unbiased estimator of
\(\operatorname{UNL}(f_{1},\cdots,f_{K})\). Moreover, for every \(\varepsilon>0\),
\[
\Pr\left(
\left|
\widehat{\operatorname{UNL}}_M
-
\operatorname{UNL}(f_{1},\cdots,f_{K})
\right|
\ge \varepsilon
\right)
\le
2\exp\left\{
-\frac{2M\varepsilon^2}{(K-1)^2}
\right\}.
\]
Equivalently, with probability at least \(1-\delta\),
\[
\left|
\widehat{\operatorname{UNL}}_M
-
\operatorname{UNL}(f_{1},\cdots,f_{K})
\right|
\le
(K-1)
\sqrt{\frac{\log(2/\delta)}{2M}}.
\]
Consequently, $\widehat{\operatorname{UNL}}_M$ is a consistent estimator of $\operatorname{UNL}(f_{1},\cdots,f_{K})$: $\widehat{\operatorname{UNL}}_M \xrightarrow{p} \operatorname{UNL}(f_{1},\cdots,f_{K})$ as $M\rightarrow \infty$.
\end{property}

\begin{property}[Monte Carlo variance bound for the importance sampling estimator]
\label{variance_bound_imp}
Let \(f_{1},\ldots, f_{K}\) be densities with respect to \(\nu\). Based on the equal weighted mixture proposal, $q(x)=1/K\sum_{k=1}^K  f_{k}(x)$, let
\(X_1,\ldots,X_M\) be iid samples from \(q\); the importance sampling estimator of UNL is:
\[\widehat{\operatorname{UNL}}_M
=
\frac1M
\sum_{i=1}^M
\frac{
\max_{1\le k\le K} f_{k}( X_i)
}{
\ q( X_i)
}.\]
It gives the variance bound:
\[
\operatorname{Var}\left(\widehat{\operatorname{UNL}}_M\right)
\le
\frac{\operatorname{UNL}(K-\operatorname{UNL})}{M}.
\]
\end{property}

\begin{property}[Combined density estimation and Monte Carlo error bound of UNL]
\label{combine_imp_error}
Let \(f_1,\ldots,f_K\) be densities with respect to \(\nu\). For sample size \(n\), let
\(\widehat f_{1,n},\ldots,\widehat f_{K,n}\) be estimated densities with respect to the same dominating measure \(\nu\). Based on the equal weighted mixture proposal, $q(x)=1/K\sum_{k=1}^K \widehat f_{k,n}(x)$, let
\(X_1,\ldots,X_M\) be iid samples from \(q\); the importance sampling estimator of UNL is:
\[\widehat {\widehat{\operatorname{UNL}}}_M
=
\frac1M
\sum_{i=1}^M
\frac{
\max_{1\le k\le K}\widehat f_{k,n}( X_i)
}{
\ q( X_i)
}.\]
Then, with probability at least \(1-\delta\),
\[
\left|
\widehat {\widehat{\operatorname{UNL}}}_M
-
\operatorname{UNL}(f_1,\ldots,f_K)
\right|
\le
(K-1)
\sqrt{\frac{\log(2/\delta)}{2M}}
+
\sum_{k=1}^K
\|\widehat f_{k,n}-f_k\|_1.
\] 
Consequently, given that \(\widehat f_{1,n},\ldots,\widehat f_{K,n}\) are consistent estimators of \( f_{1},\ldots, f_{K}\), $\widehat {\widehat{\operatorname{UNL}}}_M$ is a consistent estimator of $\operatorname{UNL}(f_1,\ldots,f_K)$: $\widehat {\widehat{\operatorname{UNL}}}_M \xrightarrow{p} \operatorname{UNL}(f_1,\ldots,f_K)$ as $M \rightarrow \infty$ and $n \rightarrow \infty$.
\end{property}



The Monte Carlo variance bound in Property \ref{variance_bound_imp} implies that when $\mathrm{UNL}=K$ (i.e., the $K$ groups are perfectly separated), $\mathrm{Var}\bigl(\widehat{\mathrm{UNL}}_M\bigr)=0$. The bound reaches its maximum when $\mathrm{UNL}=K/2$ (an  illustration of the variance bound for $K=5$ is provided in the Appendix), and when $\mathrm{UNL}=1$ (i.e., all groups are identical), it gives $\mathrm{Var}\bigl(\widehat{\mathrm{UNL}}_M\bigr)\le (K-1)/M$. 
However, when $\mathrm{UNL}=1$, the equal-weight mixture proposal $1/K\sum_{k=1}^{K} \widehat{f}_k(x)$ coincides with $\max_{1\le k\le K} \widehat{f}_k(x)$, so the importance sampling weights are constant and the estimator variance is in fact zero; the bound is therefore conservative in this regime. 

Nevertheless, either this variance bound or the bound of Property \ref{imp_unl_error} can be used to guide the choice of the Monte Carlo sample size $M$. We note that both bounds suggest $M = O(K^2)$. In the subsequent analyses, we focus on the former.  Specifically, if $\sigma_{\text{MC}}^2$ denotes a maximum acceptable Monte Carlo variance, then a conservative choice is $M = K^2/(4\sigma_{\text{MC}}^2)$
which ensures $\mathrm{Var}\bigl(\widehat{\mathrm{UNL}}_M\bigr)\le \sigma_{\text{MC}}^2$ under the bound of Property \ref{variance_bound_imp}. 

\section{UNL as a tool for discovering covariate dependence in cluster analysis}\label{sec:clustering}

Many scientific problems involve assessing how group structure relates to external variables. When group labels are observed, the UNL provides a direct way to quantify the dependence of these labels on other variables by measuring the degree of separation across groups. In many applications, however, the group structure is unobserved, and a common strategy is to discover latent group structure via clustering. 


A wide variety of clustering methods has been proposed \citep[see e.g.,][for a review]{SAXENA2017664}. In this paper, we focus on model-based clustering, with particular emphasis on Bayesian mixture models, although the proposed approach applies more broadly. We show how the UNL can be used to assess how an inferred partition relates to covariates. Depending on whether clustering is performed using only the response (the marginal approach) or via a mixture model that explicitly incorporates covariates (the conditional approach), the inferred partition may comprise clusters of individuals with similar responses or clusters of individuals with similar response–covariate relationships.


\subsection{Clustering based only on the response}
\subsubsection{The marginal approach of mixture models}
In mixture models, the observed data are modeled as rising from a mixture of simpler component distributions, with each component corresponding to a cluster \citep{Fraley2002}. If covariate information is not considered, mixture models assume the data $y_1,\dots,y_n$ (here we use $y$ to denote the data vector being clustered, to distinguish from covariates $x$ introduced later) are conditionally i.i.d. from a convex combination of parametric components \citep{wade2023bayesian}:
\begin{align}
	y_i  \iidsim \sum_{l=1}^L w_l K(\cdot \mid \theta_l) = \int K(\cdot \mid \theta) dH(\theta),
	\label{eq:mix}
\end{align}
where $K(y \mid \theta)$ is a fixed parametric density, often referred to as the kernel, with component-specific parameters contained in $\btheta = (\theta_1, \ldots, \theta_L)$. In the equivalent integral representation on the right-hand side of \eqref{eq:mix}, $H = \sum_{l=1}^L w_l \delta_{\theta_l} $ represents the mixing measure with the mixture weights $\bw = (w_1, \ldots, w_L)$ that are non-negative and sum to one. Yet another equivalent representation, useful for clustering, makes use of allocation variables $\bz = (z_1, \ldots , z_n)$:
\[
 	y_i |z_i = l, \theta_l \iidsim K(\cdot \mid \theta_l), \quad  z_i \iidsim \Cat(w_1,\ldots, w_L),
\]
where $\Cat(\cdot)$ represents the categorical distribution with parameter $\bw$.


In the Bayesian setting, the model in \eqref{eq:mix} is completed with a prior on the unknown parameters $\bw$ and $\btheta$ (or equivalently on the unknown mixing measure $H$). Additionally, Bayesian nonparametric models allow the number of clusters to be not fixed in advance, by assuming an infinite number of components and can be viewed as a limiting case of overfitted mixtures, with the Dirichlet process mixture (DPM) being the most widely used example \citep{wade2023bayesian}. The DPM model uses a Dirichlet process (DP) prior \citep{Ferguson1973} on the mixing measure ($H\sim\text{DP}(\alpha ,H_0)$). According to Sethuraman's stick-breaking representation \citep{Sethuraman1994}, $H$ can be expressed as:
\begin{align}
H = \sum_{l=1}^\infty w_l \delta_{\theta_l} 
\label{dp_prior_definition}
\end{align}
where $w_{1}=v_{1}$, $w_{l}=v_{l}\prod_{m=1}^{l-1} (1-v_{m})$, for $l\geq 2$, with $v_{l} \overset{\text{iid}}\sim \text{Beta}(1,\alpha)$, for $l\geq 1$.

Under \eqref{dp_prior_definition}, the probabilities assigned to each component decrease rapidly with the index $l$ for typical choices of $\alpha$. Consequently, the infinite mixture model can be reasonably approximated by a finite number of components. For this reason, posterior inference can be conducted using the blocked Gibbs sampler \citep{Ishwaran2001}, which truncates Sethuraman's representation in \eqref{dp_prior_definition} to a finite value, $L$, with $v_{L} = 1$ to ensure that the mixture weights sum to one. 
For example, the model structure and prior specification of a DPM mixture model with kernel as the product of multivariate normal and categorical distributions which is implemented in this paper are provided in the Appendices.



\subsubsection{Detecting dependence of the partition on covariates in the marginal approach}
In the marginal approach, clustering is performed solely on the response, thereby assigning observations with similar outcomes to the same cluster. To assess whether, and to what extent, covariates influence the clustering partition structure, we propose to use the UNL of covariates, as defined in Definitions \ref{def:UNL-continuous}–\ref{def:UNL-unified}, to quantify partition–covariate dependence by measuring the separation of covariate distributions across clusters. Low underlap values indicate substantial overlap in covariate distributions across clusters, suggesting weak dependence of the inferred partition on the covariates, whereas high underlap values reflect more distinct covariate distributions and hence stronger dependence of the partition structure on covariates. While it is difficult to visually assess covariate differences across clusters when the covariates are multivariate, the UNL provides a concise quantitative summary of how strongly a response-based partition depends on the covariates.

In a Bayesian setting, estimating the UNL for every posterior draw of the partition requires (i) estimating the density of covariates for each cluster in each sampled partition and (ii) calculating the UNL for every posterior draw of the density estimates for every partition, which is computationally expensive. We instead propose to estimate UNL for a single representative partition, specifically, the point estimate that minimizes the Jensen’s-inequality lower bound to the posterior expected variation of information \citep[VI][]{10.1214/17-BA1073}. This choice has the desirable properties of VI as a metric on partitions with an information-theoretic interpretation \citep{meilua2007comparing}, while remaining computationally tractable because the bound depends only on the posterior similarity matrix. Other choices of the representative partition, based on alternative loss functions, could also be considered, but we adopt the VI-based estimator for its benefits highlighted in \cite{10.1214/17-BA1073}.
To account for uncertainty while also retaining computational efficiency, a simple approach is to summarize the posterior with multiple representative partitions produced by the WASABI method \citep[][see the Appendices for details]{balocchi2025understanding}.

We illustrate the utility of the UNL in quantifying partition-covariate dependence through two simulated examples, which represent distinct data-generating mechanisms with different relationships between the response \(y\) and covariates \(x\). In both, the DPM is used for clustering and for estimating the covariate densities within each group, given the representative partition. The DPM  is fit using the specifications detailed in the Appendices. Posterior inference is based on 10,000 MCMC iterations after discarding an initial burn-in of 10,000 iterations. Given the density estimates, the UNL is computed via importance sampling (Algorithm \ref{alg:importance_sampling}) using a Monte Carlo sample size of $M=5000$, which yields variance upper bounds of $9/20000$ and $4/20000$ for the two examples, respectively.

\noindent\textbf{Example A (cluster assignments fully determined by a single covariate).}
We generate \(n=600\) observations according to $y_i=2\mathbf{1}(x_i \le -1)-5\mathbf{1}(x_i \ge 1)+ \epsilon_i$, $\epsilon_i \sim \mathrm{N}(0,0.1^2)$,
with \(x_i \sim \mathrm{Unif}(-3,3)\). 

\noindent\textbf{Example B (cluster assignments depend jointly on two covariates but not on either marginally).}
We generate \(n=600\) observations from
$y_i = m(x_i) + \epsilon_i$, $\epsilon_i \sim \mathrm{N}(0,0.1^2)$,
where \(x_i = (x_{i1},x_{i2})^T\), \(x_{ij} \sim \mathrm{Unif}(-2,2)\), and $m(x_i) =\mathbf{1}(\sin(x_{i1} x_{i2} \pi / 2) \le 0)-\mathbf{1}(\sin(x_{i1} x_{i2} \pi / 2) > 0)$. 


The representative partition structures obtained from the mixture model, together with the histograms of the estimated UNLs for Examples A and B, are shown in Figure \ref{example_AB_UNL_plots}. In Example A, the covariate \(x\) is almost perfectly separated across clusters, and consequently the UNL is very close to the number of inferred clusters (\(K=3\)). Example B, which is also considered in \cite{wade2025bayesian}, exhibits a different pattern: the estimated UNL for the joint distribution of \((x_1,x_2)\) is again close to the number of inferred clusters (\(K=2\)), indicating pronounced dependence of the inferred partition on \((x_1,x_2)\) jointly. By contrast, the UNL values for the marginal distributions of \(x_1\) and \(x_2\) are each near one, indicating little partition–covariate dependence when each covariate is considered marginally. This pattern illustrates a significant joint covariate effect without any significant marginal effect: neither covariate alone explains the partition, but their joint configuration does. It is also worth noting that, once the joint density has been estimated using a DPM of multivariate normals, the marginal densities are available analytically by restricting the fitted component parameters to the relevant coordinates. Consequently, the UNL of the marginals can be computed without refitting the DPM model separately for each covariate.

It is important to note that the UNL quantifies the dependence of a given partition on the covariates, rather than the dependence of the response on the covariates. Different partition structures will, in general, induce different partition–covariate relationships.



\begin{figure}[!t]
\centering
\subfigure{
\centering
\includegraphics[width=0.38\textwidth]{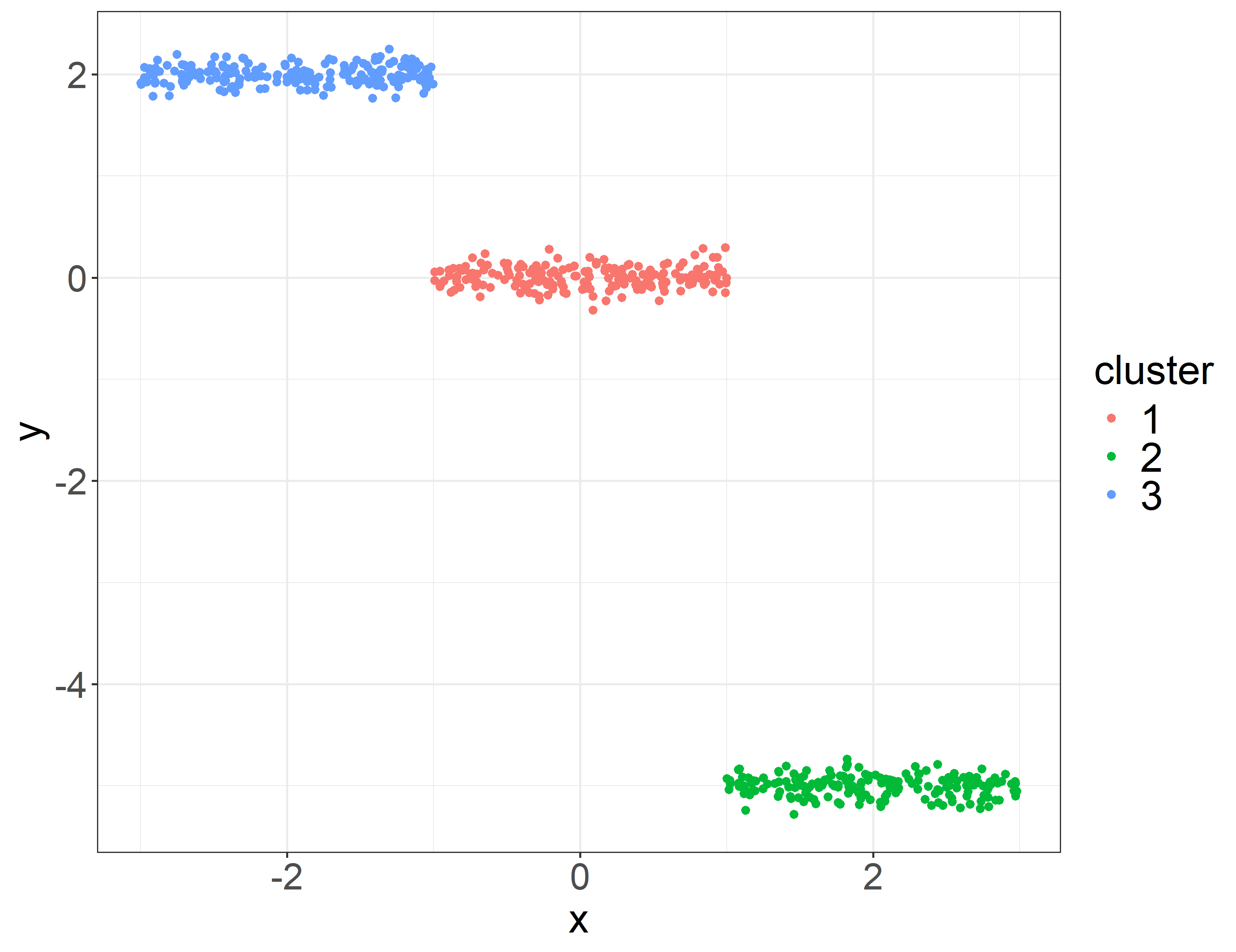}
}
\subfigure{
\centering
\includegraphics[width=0.38\textwidth]{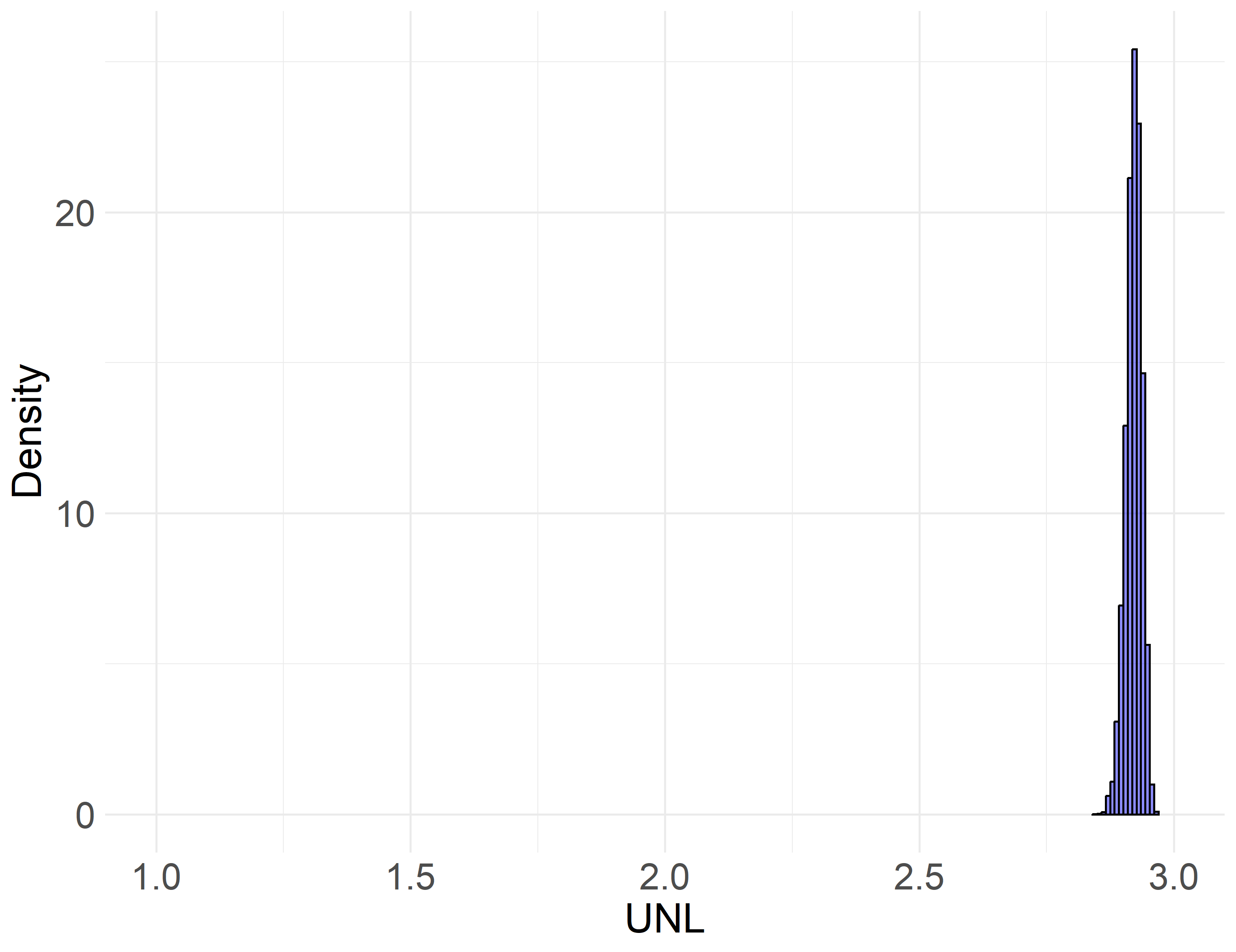}
}
\\
\centering
\subfigure{
\centering
\includegraphics[width=0.38\textwidth]{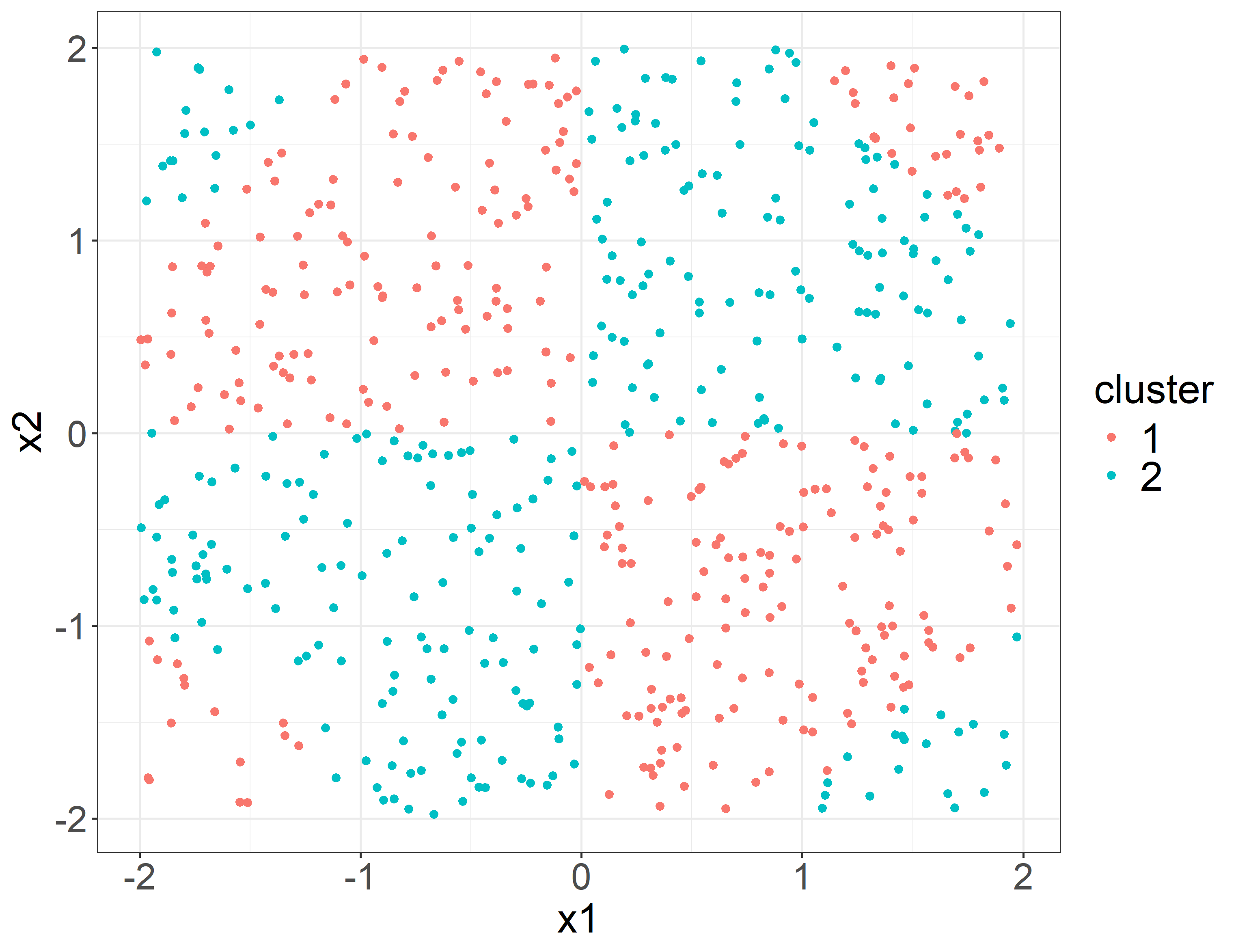}
}
\subfigure{
\centering
\includegraphics[width=0.38\textwidth]{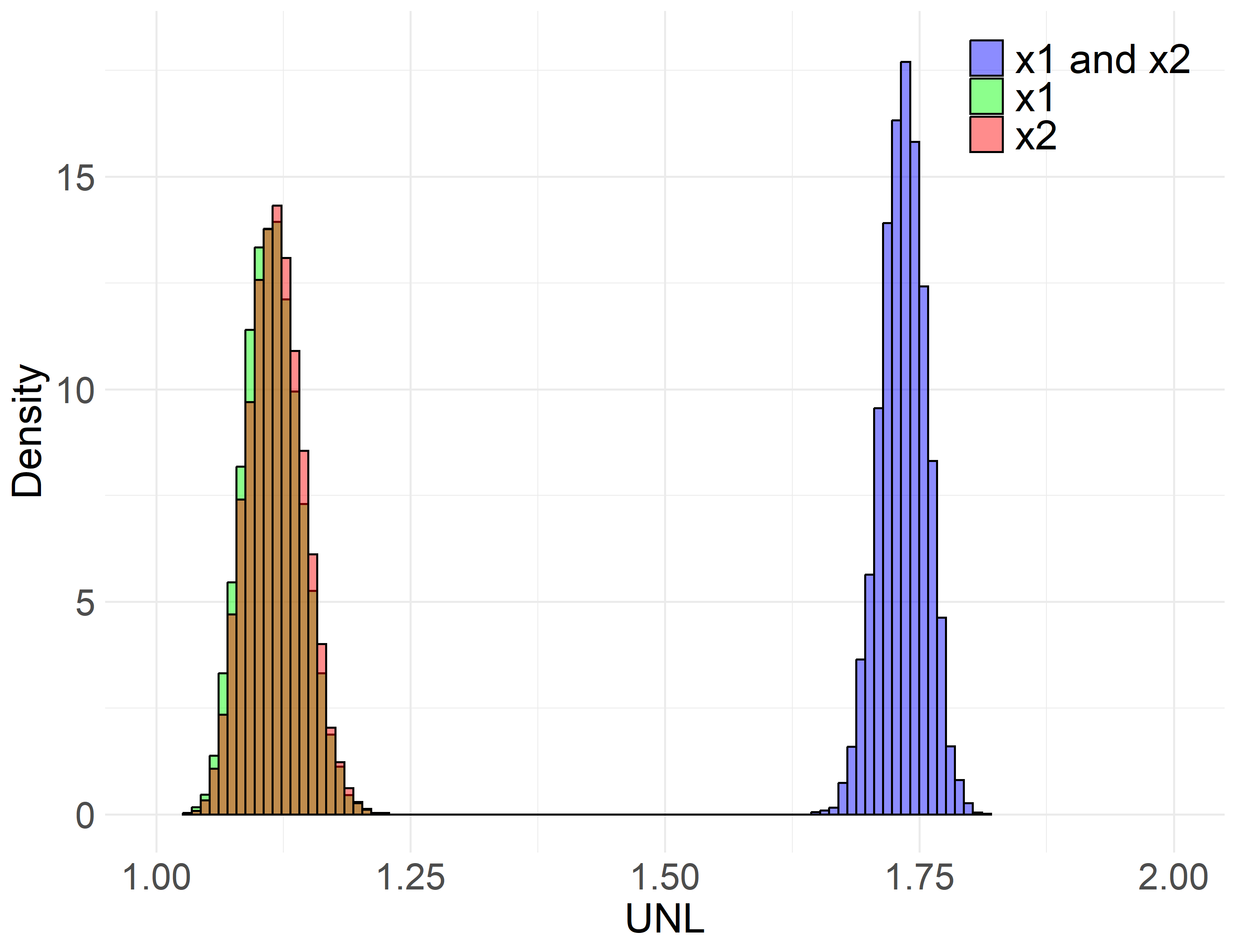}
}

\centering
\caption{Top row: Example A. Bottom row: Example B. Left: the representative partition inferred from the DPM. Right: histograms of the estimated UNL of the covariates.}
\label{example_AB_UNL_plots}
\end{figure}

\subsection{Incorporating covariate information in clustering}
\subsubsection{The conditional approach of mixture models}

Unlike in the marginal approach, the conditional approach of the mixture model explicitly models the covariate-specific density \(f(y_i | x_i)\), allowing the distribution of \(y_i\) to depend on \(x_i\) through covariate-dependent mixture weights, covariate-dependent component parameters, or both.

The general form of the conditional approach assumes that the data $y_i$ arise from a combination of covariate-specific parametric components for $i \in \{1,\cdots,n\}$:
\begin{align}
	y_i | x_i \sim \sum_{l=1}^L w_l(x_i) K(\cdot \mid \theta_l(x_i)) = \int K(\cdot | \theta) dH_{x_i}(\theta),
	\label{eq:mix_2}
\end{align}
where $K(y | \theta(x))$ is a fixed parametric density, with component-specific parameters contained in $\btheta(x) = (\theta_1(x), \ldots, \theta_L(x))$ and the mixture weights $\bw(x) = (w_1(x), \ldots, w_L(x))$ are non-negative and sum to one for every possible $x$ value. 
 The equivalent representation which makes use of latent group variables $\bz = (z_1, \ldots , z_n)$ is:
 \begin{align*}
 	y_i |x_i,z_i=l \indsim f(\cdot \mid \theta_l(x_i)), \quad  z_i \indsim \Cat(w_1(x_i),\ldots, w_L(x_i).
 \end{align*}


Similar to the marginal approach, priors on $\btheta(x)$ and $\bw(x)$ (or equivalently on $H_x$) are required for Bayesian inference of the modelling. The dependent DP (DDP) \citep{MacEachern1999} is a popular nonparametric prior for the random probability mixing measures $H_x$, which modifies the stick-breaking representation of the DP to accommodate
covariates and in full generality, and is specified as:
\begin{align}
H_x = \sum_{l=1}^\infty w_l(x), \delta_{\theta_l(x)} 
\label{ddp_prior_definition}
\end{align}
where $w_{1}(x)=v_{1}(x)$, $w_{l}(x)=v_{l}(x)\prod_{m=1}^{l-1} (1-v_{m}(x))$, for $l\geq 2$. $v_{l}(x)$ are independent stochastic processes with marginals following $\text{Beta}(1,\alpha(x))$, and $\theta_l(x)$ are also independent stochastic processes, for $l\geq 1$. Similarly to the DPM model, the number of components in \eqref{ddp_prior_definition} could be truncated to a finite value, $L$, with $v_{L}(x) = 1$, to facilitate posterior inference.

A frequently used class of DDPs is the single-weights DDP, motivated primarily by computational convenience: it reuses existing posterior sampling algorithms for DPMs by restricting the mixing weights to be covariate-invariant, i.e., $w_l(x)\equiv w_l$. For example, a popular single-weights DDP model which is often referred to as the linear dependent Dirichlet process (LDDP) mixture model \citep[see e.g.,][]{Quintana2022,wade2025bayesian} for modelling a continuous response takes the form:
\begin{align}
	f (y_i | x_i) = \sum_{l=1}^\infty w_l \phi(y_i \mid  \beta_l^T x_i,\sigma_l^2).
	\label{eq:single_weight_ddp}
\end{align}
where $\phi(\cdot)$ represents the normal density function, and $\beta_l^T x_i$ and $\sigma_l^2$ denote the mean function and the variance of component $l$ respectively. Although the model may appear highly flexible and is computationally attractive, the induced predictive mean and conditional density are in fact quite constrained. The single-weight LDDP mixture of normals yields a weighted combination of linear regression functions but lacks the local adaptation afforded by covariate-dependent weights. As a result, it may be insufficiently flexible for capturing complex regression relationships \citep{wade2025bayesian}.

The specific model structure and prior specification for the LDDP mixture model implemented in this paper are detailed in the Appendices.

\subsubsection{Detecting dependence of the partition on covariates in the conditional approach}
In the conditional approach, the mixture model groups individuals with similar response-covariate relationships. Similar to the unconditional approach, we propose to use the UNL of covariates 
to assess the influence of covariates on the inferred clustering partition.


This is particularly important for the single-weights mixture model, which posits that the inferred partition has no dependence on covariates, assuming that the UNL of the covariates should be one. However, when the regression kernel is not sufficiently flexible, the inferred partition structure may depend on the covariates. In such cases, the cluster-specific predictions are averaged regardless of covariate values, because the mixing weights do not account for the similarity between the new covariates \(x_{n+1}\) and the covariates within each cluster \citep{wade2025bayesian}. As a consequence, predictive inference for the regression function, density, and associated uncertainty under the single-weights mixture can be extremely poor.

We consider three examples to illustrate the utility of the UNL in quantifying the partition-covariate dependence in the conditional approach of mixture model. Examples C1 and C2 involve low-dimensional covariates, where dependence of the LDDP-induced partition on covariates can be visualized directly. Example D, in contrast, involves a higher-dimensional covariate space, where the dependence pattern of the partition on all covariates is difficult to observe visually, and here the UNL provides a systematic way to screen for and summarize partition–covariate dependence across many covariates.



\noindent\textbf{Example C1 (low-dimensional covariate space with distinct slopes of \(x^c\) across categories of \(x^d\)).}
We generate \(n=800\) observations according to
\[
y_i \mid x_i^c, x_{i}^d \sim
\begin{cases}
0.5\mathrm{N}(-2x_i^c, 0.4^2) + 0.5\mathrm{N}(2x_i^c, 0.4^2), & \text{if } x_{i}^d = 1,\\[4pt]
0.5\mathrm{N}(-12x_i^c + 80, 0.4^2) + 0.5\mathrm{N}(12x_i^c + 80, 0.4^2), & \text{if } x_{i}^d = 2,
\end{cases}
\]
where \(x_i^c \sim \mathrm{Unif}(-3,3)\) and \(\Pr(x_{i}^d=1)=\Pr(x_{i}^d=2)=0.5\).

\noindent\textbf{Example C2 (low-dimensional covariate space with common slopes of \(x^c\) across categories of \(x^d\)).}
We again generate \(n=800\) observations, now from
\[
y_i \mid x_i^c, x_{i}^d \sim
\begin{cases}
0.5\mathrm{N}(-12x_i^c, 0.4^2) + 0.5\mathrm{N}(12x_i^c, 0.4^2), & \text{if } x_{i}^d = 1,\\[4pt]
0.5\mathrm{N}(-12x_i^c + 80, 0.4^2) + 0.5\mathrm{N}(12x_i^c + 80, 0.4^2), & \text{if } x_{i}^d = 2,
\end{cases}
\]
with \(x_i^c \sim \mathrm{Unif}(-3,3)\) and \(\Pr(x_{i}^d=1)=\Pr(x_{i}^d=2)=0.5\).

The only difference between Examples C1 and C2 in the data generating mechanism is the slope of the \(x^c\) when \(x^d=1\), all other aspects are identical. This seemingly minor change leads to markedly different behavior under the LDDP mixture model: the representative partition inferred by LDDP contains four clusters in C1 but only two in C2 (see the left panel of Figure \ref{example_C12_UNL_plots}). In Example C1, clusters 3 and 4 consist exclusively of units with \(x^d=1\), whereas clusters 1 and 2 consist exclusively of units with \(x^d=2\). In Example C2, by contrast, \(x^d\) appears more evenly represented across the two clusters. This visible partition–covariate dependence in Example C1, and its absence in Example C2, is confirmed by the UNL values shown in the right panels of Figures \ref{example_C12_UNL_plots}.

In Example C1, the estimated UNL values for both the marginal distribution of \(x^d\) and the joint distribution \((x^c,x^d)\) are large, indicating substantial dependence of the inferred partition on \(x^d\) and thereby violating the single-weight assumption of the LDDP mixture model. By contrast, in Example C2 the UNL values 
are all close to one, indicating little partition–covariate dependence.

The difference in UNL values between the two examples has clear predictive implications. In the LDDP mixture model, predictions are averaged with mixing weights that do not vary with the covariates. In Example C1, the LDDP-induced partition exhibits moderate dependence on $x^d$; with covariate-constant weights, this leads to poor predictive performance. In Example C2, the partition shows little dependence on either covariate, thus covariate-constant weights do not degrade prediction. Heatmaps of the true and estimated density regression functions for Examples C1 and C2 conditional on $x^d=1$ are shown in Figure \ref{example_C12_heatmap_plots_cate1}, the corresponding plots for $x^d=2$ are provided in the Appendices.

\begin{figure}[!t]
\centering
\subfigure{
\centering
\includegraphics[width=0.38\textwidth]{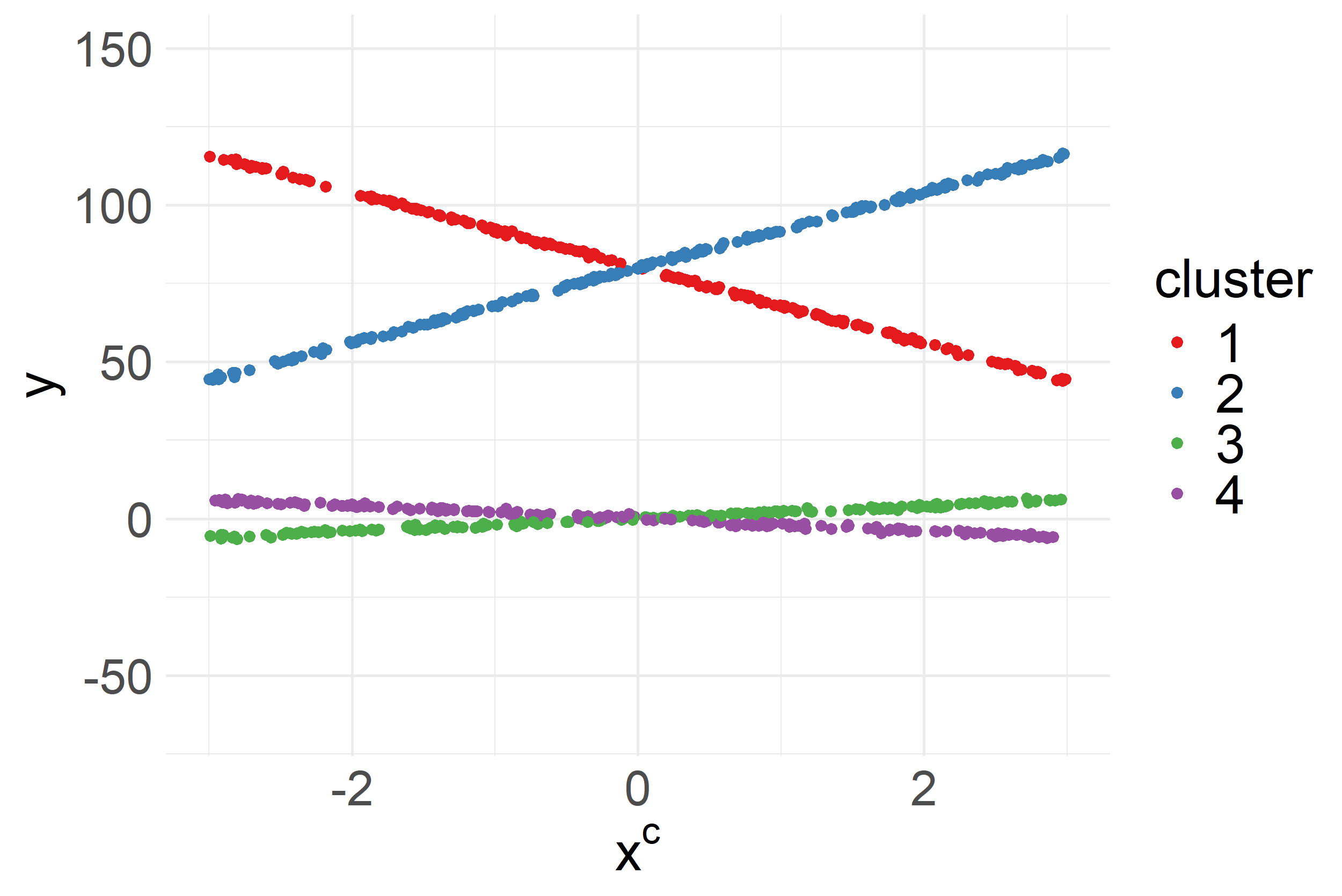}
}
\subfigure{
\centering
\includegraphics[width=0.38\textwidth]{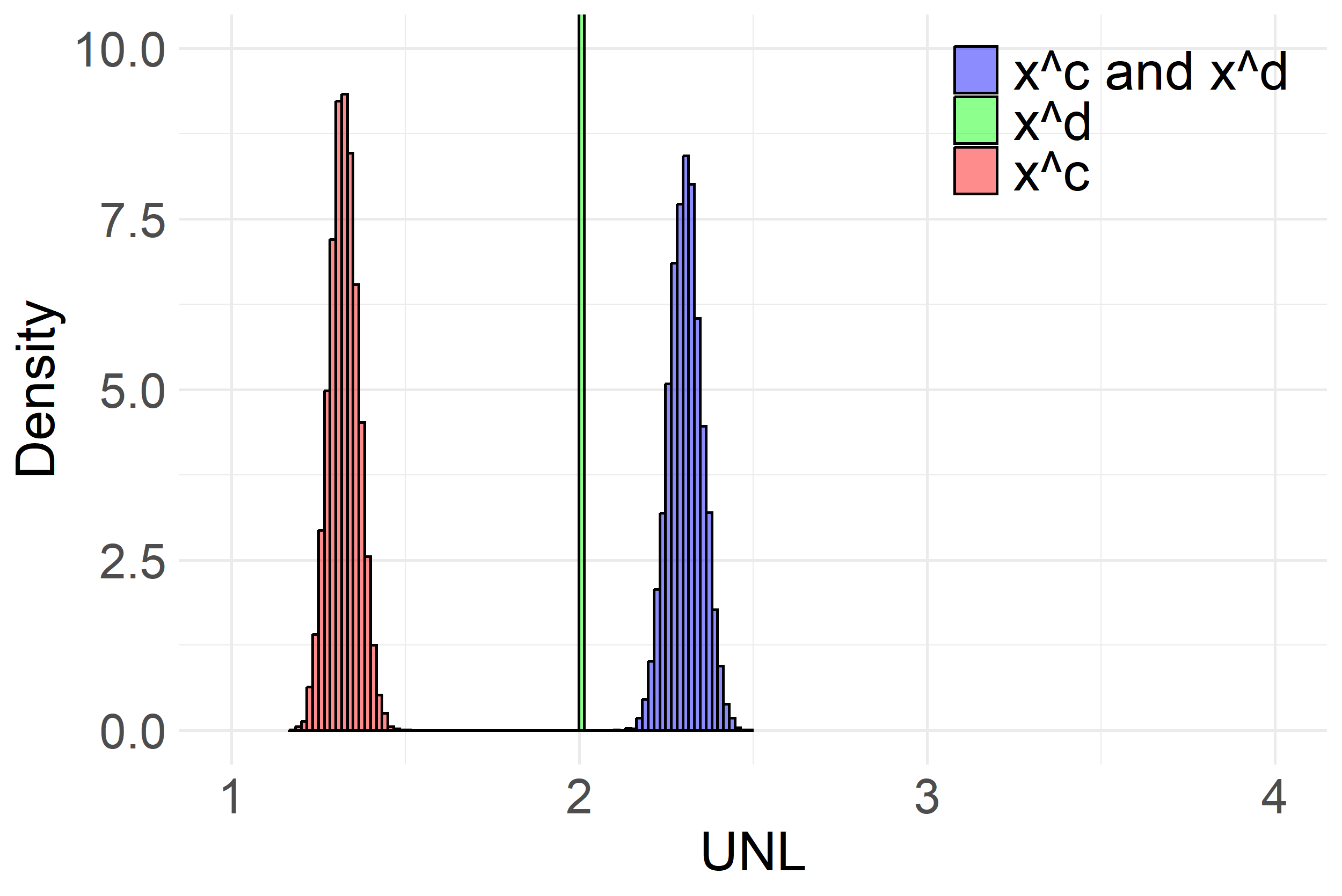}
}
\\
\subfigure{
\centering
\includegraphics[width=0.38\textwidth]{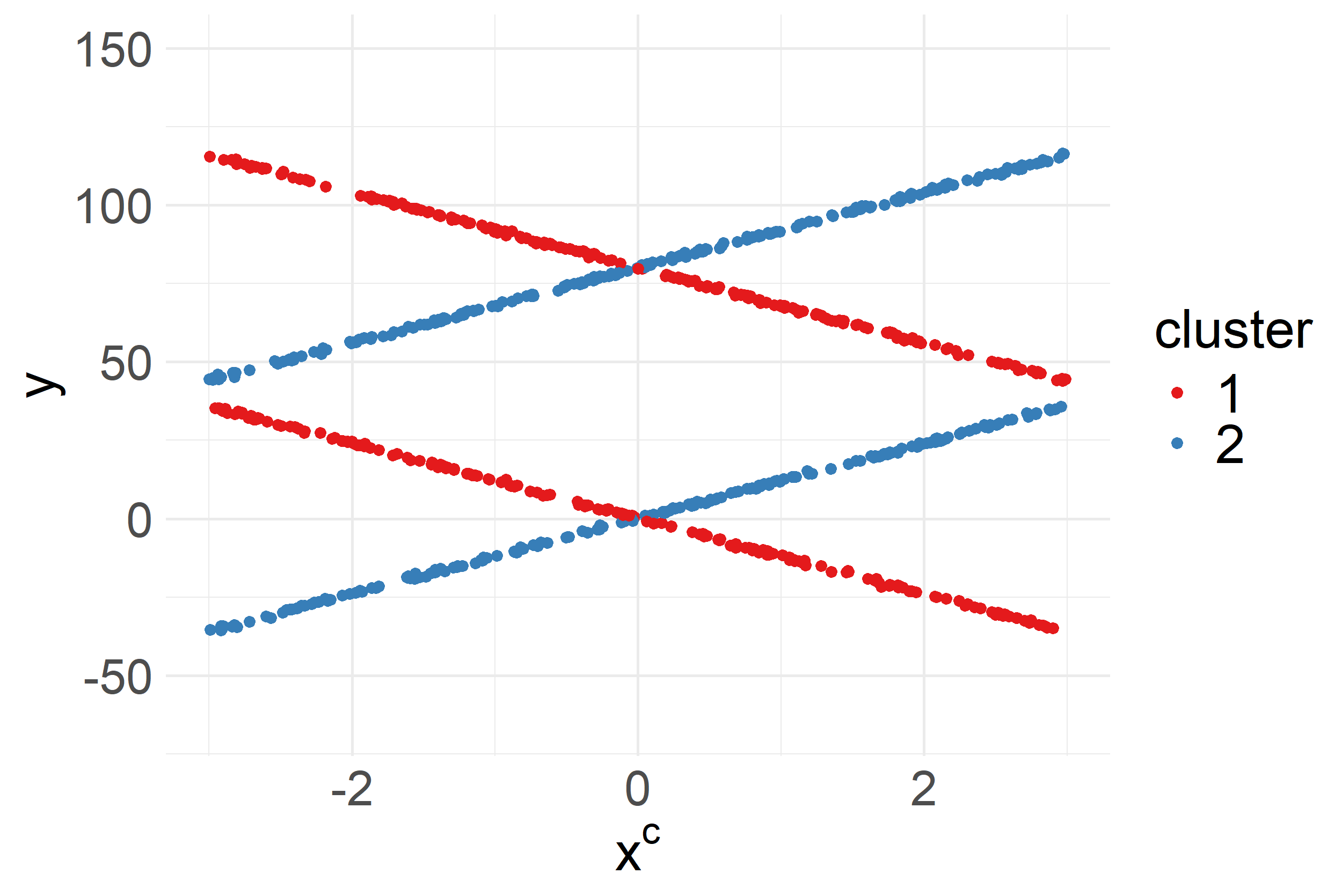}
}
\subfigure{
\centering
\includegraphics[width=0.38\textwidth]{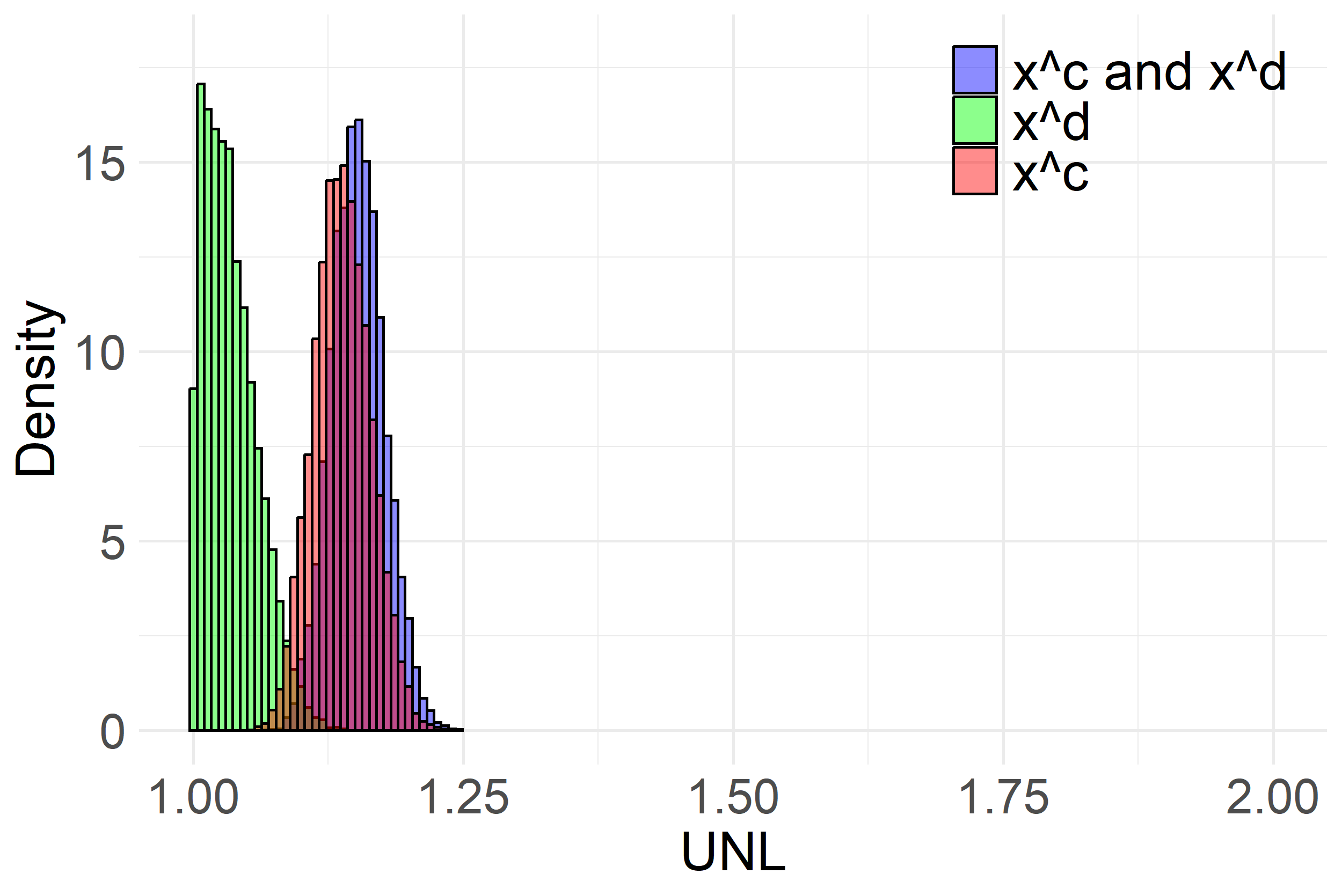}
}
\centering
\caption{Top row: Example C1. Bottom row: Example C2. Left: the representative partition inferred from the LDDP model. Right: the histograms of the estimated UNL of the covariates.}
\label{example_C12_UNL_plots}
\end{figure}

\begin{figure}[!t]
\centering
\subfigure{
\centering
\includegraphics[width=0.38\textwidth]{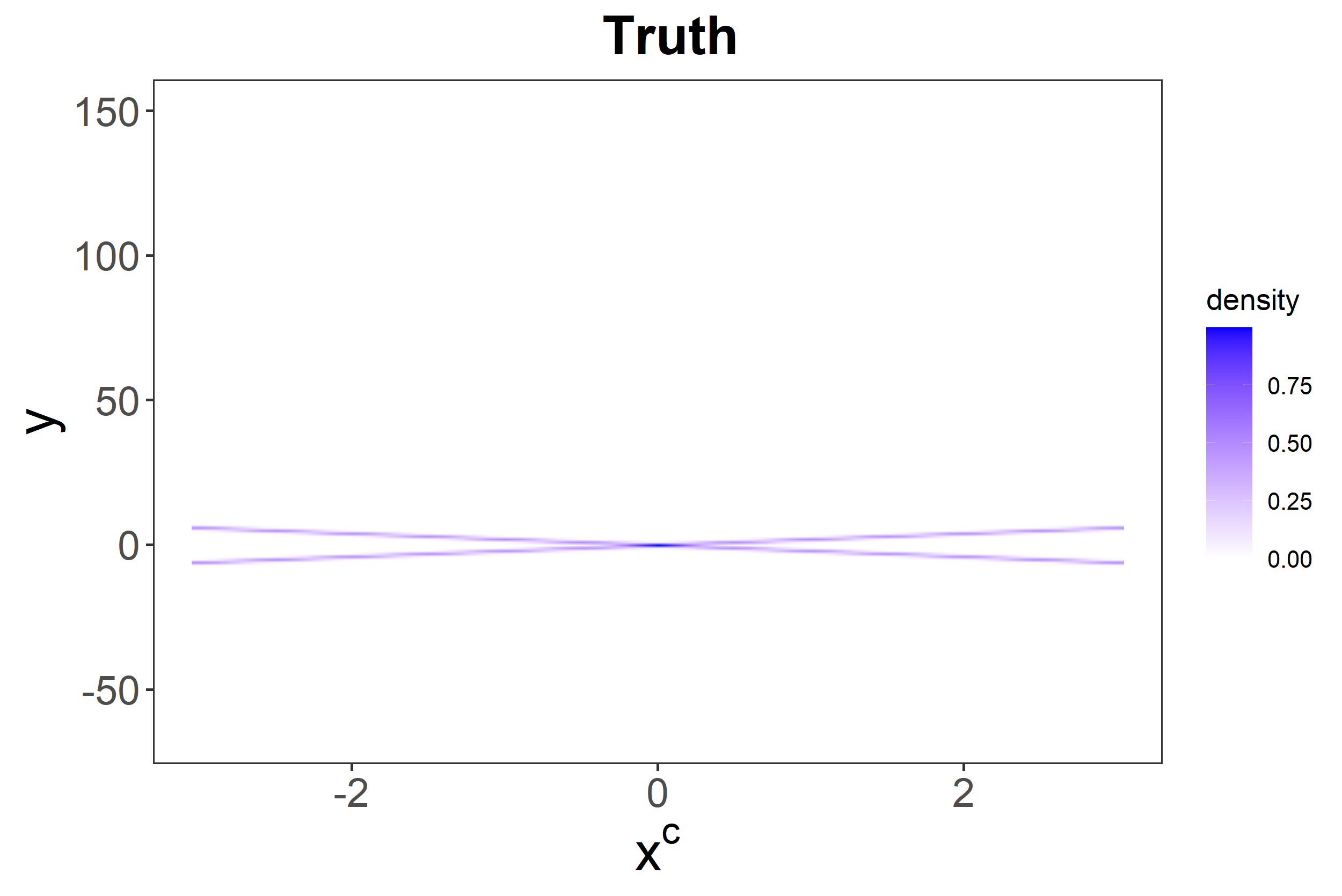}
}
\subfigure{
\centering
\includegraphics[width=0.38\textwidth]{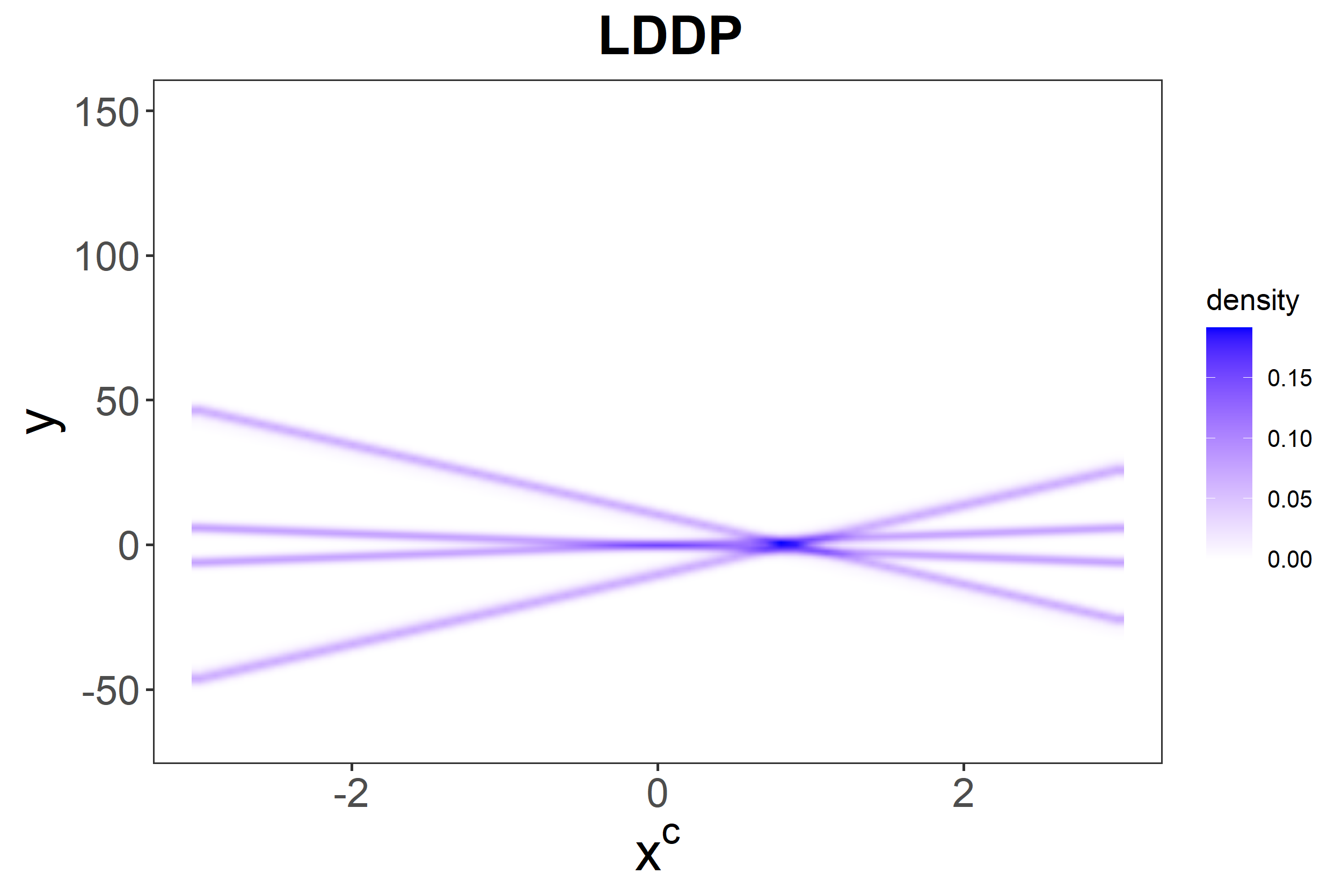}
}
\\
\subfigure{
\centering
\includegraphics[width=0.38\textwidth]{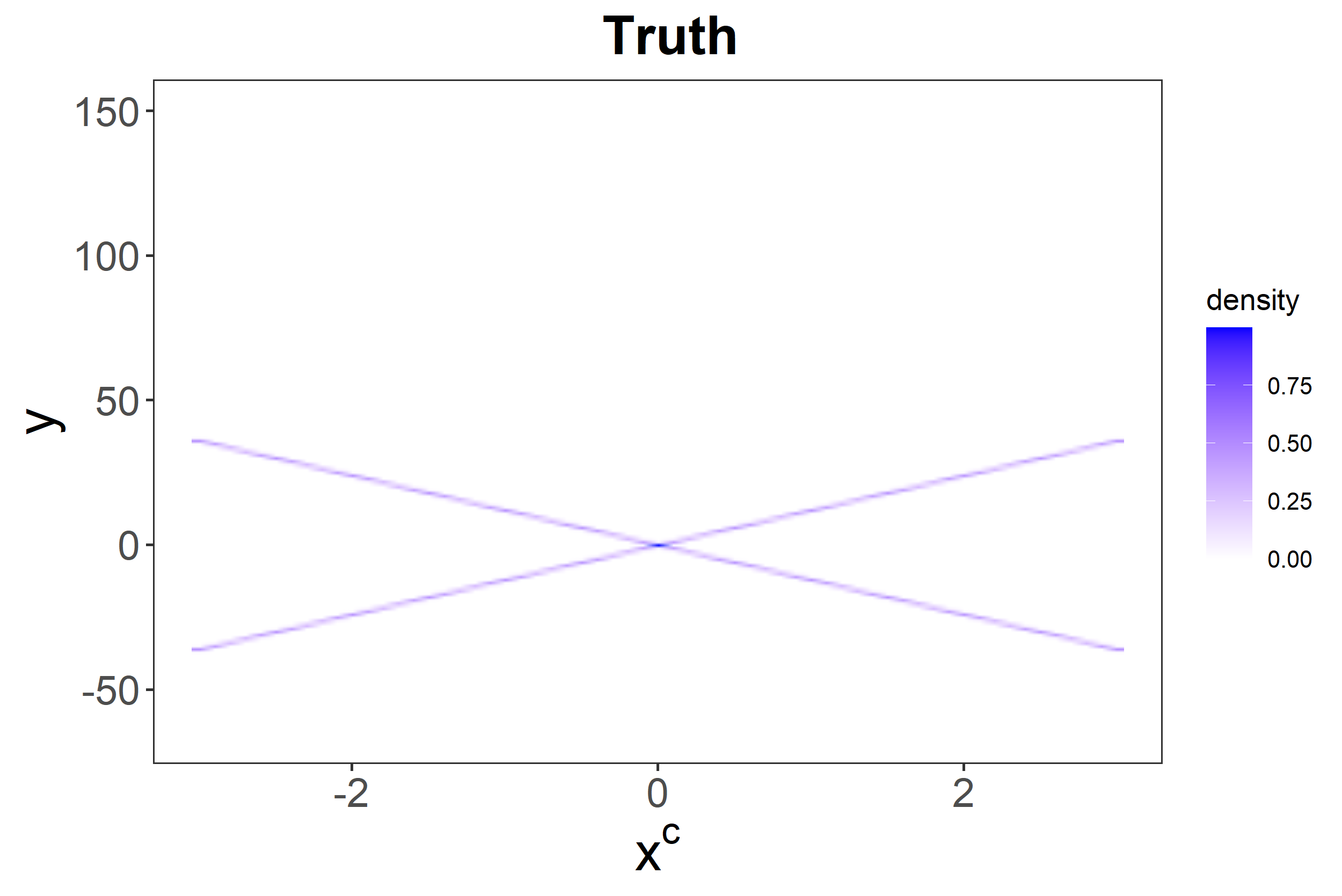}
}
\subfigure{
\centering
\includegraphics[width=0.38\textwidth]{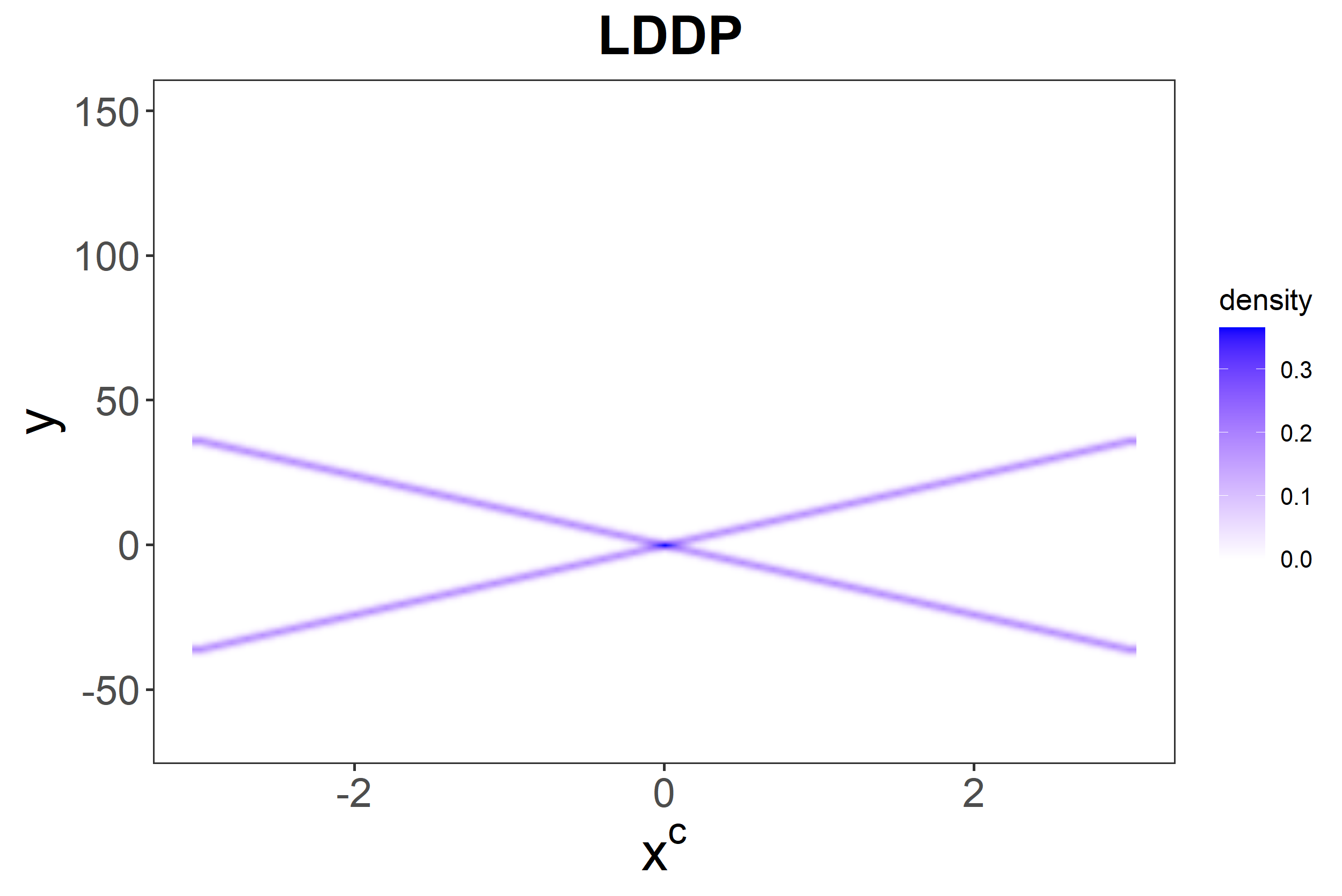}
}
\centering
\caption{Heatmap of the true and estimated density regression functions of Examples C1 and C2 conditioned on $x^d=1$. Top row: Example C1. Bottom row: Example C2.}
\label{example_C12_heatmap_plots_cate1}
\end{figure}

\noindent\textbf{Example D (high-dimensional covariate space).} We generate 1000 observations following:
$
y_i|x_i\overset{ind}{\sim}p(x_{i1})N(y_i|\beta_{1,0}+\beta_{1,1}x_{i1},\sigma_1^2)+(1-p(x_{i1}))N(y_i|\beta_{2,0}+\beta_{2,1}x_{i1},\sigma_2^2)
$,
where $p(x_{i,1})={\tau_1 exp(-\frac{\tau_1^2}{2}(x_{i,1}-\mu_1)^2 )}/{[\tau_1 exp(-{\tau_1^2}/{2}(x_{i,1}-\mu_1)^2 )+\tau_2 exp(-{\tau_2^2}/{2}(x_{i,1}-\mu_2)^2 )]},$ with $\beta_1=(0,1)^T$, $\sigma_1^2={1}/{16}$, $\beta_2=(4.5,0.1)^T$, $\sigma_1^2={1}/{8}$, $\mu_1=4$, $\mu_2=6$,$\tau_1=\tau_2=2$. Let $p=20$ and ${x}_i=(x_{i1},\dots,x_{ip})^\top \overset{\text{i.i.d.}}{\sim} \text{N}_p({\mu},{\Sigma})$, 
with mean ${\mu}=(4,\ldots,4)^{T}$ and covariance matrix ${\Sigma}$ specified as following. Let $\Sigma_{hh}=4$ for all $h$ and partition the covariate indices into the odd and even sets: 
$
\mathcal{O}=\{1,3,5,\dots,19\}$,
 $\mathcal{E}=\{2,4,6,\dots,20\}
$.
Within-group covariates are correlated with correlation $0.75$ while across-group covariates are independent.

\begin{figure}[!t]
\centering
\subfigure{
\centering
\includegraphics[width=0.38\textwidth]{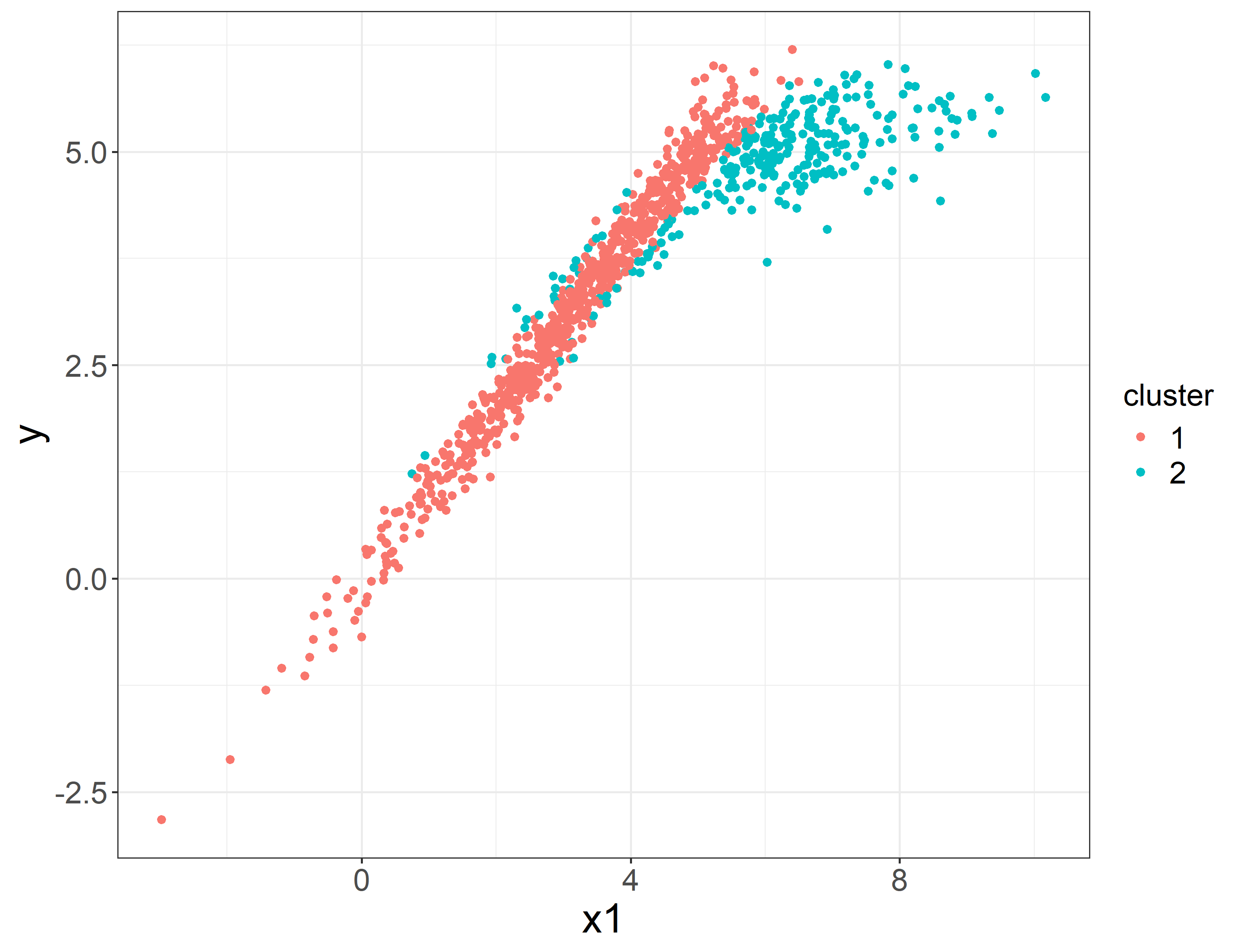}
}
\subfigure{
\centering
\includegraphics[width=0.38\textwidth]{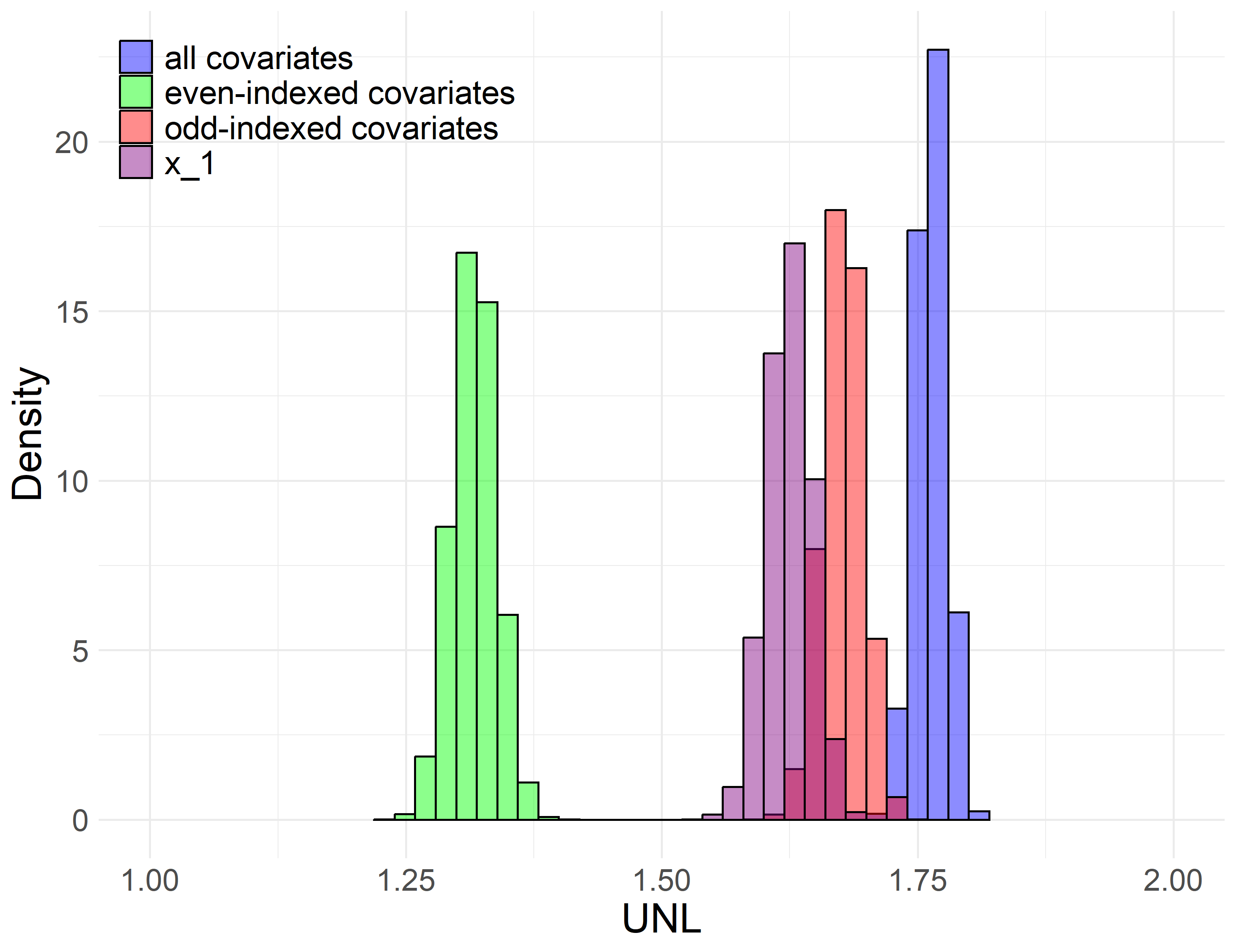}
}

\centering
\caption{Left: the representative partition of Example D inferred from the LDDP model. Right: the histograms of the estimated UNL of the covariates of Example D.}
\label{example_D_UNL_plots}
\end{figure}

\begin{figure}[!t]
\centering
\subfigure{
\centering
\includegraphics[width=0.22\textwidth]{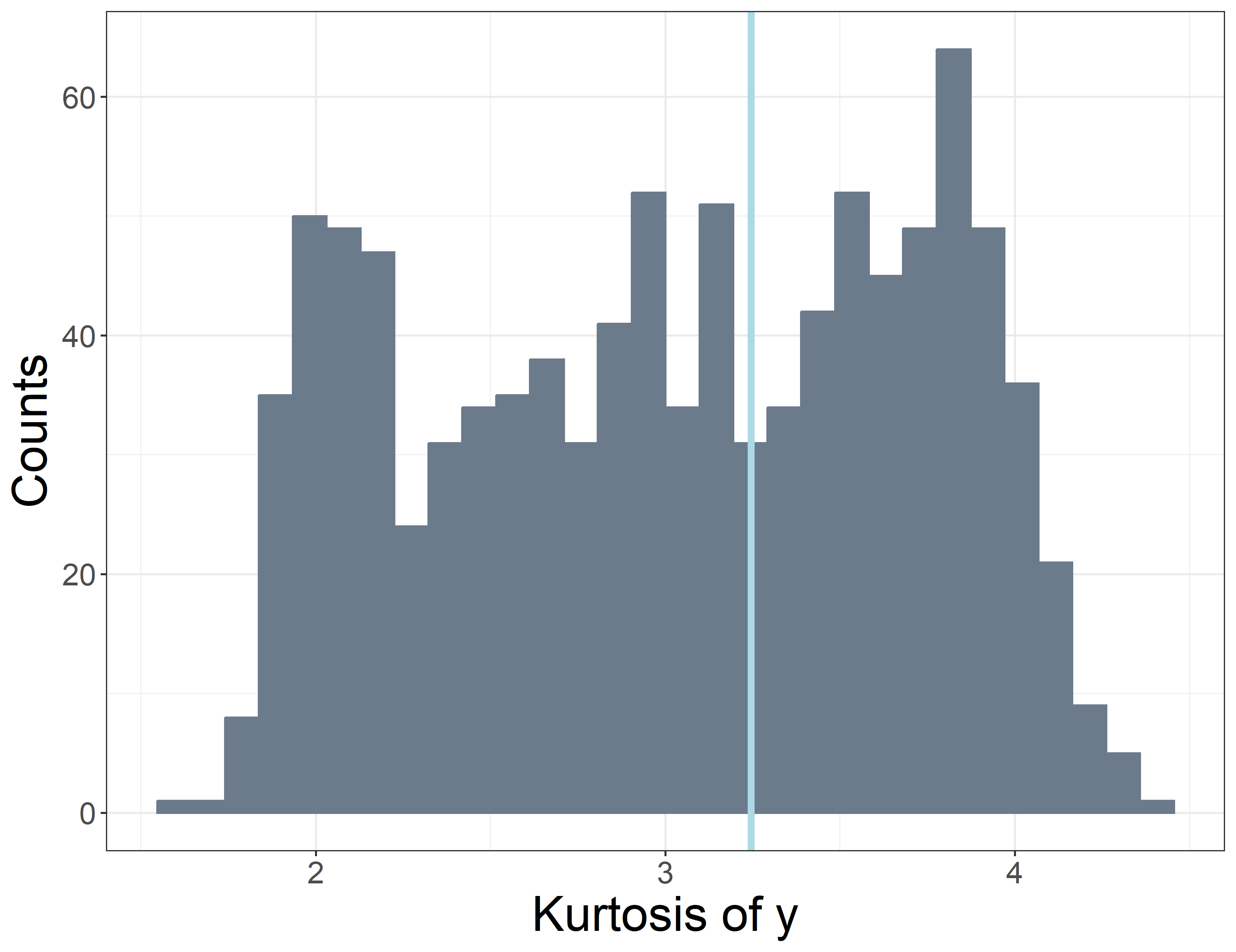}
}
\subfigure{
\centering
\includegraphics[width=0.22\textwidth]{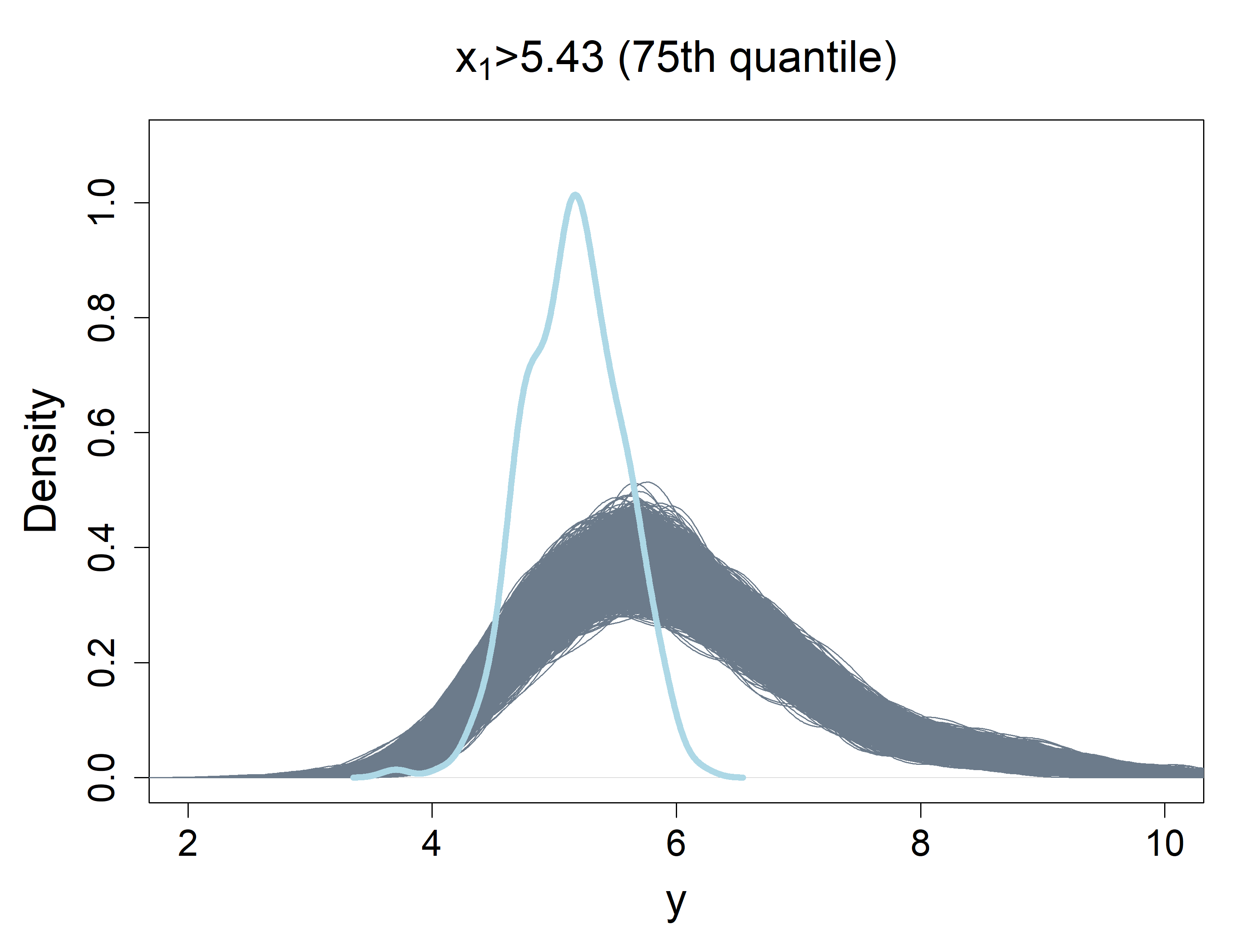}
}
\subfigure{
\centering
\includegraphics[width=0.22\textwidth]{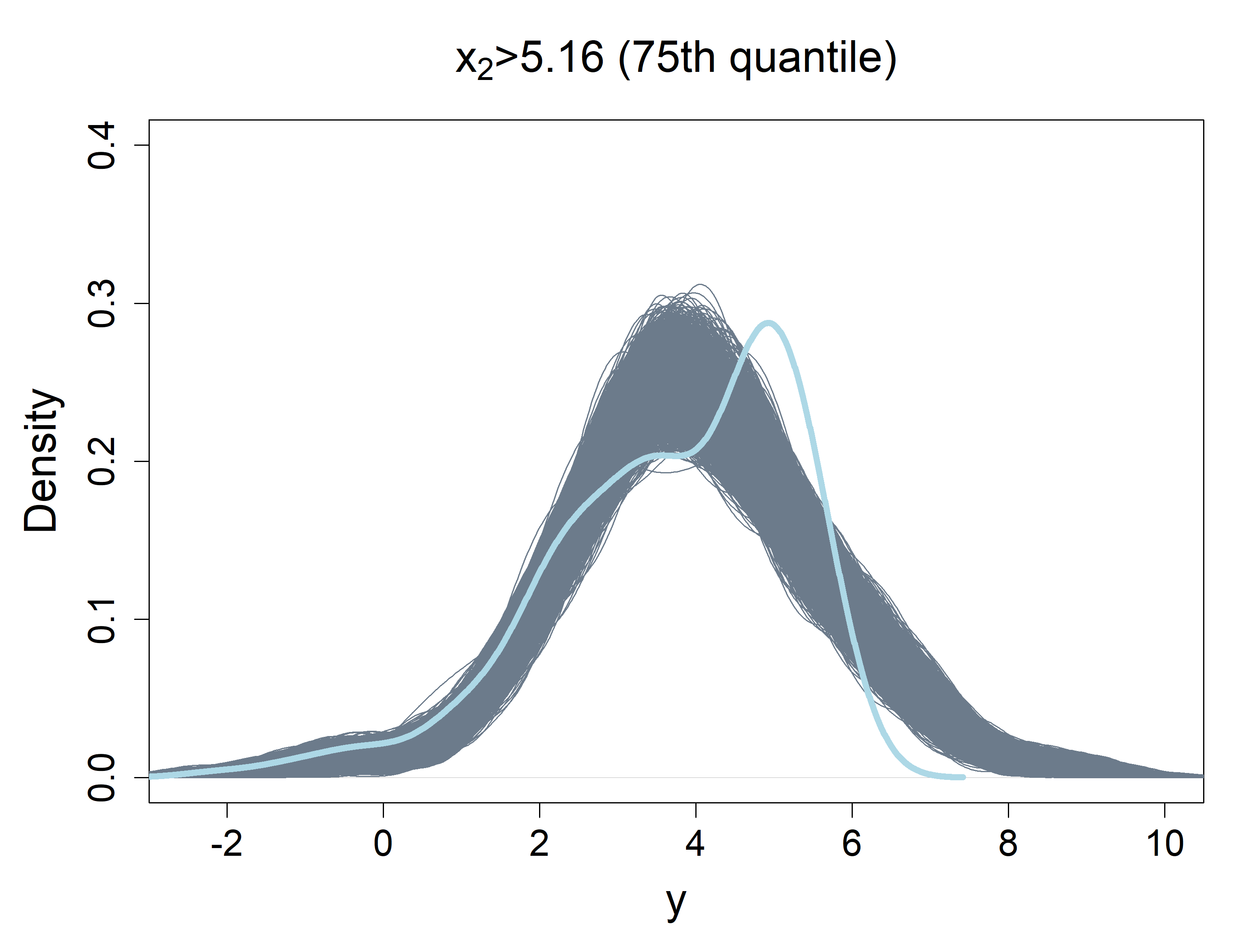}
}
\subfigure{
\centering
\includegraphics[width=0.22\textwidth]{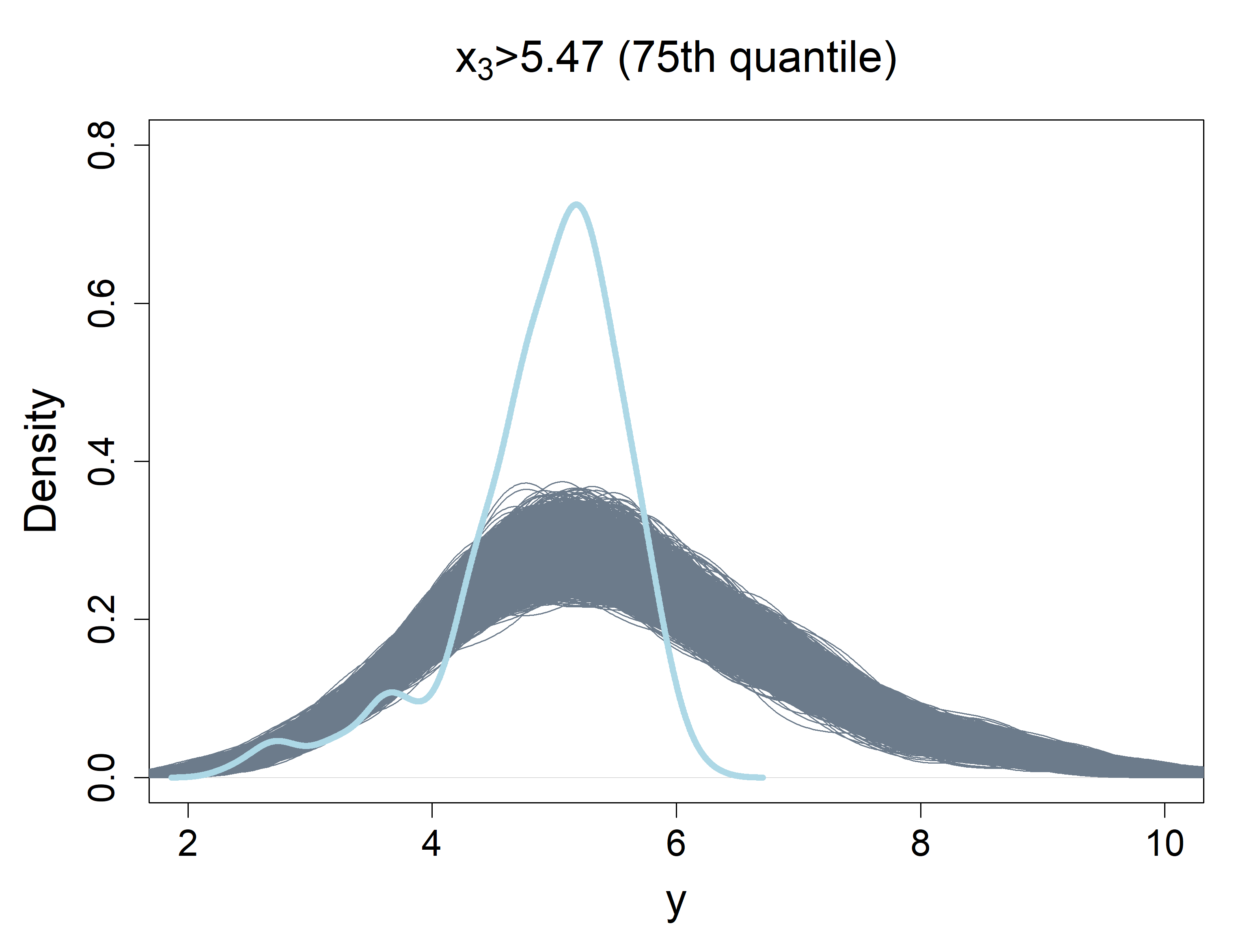}
}

\centering
\caption{Example D: posterior predictive checks for selected statistics (kurtosis and kernel density estimates), with the estimated statistics (light blue), shown alongside estimates from 5,000 datasets drawn from the posterior predictive distribution (grey).} 
\label{example_D_post_checks}
\end{figure}

\begin{figure}[!t]
\centering
\subfigure{
\centering
\includegraphics[width=0.38\textwidth]{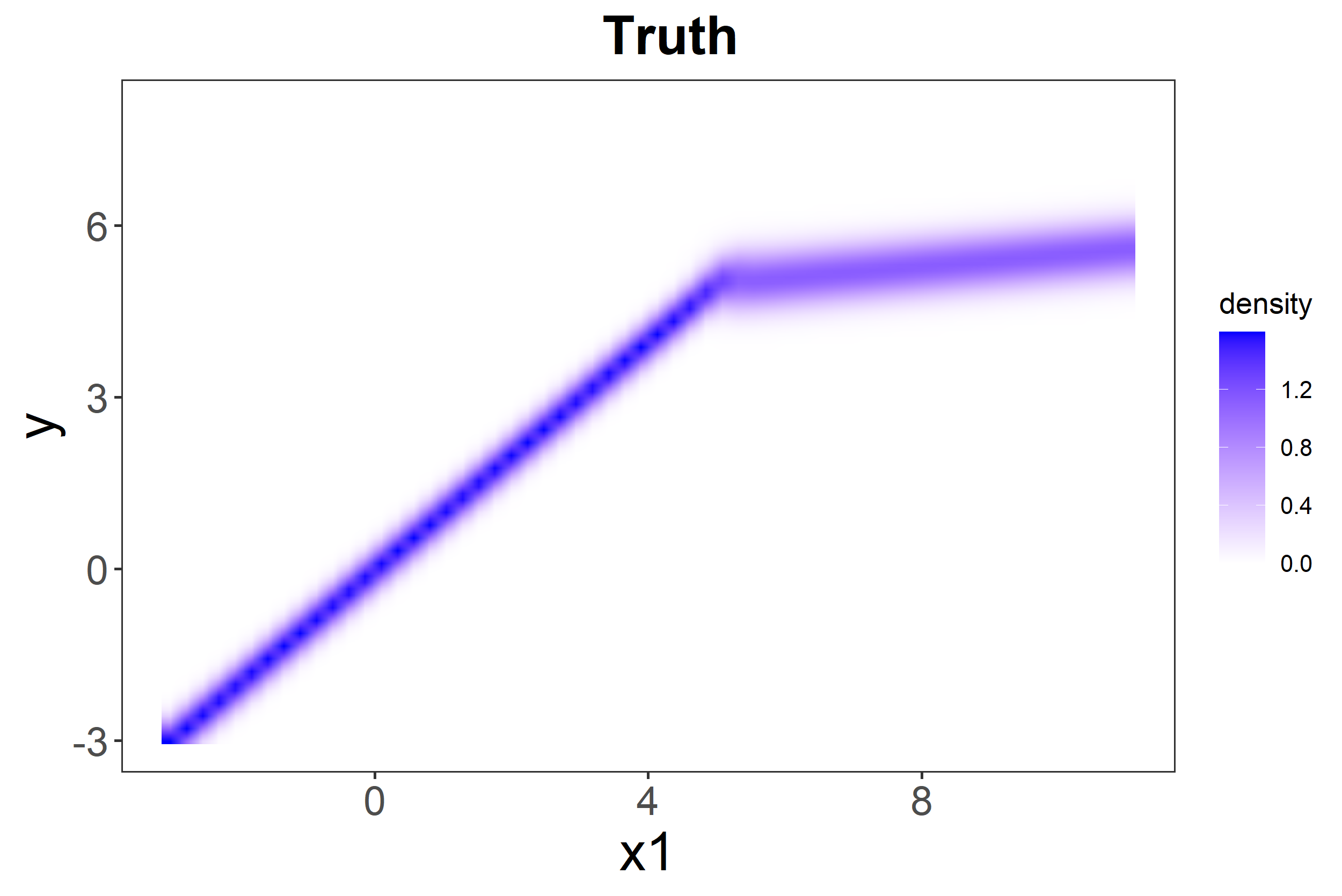}
}
\subfigure{
\centering
\includegraphics[width=0.38\textwidth]{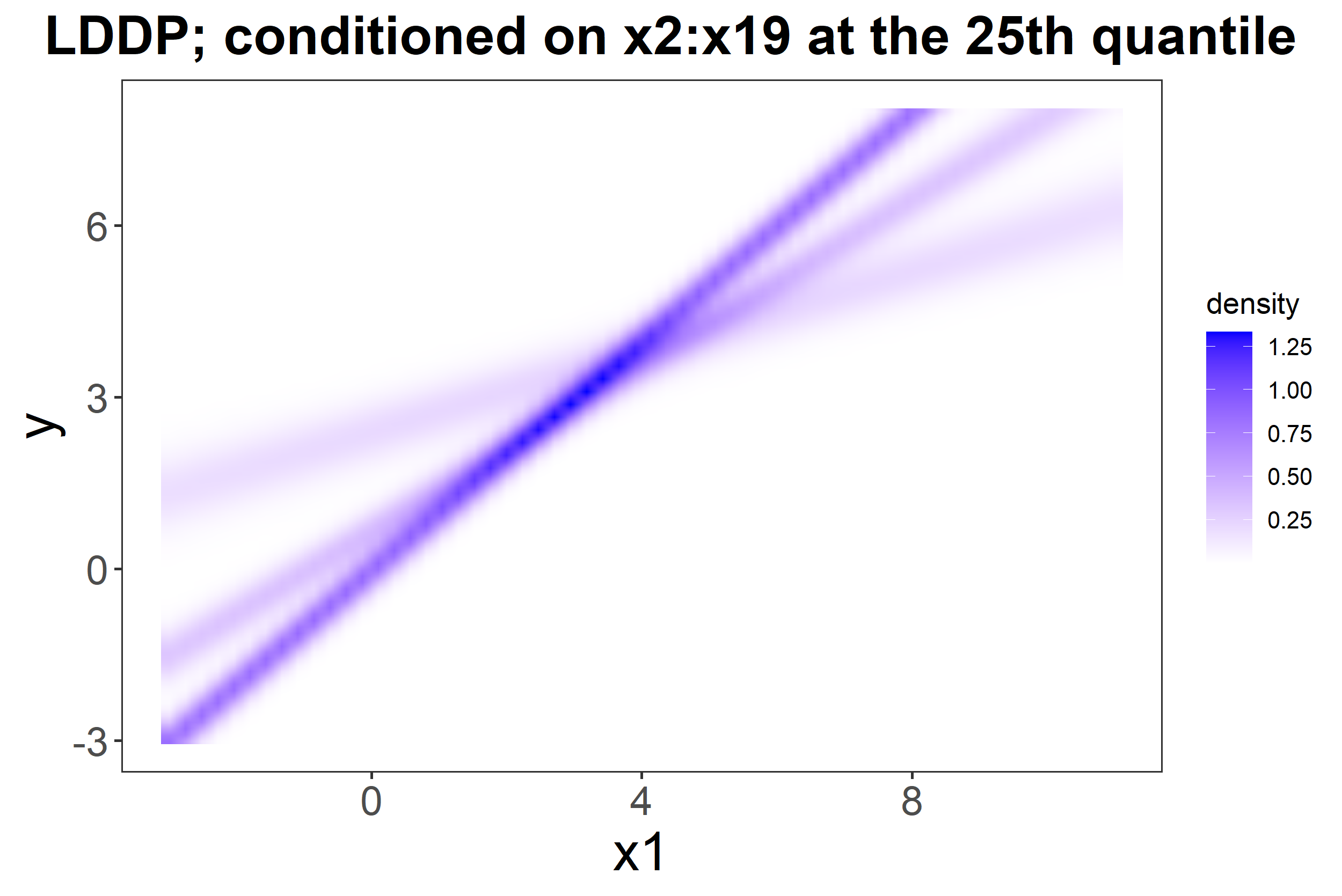}
}

\centering
\caption{Example D: heatmap of the true density regression function and the estimated density regression function conditioned on $x_2$ to $x_{19}$ at the $25\%$ quantiles.}
\label{example_D_heatmap_plots}
\end{figure}

As shown in Figure \ref{example_D_UNL_plots}, the LDDP mixture model identifies two clusters in the representative partition. 
Although one can visually detect some dependence of the partition on \(x_1\), it is difficult to assess how the clustering partition depends simultaneously on all covariates, or on the odd/even-indexed subsets of covariates. However, from the UNL values (Figure \ref{example_D_UNL_plots}), we can detect that the LDDP-induced partition exhibits much stronger dependence on the odd-indexed covariates than on the even-indexed ones. This correctly reflects the data-generating mechanism, in which only the odd-indexed covariates influence cluster allocation. Also, within the odd-indexed group, the UNL for \(x_1\) is only slightly smaller than the UNL for all odd covariates jointly, indicating that most of the partition–covariate dependence is driven by \(x_1\). This matches the construction of the example, where \(x_1\) is the only covariate that informs the cluster allocation probabilities.

In this case, the poor predictive performance of the LDDP model, which ignores covariates in the mixing weights, is therefore anticipated given the large UNL value. The discrepancy between the true conditional density and the posterior predictive conditional density is evident in Figure \ref{example_D_heatmap_plots}. Although several summary statistics of the posterior predictive distribution do not deviate severely from the observed values, the lack of fit becomes clear in the kernel density estimates of the posterior predictive distributions conditional on several covariate intervals, particularly for conditioning on \(x_1\) or \(x_3\) above their 75th percentiles. Some selected posterior predictive checks are presented in Figure \ref{example_D_post_checks} while the others are shown in the Appendices. It might be easy to overlook the lack of fit if not choosing the correct posterior check statistics, while the UNL can flag this. 

While posterior predictive checks of conditional densities may reveal such model inadequacies, the UNL offers insights into where to make model adjustments, particularly whether or not to allow weights to depend on specific covariates. Moreover, in real-world applications with high-dimensional covariates, kernel density estimates within conditioning intervals can become unstable when the sample size is small. In these settings, UNL provides a valuable complement to posterior predictive checks, offering a scalar diagnostic that can flag poor predictive performance of the single-weight mixture model when the partition depends strongly on covariates.

\section{Real data illustrations}\label{sec:app}
This section illustrates the broad applicability of the UNL for assessing partition-covariate dependence in model-based clustering, through two real-data applications. In all cases, posterior inference relies on 10000 iterations, after discarding the first 10000 as a burn-in. To evaluate the underlap via importance sampling (Algorithm \ref{alg:importance_sampling}), we use a Monte Carlo sample size of $M=5000$.

\subsection{Breast cancer genomic analysis}
Breast cancer is a heterogeneous disease whose development and progression is complex and not fully understood \citep{yersal2014biological}. Characterizing this heterogeneity is essential for improved personalized treatment, and clustering methods provide a natural framework for uncovering the latent structure.

A substantial literature has investigated prognostic and therapeutic targets in breast cancer, numerous gene expression signatures were suggested as contributing to disease severity and prognosis \citep[e.g.,][]{kallah2025breast,shokoohi2025genetic}. In this application, we focus on how the heterogeneity structure in breast cancer prognosis can be explained by the selected gene expression groups. As our clinical endpoint, we characterize the prognosis by the three clinicopathologic tumor characteristics which are used to calculate the Nottingham prognostic index (NPI) \citep{haybittle1982prognostic}: histological grade (varying from 1–3 with 3 as the most abnormal grade), lymph node stage (varying from 1–3 with 3 as the stage with the most positive lymph nodes), and tumor size (mm).

We applied our methods to the METABRIC dataset, which provides clinical and genomic data on breast tumors from five different hospitals/ research centers in the United Kingdom and Canada \citep{pereira2016somatic}. Guided by prior evidence on prognostic relevance \citep{ayoub2024impact,tang2015decreased}, we focus on two gene groups: \{MET, ESR1, ESR2\} and \{BRCA1, BECN1\}. Of the 2509 patients in the dataset, 2040 had complete data on the tumor characteristics required for the NPI; among these, 1815 also had complete mRNA expression log intensity values for the selected genes. Model-based clustering is performed on the 2040 cases with complete NPI tumor characteristics, whereas estimation of the UNL for the selected genes is restricted to the 1815 cases with complete gene-expression data.

Although the NPI is widely used in clinical practice for prognostic stratification following breast cancer surgery, it is a derived score. Clustering directly on the NPI may therefore reflect the near discrete nature of the score construction rather than its genuine latent structure (see the bottom right panel in Figure \ref{cluster_chara_tumor}). To avoid this, we fit a DPM model to the three tumor characteristics used to compute the NPI. It is an example of the marginal approach of model-based clustering as described in Section 3. The kernel of DPM model takes the product form: $
K(y;\theta)\;=\; \mathrm{N}\bigl(y_1;\mu,\sigma^{2}\bigr)\,
\mathrm{Cat}\bigl(y_2;p_{\text{g1}},p_{\text{g2}},p_{\text{g3}}\bigr)\,
\mathrm{Cat}\bigl(y_3;p_{\text{n1}},p_{\text{n2}},p_{\text{n3}}\bigr)$,
where $y_1$, $y_2$ and $y_3$ represent tumor size, histological grade and lymph node stage, respectively.

\begin{figure}[!t]
\centering
\subfigure{
\centering
\includegraphics[width=0.38\textwidth]{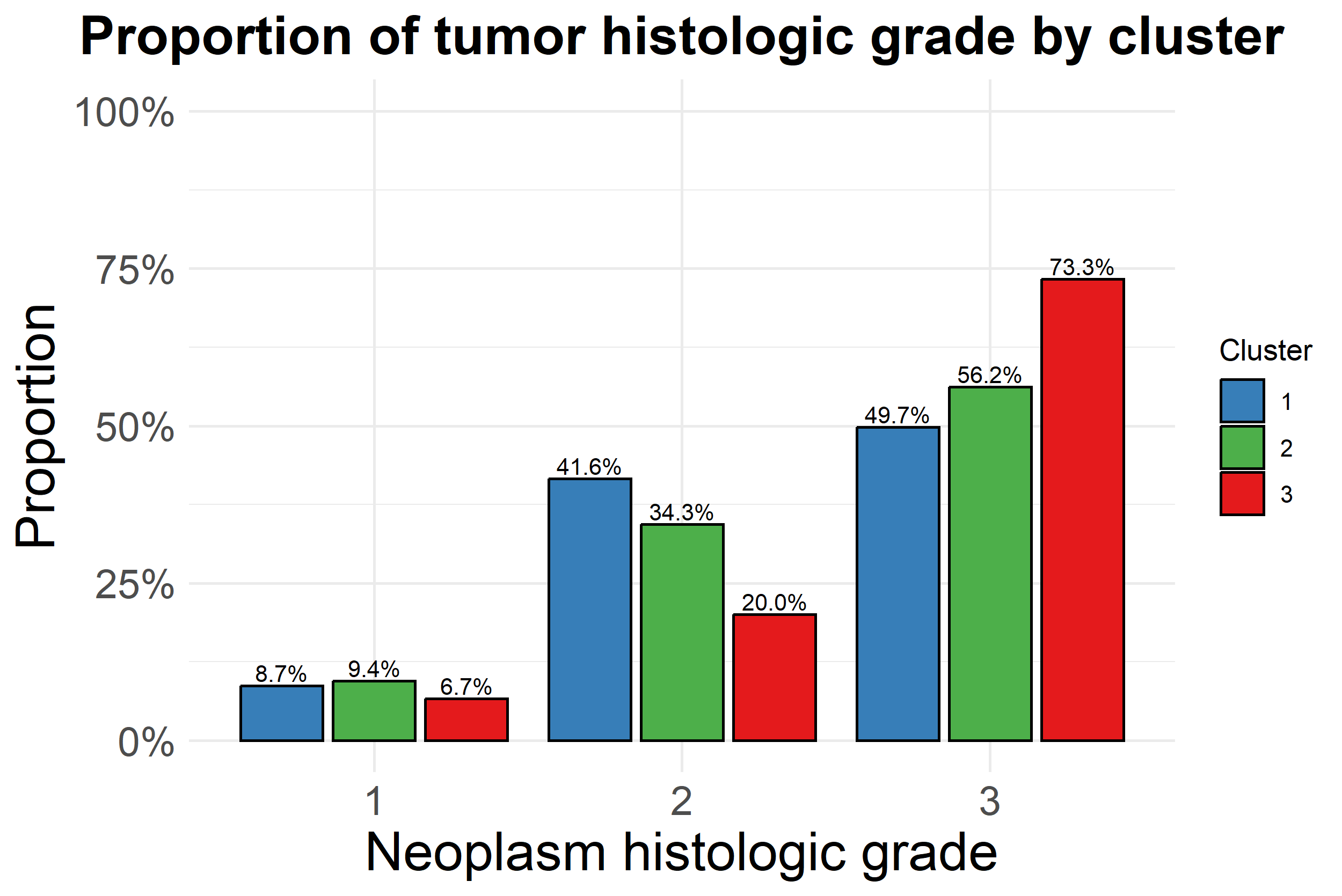}
}
\subfigure{
\centering
\includegraphics[width=0.38\textwidth]{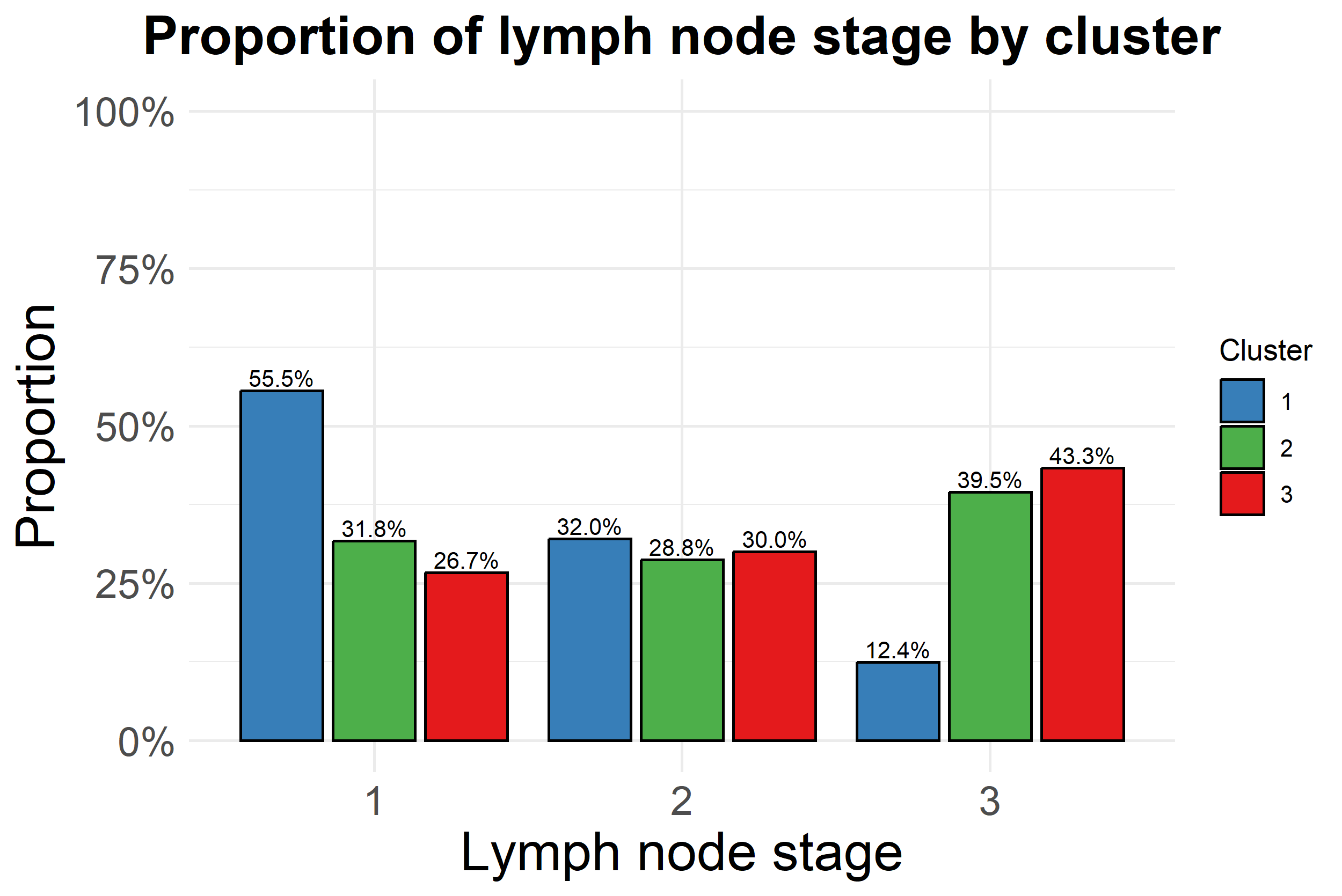}
}\\
\centering
\subfigure{
\centering
\includegraphics[width=0.38\textwidth]{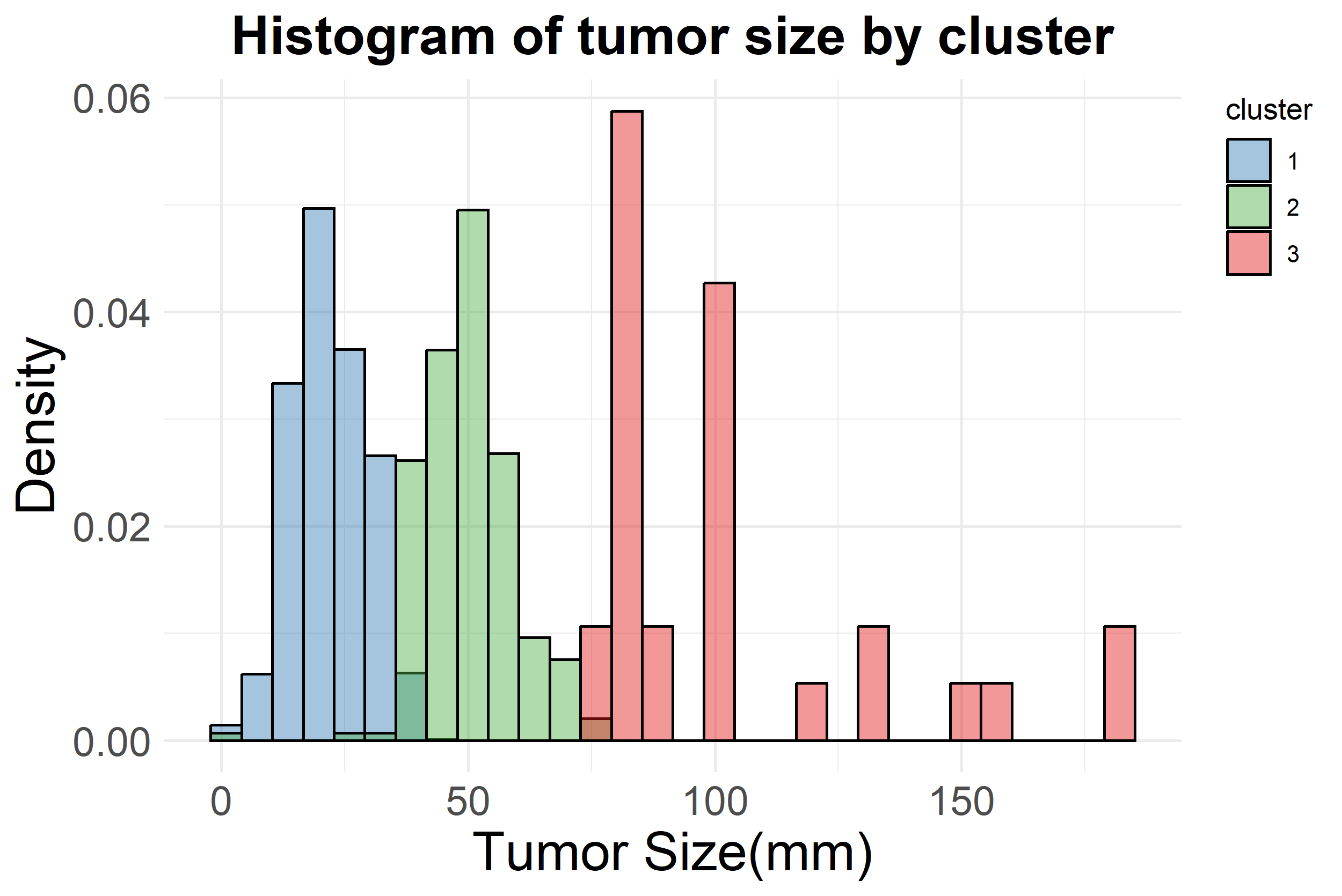}
}
\subfigure{
\centering
\includegraphics[width=0.38\textwidth]{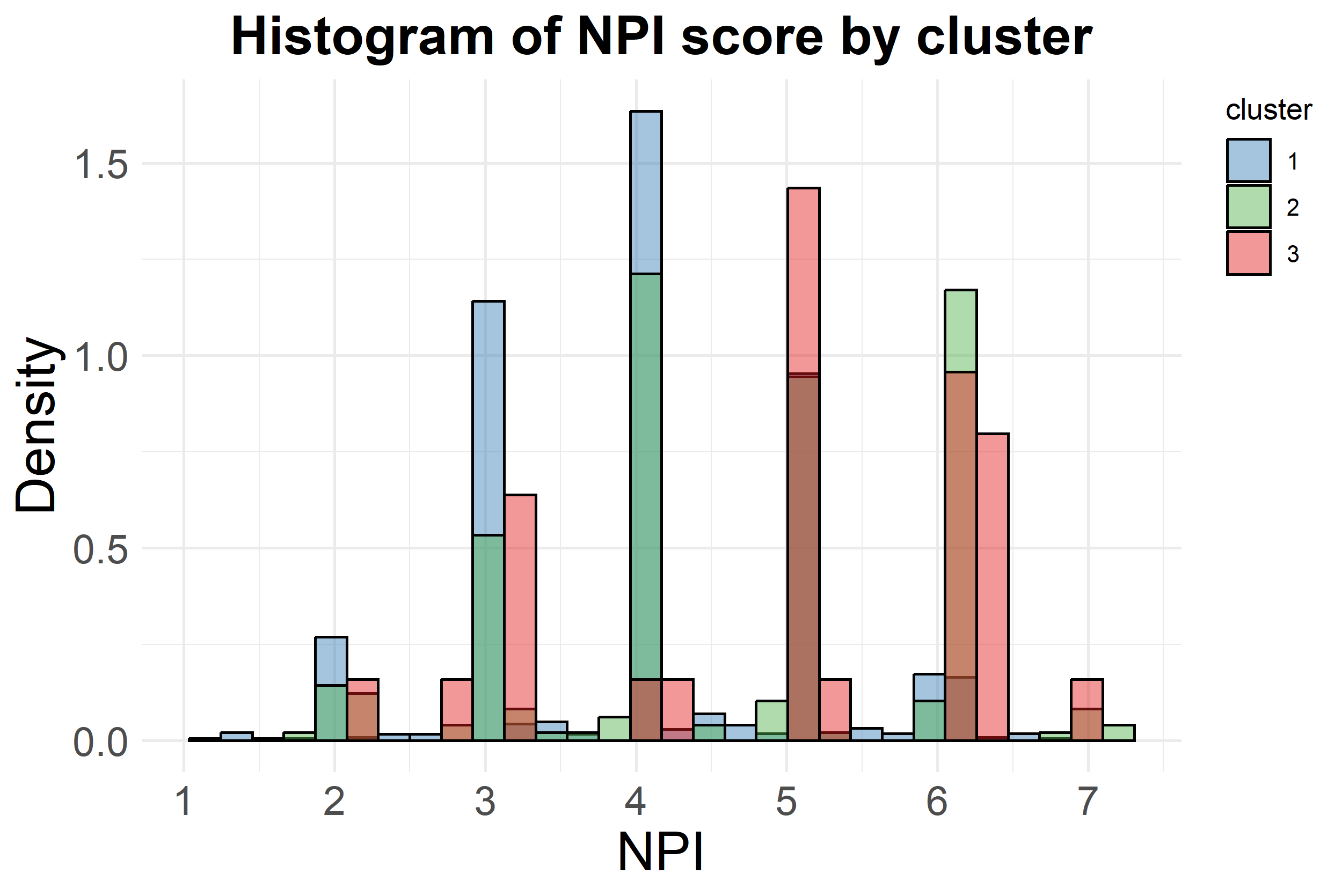}
}
\caption{Cluster profiles of breast cancer prognostic factors, based on the single representative partition.}
\label{cluster_chara_tumor}
\end{figure}

As shown in Figure \ref{cluster_chara_tumor}, cluster 1 exhibits the smallest tumor sizes, the lowest proportion of the most abnormal histological grade (grade 3), and the lowest proportion of advanced lymph node stage (stage 3); cluster 3 shows the opposite pattern. 
Therefore, cluster 1 corresponds to the most favorable prognosis and cluster 3 to the least favorable. These findings suggest that patients in cluster 3 may warrant consideration of more aggressive therapies.

We next quantify the dependence of the clustering on the two gene groups, \{MET, ESR1, ESR2\} and \{BRCA1,BECN1\}, via the UNL. Figure \ref{UNL_BC_gene_hist} displays the histograms of the estimated UNL for each gene groups. The UNL of \{MET, ESR1\} is larger than UNL of \{MET, ESR2\} with high probability ($\text{P}(\text{UNL}_{\text{MET\&ESR1}}>\text{UNL}_{\text{MET\&ESR2}})=0.99$), suggesting that the coexpression of MET and ESR1 is more informative for the clinicopathologic prognosis clusters than MET and ESR2. \cite{tang2015decreased} reported an association with prognosis for BECN1 but not BRCA1. Consistent with this, the estimated marginal underlap distribution for BECN1 is slightly shifted upward relative to that for BRCA1 (despite substantial overlap in their histograms), with $\text{P}(\text{UNL}_{\text{BECN1}}>\text{UNL}_{\text{BRCA1}}=0.71)$. It should be mentioned that, by themselves, the mRNA expression of the selected genes only accounts for a small fraction of the three prognosis groups, incorporating the expression of other gene groups and/or using multi-omics methods is suggested for further investigations.

\begin{figure}[!t]
\centering
\subfigure{
\centering
\includegraphics[width=0.38\textwidth]{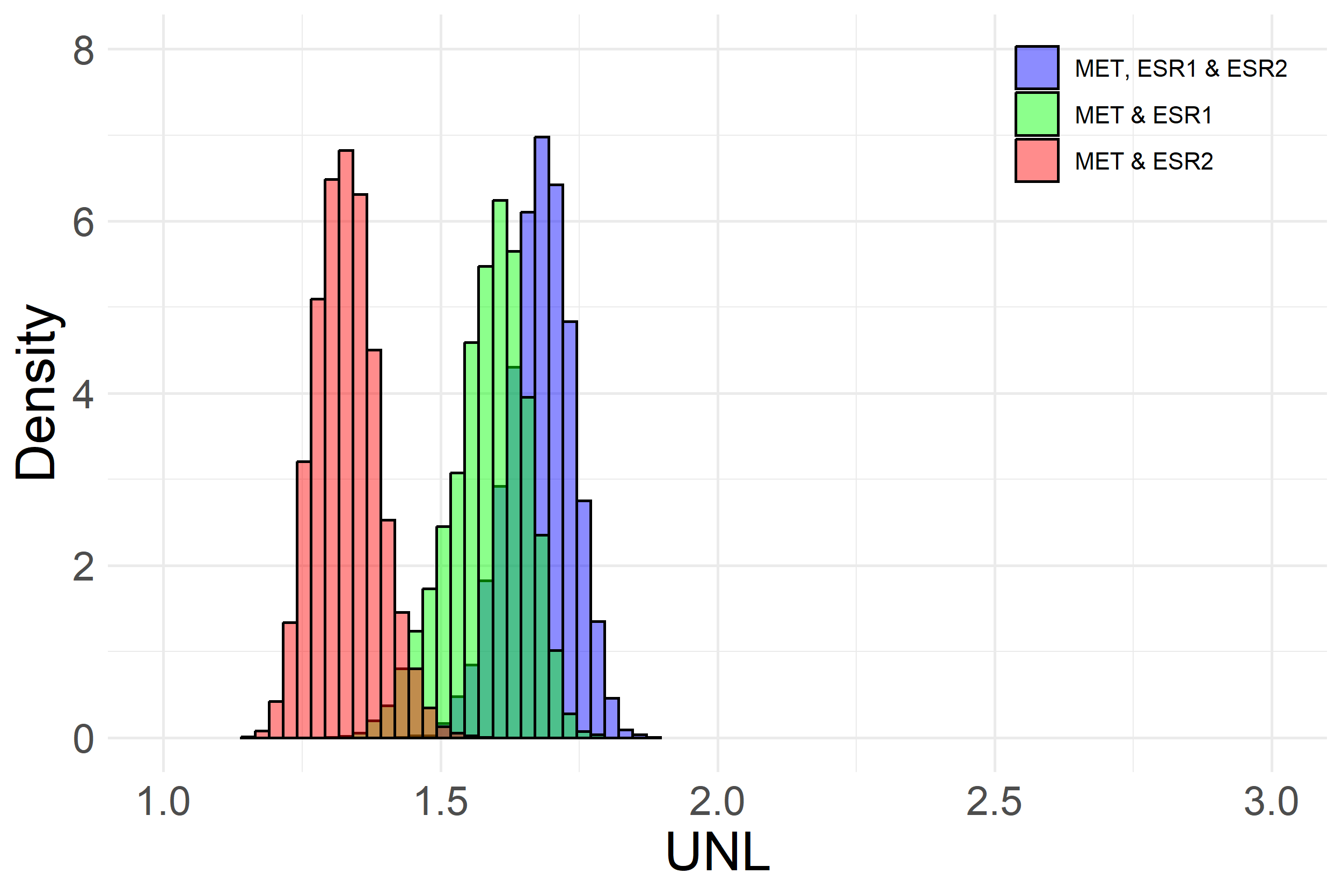}
}
\subfigure{
\centering
\includegraphics[width=0.38\textwidth]{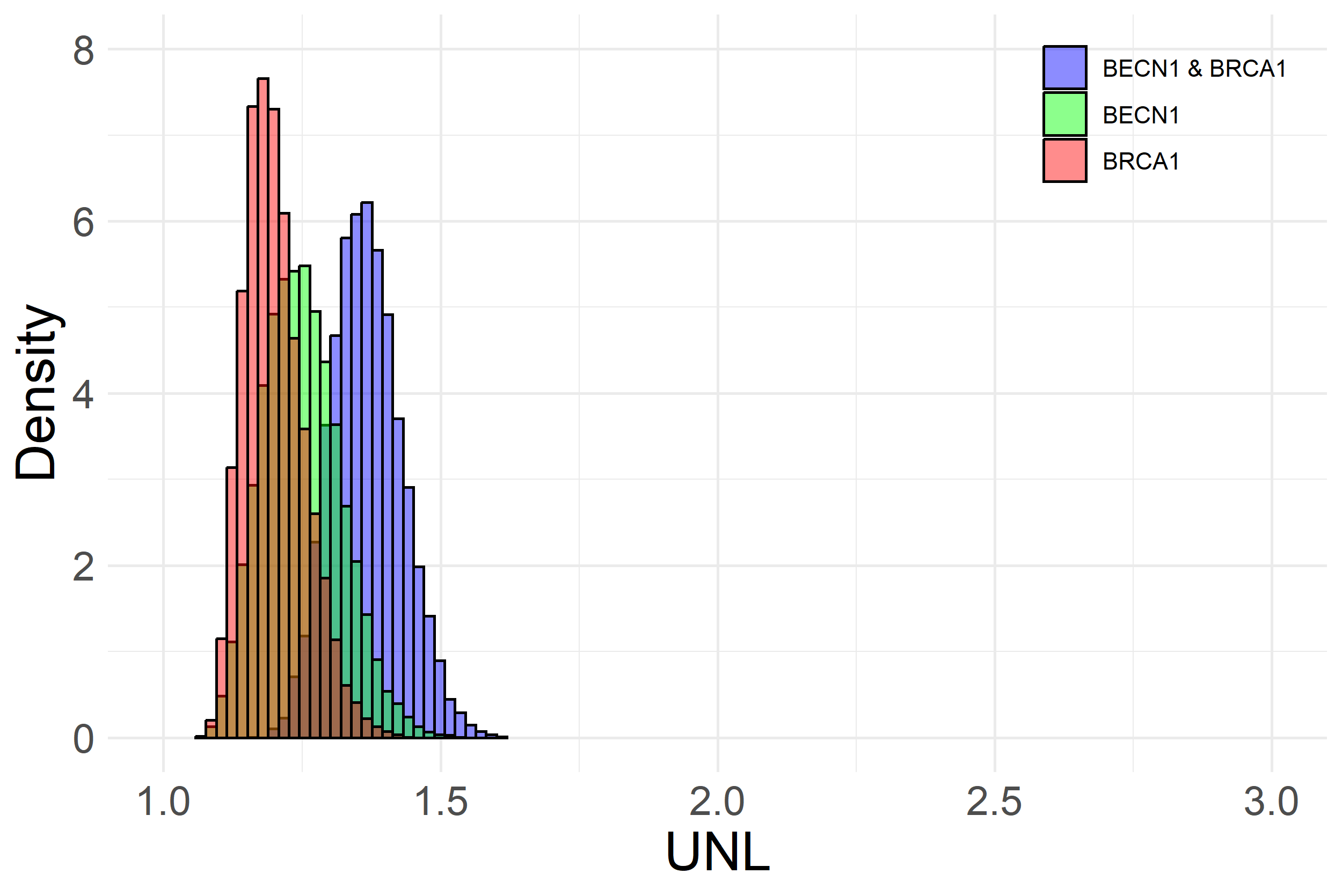}
}
\centering
\caption{Left panel: the histogram of estimated underlap for the \{MET, ESR1, ESR2\} gene group based on the single representative partition. Right panel: the histogram of estimated underlap for the \{BRCA1, BECN1\} gene group based on the single representative partition.}
\label{UNL_BC_gene_hist}
\end{figure}

\subsection{Pregnancy term toxicology analysis}
Our second application, drawn from epidemiology studies, has been used in numerous works for validating Bayesian nonparametric density regression models \citep[e.g.,][]{dunson2008kernel,rigon2021tractable,rodriguez2025density}. The primary objective is to assess how maternal exposure to dichlorodiphenyldichloroethylene (DDE), a metabolite of the pesticide DDT, relates to gestational age at delivery (GAD). Despite concerns about adverse health effects, DDT remains in use for malaria control in some regions, motivating careful investigation of potential preterm pregnancy risks \citep{rodriguez2025density}.

The dataset contains 2312 women with third–trimester maternal serum DDE concentrations and GAD recorded, and deliveries prior to 37 completed weeks are considered preterm while deliveries after 42 completed weeks are considered postterm. While many epidemiologic studies dichotomize GAD at the 37–week threshold to model preterm birth as a binary outcome, such dichotomization discards information about GAD tail behavior. In particular, morbidity and mortality risks increase sharply as GAD decreases within the preterm range, making the left tail of the GAD distribution especially salient \citep{rigon2021tractable}. Accordingly, we analyze how the entire conditional distribution of GAD varies with DDE, with particular focus on left–tail changes.

As shown in several studies of this dataset \citep[see e.g.,][]{dunson2008kernel,rigon2021tractable}, the probability that GAD falls below clinical thresholds varies with DDE. This naturally motivates mixtures with covariate-dependent weights, in which exposure can alter the prevalence of latent subpopulations (e.g., preterm/term/postterm). However, if these probability shifts arise primarily from within-group dependence, allowing highly flexible weight functions may be unnecessary and invites overfitting at substantial computational cost. We therefore adopt the LDDP model, which holds the weights constant across DDE while allowing component means to depend on DDE. The LDDP model is computationally simpler and yields stable inference for conditional densities when the signal of covariates in the composition of clusters is weak. 


We apply the LDDP model to the data $(y_i, x_i) = (\text{GAD}_i
, \text{DDE}_i)$, for $i = 1, \cdots, 2312$, and consider a simple linear structure for the mean of each mixture component: $\mu(x_i)=\beta_0+\beta_1 x_i.$ Here we focused on the single representative partition from the LDDP clustering shown in Figure \ref{dde_cluster_partition_dde_UNL_hist}. 91.8\% of women in Cluster 1 delivered between 37 and 42 weeks, 72.3\% in Cluster 2 delivered before 37 weeks, and everyone in Cluster 3 delivered after 42 weeks. We therefore refer to Clusters 1–3 as the term, preterm, and postterm groups, respectively. 

We assessed the dependence of the inferred partition on DDE using the UNL. Figure \ref{dde_cluster_partition_dde_UNL_hist} also displays the histogram of the estimated UNL for DDE. The posterior sampling distribution of the UNL concentrates near one (posterior mean \(=1.21\)), indicating little residual dependence of the partition on DDE. This is consistent with the single-weights assumption of the LDDP, under which cluster weights do not vary with DDE. Posterior predictive checks of several statistics and conditional densities (in the Appendices) further support an adequate fit of the LDDP model.


\begin{figure}[!t]
\centering
\subfigure{\includegraphics[width=0.38\textwidth]{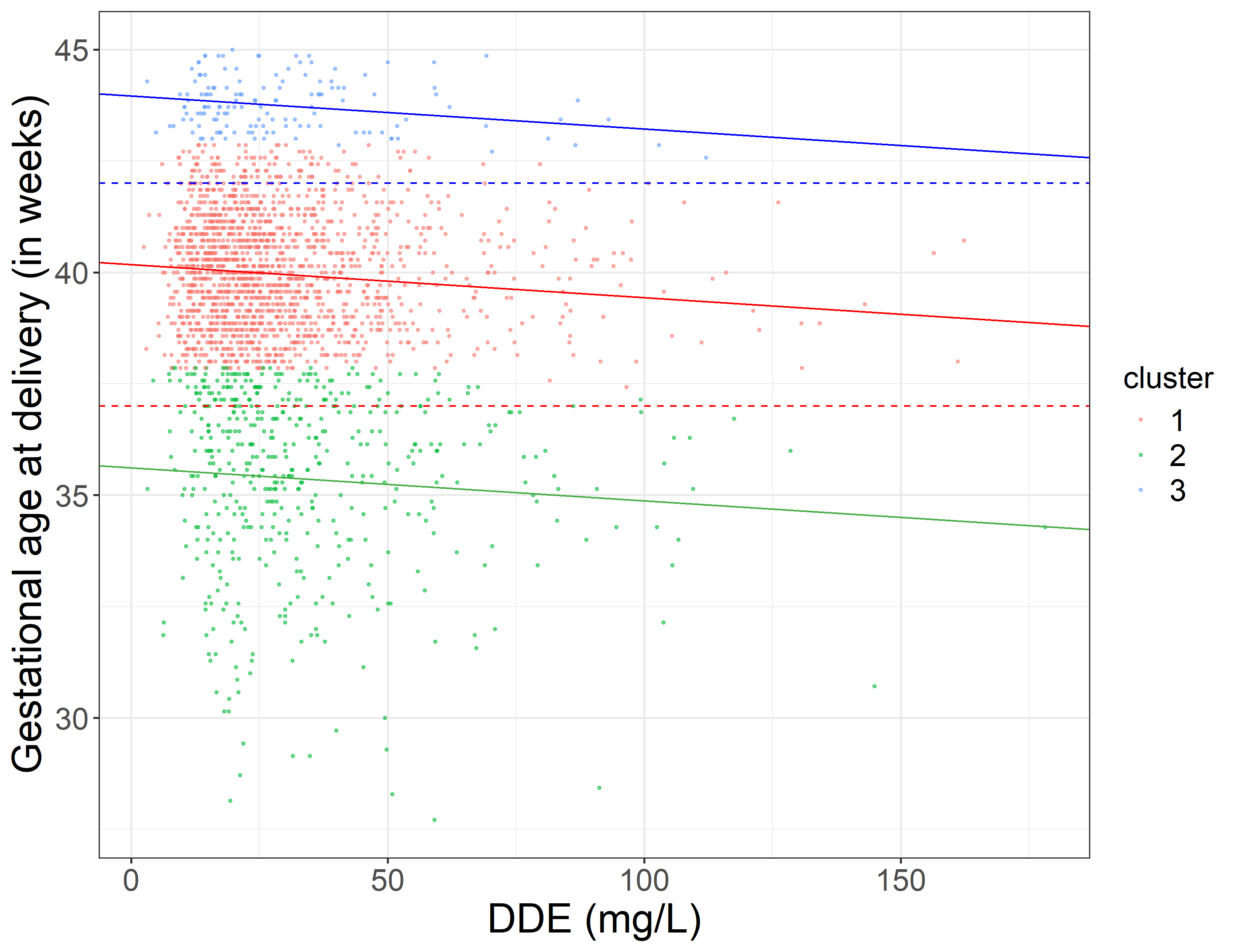}
}
\subfigure{\includegraphics[width=0.38\textwidth]{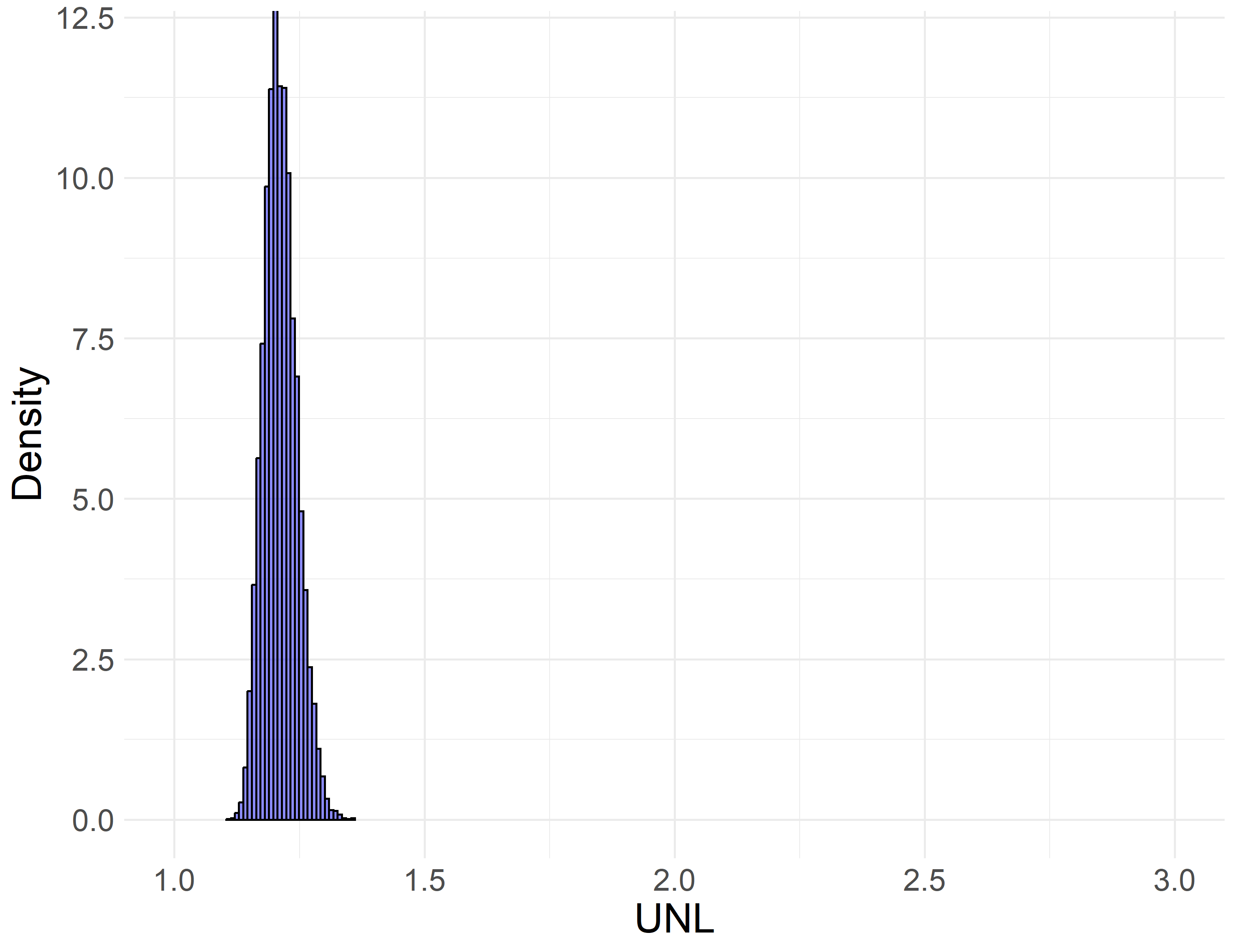}
}
\caption{Left: Single representative partition structure of the LDDP clustering of the pregnancy term toxicology analysis. The red and blue lines represent the 37-week preterm threshold and the 42-week postterm threshold used in clinical practice, respectively. Right: The histogram of estimated UNL for DDE based on the partition shown in the left.}
\label{dde_cluster_partition_dde_UNL_hist}
\end{figure}

Based on the UNL values, there is little evidence that subpopulation prevalences change with DDE levels, therefore letting the weights vary with DDE risks overfitting with negligible gain. We conduct a cross-check with a more flexible model, the logit stick-breaking process (LSBP) model \citep{rigon2021tractable} which allows a flexible structure of covariate-dependent weights. We employ, for the standardized data, the same hyperparameter values used in \cite{rigon2021tractable}. In Figure \ref{dde_conditional_dens}, the posterior conditional densities obtained using the LDDP model and the LSBP model shows great similarity, even at the extreme 0.99 quantile, indicating little practical gain from allowing DDE-dependent weights. 

\begin{figure}[!t]
\centering
\subfigure{
\centering
\includegraphics[width=0.22\textwidth]{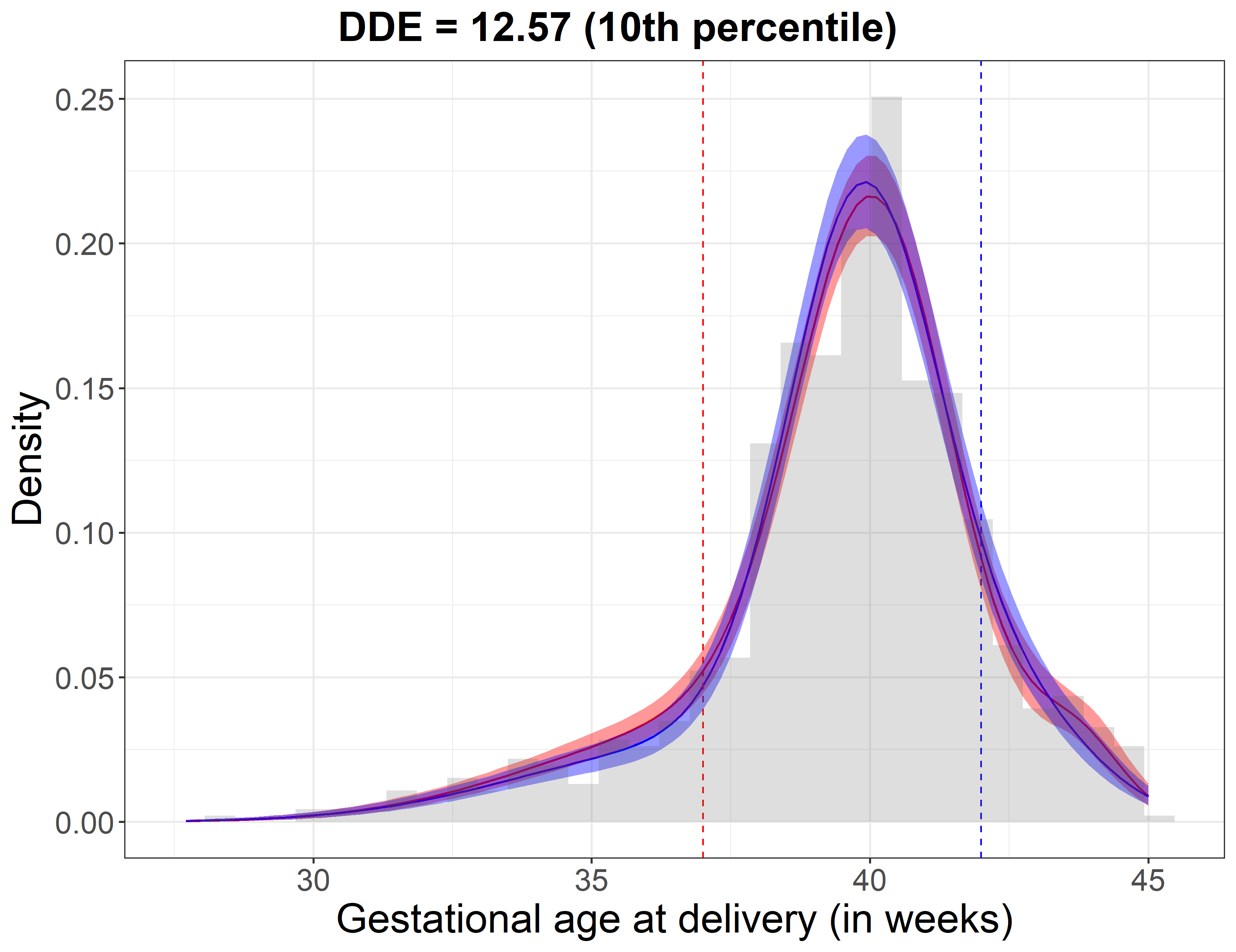}
}
\subfigure{
\centering
\includegraphics[width=0.22\textwidth]{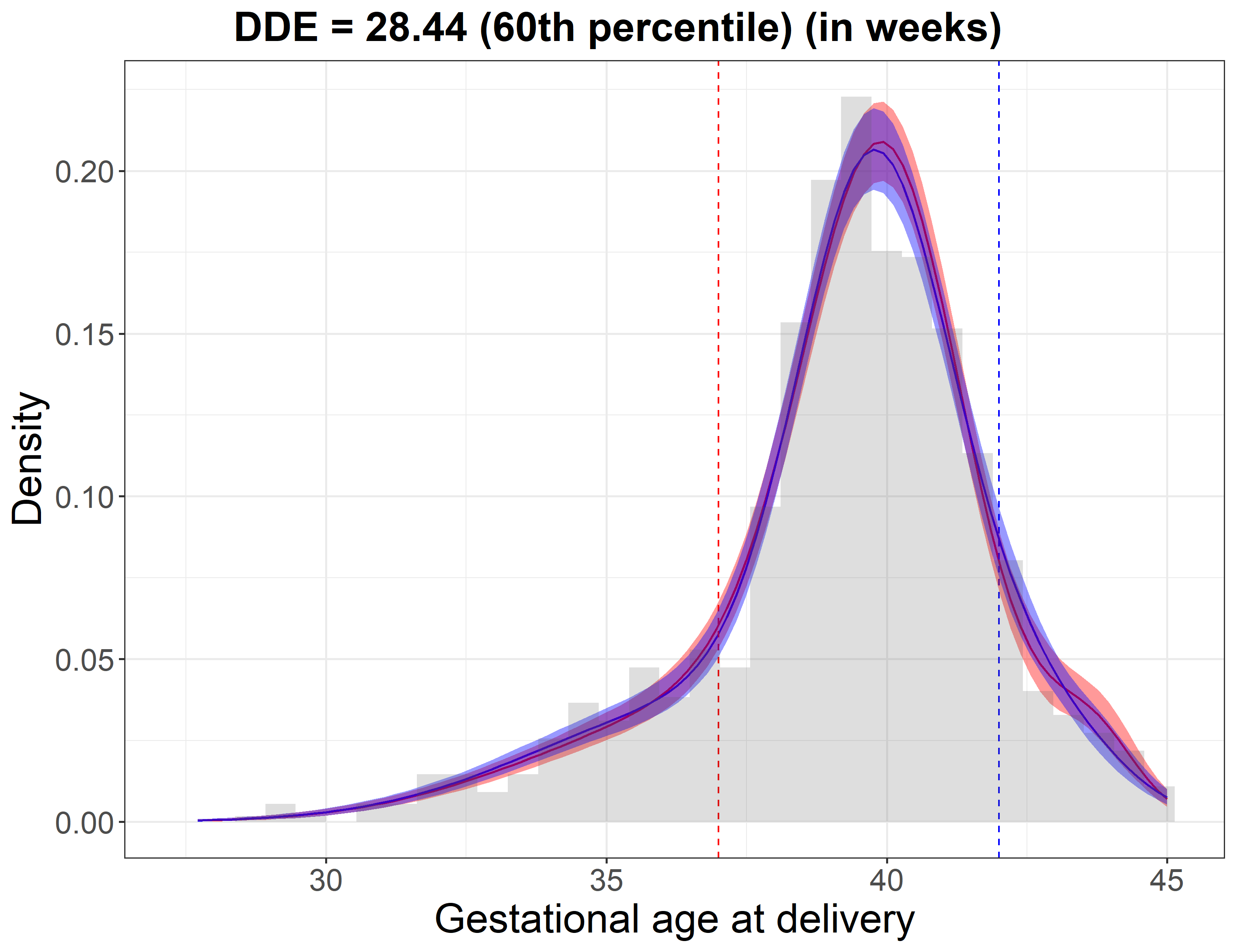}
}
\subfigure{
\centering
\includegraphics[width=0.22\textwidth]{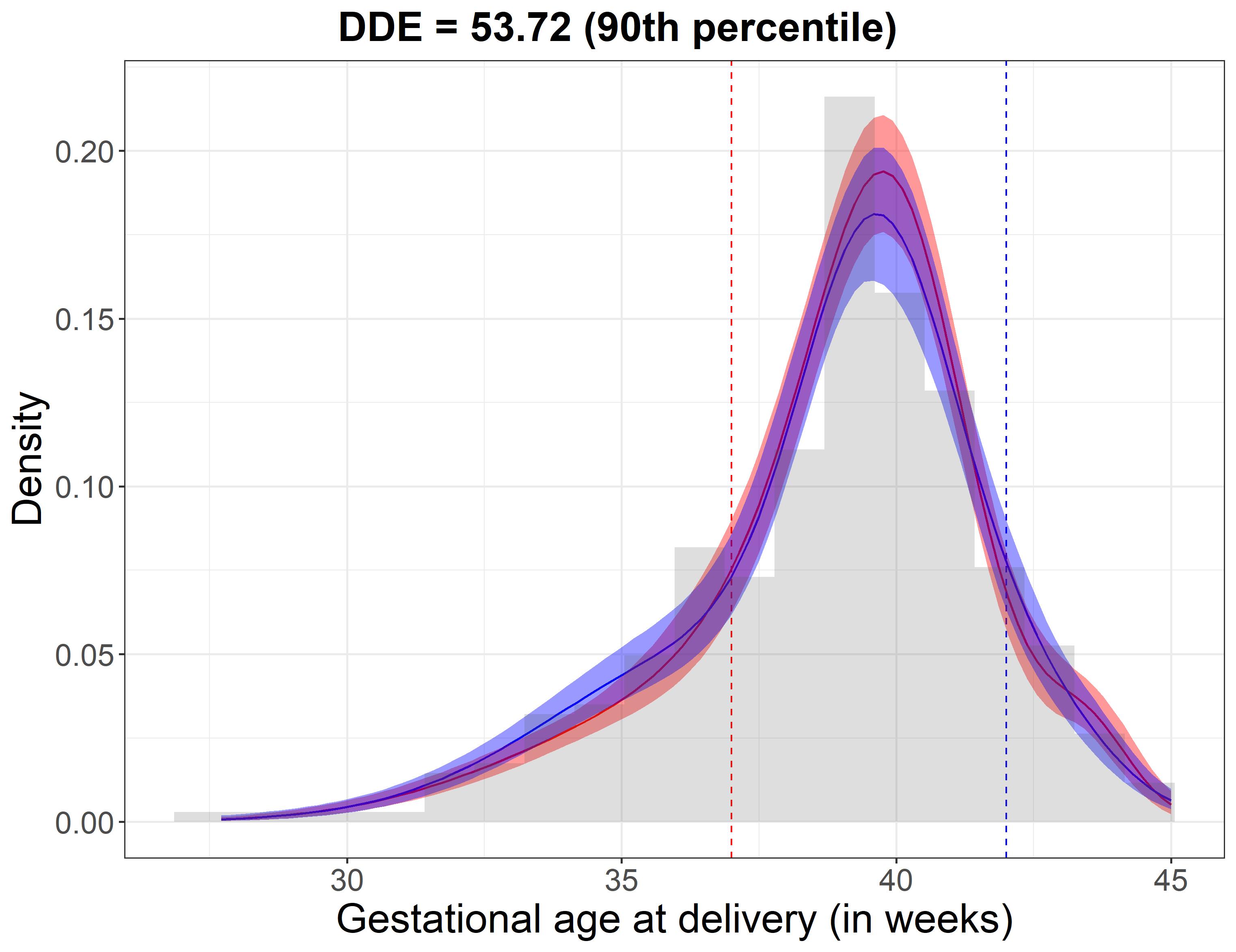}
}
\subfigure{
\centering
\includegraphics[width=0.22\textwidth]{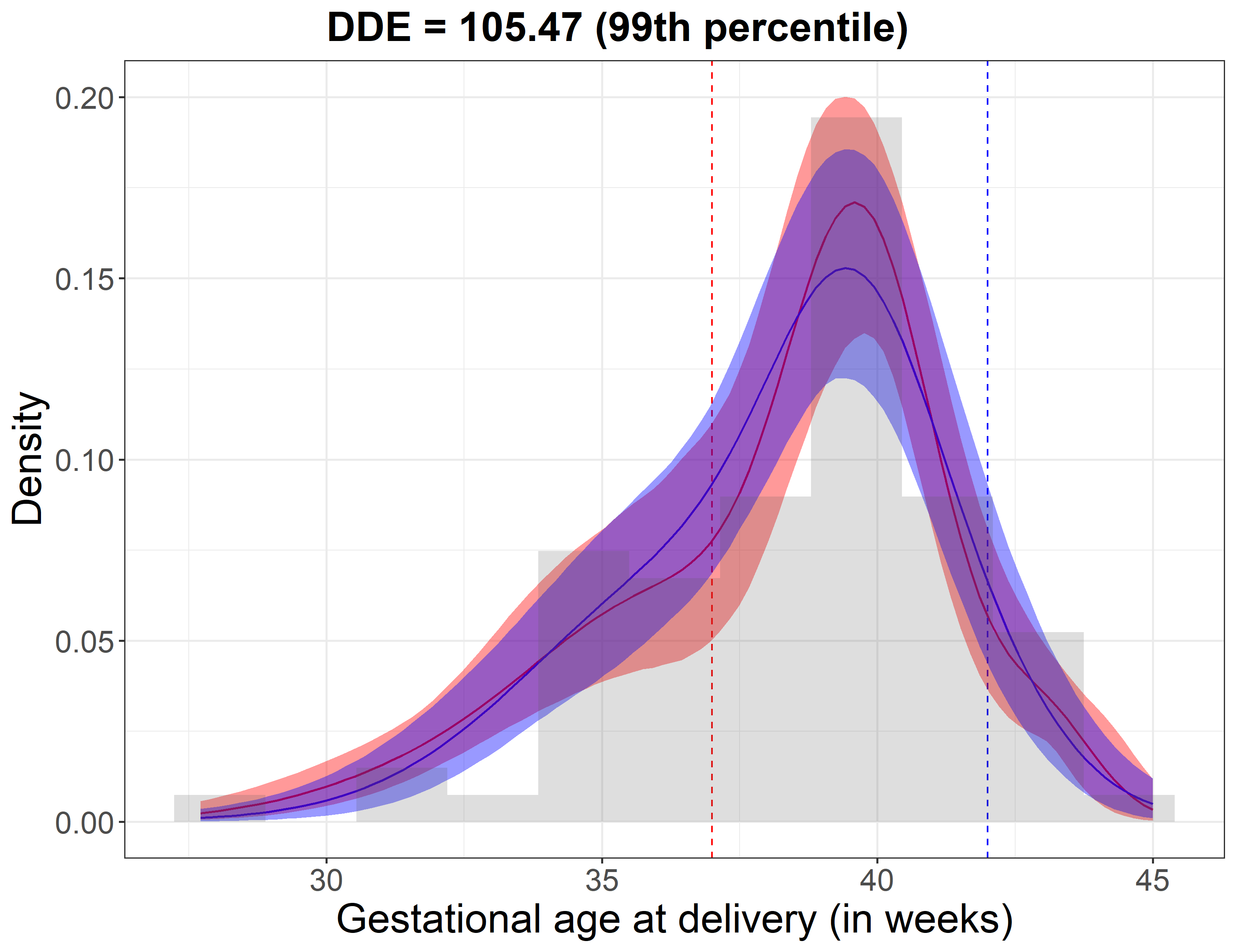}
}

\centering
\caption{Lines and ribbons: posterior mean and 95\% pointwise credible interval of
the conditional density for GAD for selected percentiles of DDE. Blue and red represents results from the LDDP and  
LSBP models, respectively. Grey area: conditional histograms obtained by grouping the GAD values according to a binning of the DDE with cutoffs at the central values
of subsequent quantiles \citep[see][for details]{rigon2021tractable}.}
\label{dde_conditional_dens}
\end{figure}

\section{Discussion}
We generalize the UNL as a multi-group measure of distributional separation for multivariate variables, applicable to various settings, and we further interpret it as a measure of statistical dependence between the group label and variables of interest. We provide a concrete application of the UNL in clustering: assessing partition-covariate dependence. We focus on model-based clustering for illustration. 
In particular, for mixture models with single weights, the UNL provides a targeted scalar diagnostic: it can flag when the inferred partition exhibits substantial dependence on covariates, thereby indicating that covariate-dependent mixing weights are warranted. While posterior predictive checks can diagnose general predictive inadequacy (with appropriately  chosen statistics, which is not always straightforward), the UNL directly highlights whether allowing the weights to depend on specific covariates is useful to adjust and improve the model.

Several extensions are of interest for future work. First, the variance bound derived in this paper is conditional on the estimated densities. In the Bayesian framework, uncertainty from density estimation is propagated naturally to the UNL estimate. However, propagation of partition uncertainty from Bayesian clustering to the UNL could be further explored, beyond the use of a single representative partition or the multiple representative partitions produced by the WASABI method \citep{balocchi2025understanding}. Methodologically, localized variants of the UNL could be developed to target separation over restricted regions of the variable space. Finally, it would be of interest to explore the use of UNL as a criterion for selecting the target dimensions in dimension reduction and for comparing competing dimension reduction methods, particularly in settings where preserving group separation is the primary objective.

The code necessary to reproduce the results 
in Section 3 and 
Section 4 is available at \url{https://github.com/Zhaoxi99/Generalized_UNL_clustering}.

\section*{Acknowledgements}
We acknowledge T Cannings, F Lindgren, V Elvira, and G Clart\'e for helpful discussions. SW was supported by the Engineering and Physical Sciences Research Council (EPSRC), grant no. EP/Y028783/1. 

\appendix

\section{Proof of the Bayes risk prior-shift bound}
For completeness, we restate all propositions and properties listed in the manuscript.

\noindent\textbf{Proposition \ref{prior-shift_bound}} (Bayes risk prior-shift bound).
Let \(\pi=(\pi_1,\ldots,\pi_K)\) and \(\rho=(1/K,\ldots,1/K)\) be the two group prevalence vectors, and let the associated Bayes risk be denoted by $\mathrm{BR}^*_{\pi}$ and $\mathrm{BR}^*_{\rho}$, respectively. Then
$
|\mathrm{BR}^*_{\pi}-\mathrm{BR}^*_{\rho}|
\le
\|\pi-\rho\|_1,
$
or equivalently, 
\[
\left|
\mathrm{BR}^*_{\pi}
-
\left\{
1-\frac1K\operatorname{UNL}(f_1,\ldots,f_K)
\right\}
\right|
\le
\left\|
\pi-\frac1K\mathbf 1
\right\|_1,
\]
where \(\mathbf 1=(1,\ldots,1)^\top\). Thus the difference between the
Bayes risk under group prevalences \(\pi\) and the balanced-prior Bayes
risk associated with the UNL is controlled by the size of the prior shift away from the uniform prior.

\begin{proof}
By the triangle inequality and the Lipschitz property of the maximum function,
\[
\begin{aligned}
|\mathrm{BR}^*_{\pi}-\mathrm{BR}^*_{\rho}|
&=
\left|
\int_{\mathcal X}
\left\{
\max_k \pi_k f_k(x)
-
\max_k \rho_k f_k(x)
\right\}
\,d\nu(x)
\right| \\
&\le
\int_{\mathcal X}
\left|
\max_k \pi_k f_k(x)
-
\max_k \rho_k f_k(x)
\right|
\,d\nu(x) \\
&\le
\int_{\mathcal X}
\max_k
|(\pi_k-\rho_k)f_k(x)|
\,d\nu(x) \\
&\le
\int_{\mathcal X}
\sum_{k=1}^K
|\pi_k-\rho_k|f_k(x)
\,d\nu(x)\le
\sum_{k=1}^K
|\pi_k-\rho_k|
=
\|\pi-\rho\|_1=\left\|
\pi-\frac1K\mathbf 1
\right\|_1.
\end{aligned}
\]

\end{proof}

\section{Connecting the underlap coefficient with total variation}

\noindent\textbf{Proposition~\ref{prop:UNL_TV_K2}} (UNL's relationship with total variation distance when $K=2$).
Suppose $P_1$ and $P_2$ are probability measures absolutely continuous with respect to $\nu$, with Radon-Nikodym derivatives $f_1$ and $f_2$, then $\text{UNL}(f_1,f_2)=1+\text{TV}(P_1,P_2)$

\begin{proof}
When $K=2$, the UNL has a direct analytical link with the Weitzman's overlap coefficient \citep{weitzman1970measures}:
\label{proof:UNL_TV_K2}
\begin{equation}
    \text{UNL}(f_1,f_2)=2-\text{OVL}(f_1,f_2)
    \label{UNL_OVL_2CLASS}
\end{equation}
According to \citep{schmid2006nonparametric}, the Weitzman's overlap could be expressed as 
\begin{equation}
    \text{OVL}(f_1,f_2)=1-\text{TV}(P_1,P_2)
    \label{ovl_tv}
\end{equation}
Substituting \eqref{ovl_tv} into \eqref{UNL_OVL_2CLASS}, we have proposition \ref{prop:UNL_TV_K2} proved.
\end{proof}

\noindent\textbf{Proposition \ref{prop:UNL_equals_TV}} (UNL equals total variation norm of a vector-valued measure consisted of K probability measures).
The UNL of the densities $f_1,\dots,f_K$ satisfies
\[
\mathrm{UNL}(f_1,\dots,f_K)
=
\int_{\mathcal X}\max_{1\le k\le K}f_k(x)d\nu(x)
=
\lVert\mu\rVert_{\mathrm{TV},\infty}.
\]

\begin{proof}
We prove the two inequalities and hence the equality.

\textbf{(i) Upper bound:}
Let $\pi=\{A_1,\dots,A_m\}$ be an arbitrary finite partition of $\mathcal X$.
For every block $A_j$ and every index $k$ we have the pointwise inequality
$\max_{r}f_r(x)\ge f_k(x)$, hence by monotonicity of the integral
\[
\int_{A_j}\max_{r}f_r d\nu
\ge
\int_{A_j} f_kd\nu.
\]
Taking the maximum over $k$ on the right hand side gives
\(
\int_{A_j}\max_{r}f_r d\nu \ge \max_k\int_{A_j} f_k d\nu.
\)
Summing over $j$ yields
\[
\int_{\mathcal X}\max_{r}f_r\,d\nu
=
\sum_{j=1}^{m}\int_{A_j}\max_{r}f_r\,d\nu
\;\ge\;
\sum_{j=1}^{m}\max_{k}\int_{A_j}f_k\,d\nu
=
\sum_{A\in\pi}\max_{k}P_k(A).
\]
Since the right--hand side is bounded above by the supremum over all partitions, we have
\begin{equation}
\int_{\mathcal X}\max_{r}f_r\,d\nu
\;\ge\;
\lVert\mu\rVert_{\mathrm{TV},\infty}.
\label{inequal_1_unl_tv}
\end{equation}

\textbf{(ii) Lower bound via a canonical partition:}
Define measurable sets
\[
A_k
=\Bigl\{\,x\in\mathcal X:
   f_k(x)\ge f_r(x)\text{ for all }r,
   \text{ and }k\text{ is the smallest such index}\Bigr\},
\quad k=1,\dots,K.
\]
The "smallest-index" rule breaks ties, so $\{A_1,\dots,A_K\}$ is a measurable partition of $\mathcal X$.  
On each $A_k$ we have
$\max_r f_r(x)=f_k(x)$, hence
\begin{equation*}
    \sum_{k=1}^{K} \max_r \int_{A_k} f_r\,d\nu
=
\sum_{k=1}^{K}\int_{A_k} f_k\,d\nu
=
\int_{\mathcal X}\max_r f_r\,d\nu.
\end{equation*}
Because this value is attained by the specific partition $\{A_k\}$, it is not smaller than the supremum taken over all partitions:
\begin{equation}
    \lVert\mu\rVert_{\mathrm{TV},\infty}
=
\sup_{\pi}\sum_{A\in\pi}\max_k P_k(A)
\;\ge\;
\int_{\mathcal X}\max_r f_r\,d\nu.
\label{inequal_2_unl_tv}
\end{equation}

Inequalities \eqref{inequal_1_unl_tv} and \eqref{inequal_2_unl_tv} together imply
\(
\int_{\mathcal X}\max_{r}f_r\,d\nu
=
\lVert\mu\rVert_{\mathrm{TV},\infty},
\)
proving the claim.
\end{proof}

\section{Properties of the underlap coefficient}

\noindent\textbf{Property \ref{marginal_monotonicity_UNL_property_continuous}}(Marginal monotonicity of UNL for continuous variables).
Let $I\subseteq \{1,2,\dots,p\}$ index a subset of variables and $I^c$ denote its complement. For each class $k=1,\cdots,K$, $X_k = (X_{kI},X_{kI^c})$ is partitioned into two sets, taking values in $R^{|I|}$ and $R^{p-|I|}$, respectively. 
Define the marginal probability density function of $X_{kI}$ as:
\[
g_k(x_{I})
= \int_{R^{p-|I|}}
f_k(x)d x_{I^c},\qquad k=1,\cdots,K. 
\]
The UNL has the property:
$
\mathrm{UNL}(f_1,\dots,f_K)
\ge
\mathrm{UNL}(g_1,\dots,g_K).
$

\begin{proof}
By definition of the underlap coefficient for continuous variables,
\[
  \mathrm{UNL}(f_1,\dots,f_K)
  =
  \int_{R^|I|}\int_{\mathrm{R}^{p-|I|}}
    \max_{1\le k\le K} f_k(x)\,dx_{I^c}\,dx_{I},
\]
where we have written $x = (x_{I},x_{I^c})$ according to the
decomposition $\mathrm{R}^p = \mathrm{R}^{|I|} \times \mathrm{R}^{p-|I|}$.
Similarly, the underlap coefficient based on the marginal densities
$g_k$ of $X_{kI}$ is
\[
  \mathrm{UNL}(g_1,\dots,g_K)
  =
  \int_{\mathrm{R}^{|I|}}
    \max_{1\le k\le K} g_k(x_{I})\,dx_{I}
  =
  \int_{\mathrm{R}^{|I|}}
    \max_{1\le k\le K}
      \int_{\mathrm{R}^{p-|I|}} f_k(x)\,dx_{I^c}\,dx_{I}.
\]
Fix an arbitrary value $x_{I}\in \mathrm{R}^{|I|}$. For each class $k=1,\dots,K$,
define
\[
  h_k(x_{I^c})
  =
  f_k(x_{I},x_{I^c}),
  \qquad x_{I^c} \in \mathrm{R}^{p-|I|}.
\]
For every $x_{I^c}$ we clearly have
\[
  \max_{1\le k\le K} h_k(x_{I^c})
  \ge
  h_k(x_{I^c})
  \quad\text{for each }k,
\]
and therefore, integrating over $x_{I^c}$ gives
\[
  \int_{\mathrm{R}^{p-|I|}}
    \max_{1\le k\le K} h_k(x_{I^c})\,dx_{I^c}
  \ge
  \int_{\mathrm{R}^{p-|I|}} h_k(x_{I^c})\,dx_{I^c}
  \quad\text{for each }k.
\]
Taking the maximum over $k$ on the right-hand side yields
\[
  \int_{\mathrm{R}^{p-|I|}}
    \max_{1\le k\le K} h_k(x_{I^c})\,dx_{I^c}
  \ge
  \max_{1\le k\le K}
    \int_{\mathrm{R}^{p-|I|}} h_k(x_{I^c})\,dx_{I^c}.
\]
Rewriting in terms of $f_k$ and $g_k$, we obtain, for every fixed
$x_{I}\in \mathrm{R}^{|I|}$,
\[
  \int_{\mathrm{R}^{p-|I|}}
    \max_{1\le k\le K} f_k(x)\,dx_{I^c}
  \ge
  \max_{1\le k\le K}
    \int_{\mathrm{R}^{p-|I|}} f_k(x)\,dx_{I^c}
  =
  \max_{1\le k\le K} g_k(x_{I}).
\]
Finally, integrate both sides with respect to $x_{I}\in \mathrm{R}^{|I|}$:
\[
\begin{aligned}
  \mathrm{UNL}(f_1,\dots,f_K)
  &= \int_{\mathrm{R}^{|I|}}\int_{\mathrm{R}^{p-|I|}}
       \max_{1\le k\le K} f_k(x)\,dx_{I^c}\,dx_{I} \\
  &\ge \int_{\mathrm{R}^{|I|}}
          \max_{1\le k\le K} g_k(x_{I})\,dx_{I}
   =
     \mathrm{UNL}(g_1,\dots,g_K).
\end{aligned}
\]
This proves the marginal monotonicity property of UNL in the continuous case.
\end{proof}

\begin{property}[Marginal monotonicity of UNL for discrete variables]
\label{marginal_monotonicity_UNL_property_discrete}
For each class $k=1,\dots,K$, let $X_k=(X_{k1},\dots,X_{kp})$ take values in the
finite or countably infinite product space
\(
S=S_1\times\cdots\times S_p
\),
with group-specific probability mass function $p_k$ as in
Definition~\ref{def:UNL-discrete}.
Fix an index set $J\subseteq\{1,\dots,p\}$ and write $J^c=\{1,\dots,p\}\setminus J$.
Let
\[
S_J=\prod_{j\in J} S_j,
\qquad
S_{J^c}=\prod_{j\in J^c} S_j,
\qquad
x=(x_J,x_{J^c})\in S_J\times S_{J^c}=S.
\]
Define the marginal probability mass function of the subvector $X_{kJ}=(X_{kj})_{j\in J}$ by
\[
g_k(x_J)
=
\sum_{x_{J^c}\in S_{J^c}} p_k(x_J,x_{J^c})
=
\Pr(X_{kJ}=x_J),
\qquad x_J\in S_J,\;k=1,\dots,K.
\]
Then
\[
\mathrm{UNL}(p_1,\dots,p_K)
\ge
\mathrm{UNL}(g_1,\dots,g_K).
\]
\end{property}

\begin{proof}
By Definition~\ref{def:UNL-discrete},
\[
\mathrm{UNL}(p_1,\dots,p_K)
=
\sum_{x_J\in S_J}\sum_{x_{J^c}\in S_{J^c}}
\max_{1\le k\le K} p_k(x_J,x_{J^c}).
\]
Similarly,
\[
\mathrm{UNL}(g_1,\dots,g_K)
=
\sum_{x_J\in S_J}\max_{1\le k\le K} g_k(x_J)
=
\sum_{x_J\in S_J}\max_{1\le k\le K}
\sum_{x_{J^c}\in S_{J^c}} p_k(x_J,x_{J^c}).
\]

Fix an arbitrary $x_J\in S_J$ and define, for each $k=1,\dots,K$,
\[
h_k(x_{J^c}) = p_k(x_J,x_{J^c}),
\qquad x_{J^c}\in S_{J^c}.
\]
For every $x_{J^c}\in S_{J^c}$ we have $\max_{1\le k\le K} h_k(x_{J^c})\ge h_k(x_{J^c})$ for each $k$.
Summing over $x_{J^c}$ gives
\[
\sum_{x_{J^c}\in S_{J^c}} \max_{1\le k\le K} h_k(x_{J^c})
\ge
\sum_{x_{J^c}\in S_{J^c}} h_k(x_{J^c})
\qquad\text{for each }k.
\]
Taking the maximum over $k$ on the right-hand side yields
\[
\sum_{x_{J^c}\in S_{J^c}} \max_{1\le k\le K} h_k(x_{J^c})
\ge
\max_{1\le k\le K}\sum_{x_{J^c}\in S_{J^c}} h_k(x_{J^c}).
\]
Rewriting in terms of $p_k$ and $g_k$, we obtain for each fixed $x_J\in S_J$,
\[
\sum_{x_{J^c}\in S_{J^c}}
\max_{1\le k\le K} p_k(x_J,x_{J^c})
\ge
\max_{1\le k\le K} g_k(x_J).
\]
Finally, summing both sides over $x_J\in S_J$ gives
\[
\begin{aligned}
\mathrm{UNL}(p_1,\dots,p_K)
&=
\sum_{x_J\in S_J}\sum_{x_{J^c}\in S_{J^c}}
\max_{1\le k\le K} p_k(x_J,x_{J^c}) \\
&\ge
\sum_{x_J\in S_J}\max_{1\le k\le K} g_k(x_J)
=
\mathrm{UNL}(g_1,\dots,g_K),
\end{aligned}
\]
which proves the claim.
\end{proof}

\begin{property}[Marginal monotonicity of UNL for mixed-type variables]
\label{prop:marginal-monotonicity-mixed}
For each class $k=1,\dots,K$, let
\[
X_k = \bigl(X_k^{c},X_k^{d}\bigr),
\qquad
X_k^{c}=(X_{k1}^{c},\dots,X_{kp^c}^{c})\in\mathrm{R}^{p^c},
\qquad
X_k^{d}=(X_{k1}^{d},\dots,X_{kp^d}^{d})\in S,
\]
where $S=S_1\times\cdots\times S_{p^d}$ and $S_j$ is the state space of the
$j$th categorical variable (Definition~\ref{def:UNL-discrete}).
Let $f_k(x^{c},x^{d})$ denote the joint density/mass of $(X_k^{c},X_k^{d})$
with respect to the product of the Lebesgue measure in $R^{p^c}$ and the counting measure in $S$.

Fix index sets $I\subseteq\{1,\dots,p^c\}$ and $J\subseteq\{1,\dots,p^d\}$,
and write $I^c=\{1,\dots,p^c\}\setminus I$ and $J^c=\{1,\dots,p^d\}\setminus J$.
Define
\[
x^{c}=(x_I^{c},x_{I^c}^{c})\in\mathrm{R}^{|I|}\times\mathrm{R}^{p^c-|I|},
\qquad
x^{d}=(x_J^{d},x_{J^c}^{d})\in S_J\times S_{J^c},
\]
where $S_J=\prod_{j\in J} S_j$ and $S_{J^c}=\prod_{j\in J^c} S_j$.

Define the marginal density of $(X_{kI}^{c},X_{kJ}^{d})$ by
\[
g_k(x_I^{c},x_J^{d})
=
\sum_{x_{J^c}^{d}\in S_{J^c}}
\int_{\mathrm{R}^{p^c-|I|}}
f_k(x^{c},x^{d})\,dx_{I^c}^{c},
\qquad (x_I^{c},x_J^{d})\in \mathrm{R}^{|I|}\times S_J.
\]
Then
\[
\mathrm{UNL}(f_1,\dots,f_K)
\ge
\mathrm{UNL}(g_1,\dots,g_K).
\]
\end{property}

\begin{proof}
By definition of UNL for mixed-type variables,
\[
\mathrm{UNL}(f_1,\dots,f_K)
=
\sum_{x_J^{d}\in S_J}\sum_{x_{J^c}^{d}\in S_{J^c}}
\int_{\mathrm{R}^{|I|}}\int_{\mathrm{R}^{p^c-|I|}}
\max_{1\le k\le K}
f_k(x_I^{c},x_J^{d})\,
dx_{I^c}^{c}\,dx_I^{c}.
\]
On the other hand,
\[
\mathrm{UNL}(g_1,\dots,g_K)
=
\sum_{x_J^{d}\in S_J}
\int_{\mathrm{R}^{|I|}}
\max_{1\le k\le K} g_k(x_I^{c},x_J^{d})\,dx_I^{c}
=
\sum_{x_J^{d}\in S_J}
\int_{\mathrm{R}^{|I|}}
\max_{1\le k\le K}
\sum_{x_{J^c}^{d}\in S_{J^c}}
\int_{\mathrm{R}^{p^c-|I|}}
f_k(\cdot)\,dx_{I^c}^{c}\,dx_I^{c}.
\]

Fix arbitrary $(x_I^{c},x_J^{d})\in\mathrm{R}^{|I|}\times S_J$ and define, for
each $k=1,\dots,K$,
\[
h_k(x_{I^c}^{c},x_{J^c}^{d})
=
f_k\bigl((x_I^{c},x_{I^c}^{c}),(x_J^{d},x_{J^c}^{d})\bigr),
\qquad
(x_{I^c}^{c},x_{J^c}^{d})\in \mathrm{R}^{p^c-|I|}\times S_{J^c}.
\]
For every $(x_{I^c}^{c},x_{J^c}^{d})$ we have
$\max_{1\le k\le K} h_k(x_{I^c}^{c},x_{J^c}^{d})\ge h_k(x_{I^c}^{c},x_{J^c}^{d})$
for each $k$. Hence,
\[
\sum_{x_{J^c}^{d}\in S_{J^c}}
\int_{\mathrm{R}^{p^c-|I|}}
\max_{1\le k\le K} h_k(x_{I^c}^{c},x_{J^c}^{d})\,dx_{I^c}^{c}
\ge
\sum_{x_{J^c}^{d}\in S_{J^c}}
\int_{\mathrm{R}^{p^c-|I|}}
h_k(x_{I^c}^{c},x_{J^c}^{d})\,dx_{I^c}^{c}
\qquad \text{for each }k.
\]
Taking the maximum over $k$ on the right-hand side yields
\[
\sum_{x_{J^c}^{d}\in S_{J^c}}
\int_{\mathrm{R}^{p^c-|I|}}
\max_{1\le k\le K} h_k(x_{I^c}^{c},x_{J^c}^{d})\,dx_{I^c}^{c}
\ge
\max_{1\le k\le K}
\sum_{x_{J^c}^{d}\in S_{J^c}}
\int_{\mathrm{R}^{p^c-|I|}}
h_k(x_{I^c}^{c},x_{J^c}^{d})\,dx_{I^c}^{c}.
\]
Rewriting in terms of $f_k$ and $g_k$, this becomes, for each fixed
$(x_I^{c},x_J^{d})$,
\[
\sum_{x_{J^c}^{d}\in S_{J^c}}
\int_{\mathrm{R}^{p^c-|I|}}
\max_{1\le k\le K}
f_k\bigl((x_I^{c},x_{I^c}^{c}),(x_J^{d},x_{J^c}^{d})\bigr)\,dx_{I^c}^{c}
\ge
\max_{1\le k\le K} g_k(x_I^{c},x_J^{d}).
\]
Finally, integrate both sides over $x_I^{c}\in\mathrm{R}^{|I|}$ and sum over
$x_J^{d}\in S_J$ to obtain
\[
\mathrm{UNL}(f_1,\dots,f_K)
\ge
\mathrm{UNL}(g_1,\dots,g_K),
\]
which establishes the marginal monotonicity property for mixed-type variables.
\end{proof}


\noindent\textbf{Property \ref{prop:UNL-continuous}}(Transformation invariance of UNL for continuous variables).
Let $X_1,\dots,X_k\in\mathrm{R}^p$ have probability density functions $f_1,\dots,f_k$. Consider a continuously differentiable, invertible map \(\psi:\mathrm{R}^p\to\mathrm{R}^p\) with everywhere–positive Jacobian determinant
\(|J\psi(x)|>0\).  
The UNL satisfies
\[
  \text{UNL}(f_1,\dots,f_k)
  =
  \text{UNL}(f_1^{\psi},\dots,f_k^{\psi}),
\]
where 
$
  f_i^{\psi}(u)
  =
  f_i\bigl(\psi^{-1}(u)\bigr)/|J\psi^{-1}(u)|
$ 
is the density of $U_i=\psi(X_i)$.

\begin{proof}
For \(U_i=\psi(X_i)\) the corresponding densities are  
\[
  f_i^{\psi}(u)
  =
  f_i\bigl(\psi^{-1}(u)\bigr)
  \bigl|J\psi^{-1}(u)\bigr|,
  \qquad u\in\mathrm{R}^p .
\]

The underlap of the transformed variables is  
\begin{align*}
  \text{UNL}\bigl(f_1^{\psi},\dots,f_k^{\psi}\bigr)
  &= \int_{\mathrm{R}^p}
       \max_{i}
         \Bigl\{f_i\bigl(\psi^{-1}(u)\bigr)
               |J\psi^{-1}(u)|\Bigr\}
       \,du
       \notag \\[4pt]
  &=
       \int_{\mathrm{R}^p}
         \frac{\max_{i}\{f_i(x)\}}{|J(\psi(x))|}
         |J\psi(x)|\,dx \quad (\text{using} \; \,u=\psi(x),\;du=|J\psi(x)|dx)
       \notag \\[4pt]
  &= \int_{\mathrm{R}^p}\max_{i}f_i(x)dx
     =
     \text{UNL}(f_1,\dots,f_k).
\end{align*}
\end{proof}

\begin{property}[Transformation invariance of UNL for discrete variables]
\label{prop:UNL-discrete}
Let $X_1,\dots,X_k$ take values in a countable set 
${S}$ and let 
$p_i(x)=\Pr\{X_i=x\}$ be the corresponding probability mass functions.
If $\phi:{S}\to {S}$ is a bijection, then
\[
  \text{UNL}(p_1,\dots,p_k)
  \;=\;
  \text{UNL}(p_1^{\phi},\dots,p_k^{\phi}),
\]
where $p_i^{\phi}(u)=p_i\bigl(\phi^{-1}(u)\bigr)$ is the pmf of
$U_i=\phi(X_i)$ and
\[
  \text{UNL}(p_1,\dots,p_k)
  =\sum_{s\in\mathcal{S}}\max_{1\le i\le k} p_i(s).
\]
\end{property}

\begin{proof}
\[
\begin{aligned}
\text{UNL}\bigl(p_1^{\phi},\dots,p_k^{\phi}\bigr)
  &=\sum_{u\in{S}}
      \max_i p_i^{\phi}(u)
   =\sum_{u\in{S}}
      \max_i p_i\bigl(\phi^{-1}(u)\bigr)                \\[4pt]
  &=\sum_{s\in{S}}
      \max_i p_i(x)
   =\text{UNL}\bigl(p_1,\dots,p_k\bigr),
\end{aligned}
\]
since $\phi$ is a bijection and thus yields a one–to–one
reindexing of the summation over $\mathcal{S}$.
\end{proof}

\begin{property}[Transformation invariance of UNL for mixed type variables]
\label{prop:UNL-mixed}
For group $k = 1,\cdots K$, $X_k=(X_k^{c},X_k^{d})$ with
$X_k^{c}\in\mathrm{R}^{p}$ and
$X_k^{d}\in S$, set the joint probability density functions as
\(
  f_k(x^{c},x^{d})
  = f_k^{\,c}(x^{c}\mid x^{d})\,
    p_k^{\,d}(x^{d}).
\)
Let
\(
    \Psi(x^c,x^d)=\bigl(\psi(x^c),\phi(x^d)\bigr)
\)
with  \(\psi:\mathrm R^{p}\to\mathrm R^{p}\)
a continuously differentiable, invertible map with everywhere–positive Jacobian determinant \(| J\psi(z)|>0\)
and \(\phi: S \to S\) a bijection. Then
\[
  \text{UNL}(f_1,\dots,f_k)
  =
  \text{UNL}(f_1^{\Psi},\dots,f_k^{\Psi}),
\]
where \(
  f_i^{\Psi}(u,s')
  =f_i(\psi^{-1}(u),\phi^{-1}(s'))/| J\psi(u)|
\) is the joint density of $(U_i,S'_i)=(\psi(X_i),\phi(X_i))$
\end{property}
\begin{proof}
\[
\begin{aligned}
\text{UNL}(f_1^{\Psi},\dots,f_k^{\Psi})
  &=\sum_{s'\in S}
     \int_{\mathrm R^{p}}
       \max_{i}
       \Bigl\{
         f_i(\psi^{-1}(u),\phi^{-1}(s'))
         | J\psi^{-1}(u)|
       \Bigr\}
     du                                            \\[4pt]
  &=
     \sum_{x^d\in S}
     \int_{\mathrm R^{p}}
       \frac{\max_{i}f_i(x^c,x^d)}{ |J\psi(x^c)|}
       | J\psi(x^c)|
     dx^c   \quad (\text{using} \; u=\psi(x^c),\,du=|J\psi(x^c)|dx^c)                                       \\[4pt]
  &=\sum_{x^d\in S}
     \int_{\mathrm R^{p}}
       \max_{i}f_i(x^c,x^d)dz
   =\text{UNL}(f_1,\dots,f_k).
\end{aligned}
\]
\end{proof}

\noindent\textbf{Property \ref{UNL-linear-reduction}}(Monotonicity of UNL for continuous variables under linear dimension reduction).
Let $A\in\mathrm{R}^{q\times p}$ have full row rank $q\le p$. For each $k=1,\dots,K$, let $f_k$ be the probability density
of $X_k$ on $\mathrm{R}^p$, and let $g_k$ be the probability density
of $AX_k$ on $\mathrm{R}^q$. Then
\[
  \mathrm{UNL}(f_1,\dots,f_K)\ge\mathrm{UNL}(g_1,\dots,g_K).
\]

\begin{proof}
Choose an invertible matrix
$T\in\mathrm{R}^{p\times p}$ whose first $q$ rows is identical of $A$. Then we could write $A=BT,\;B=[\,I_q\;\;0\,]$, where $I_q$ is the $q\times q$ identity matrix.

By the transformation invariance of UNL in Property \ref{prop:UNL-continuous}, then
\begin{align}
\label{eq:dimension_reduct1}
\mathrm{UNL}(f_1,\dots,f_K)=\mathrm{UNL}(f_1^T,\dots,f_K^T).    
\end{align}

For $k=1,\dots,K$, $TX_k = (U_{1k},U_{2k})$ where $U_{1k}\in\mathrm{R}^q$ denotes the first $q$ coordinates
and $U_{2k}\in\mathrm{R}^{p-q}$ denotes the remaining coordinates, $g_k$ is then the probability density
of $U_{1k}$ on $\mathrm{R}^q$. By the marginal monotonicity of UNL in Property \ref{marginal_monotonicity_UNL_property_continuous}, 
\begin{align}
\label{eq:dimension_reduct2}
\mathrm{UNL}(f_1^T,\dots,f_K^T)\le \mathrm{UNL}(g_1,\dots,g_K).    
\end{align}

Combining \eqref{eq:dimension_reduct1} and \eqref{eq:dimension_reduct2} yields
\[
  \mathrm{UNL}(f_1,\dots,f_K)
  \;\ge\;
  \mathrm{UNL}(g_1,\dots,g_K),
\]
as claimed.
\end{proof}

\begin{property}[Invariance of the UNL upon adding new groups built by mixing original groups]
Let $X_1,\dots,X_k\in\mathrm{R}^p$ have probability density functions 
$f_1,\dots,f_k$. Consider a new density constructed as an arbitrary mixture of the existing ones,
\[
  f_{K+1}(x)
  =
  \sum_{k=1}^{K} w_k f_k(x),
  \quad
  w_k\ge 0,
  \quad
  \sum_{k=1}^{K} w_k = 1.
\]
\[
  \text{UNL}\!\bigl(f_{1},\dots,f_{K},f_{K+1}\bigr)
  =
  \text{UNL}\!\bigl(f_{1},\dots,f_{K}\bigr)
  \le K.
\]

\end{property}

\begin{proof}
For every \(x\in\mathcal X\) the inequality
\[
  f_{K+1}(x)
  =
  \sum_{k=1}^{K} w_k\,f_k(x)
  \le
  \max_{1\le k\le K} f_k(x)
  = M(x)
\]
holds because each weight \(w_k\) is non–negative and bounded by 1.  Consequently,
\[
  \max\bigl\{f_1(x),\dots,f_K(x),f_{K+1}(x)\bigr\}=M(x)
  \quad\text{for all }x\in\mathcal X.
\]
Integrating both sides over~\(\mathcal X\) shows that adding the mixture leaves the underlap unchanged,
\[
  \text{UNL}\!\bigl(f_{1},\dots,f_{K},f_{K+1}\bigr)
  =
  \text{UNL}\!\bigl(f_{1},\dots,f_{K}\bigr)
  \le K.
\]
\end{proof}

Analogous invariance results for discrete and mixed type variables can be obtained by straightforward adaptations of the same argument and are therefore omitted. Introducing a new group whose distribution is merely a resampling of the existing distributions can never create a region in \(\mathcal X\) where it dominates the original pointwise maximum. Hence the overall degree of separation, as quantified by the UNL, remains exactly the same.  Observing no increase in the UNL after adding an alleged new group indicates that the added group is not statistically distinct from the existing groups. Conversely, only a genuinely new distribution could raise the UNL above its previous value.

\section{Properties of the importance sampling estimator of UNL}

\noindent\textbf{Property~\ref{dense_unl_error}} (Error bound from density estimation and plug-in consistency of the UNL).
Let \(f_1,\ldots,f_K\) be densities with respect to \(\nu\). For sample size \(n\), let
\(\widehat f_{1,n},\ldots,\widehat f_{K,n}\) be estimated densities with respect
to the same dominating measure \(\nu\). 

Then
\[
\left|
\operatorname{UNL}(\widehat f_{1,n},\ldots,\widehat f_{K,n})
-
\operatorname{UNL}(f_1,\ldots,f_K)
\right|
\le
\sum_{k=1}^K
\|\widehat f_{k,n}-f_k\|_1,
\]
where
\[
\|\widehat f_{k,n}-f_k\|_1
=
\int_{\mathcal X}|\widehat f_{k,n}(x)-f_k(x)|\,d\nu(x).
\]

Consequently, if $\widehat f_{k,n}$ is a consistent estimator of $f_k$ for $k=1,\cdots,K$: $\|\widehat f_{k,n}-f_k\|_1\xrightarrow{p}0$,
then $\mathrm{UNL}(\widehat f_{1,n},\ldots,\widehat f_{K,n})$ is a consistent estimator of $\mathrm{UNL}(f_1,\ldots,f_K)$: $\mathrm{UNL}(\widehat f_{1,n},\ldots,\widehat f_{K,n})
\xrightarrow{p}
\mathrm{UNL}(f_1,\ldots,f_K).$

\begin{proof}
By the definition of the UNL and the triangle inequality,

\begin{align*}
\left|
\operatorname{UNL}(\widehat f_{1,n},\ldots,\widehat f_{K,n})
-
\operatorname{UNL}(f_1,\ldots,f_K)
\right| & =
\left|
\int_{\mathcal X}
\left\{
\max_{1\le k\le K}\widehat f_{k,n}(x)
-
\max_{1\le k\le K} f_k(x)
\right\}
\,d\nu(x)
\right| \\
&\le
\int_{\mathcal X}
\left|
\max_{1\le k\le K}\widehat f_{k,n}(x)
-
\max_{1\le k\le K} f_k(x)
\right|
\,d\nu(x).
\end{align*}

As $\left|
\max_{1\le k\le K}\widehat f_{k,n}(x)
-
\max_{1\le k\le K}f_k(x)
\right|
\le
\max_{1\le k\le K}\left|
\widehat f_{k,n}(x)-\widehat f_k(x)
\right|
\le
\sum_{k=1}^K |\widehat f_{k,n}(x)-f_k(x)|,$
we obtain:
\[
\begin{aligned}
\left|
\operatorname{UNL}(\widehat f_{1,n},\ldots,\widehat f_{K,n})
-
\operatorname{UNL}(f_1,\ldots,f_K)
\right| \le
\int_{\mathcal X}
\sum_{k=1}^K |\widehat f_{k,n}(x)-f_k(x)|\,d\nu(x)
=
\sum_{k=1}^K
\|\widehat f_{k,n}-f_k\|_1.
\end{aligned}
\]
If
\[
\|\widehat f_{k,n}-f_k\|_1\xrightarrow{p}0 \quad \text{as } n \rightarrow \infty,
\quad
\text{for each } k=1,\ldots,K, 
\]
then, since \(K\) is fixed,
\[
\sum_{k=1}^K
\|\widehat f_{k,n}-f_k\|_1
\xrightarrow{p}0 \quad \text{as } n \rightarrow \infty.
\]
Hence, for every \(\varepsilon>0\),
\[
\Pr\left(
\left|
\operatorname{UNL}(\widehat f_{1,n},\ldots,\widehat f_{K,n})
-
\operatorname{UNL}(f_1,\ldots,f_K)
\right|
>\varepsilon
\right)
\le
\Pr\left(
\sum_{k=1}^K
\|\widehat f_{k,n}-f_k\|_1
>\varepsilon
\right)
\to 0.
\]
Therefore,
\[
\operatorname{UNL}(\widehat f_{1,n},\ldots,\widehat f_{K,n})
\xrightarrow{p}
\operatorname{UNL}(f_1,\ldots,f_K) \quad \text{as } n \rightarrow \infty.
\]
\end{proof}

\noindent\textbf{Property~\ref{imp_unl_error}} (Monte Carlo error bound and consistency of the importance-sampling estimator).
Let \(f_{1},\ldots, f_{K}\) be densities with respect to \(\nu\). Based on the equal weighted mixture proposal, $q(x)=1/K\sum_{k=1}^K  f_{k}(x)$, let
\(X_1,\ldots,X_M\) be iid samples from \(q\), the importance sampling estimator of UNL is:
\[\widehat{\operatorname{UNL}}_M
=
\frac1M
\sum_{i=1}^M
\frac{
\max_{1\le k\le K} f_{k}( X_i)
}{
\ q( X_i)
}.\]

\(\widehat{\operatorname{UNL}}_M\) is an unbiased estimator of
\(\operatorname{UNL}(f_{1},\cdots,f_{K})\). Moreover, for every \(\varepsilon>0\),
\[
\Pr\left(
\left|
\widehat{\operatorname{UNL}}_M
-
\operatorname{UNL}(f_{1},\cdots,f_{K})
\right|
\ge \varepsilon
\right)
\le
2\exp\left\{
-\frac{2M\varepsilon^2}{(K-1)^2}
\right\}.
\]
Equivalently, with probability at least \(1-\delta\),
\[
\left|
\widehat{\operatorname{UNL}}_M
-
\operatorname{UNL}(f_{1},\cdots,f_{K})
\right|
\le
(K-1)
\sqrt{\frac{\log(2/\delta)}{2M}}.
\]
Consequently, $\widehat{\operatorname{UNL}}_M$ is a consistent estimator of $\operatorname{UNL}(f_{1},\cdots,f_{K})$: $\widehat{\operatorname{UNL}}_M \xrightarrow{p} \operatorname{UNL}(f_{1},\cdots,f_{K})$ as $M\rightarrow \infty$.

\begin{proof}
First, by the nature of the importance sampling estimator,
\[
\mathrm E_q\left[
\frac{\max_{1\le k\le K}f_{k}(X)}{q(X)}
\right]
=
\int_{\mathcal X}
\frac{\max_{1\le k\le K}f_{k}(x)}{q(x)}
q(x)\,d\nu(x)
=
\operatorname{UNL}(f_{1},\cdots,f_{K}),
\]
\(\widehat{\operatorname{UNL}}_M\) is unbiased.

We denote $W_i={\max_{1\le k\le K} f_{k}(X_i)}/{q(X_i)}$ for $i = 1,2,\cdots , M$. Under the equally weighted mixture proposal, $q(x)=1/K \sum_{k=1}^K f_{k}(x)$, we have \(1\le W_i\le K\). Since $W_i$ are independent and bounded, Hoeffding's inequality gives
\[
\Pr\left(
\left|
\frac1M\sum_{i=1}^M W_i-\mathrm E(W_i)
\right|
\ge \varepsilon
\right)=
\Pr\left(
\left|
\widehat{\operatorname{UNL}}_M
-
\operatorname{UNL}(f_{1},\cdots,f_{K})
\right|
\ge \varepsilon
\right)
\le
2\exp\left\{
-\frac{2M\varepsilon^2}{(K-1)^2}
\right\}
.
\]
Solving the inequality
$2\exp\left\{
-{2M\varepsilon^2}/{(K-1)^2}
\right\}
\le \delta$ for \(\varepsilon\) gives
\[\Pr\left(\left|
\widehat{\operatorname{UNL}}_M
-
\operatorname{UNL}(f_{1},\cdots,f_{K})
\right|
\le
(K-1)
\sqrt{\frac{\log(2/\delta)}{2M}}\right)\ge 1-\delta.\]

Finally, for any fixed \(\varepsilon>0\), $2\exp\left\{
-{2M\varepsilon^2}/{(K-1)^2}
\right\}
\to 0,
\quad \text{as } M\to\infty.$ Therefore, $\widehat{\operatorname{UNL}}_M
\xrightarrow{p}
\operatorname{UNL}(f_{1},\cdots,f_{K})$, as $ M\to\infty$.

\end{proof}

\noindent\textbf{Property~\ref{variance_bound_imp}} (Monte Carlo variance bound for the importance sampling estimator).
Let \(f_{1},\ldots, f_{K}\) be densities with respect to \(\nu\). Based on the equal weighted mixture proposal, $q(x)=1/K\sum_{k=1}^K  f_{k}(x)$, let
\(X_1,\ldots,X_M\) be iid samples from \(q\), the importance sampling estimator of UNL is:
\[\widehat{\operatorname{UNL}}_M
=
\frac1M
\sum_{i=1}^M
\frac{
\max_{1\le k\le K} f_{k}( X_i)
}{
\ q( X_i)
}.\]
It gives the variance bound:
\[
\operatorname{Var}\left(\widehat{\operatorname{UNL}}_M\right)
\le
\frac{\operatorname{UNL}(K-\operatorname{UNL})}{M}.
\]

\begin{proof}
The variance of the importance estimator of UNL could be written as:
\begin{align}
\label{exact_variance_imp1}
\mathrm{Var}\Bigl(\widehat{\text{UNL}}_M\Bigr)&=\frac{1}{M}\mathrm{Var} \Bigl(\frac{\max_{1\le k\le K} f_k(x_i)}{q(x_i)}\Bigr) \nonumber\\
&=\frac{1}{M}\Bigl[\mathrm{E}_{q} \Bigl(\frac{(\max_{1\le k\le K} f_k(x_i))^2}{(q(x_i))^2}\Bigr)-\Bigl(\mathrm{E}_{q} \Bigl(\frac{\max_{1\le k\le K} f_k(x_i)}{q(x_i)}\Bigr)\Bigr)^2\Bigr]  \nonumber\\
&=\frac{1}{M}\Bigl[\int_{\mathcal X} \frac{(\max_{1\le k\le K} f_k(x))^2}{q(x)}d\nu(x)-{\text{UNL}}^2\Bigr].
\end{align}

Given that the proposal takes the form $q(x)=1/K\sum_{k=1}^{K} f_k(x)$, the exact variance of $\mathrm{Var}\Bigl(\widehat{\text{UNL}}\Bigr)$ in \eqref{exact_variance_imp1} could be written as:
\begin{align}
\label{exact_variance_imp2}
\mathrm{Var}\Bigl(\widehat{\text{UNL}}_M\Bigr)&=\frac{1}{M}\Bigl[\int_{\mathcal X} \frac{(\max_{1\le k\le K} f_k(x))^2}{\frac1K\sum_{k=1}^{K} f_k(x)}d\nu(x)-{\text{UNL}}^2\Bigr].
\end{align}

As ${\max_{1\le k\le K}f_k(x)}/
               {(K^{-1}\sum_{k=1}^{K} f_k(x))}
        \le K$, thus
\begin{align}
\label{bound_imp1}
\int_{\mathcal X} \frac{(\max_{1\le k\le K} f_k(x))^2}{\frac1K\sum_{k=1}^{K} f_k(x)}d\nu(x)\le K\int_{\mathcal X} \max_{1\le k\le K} f_k(x)d\nu(x)=K {\text{UNL}}.
\end{align}

Substituting \eqref{bound_imp1} back to \eqref{exact_variance_imp2}, we have:
\begin{align*}
\label{bound_imp2}
\mathrm{Var}\Bigl(\widehat{\text{UNL}}_M\Bigr)\le\Bigl(K {\text{UNL}}-{\text{UNL}}^2\Bigr)/M={\text{UNL}\,(K-\text{UNL})}/{M}.
\end{align*}
\end{proof}

\noindent\textbf{Property \ref{combine_imp_error}} (Combined density estimation and Monte Carlo error bound of UNL).
Let \(f_1,\ldots,f_K\) be densities with respect to \(\nu\). For sample size \(n\), let
\(\widehat f_{1,n},\ldots,\widehat f_{K,n}\) be estimated densities with respect to the same dominating measure \(\nu\). Based on the equal weighted mixture proposal, $q(x)=1/K\sum_{k=1}^K \widehat f_{k,n}(x)$, let
\(X_1,\ldots,X_M\) be iid samples from \(q\), the importance sampling estimator of UNL is:
\[\widehat {\widehat{\operatorname{UNL}}}_M
=
\frac1M
\sum_{i=1}^M
\frac{
\max_{1\le k\le K}\widehat f_{k,n}( X_i)
}{
\ q( X_i)
}.\]
Then, with probability at least \(1-\delta\),
\[
\left|
\widehat {\widehat{\operatorname{UNL}}}_M
-
\operatorname{UNL}(f_1,\ldots,f_K)
\right|
\le
(K-1)
\sqrt{\frac{\log(2/\delta)}{2M}}
+
\sum_{k=1}^K
\|\widehat f_{k,n}-f_k\|_1.
\]

Consequently, $\widehat {\widehat{\operatorname{UNL}}}_M$ is a consistent estimator of $\operatorname{UNL}(f_1,\ldots,f_K)$: $\widehat {\widehat{\operatorname{UNL}}}_M \xrightarrow{p} \operatorname{UNL}(f_1,\ldots,f_K)$ as $M \rightarrow \infty$ and $n \rightarrow \infty$.

\begin{proof}
According to the triangle inequality, we have
\[
\left|
\widehat {\widehat{\operatorname{UNL}}}_M
-
\operatorname{UNL}(f_1,\ldots,f_K)
\right| \le
\left|
\widehat {\widehat{\operatorname{UNL}}}_M
-
\operatorname{UNL}(\widehat f_{1,n},\ldots,\widehat f_{K,n})
\right|
+
\left|
\operatorname{UNL}(\widehat f_{1,n},\ldots,\widehat f_{K,n})
-
\operatorname{UNL}(f_1,\ldots,f_K)
\right|.
\]
Combining Property \ref{dense_unl_error} and Property \ref{imp_unl_error} gives the result.
\end{proof}

\begin{figure}[!t]
\centering
\includegraphics[width=0.42\textwidth]{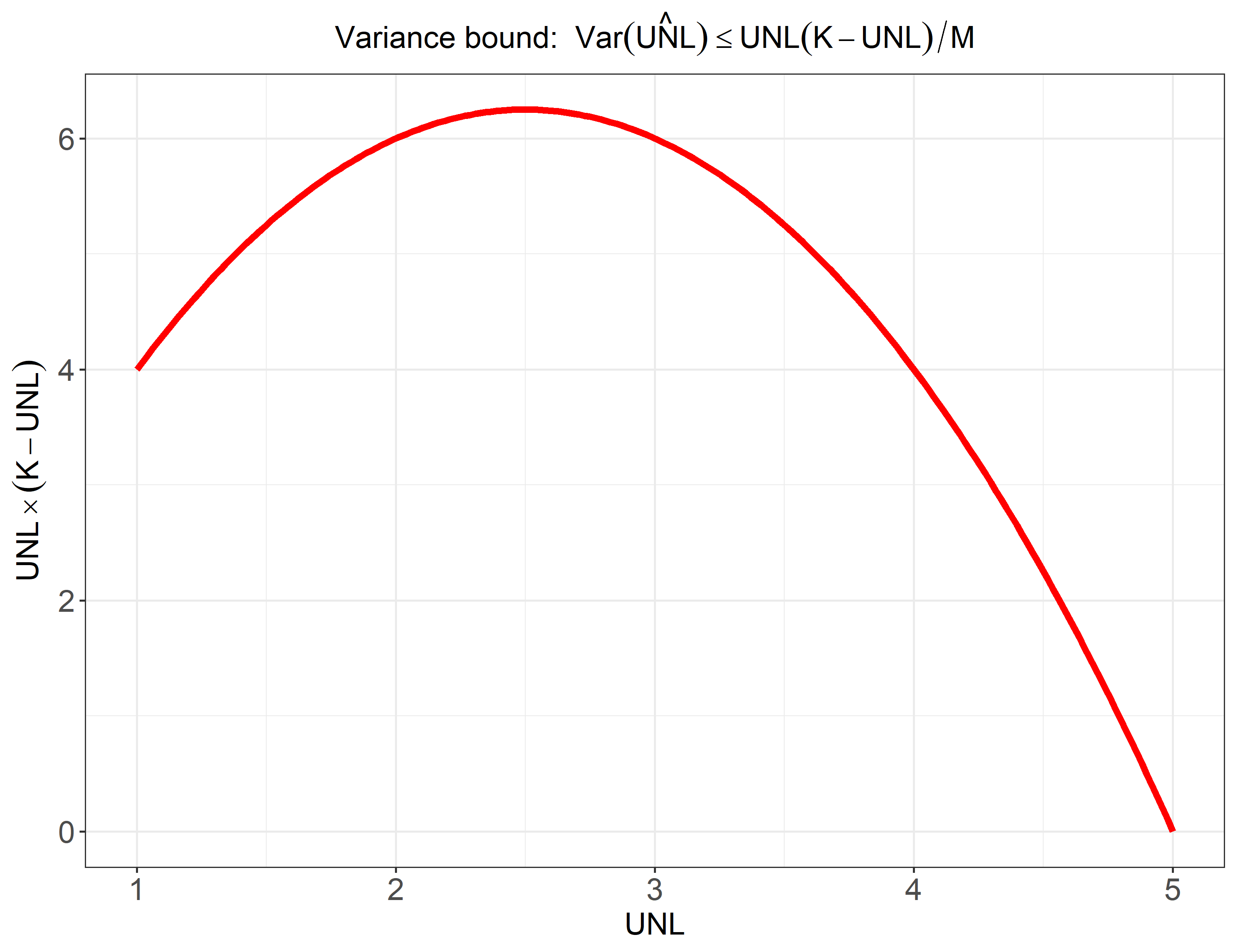}
\centering
\caption{Visual illustration of the variance bound of the importance sampling UNL estimator when $K=5$.}
\label{variance_bound_plot}
\end{figure}

\section{Prior specification of the DPM model}

We model observations $\{y_i\}_{i=1}^{N}=\{(y_i^{c},\,y^{d}_i)\}_{i=1}^N$ with a truncated approximation of DPM model of $L$ components. In this paper, we fit the DPM model using truncation level $L=10$. We denote $p^c$ as the number of dimension of $y_i^{c}$ and $p^d$ as the number of dimension of $y^{d}$. For
$i=1,\ldots,N$ and $l=1,\ldots,L$,
\begin{align*}
y_i^{c} | z_i=l,\mu_l,\Sigma_l &\iidsim \text{N}_p(\mu_l,\Sigma_l), \\
y^{d}_i | z_i=l,\{\boldsymbol\pi^{(k)}_l\}_{k=1}^{p^d}
&\iidsim \prod_{k=1}^{p^d} \mathrm{Categorical}\left(\boldsymbol\pi^{(k)}_l\right), \\
\Pr(z_i=l)&=w_l,\qquad
w_l=v_l\!\!\prod_{h<l}(1-v_h),\\
v_l \mid \alpha &\stackrel{\text{iid}}{\sim} \mathrm{Beta}(1,\alpha) \quad \text{for} \; l<L,\qquad v_L= 1.
\end{align*}

We adopt the normal-inverse Wishart conjugate prior
\begin{align*}
\mu_l &\sim \text{N}_p\!\left(m_0,\;L_0\right), \\
\Sigma_l &\sim \text{IW}\!\left(\nu_0,\;S_0\right),
\end{align*}
with $m_0$ and $L_0$ as the mean and covariance matrix of the kernel means and $\nu_0$ and $S_0$ as the degrees of freedom and scale matrix of the kernel covariance matrices.

We set $(m_0,L_0,\nu_0,S_0)$ in a data-adaptive yet weakly informative manner
using an initial K-means partition of the continuous variables $\{y_i^{c}\}$
(with small $K$, e.g., $K\in[3,10]$). Let $\widehat c_i\in\{1,\ldots,K\}$ denote cluster labels inferred by the K-means algorithm and $\widehat\mu_k$ denote the centroids of clusters inferred by the K-means algorithm.

The prior for component means is set to be centered at the sample mean,
\[
m_0= \frac{1}{N}\sum_{i=1}^N y_i^{c},
\]
and the prior covariance of the component means is set to the empirical covariance of the K-means centroids,
\[
L_0=\mathrm{Cov}\big(\widehat\mu_1,\ldots,\widehat\mu_K\big).
\]

The degree of freedom is set to $\nu_0=p+2$, unless otherwise specified.
Let
\[
\widehat\Sigma_{\mathrm{w}}=\frac{1}{K}\sum_{k=1}^K \mathrm{Cov}\left(\{y_i^{d}:\widehat c_i=k\}\right)
\]
be the average within-cluster covariance from K-means algorithm. Choose $S_0=(\nu_0+p+1)\,\widehat\Sigma_{\mathrm{w}}$ so that the
mode of the inverse-Wishart prior matches $\widehat\Sigma_{\mathrm{w}}$.

For each categorical variable $k=1,\cdots,p^d$ and $l=1,\cdots,L$,
we assign independent Dirichlet priors
\[
\boldsymbol\pi^{(k)}_l \sim \mathrm{Dirichlet}\!\big(\boldsymbol\eta^{(k)}\big),
\qquad \boldsymbol\eta^{(k)}=\tau_k\,\widehat{\boldsymbol p}^{(k)},
\]
where $\widehat{\boldsymbol p}^{(k)}$ are the empirical marginal proportions of the
$k$th categorical variable. We take $\tau_k=10$ (a prior "pseudo-count" of 10),
which is weak relative to most $N$ values.

The DP precision $\alpha$ controls the expected number of occupied components.
We use a weakly informative gamma prior
\[
\alpha \sim \mathrm{Gamma}(a_\alpha,b_\alpha),
\]
with defaults $a_\alpha=b_\alpha=2$, giving $\text{E}(\alpha)=1$ and
$\mathrm{Var}(\alpha)=0.5$, which is also used in \cite{RJ-2021-066}.

\section{Prior specification of the LDDP model}

We model $y_i \in \mathrm{R}$ with covariate vector $x_i$ of length $p$ through a truncated approximation of LDDP model of $L$ components. In this paper, we fit the DPM model using truncation level $L=20$. For $i=1,\dots,N$ and $l=1,\dots,L$,
\begin{align*}
&y_i | z_i=l, \beta_l, \tau_l \sim \text{N}\big(x_i^\top \beta_l,\tau_l^{-1}\big)\\
&\Pr(z_i=l)=w_l,\qquad
w_l=v_l\prod_{h<l}(1-v_h),\\
&v_l | \alpha \iidsim \mathrm{Beta}(1,\alpha) \quad \text{for} \; l<L,\qquad v_L= 1.\\
\end{align*}

We set a Gaussian base measure with unknown mean and covariance on component-specific regression coefficients:
\begin{align*}
\beta_l | \mu, \Sigma &\sim \text{N}_p(\mu, \Sigma), \qquad l=1,\dots,L,
\end{align*}
and we place conjugate hyperpriors on $(\mu,\Sigma)$ in order to allow for some more flexibility:
\begin{align*}
\mu &\sim \mathrm{N}_p(m_0,\ S_0),\\
\Sigma &\sim \mathrm{IW}\!\big(\nu,\ (\nu \Psi)^{-1}\big),
\end{align*}
with $m_0$ and $S_0$ as the mean and covariance matrix of normal prior and $\nu$ and $(\nu \Psi)^{-1}$ as the degrees of freedom and scale matrix of the inverse-Wishart prior.


The conjugate gamma priors are also placed on component precisions and the DP concentration parameter:
\begin{align*}
\tau_l &\sim \mathrm{Gamma}(a,b), \qquad l=1,\dots,L,\\
\alpha &\sim \mathrm{Gamma}(a_\alpha,b_\alpha).
\end{align*}

We set the hyperparameter values in a data-adaptive yet weakly informative manner using the results of ordinary least squares (OLS) regression. Let $X$ denotes the $N \times p$ design matrix of the OLS regression, $\widehat\beta_{\text{OLS}}$ and $\widehat\sigma^2$ denote the estimated coefficients and residual variance from the OLS regression. We set the following values (which are also used in \citep{RJ-2021-066}) unless otherwise specified:
\begin{align*}
&m_0 = \widehat\beta_{\text{OLS}}, \qquad S_0 = \widehat\sigma^{\,2}\,(X^\top X)^{-1},\\
&\nu = p+2, \qquad \Psi = 30\,\widehat\sigma^{\,2}\,(X^\top X)^{-1}, \\
&a = 2, \qquad b = \frac{\widehat\sigma^{\,2}}{2}, \\
&a_\alpha = 2, \qquad b_\alpha = 2. \\
\end{align*}

\clearpage

\section{Some plots for the examples in Section 3}

\begin{figure}[htbp]
\centering
\subfigure{
\centering
\includegraphics[width=0.42\textwidth]{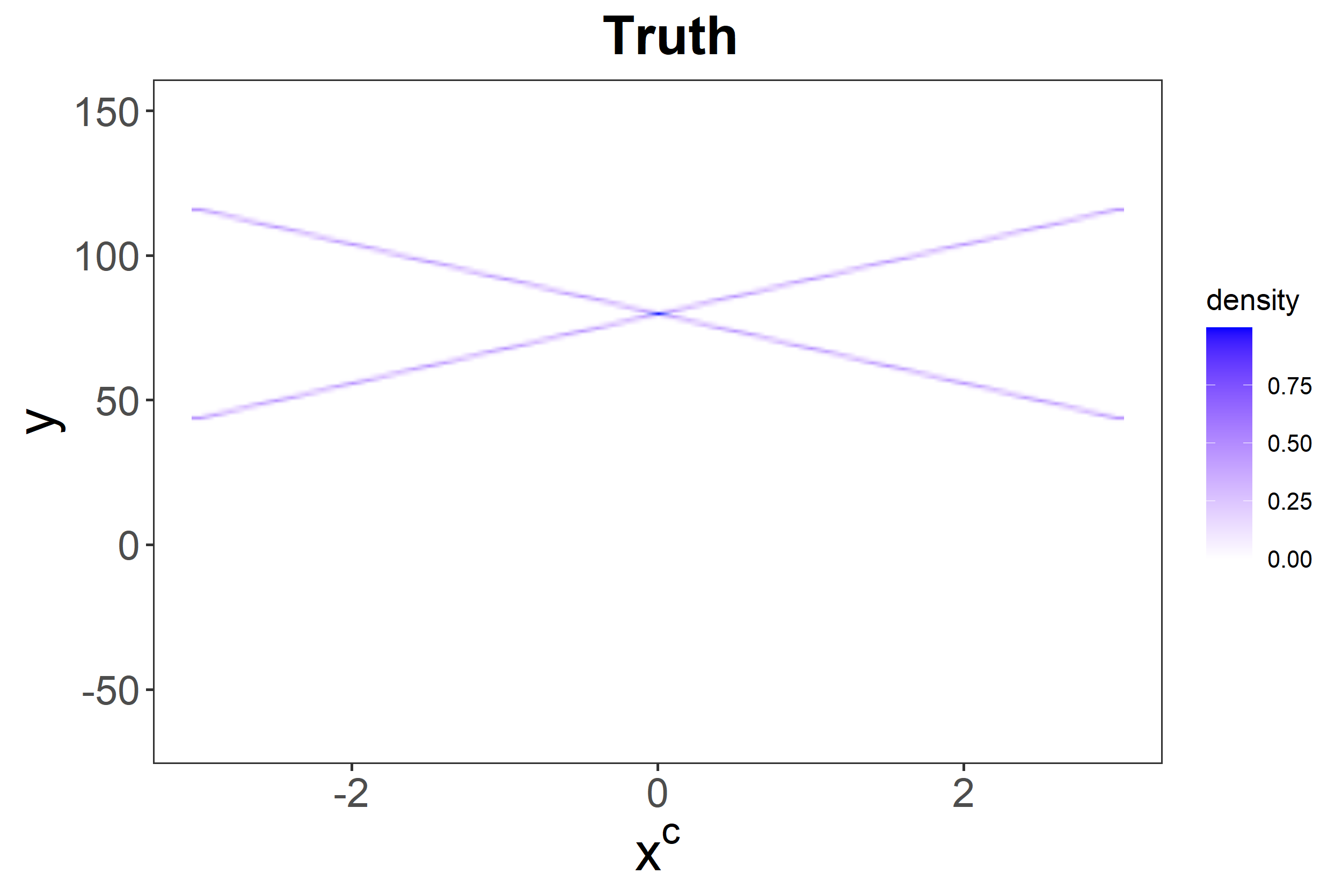}
}
\subfigure{
\centering
\includegraphics[width=0.42\textwidth]{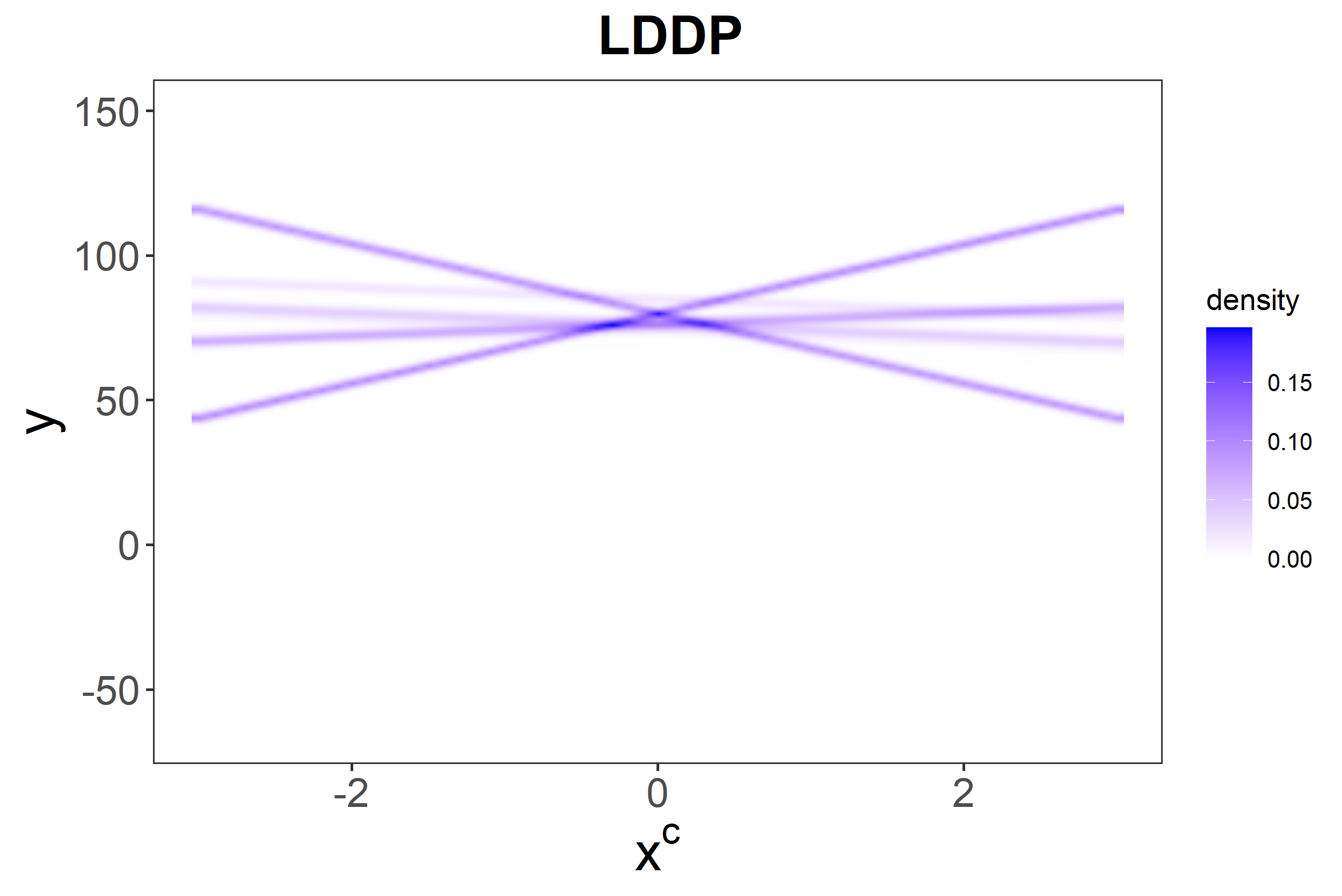}
}
\\
\subfigure{
\centering
\includegraphics[width=0.42\textwidth]{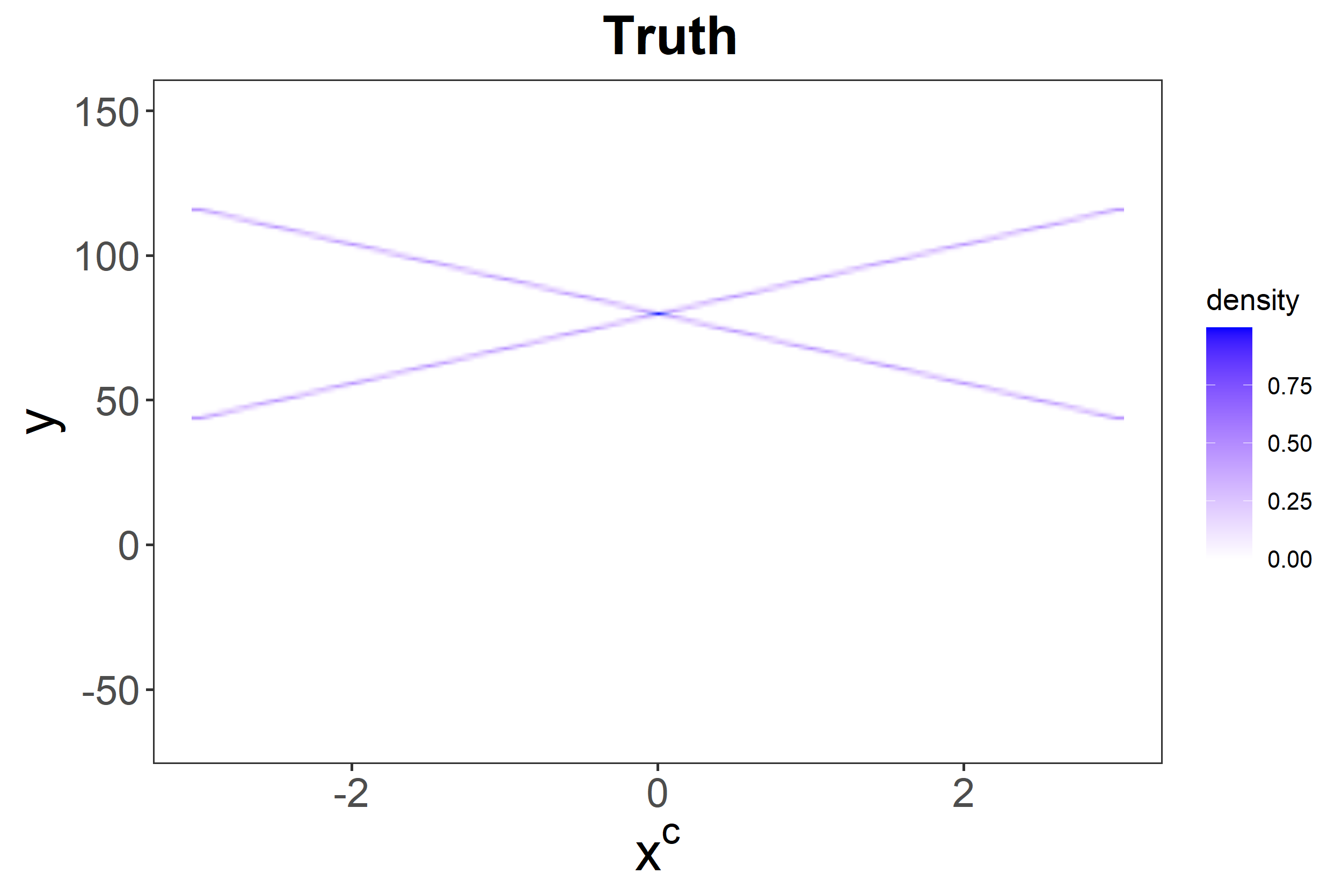}
}
\subfigure{
\centering
\includegraphics[width=0.42\textwidth]{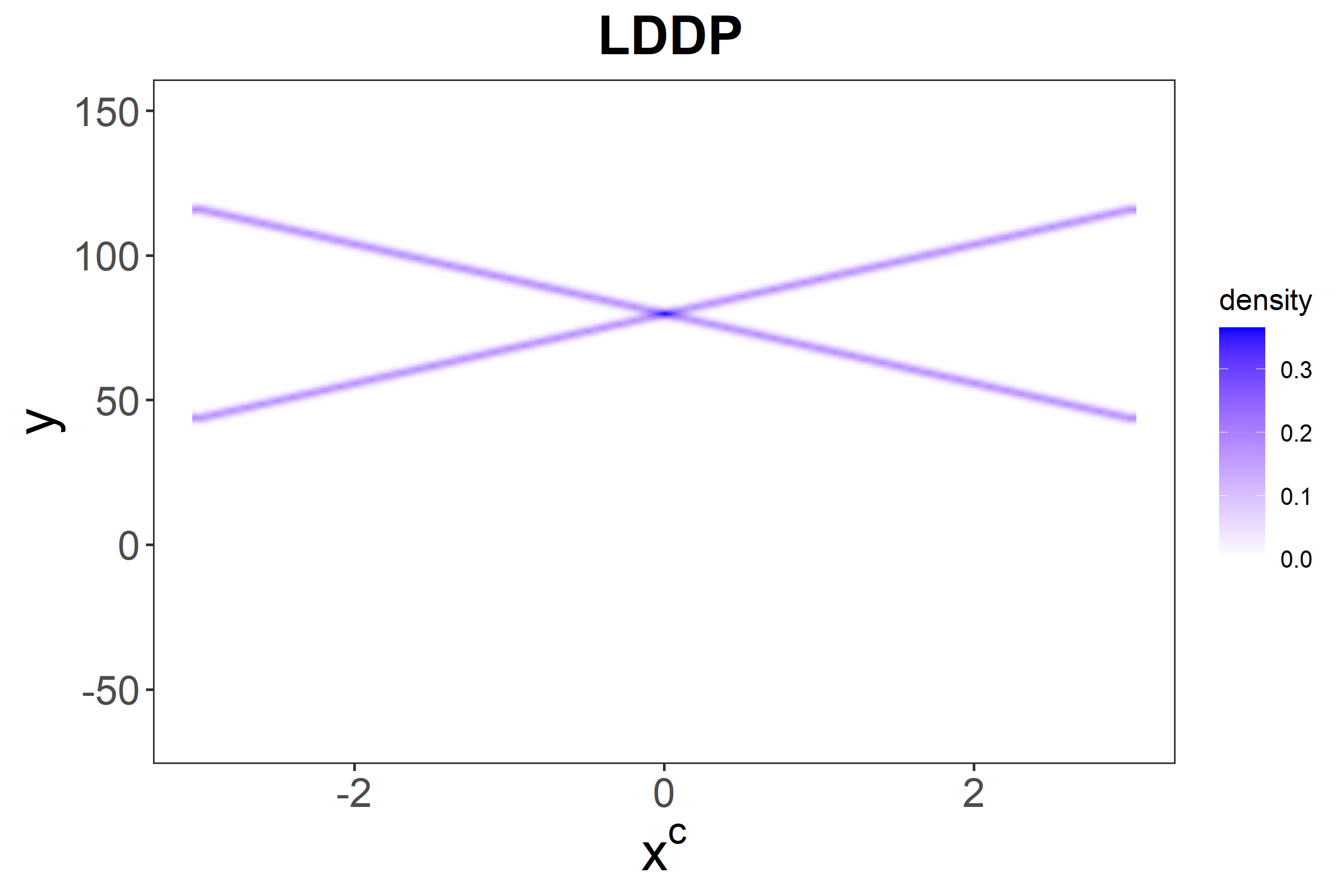}
}

\centering
\caption{Heatmap of the true and estimated density regression functions of Examples C1 and C2 conditioned on $x^d=2$. Top row: Example C1. Bottom row: Example C2.}
\label{example_C12_heatmap_plots_cate2}
\end{figure}

In the LDDP mixture model, predictions are averaged with mixing weights that do not vary with the covariates. In Example C1, the LDDP-induced partition exhibits moderate dependence on $x^d$; with covariate-constant weights, this leads to poor predictive performance. In Example C2, the partition shows little dependence on either covariate, thus covariate-constant weights do not degrade prediction.

\begin{figure}[htbp]
\centering
\subfigure{
\centering
\includegraphics[width=0.31\textwidth]{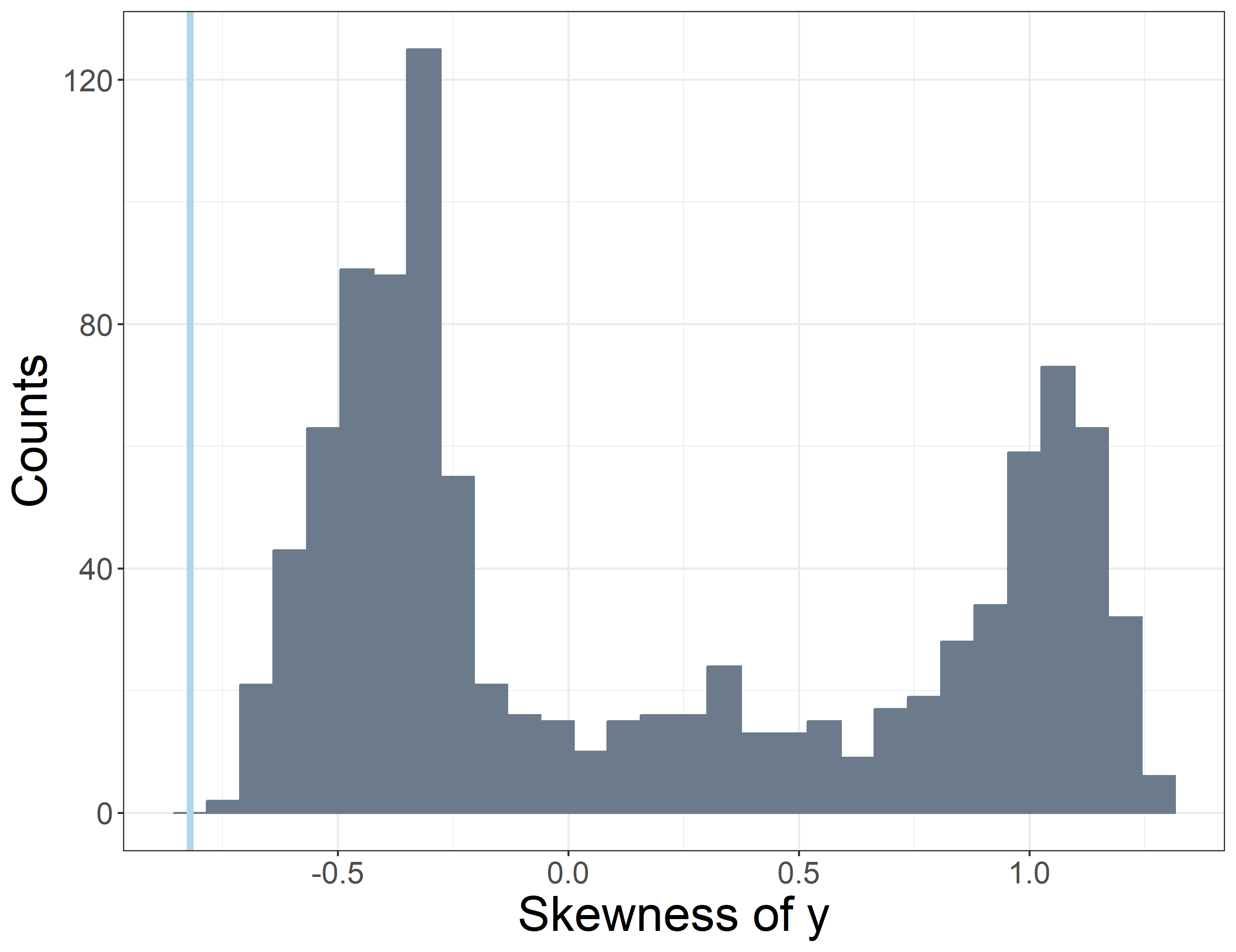}
}
\subfigure{
\centering
\includegraphics[width=0.31\textwidth]{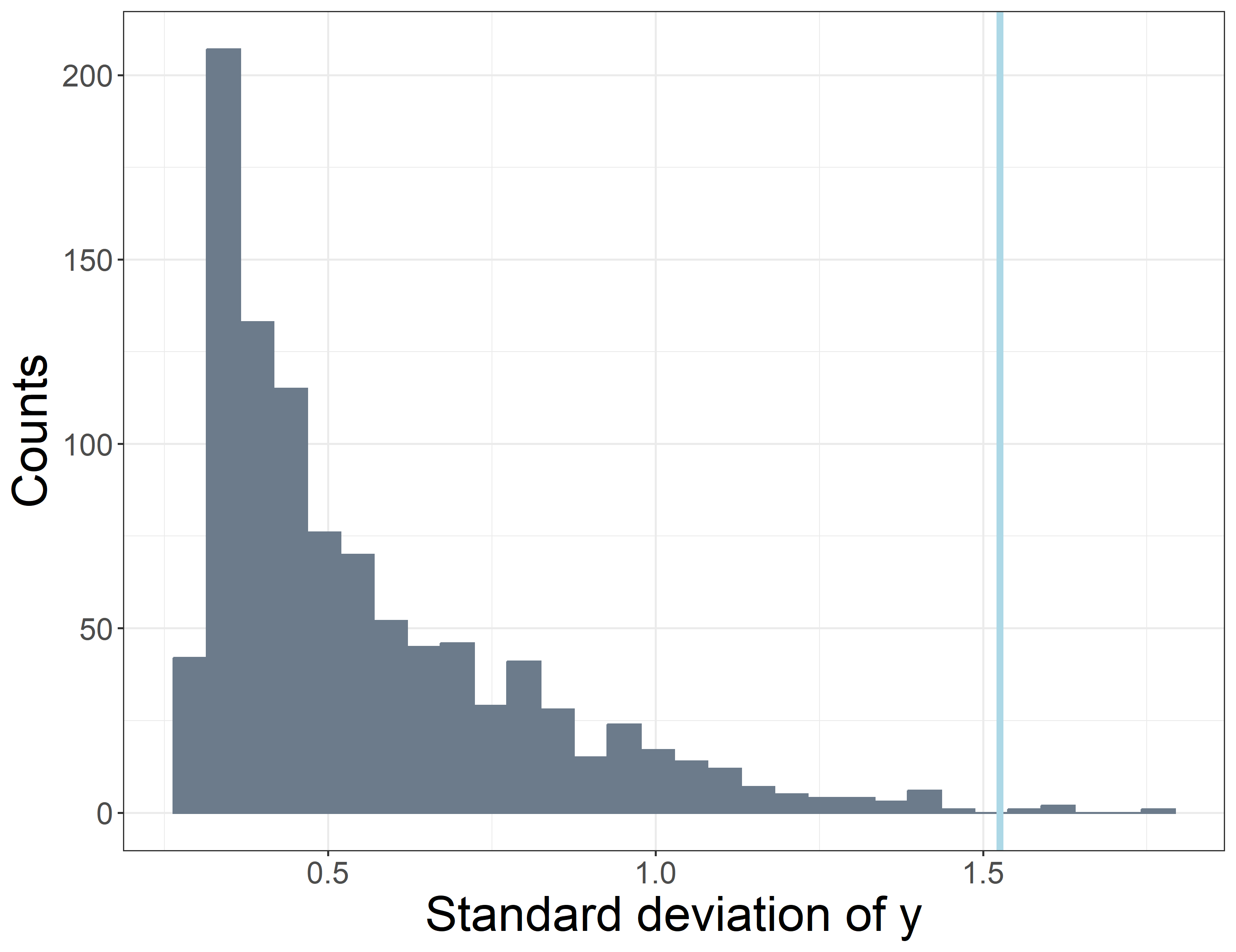}
}
\subfigure{
\centering
\includegraphics[width=0.31\textwidth]{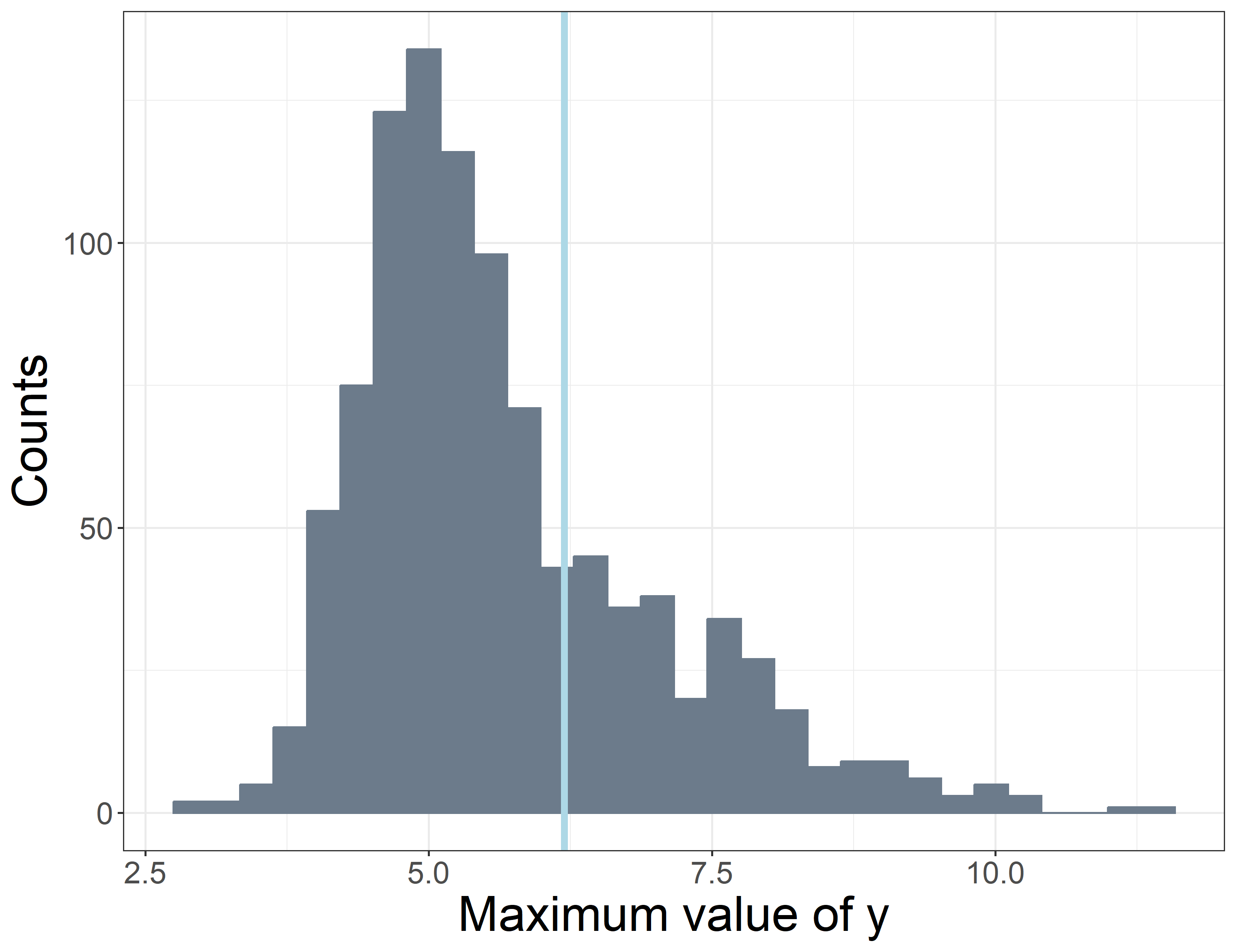}
}
\\
\subfigure{
\centering
\includegraphics[width=0.31\textwidth]{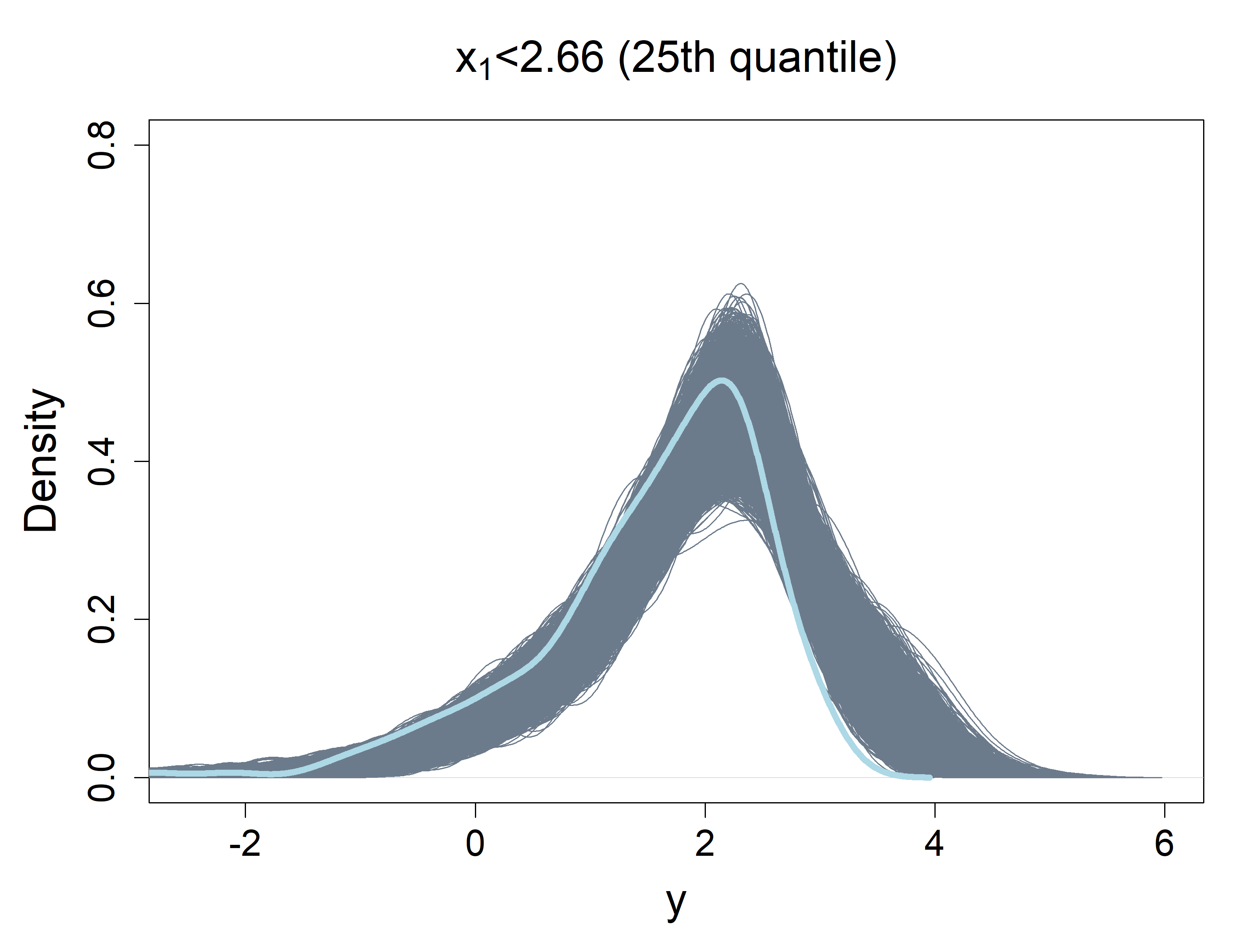}
}
\subfigure{
\centering
\includegraphics[width=0.31\textwidth]{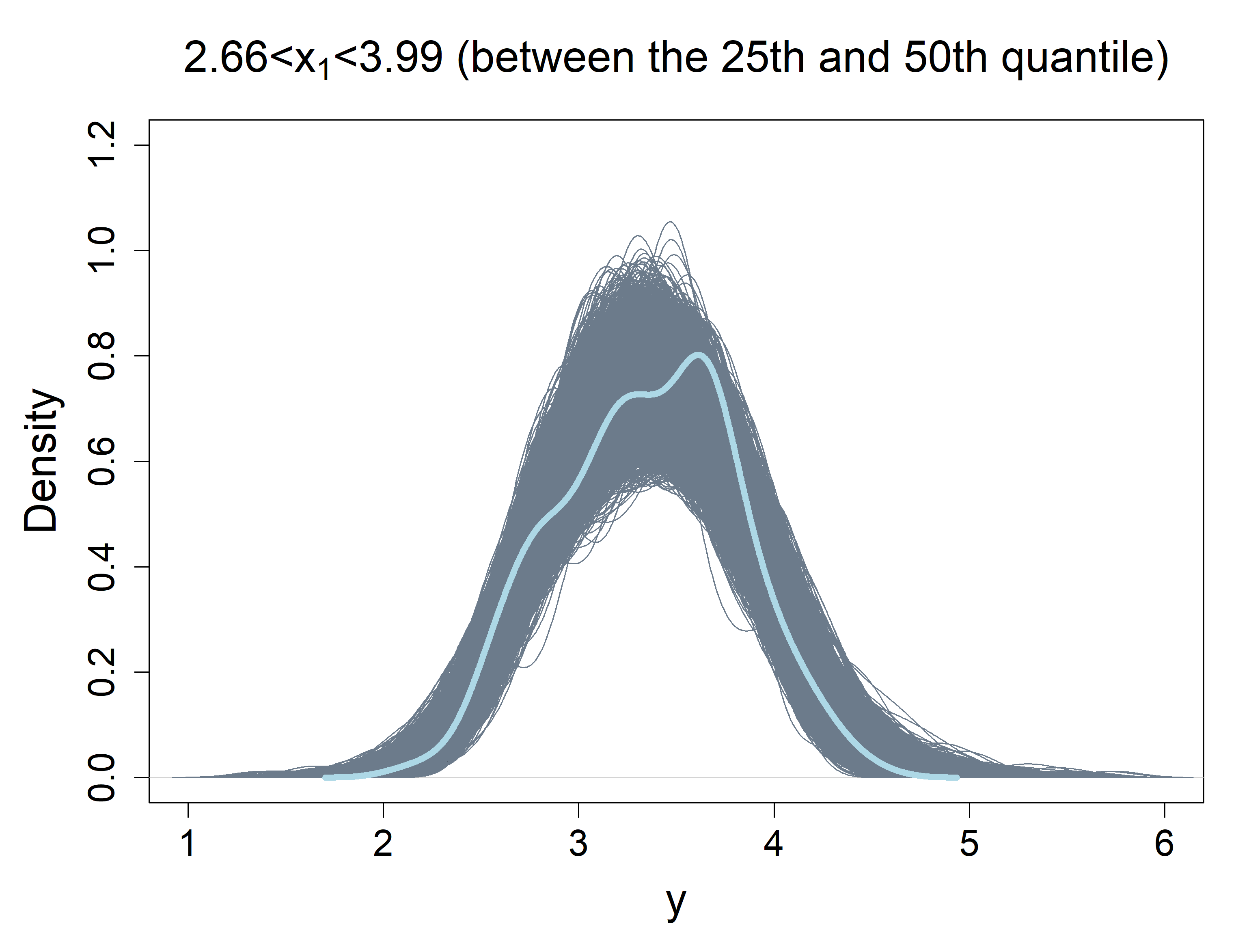}
}
\subfigure{
\centering
\includegraphics[width=0.31\textwidth]{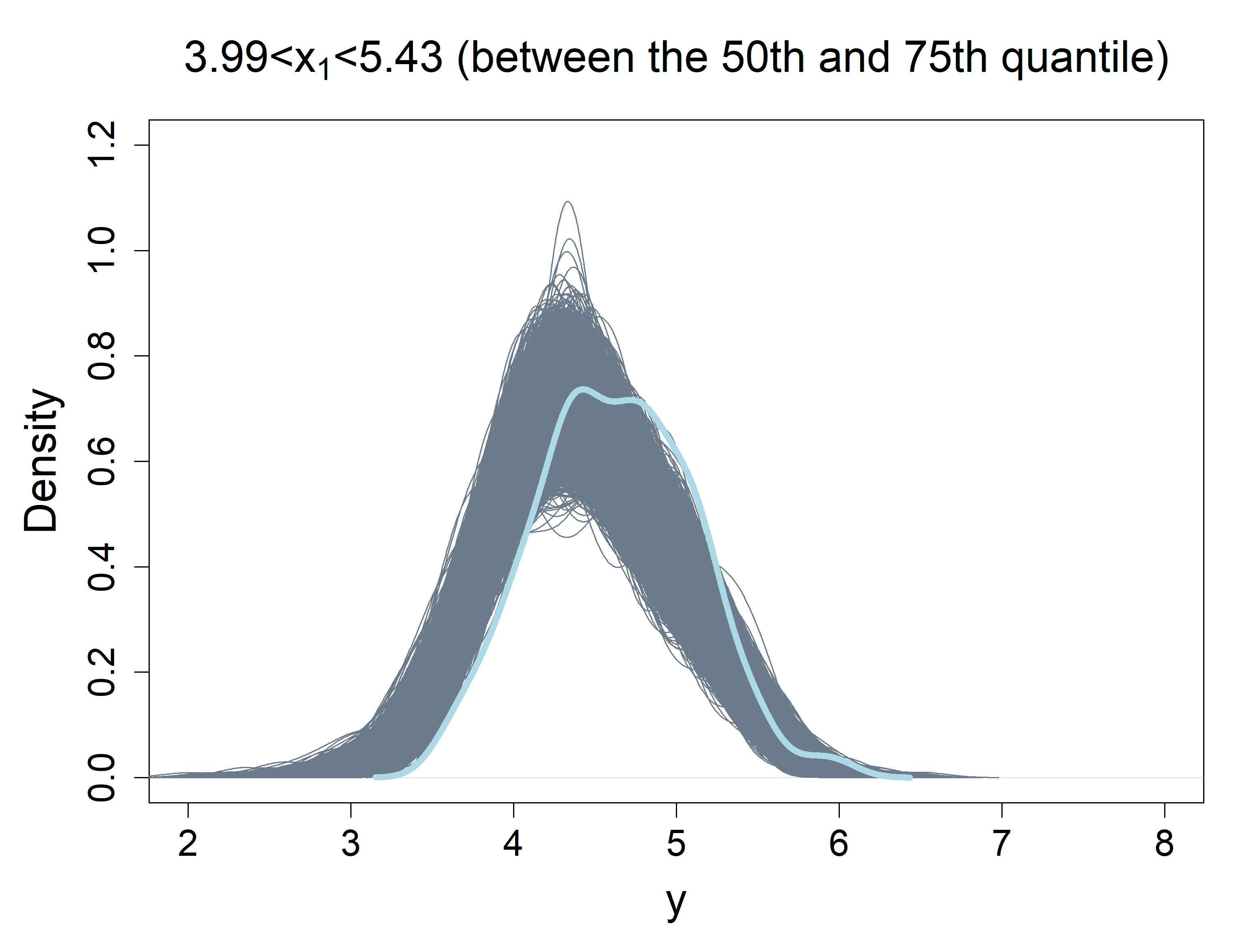}
}
\\
\subfigure{
\centering
\includegraphics[width=0.31\textwidth]{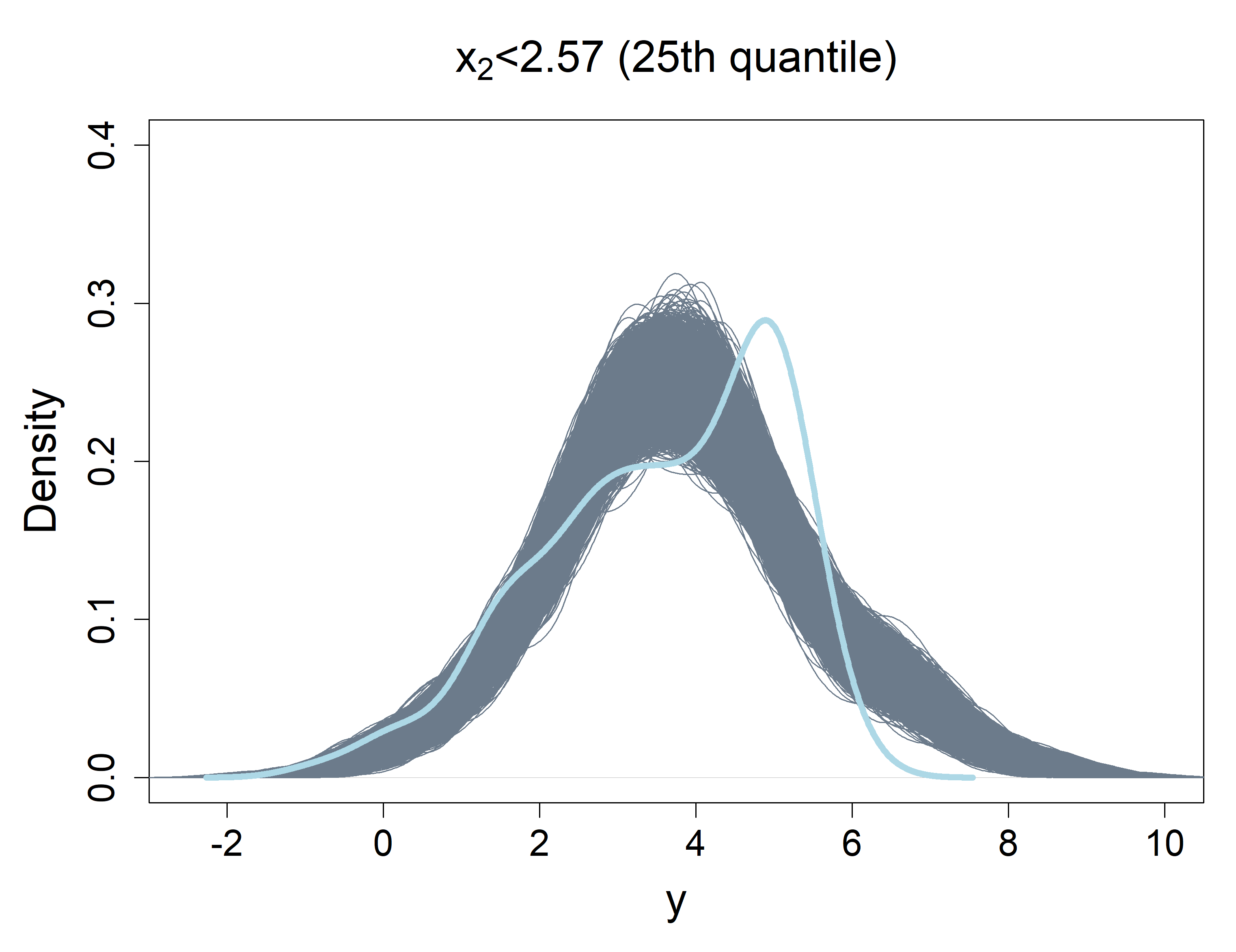}
}
\subfigure{
\centering
\includegraphics[width=0.31\textwidth]{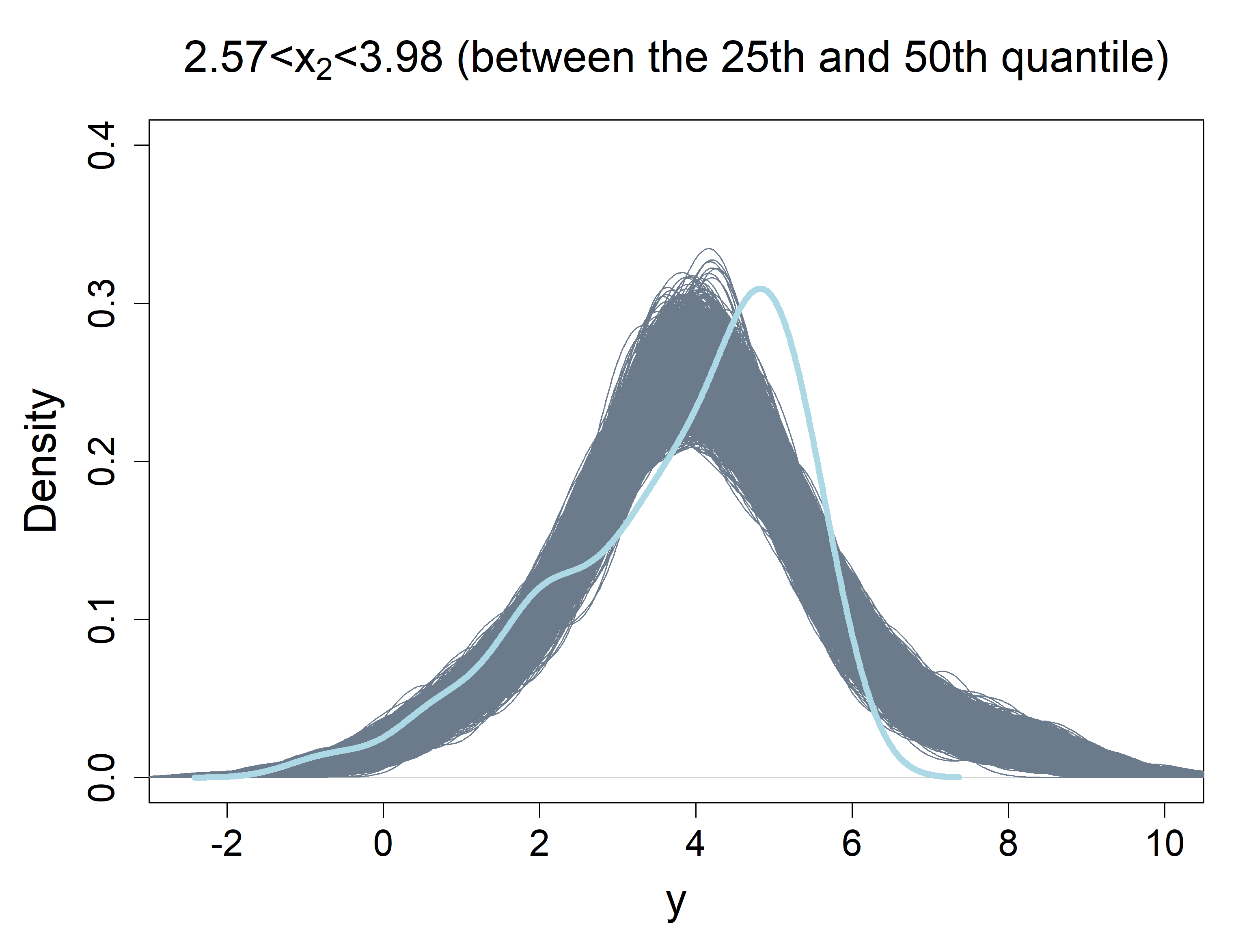}
}
\subfigure{
\centering
\includegraphics[width=0.31\textwidth]{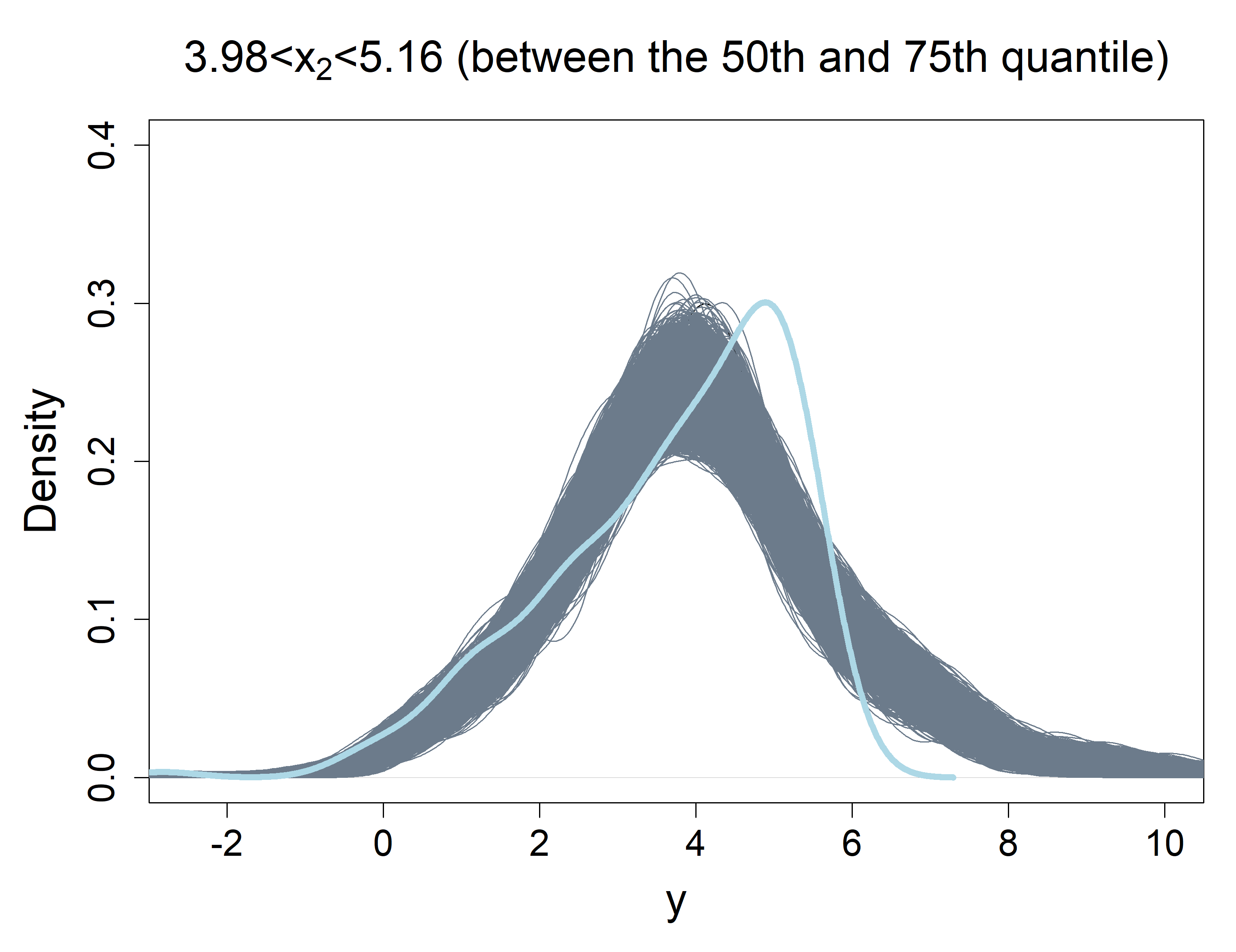}
}
\\
\subfigure{
\centering
\includegraphics[width=0.31\textwidth]{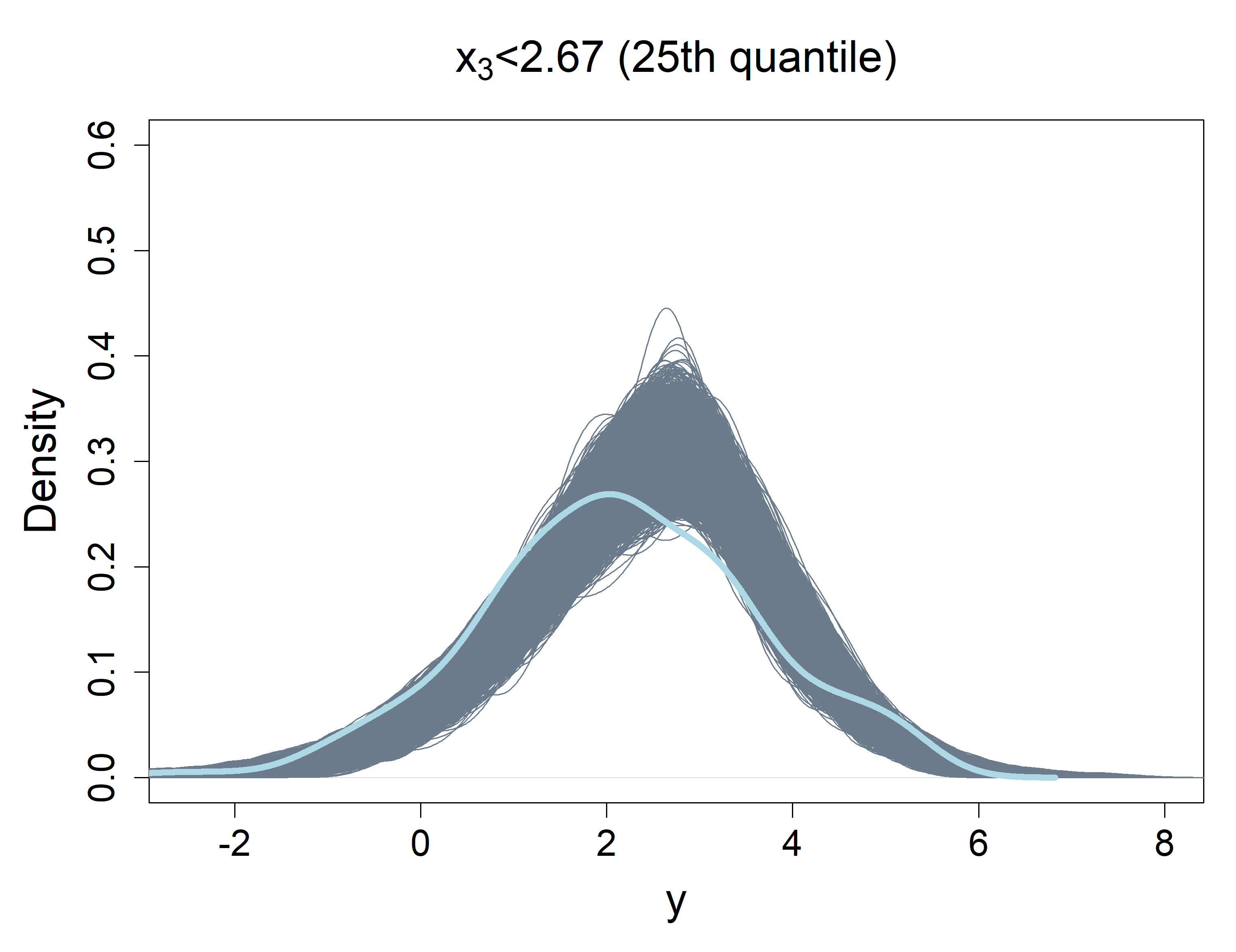}
}
\subfigure{
\centering
\includegraphics[width=0.31\textwidth]{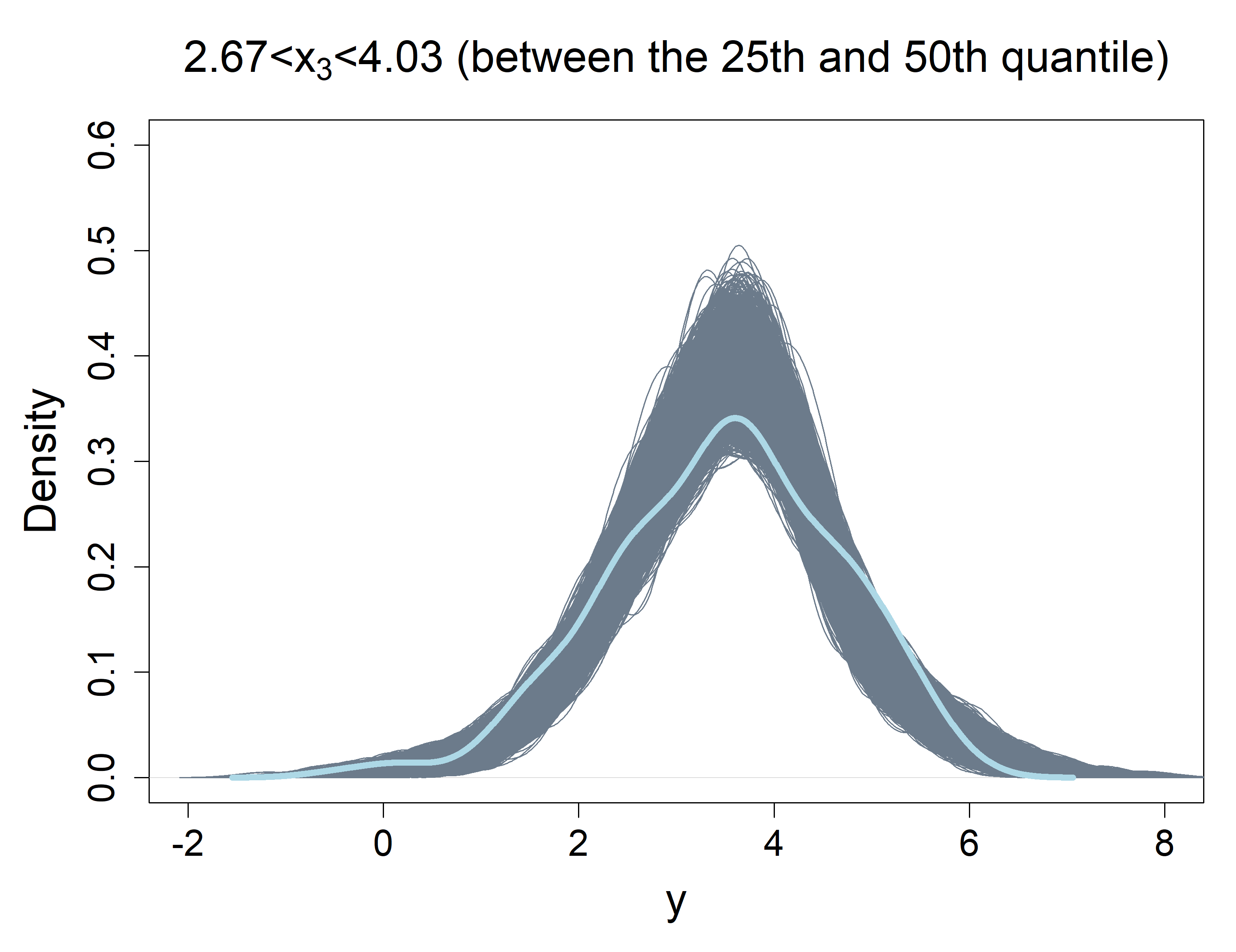}
}
\subfigure{
\centering
\includegraphics[width=0.31\textwidth]{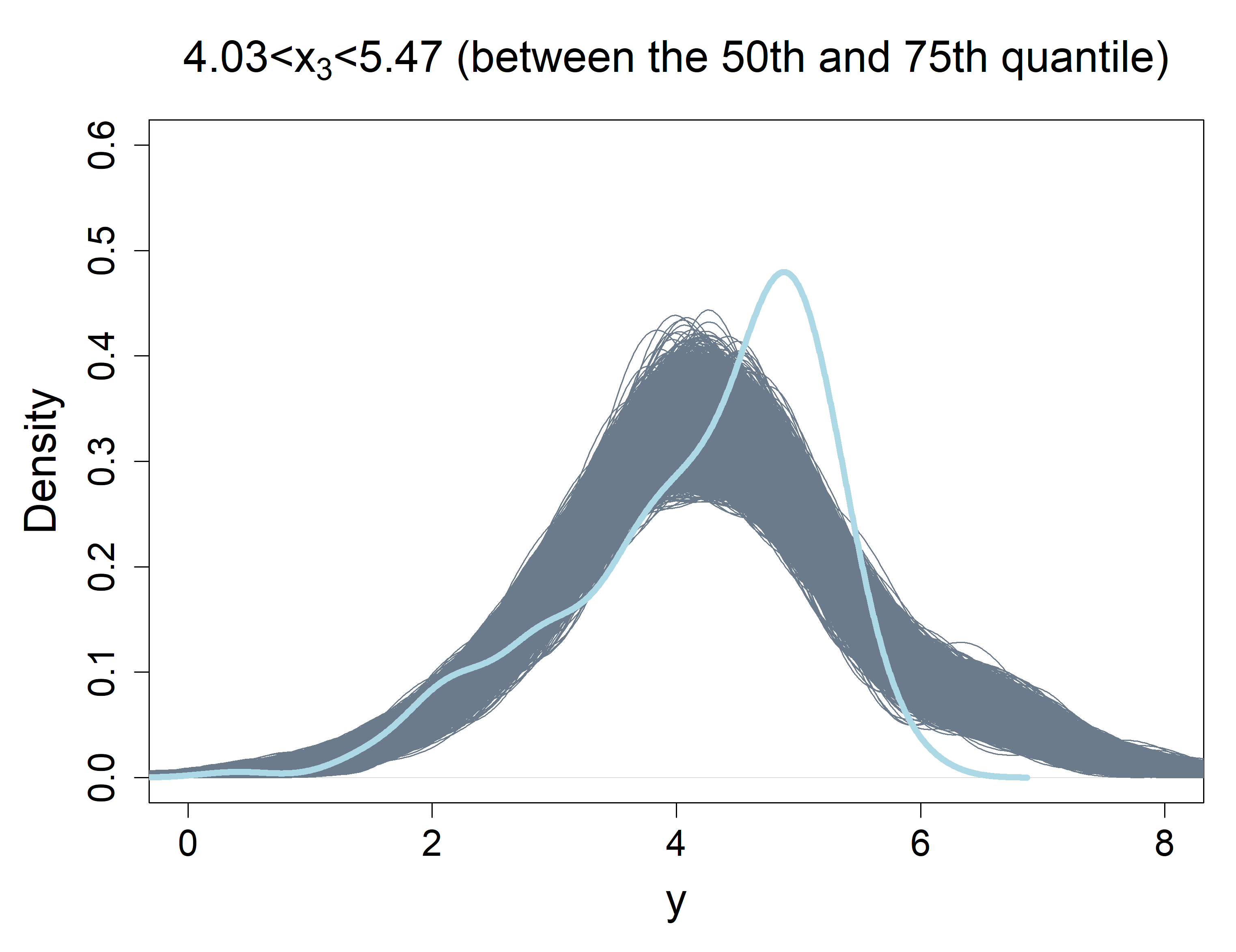}
}

\centering
\caption{Posterior predictive checks of Example D. Top row: estimated kurtosis, skewness, standard deviation of the response with the observed maximum value (light blue), shown alongside estimates from the 5000 datasets drawn from the posterior predictive distribution (grey). Rest rows: kernel density estimate of the observed response (light blue), shown alongside kernel density estimates from the 5000 datasets drawn from the posterior predictive distribution conditional on several covariate intervals (grey).}
\label{example_D_post_checks2}
\end{figure}

\begin{figure}[htbp]
\centering
\subfigure{
\centering
\includegraphics[width=0.42\textwidth]{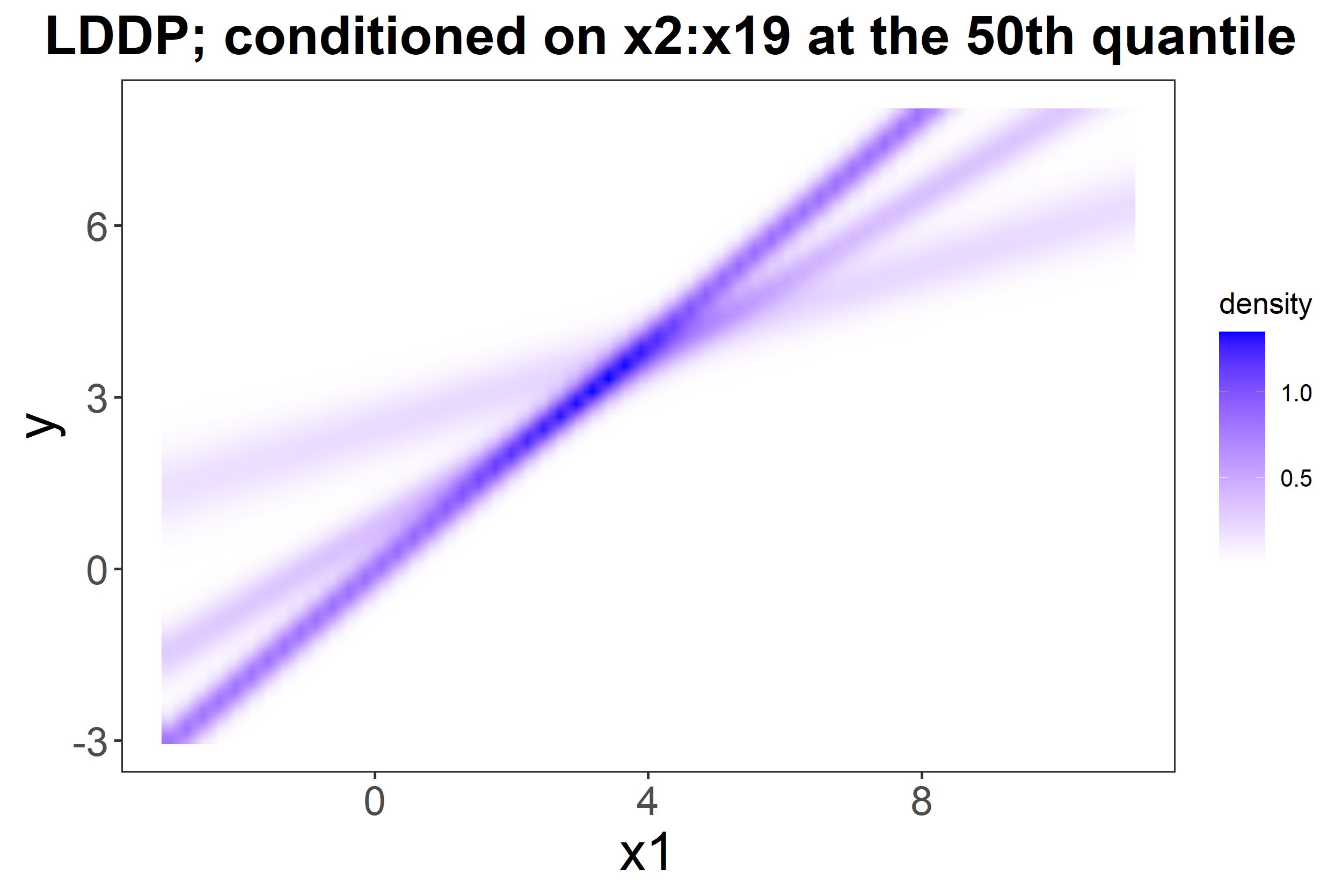}
}
\subfigure{
\centering
\includegraphics[width=0.42\textwidth]{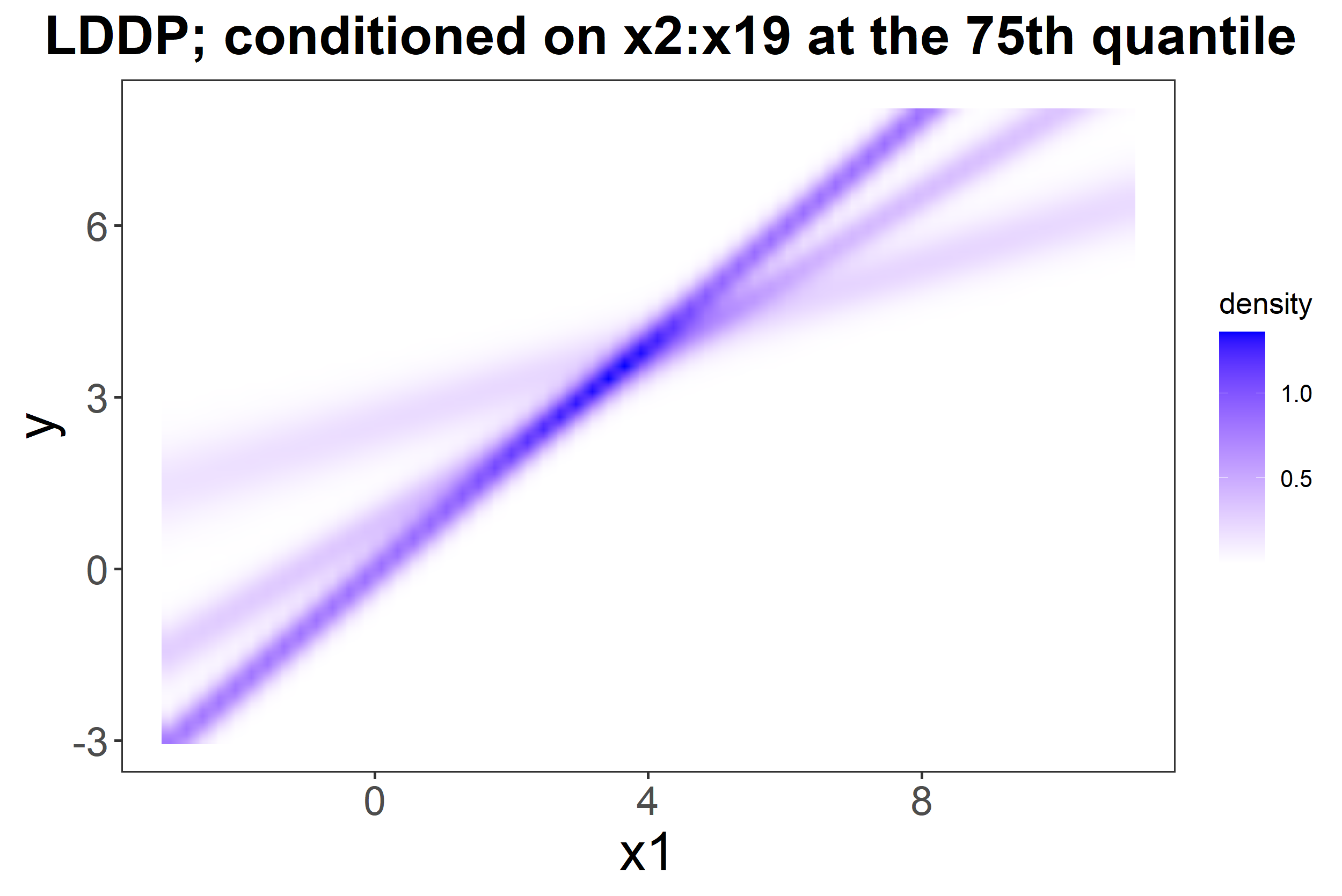}
}

\centering
\caption{Heatmap of estimated density regression functions conditioned on $x_2$ to $x_{19}$ at the $50\%$ and $75\%$ quantiles of Example D.}
\label{example_D_heatmap_plots}
\end{figure}

\newpage

\section{UNL analysis based on WASABI solutions for the real data applications}

\subsection{Breast cancer genomic analysis}
To choose the number of particles in WASABI \citep{balocchi2025understanding}, we construct an elbow plot in Figure \ref{UNL_BC_gene_elbow_summary}, which suggests a four-particle solution, achieving a balance between parsimony and minimizing the objective. And from the summary plots in Figure \ref{UNL_BC_gene_elbow_summary}, we can see that the third and the fourth particle of the WASABI solution has negligible weights, so the following analysis is only based on the first two particles which both contain three clusters. As shown in Figure \ref{cluster_chara_tumor_elbowp1} and Figure \ref{cluster_chara_tumor_elbowp2}, in both particles, cluster 1 corresponds to the most favorable prognosis and cluster 3 to the least favorable.

\begin{figure}[htbp]
\centering
\subfigure{
\centering
\includegraphics[width=0.46\textwidth]{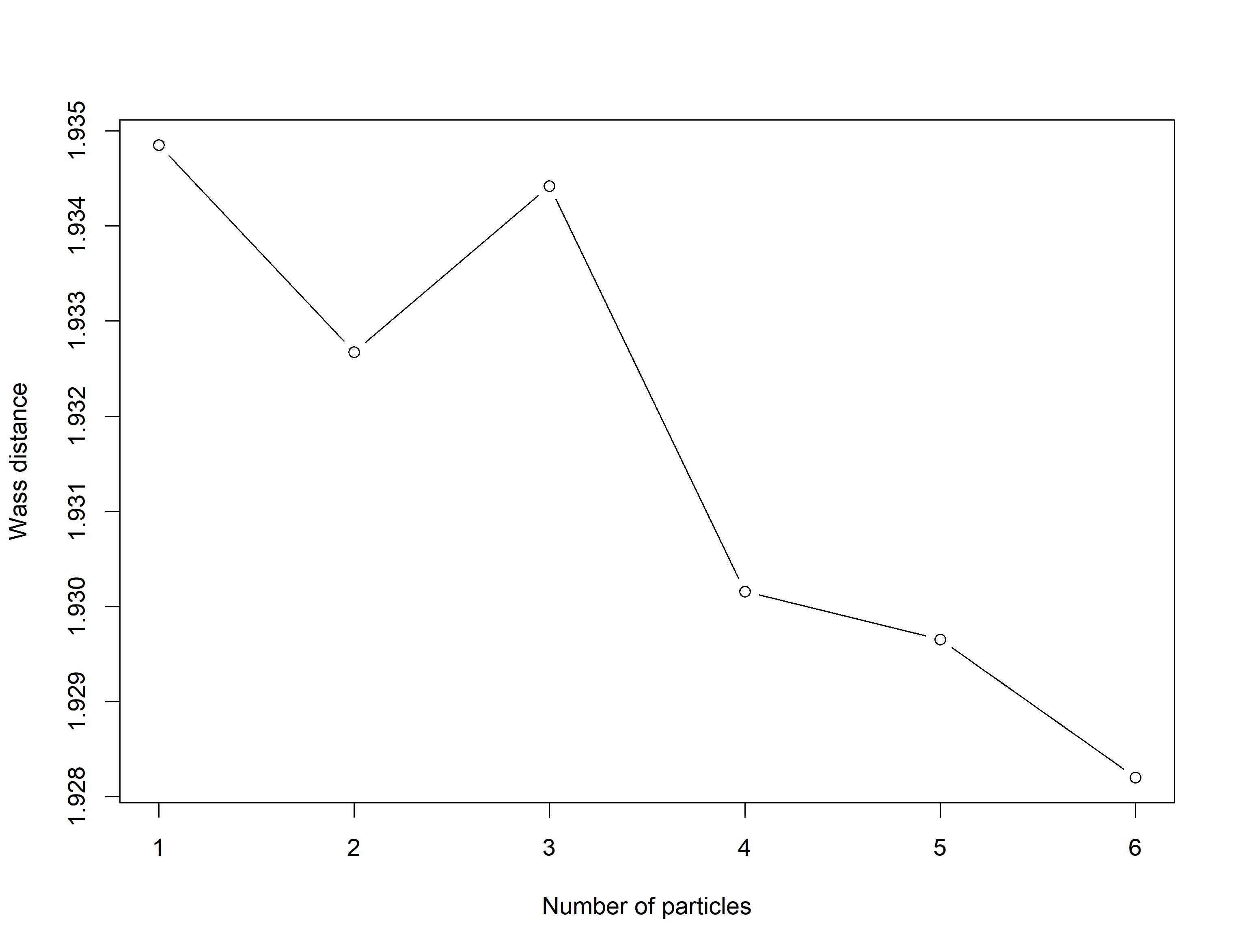}
}
\subfigure{
\centering
\includegraphics[width=0.46\textwidth]{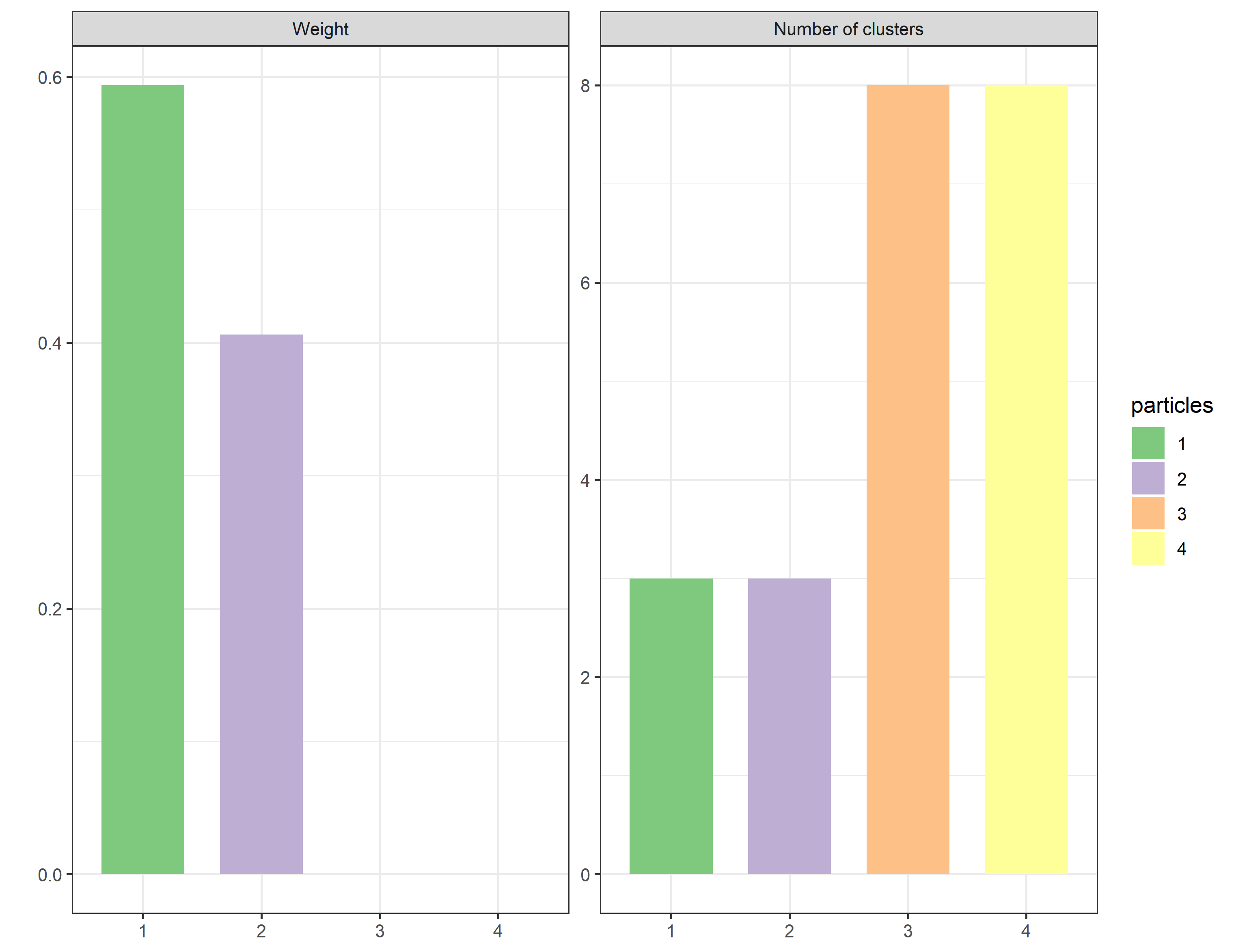}
}
\centering
\caption{Left panel: Elbow plots containing the Wasserstein distances achieved by the WASABI approximation for the breast cancer genomic analysis. Right panel: WASABI summaries of the four-particle clustering solution for the breast cancer genomic analysis.}
\label{UNL_BC_gene_elbow_summary}
\end{figure}

\begin{figure}[!t]
\centering
\subfigure{
\centering
\includegraphics[width=0.46\textwidth]{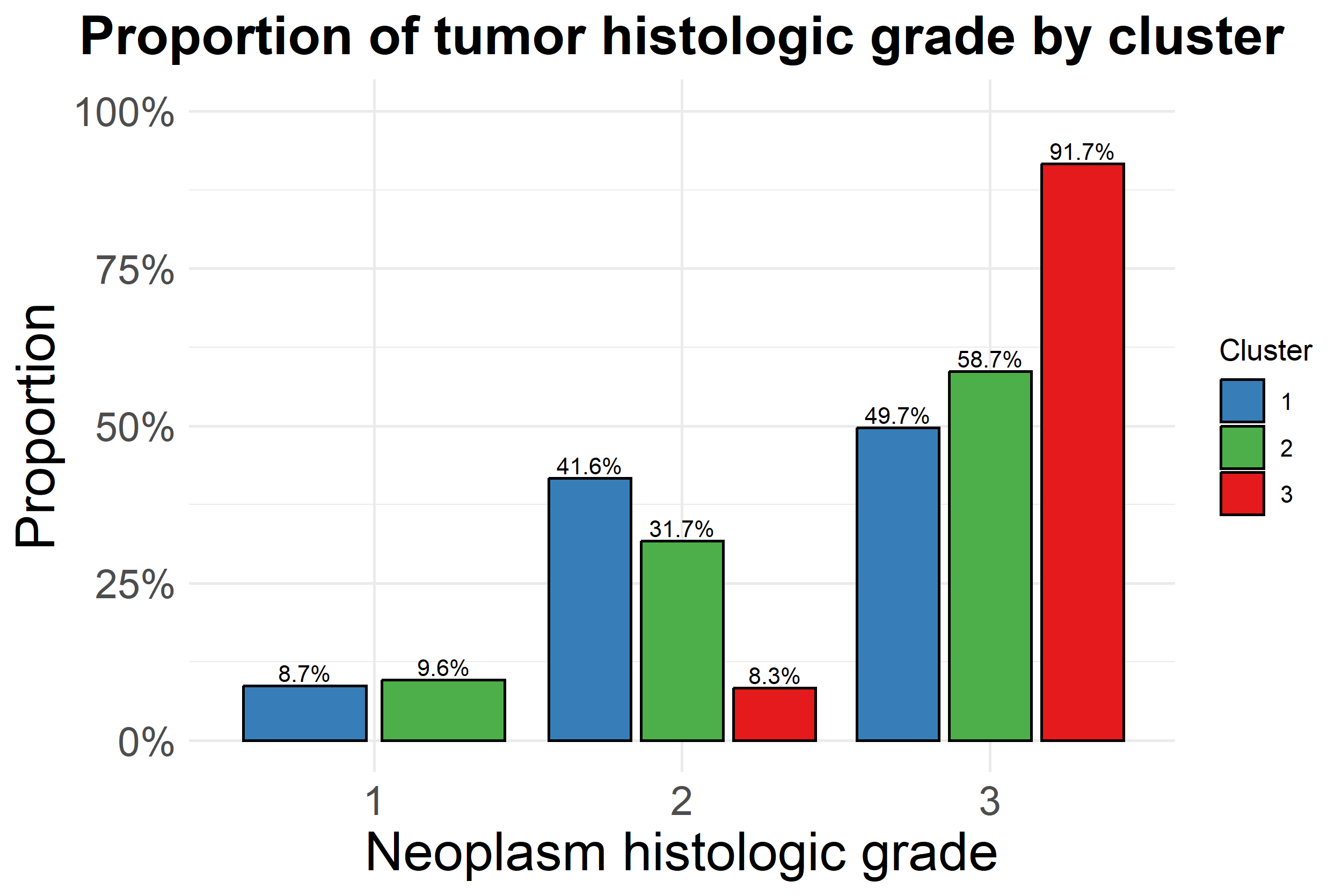}
}
\subfigure{
\centering
\includegraphics[width=0.46\textwidth]{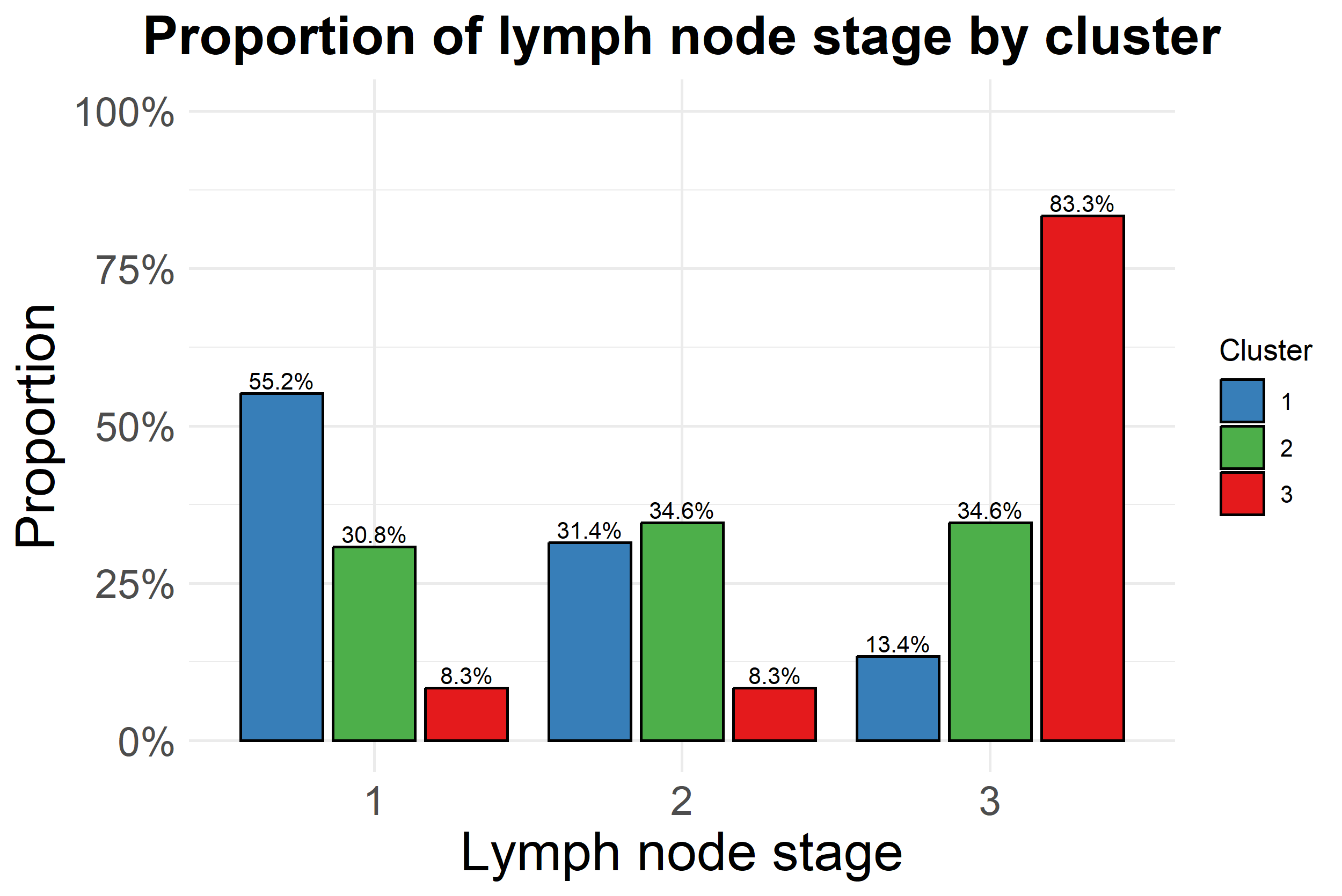}
}\\
\centering
\subfigure{
\centering
\includegraphics[width=0.46\textwidth]{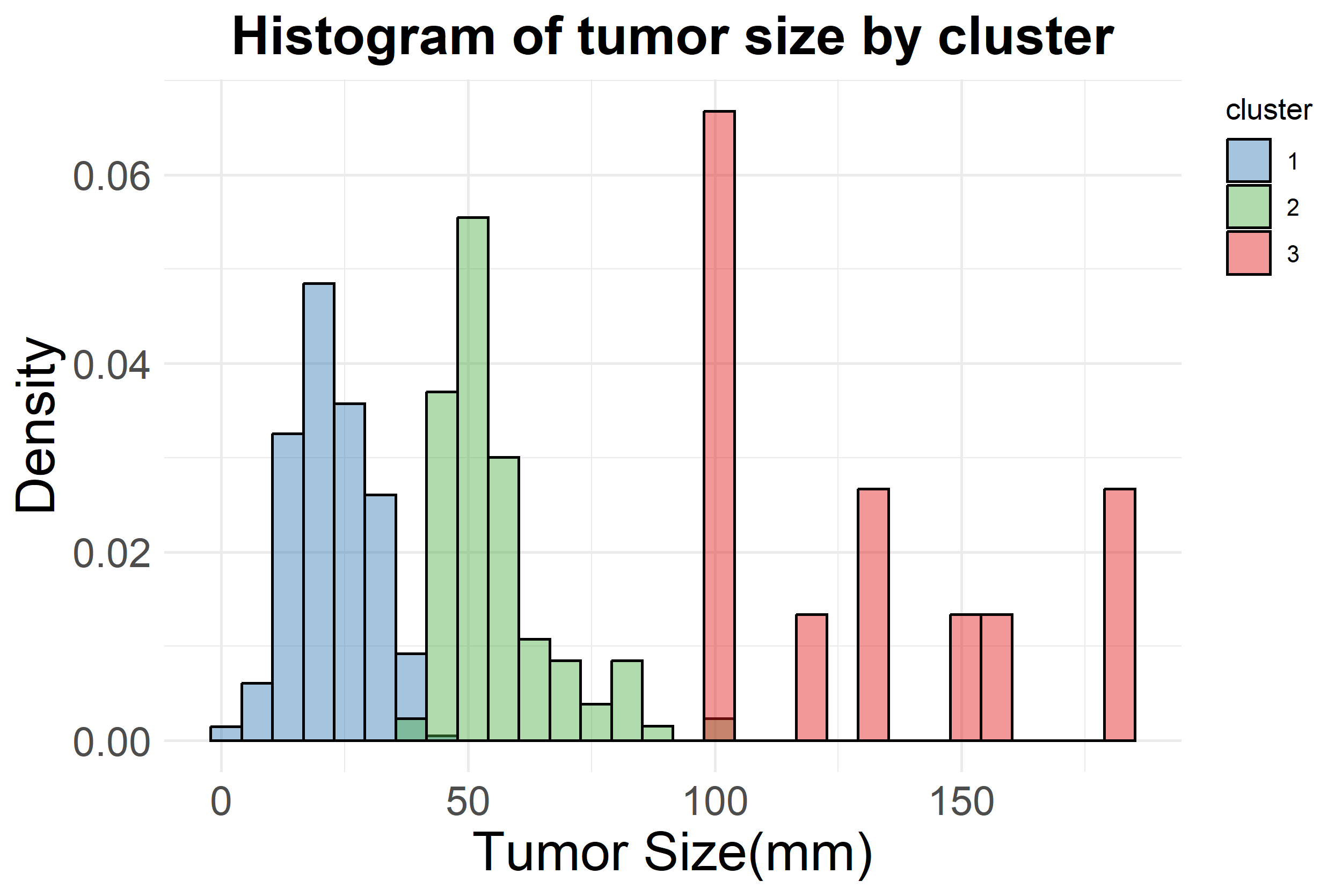}
}
\subfigure{
\centering
\includegraphics[width=0.46\textwidth]{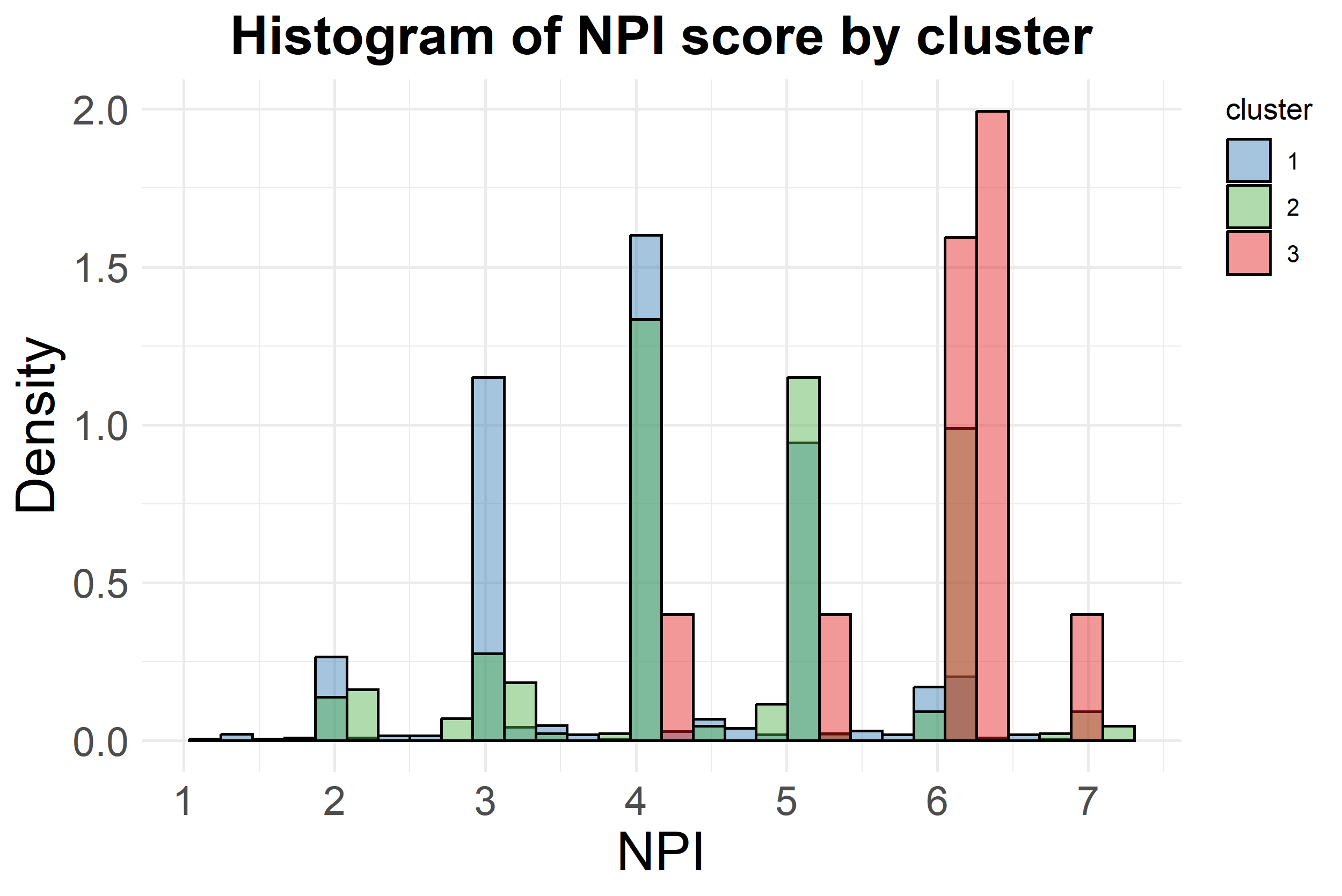}
}
\caption{Cluster profiles of breast cancer prognostic factors, based on the first particle from the four-particle WASABI solution.}
\label{cluster_chara_tumor_elbowp1}
\end{figure}

\begin{figure}[!t]
\centering
\subfigure{
\centering
\includegraphics[width=0.46\textwidth]{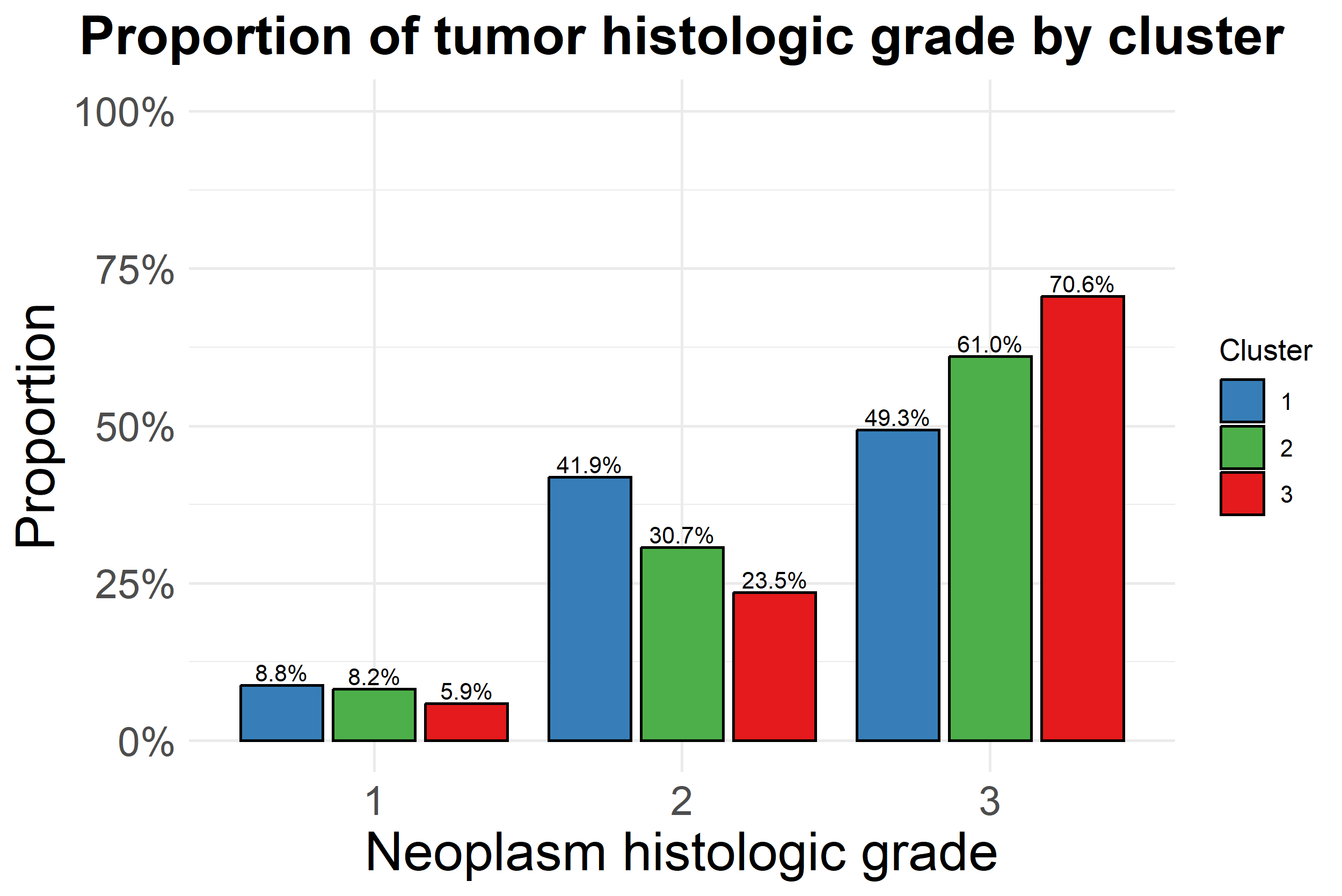}
}
\subfigure{
\centering
\includegraphics[width=0.46\textwidth]{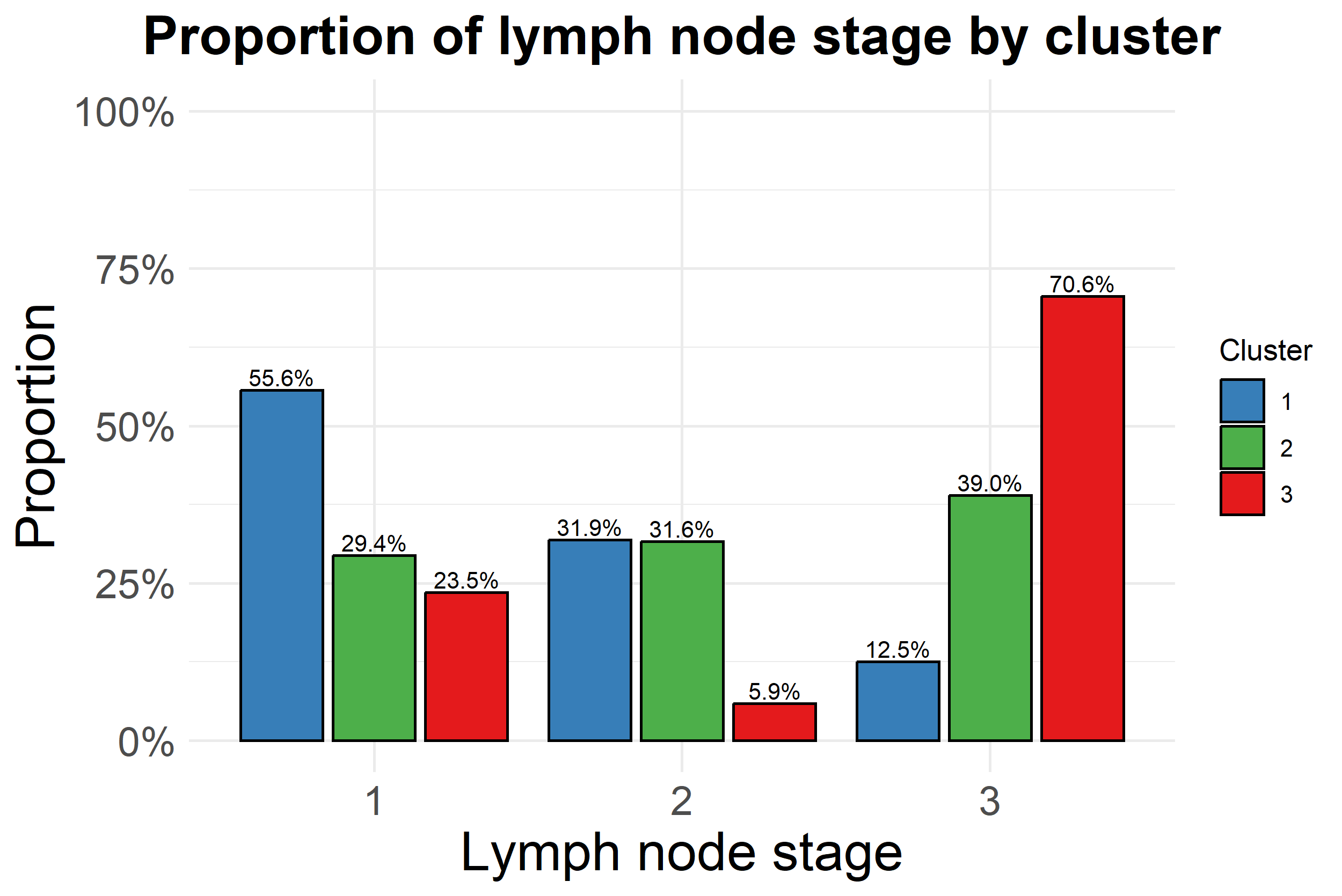}
}\\
\centering
\subfigure{
\centering
\includegraphics[width=0.46\textwidth]{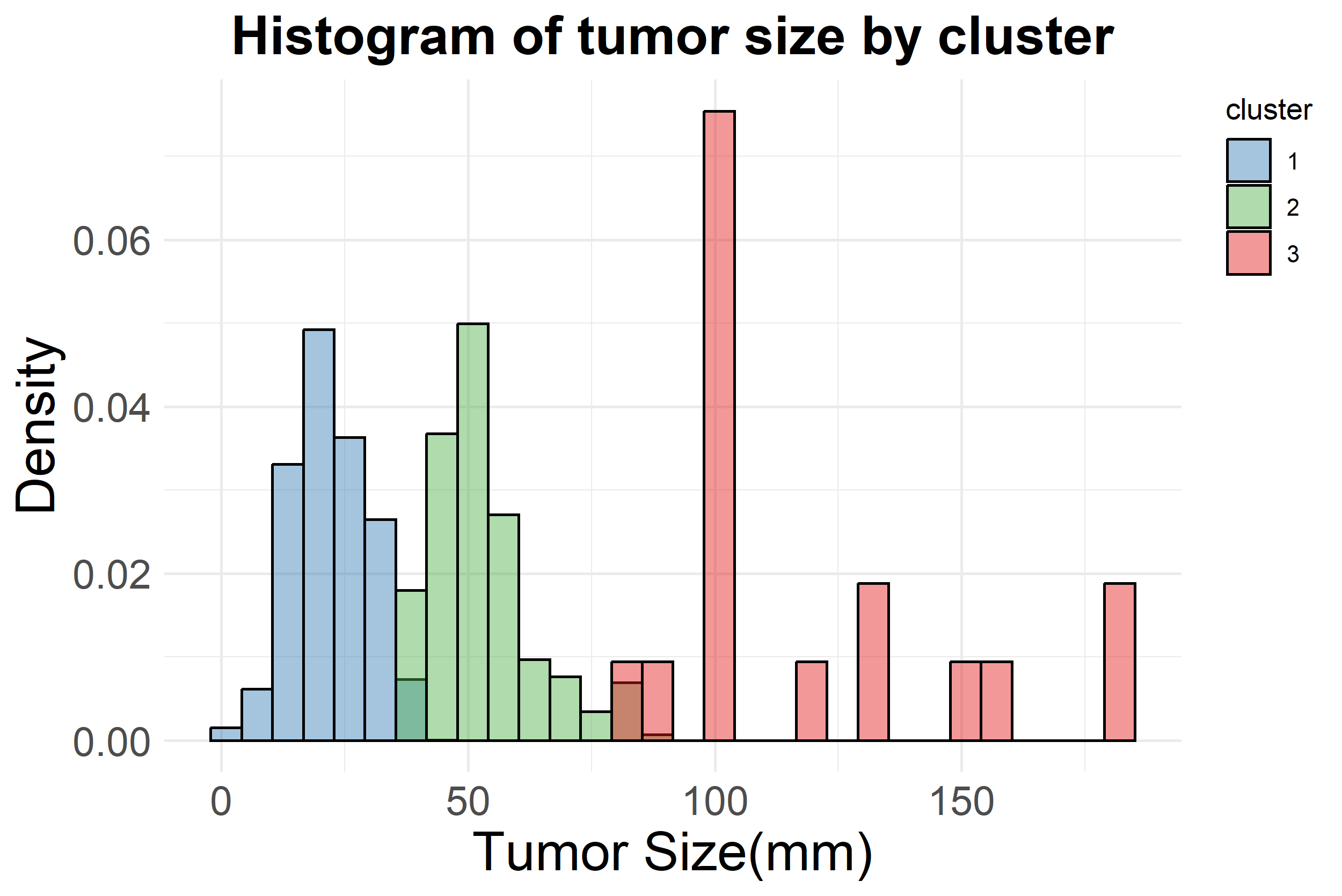}
}
\subfigure{
\centering
\includegraphics[width=0.46\textwidth]{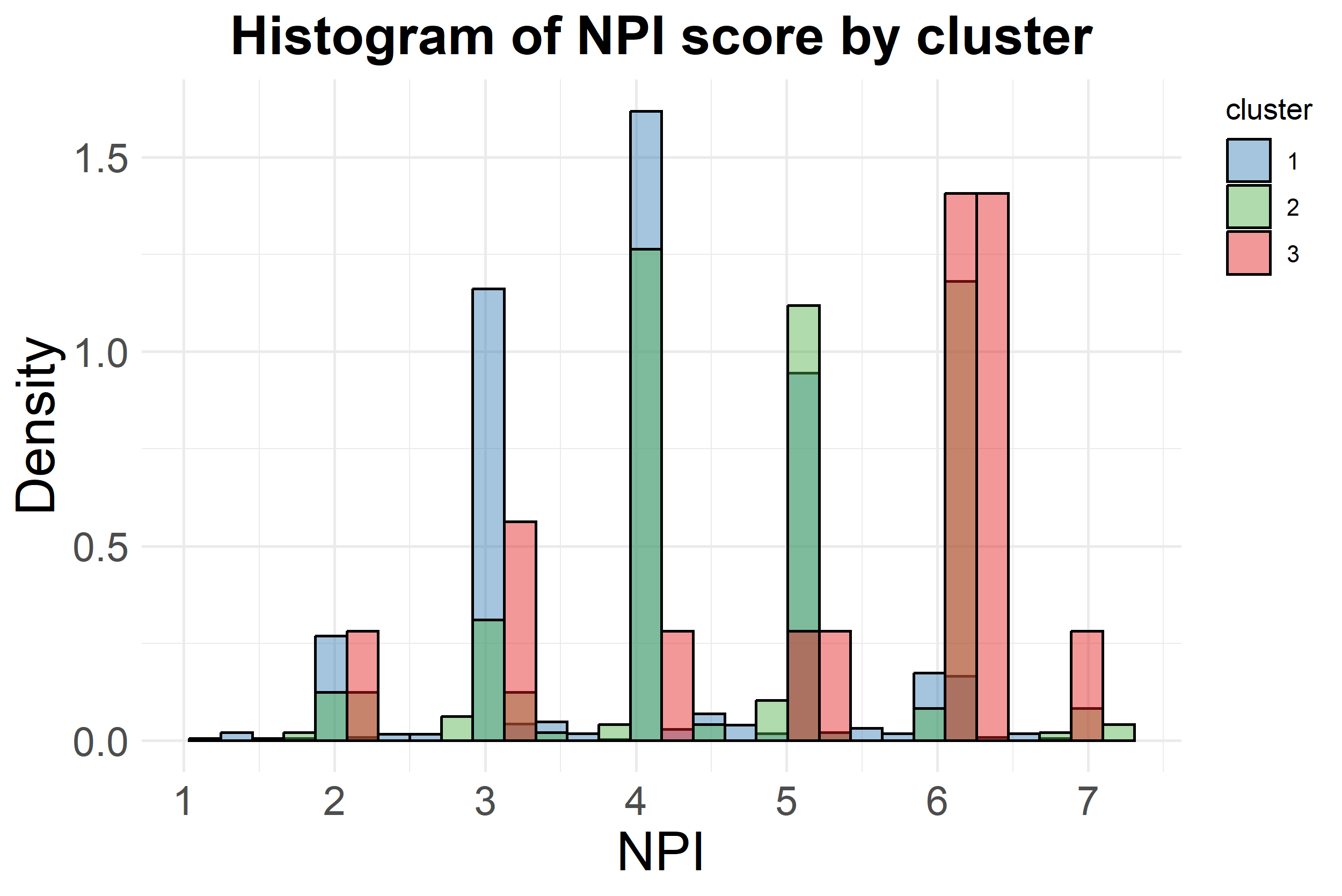}
}
\caption{Cluster profiles of breast cancer prognostic factors, based on the second particle from the four-particle WASABI solution.}
\label{cluster_chara_tumor_elbowp2}
\end{figure}

\begin{figure}[!t]
\centering
\subfigure{
\centering
\includegraphics[width=0.46\textwidth]{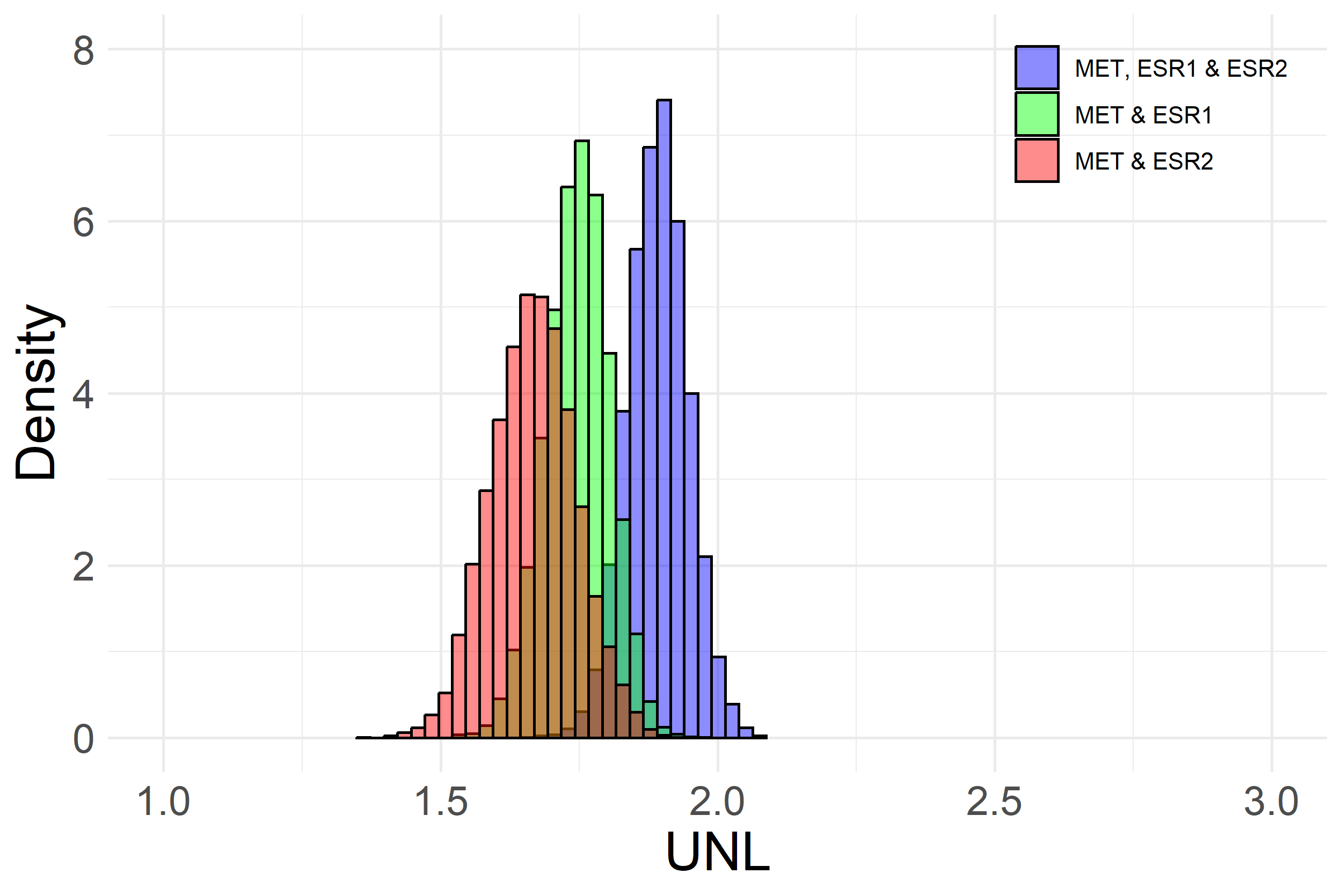}
}
\subfigure{
\centering
\includegraphics[width=0.46\textwidth]{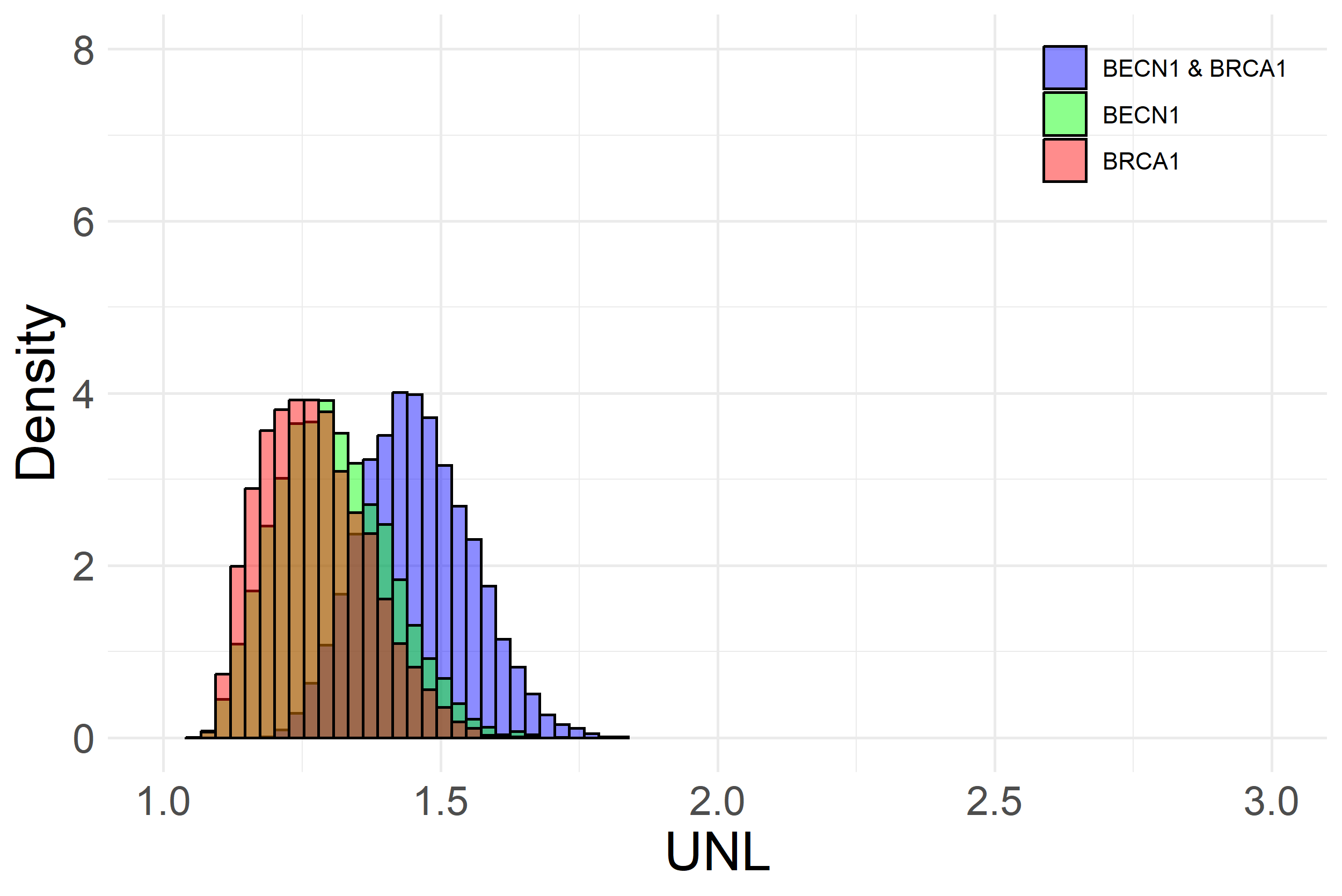}
}
\centering
\caption{Left panel: the histogram of estimated underlap for the \{MET, ESR1, ESR2\} gene group based on the first particle from the four-particle WASABI solution. Right panel: the histogram of estimated underlap for the \{BRCA1, BECN1\} gene group based on the first particle from the four-particle WASABI solution.}
\label{UNL_BC_gene_hist_elbowp1}
\end{figure}

We next quantify the dependence of the clustering on the two gene groups, \{MET, ESR1, ESR2\} and \{BRCA1,BECN1\}, via the underlap coefficient. Densities of mRNA log-expression and the underlap coefficient are estimated in the same way as described for the single representative partition solution. Figure \ref{UNL_BC_gene_hist_elbowp1} and Figure \ref{UNL_BC_gene_hist_elbowp2} display the histograms of the estimated underlap coefficient for each gene groups in both particles. we can arrive at a similar conclusion to the single representative partition solution that the coexpression of MET and ESR1 is more informative for the clinicopathologic prognosis clusters than MET and ESR2, and BECN1 is slightly shifted more informative than BRCA1.

\begin{figure}[!t]
\centering
\subfigure{
\centering
\includegraphics[width=0.46\textwidth]{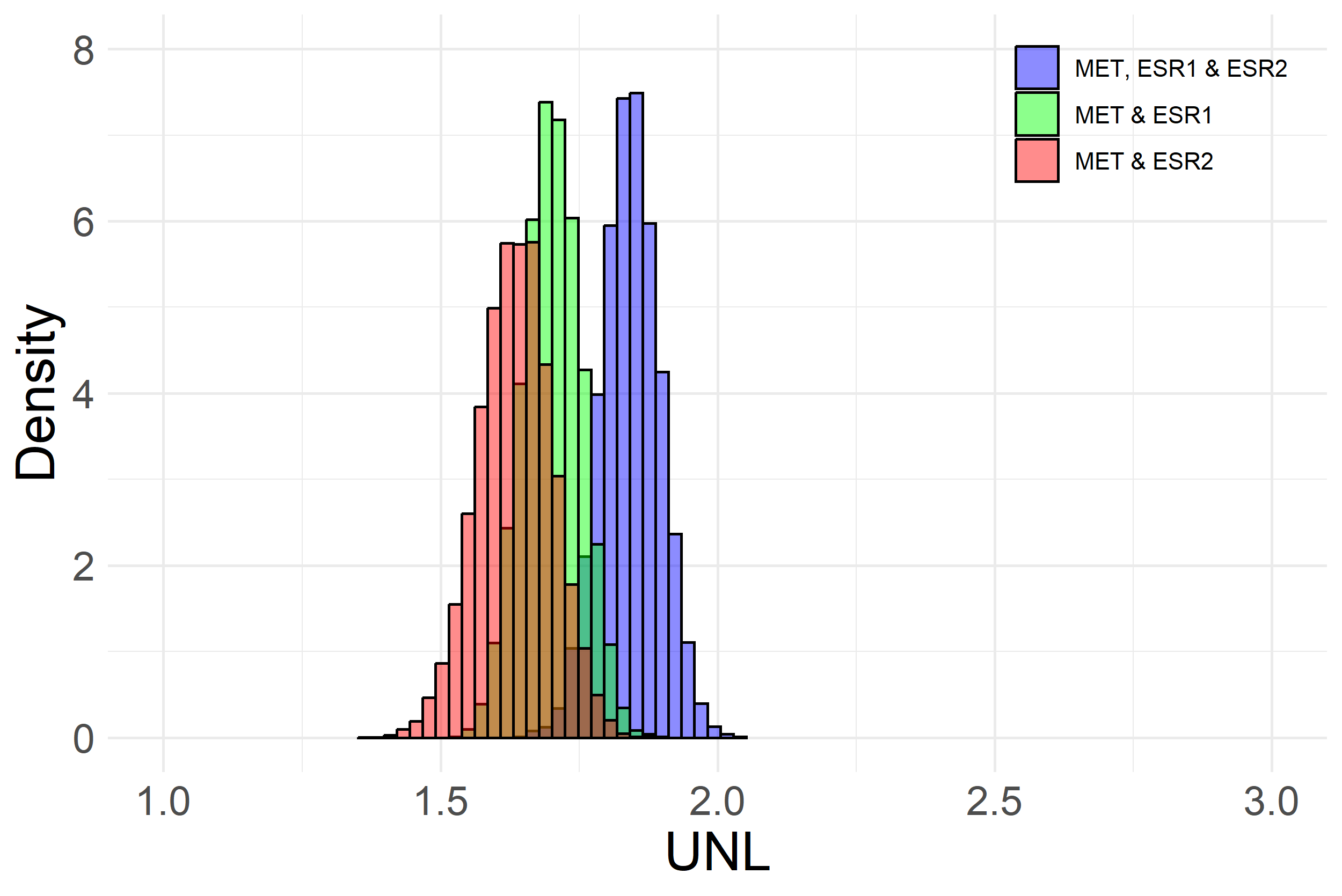}
}
\subfigure{
\centering
\includegraphics[width=0.46\textwidth]{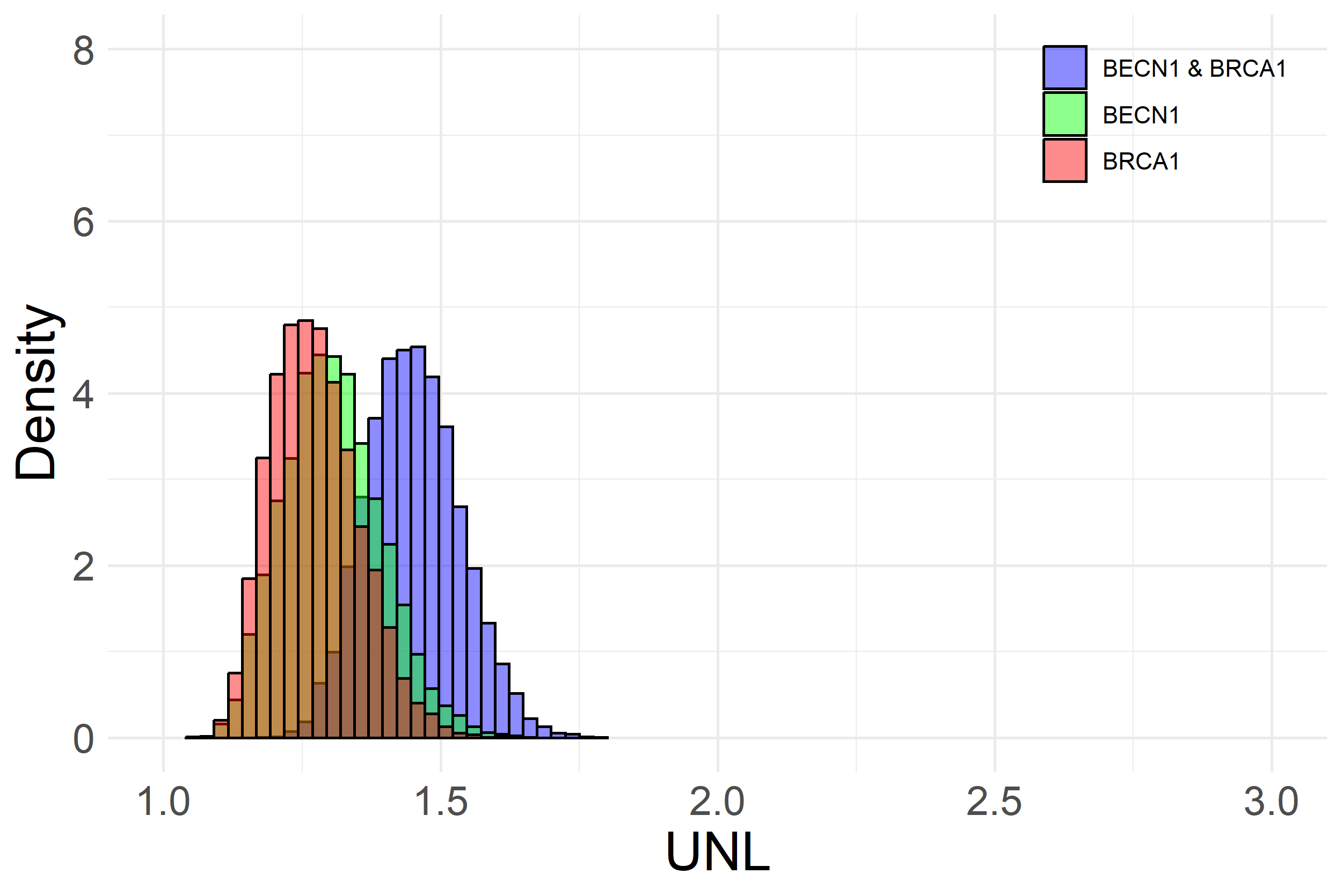}
}
\centering
\caption{Left panel: the histogram of estimated underlap for the \{MET, ESR1, ESR2\} gene group based on the second particle from the four-particle WASABI solution. Right panel: the histogram of estimated underlap for the \{BRCA1, BECN1\} gene group based on the second particle from the four-particle WASABI solution.}
\label{UNL_BC_gene_hist_elbowp2}
\end{figure}

\clearpage
\subsection{Pregnancy term toxicology analysis}
To choose the number of particles in WASABI \citep{balocchi2025understanding}, we construct an elbow plot in Figure \ref{UNL_dde_elbow_summary}, which suggests a four-particle solution, achieving a balance between parsimony and minimizing the objective. And from the summary plots in Figure \ref{UNL_dde_elbow_summary}, we can see that the third and the fourth particle of the WASABI solution has negligible weights, so the following analysis is only based on the first two particles which both contain three clusters. As shown in Figures \ref{dde_cluster_partition_dde_UNL_hist_elbowp1} and \ref{dde_cluster_partition_dde_UNL_hist_elbowp2}, the first particle contains two clusters, roughly corresponding to the postterm and non-postterm groups, whereas the second particle contains three clusters, roughly corresponding to the preterm, term, and postterm groups.

We assessed the dependence of the inferred partition on DDE using the underlap coefficient. Densities of DDE and the underlap coefficient are estimated in the same way as described for the single representative partition solution. Figure \ref{dde_cluster_partition_dde_UNL_hist_elbowp1} and Figure \ref{dde_cluster_partition_dde_UNL_hist_elbowp1} also displays the histogram of the estimated underlap coefficient for DDE in the first two particles. The posterior sampling distribution of the UNL concentrates near one in both particles, indicating little residual dependence of the partition on DDE.

\begin{figure}[!t]
\centering
\subfigure{
\centering
\includegraphics[width=0.46\textwidth]{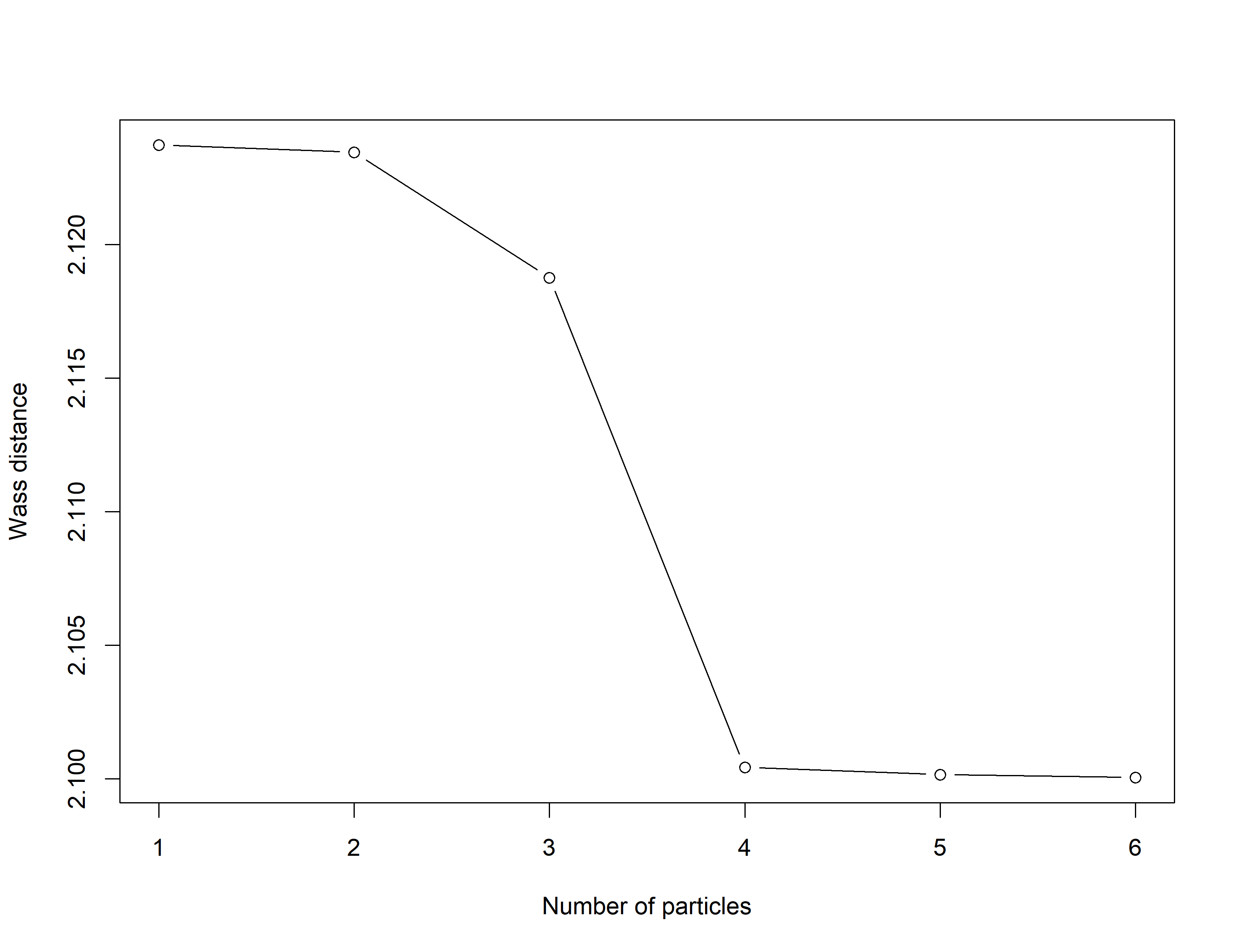}
}
\subfigure{
\centering
\includegraphics[width=0.46\textwidth]{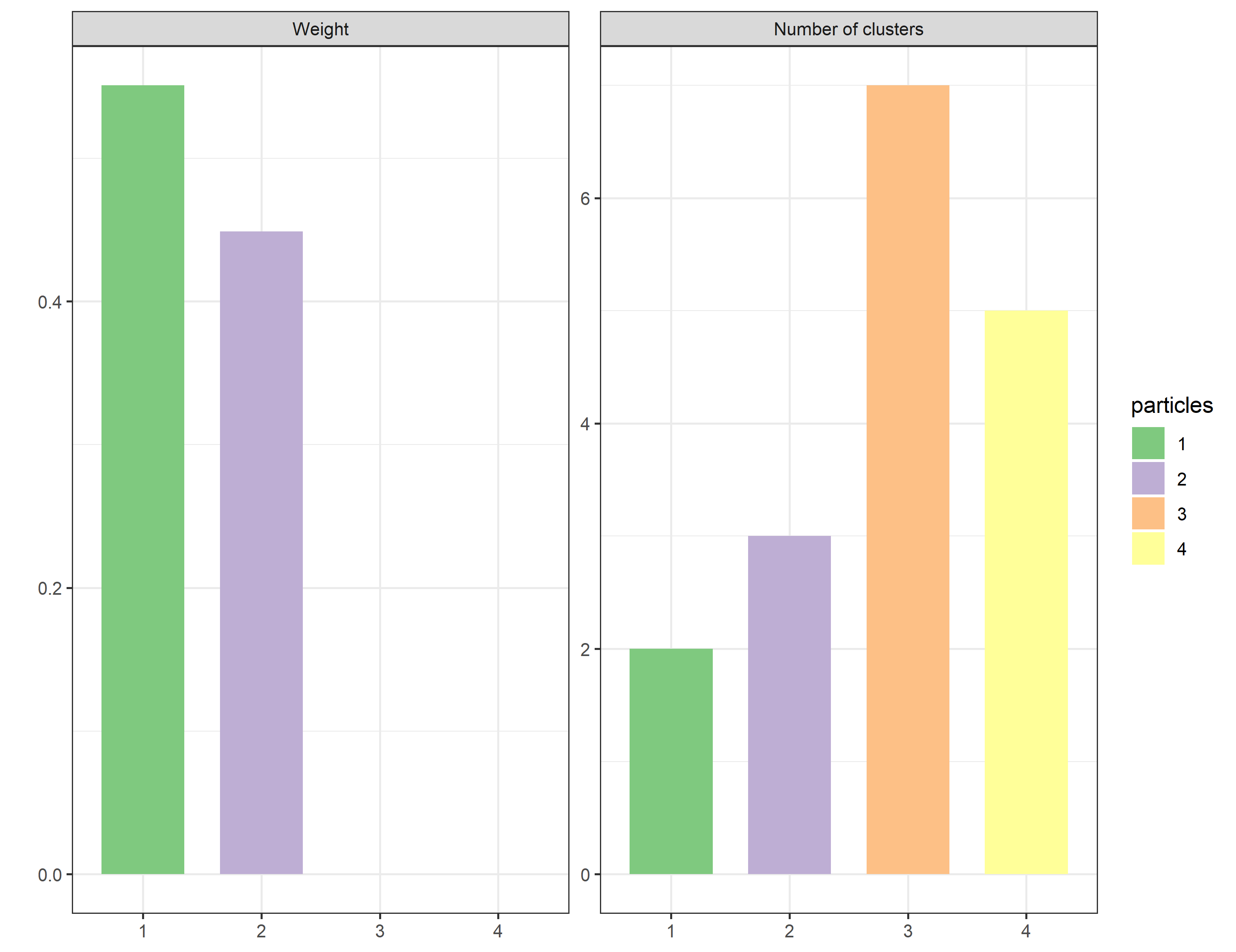}
}
\centering
\caption{Left panel: Elbow plots containing the Wasserstein distances achieved by the WASABI approximation for the pregnancy term toxicology analysis. Right panel: WASABI summaries of the four-particle clustering solution for the pregnancy term toxicology analysis.}
\label{UNL_dde_elbow_summary}
\end{figure}

\begin{figure}[!t]
\centering
\subfigure{
\centering
\includegraphics[width=0.46\textwidth]{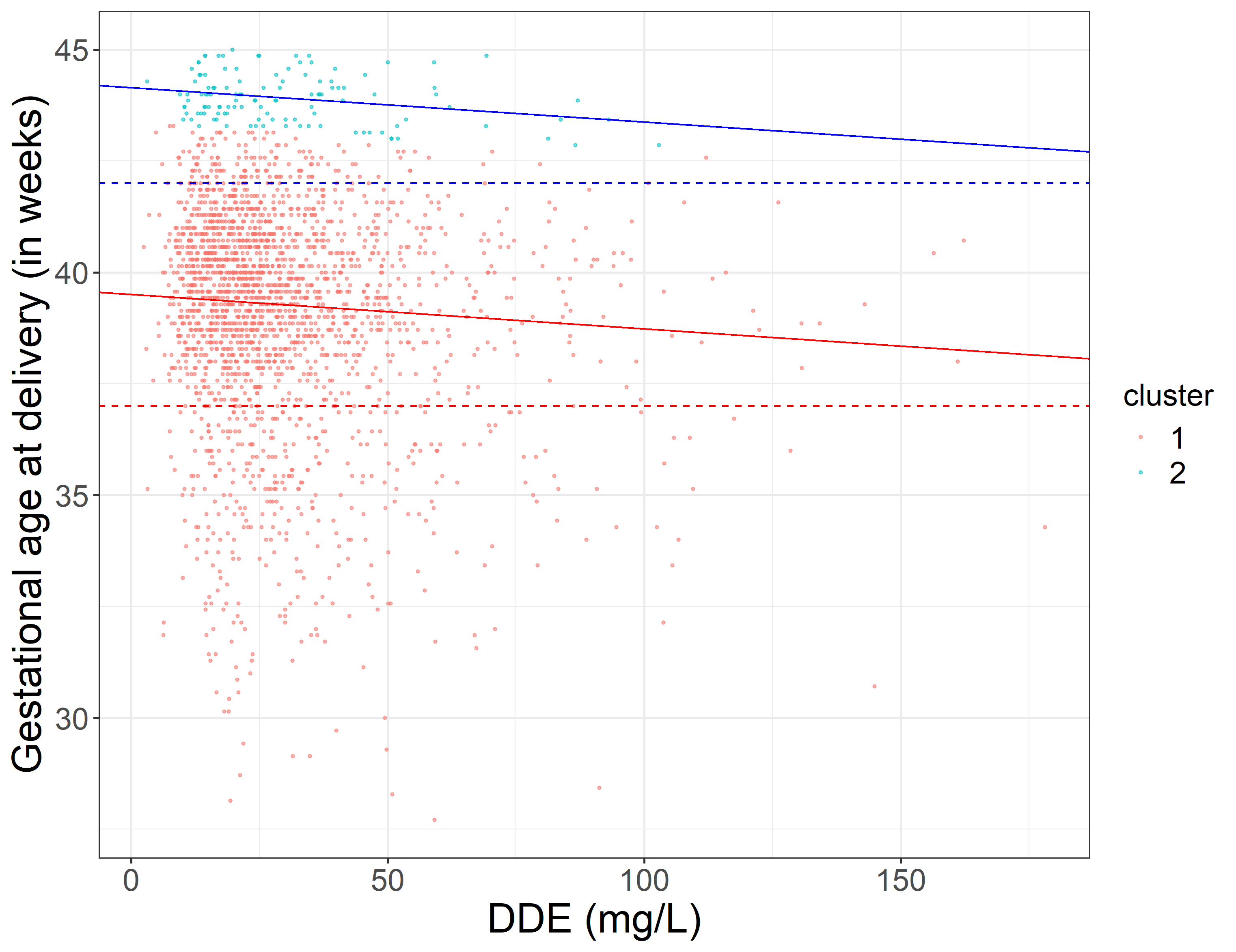}
}
\subfigure{
\centering
\includegraphics[width=0.46\textwidth]{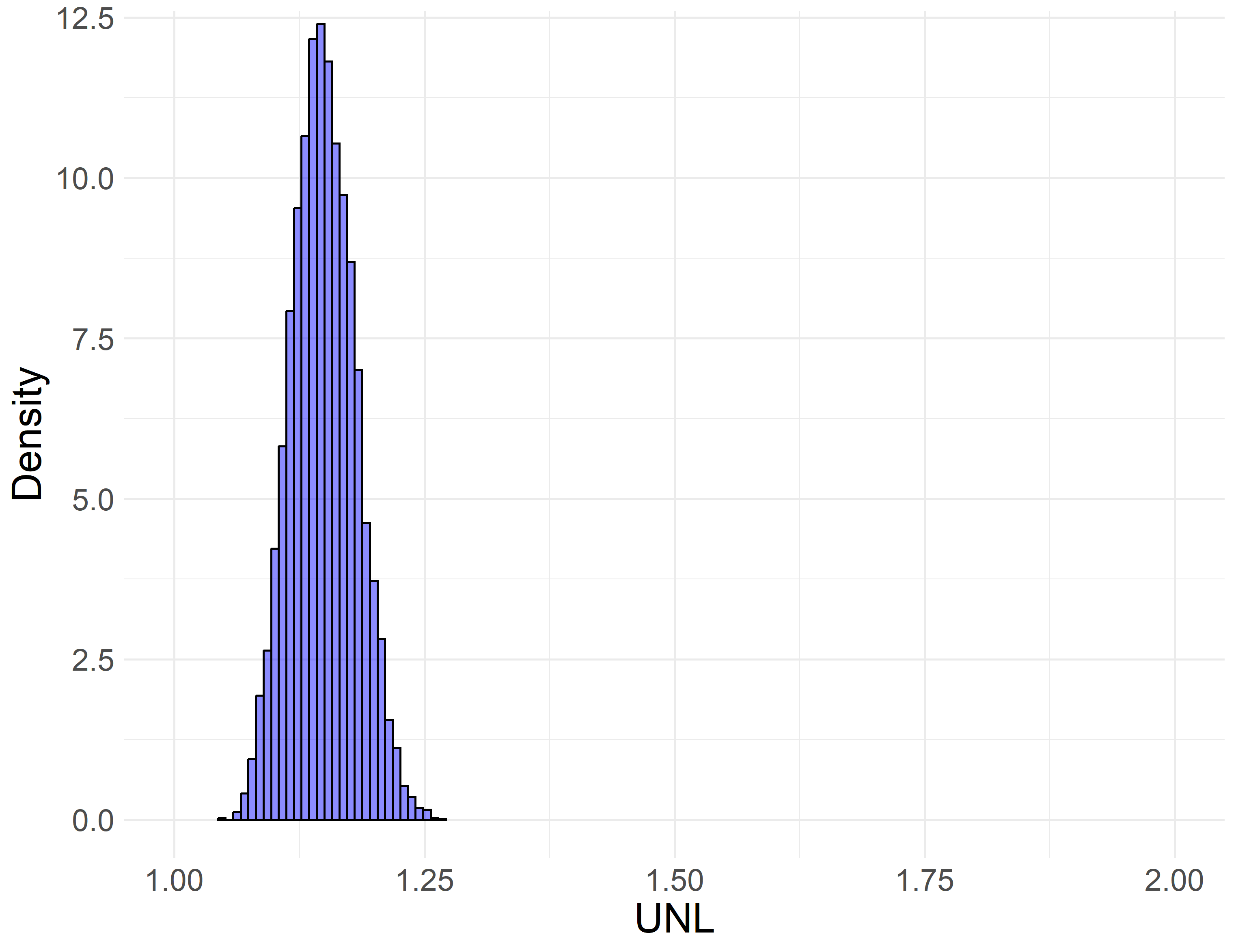}
}
\centering
\caption{Left panel: Partition structure of the first particle of the four-particle WASABI solution of the LDDP clustering of the pregnancy term toxicology analysis. The red line represents the 37-week preterm threshold used in clinical practice. The blue line represents the 42-week postterm threshold used in clinical practice. Right panel: The histogram of estimated underlap coefficient for DDE based on the partition structure of the first particle of the four-particle WASABI solution of LDDP.}
\label{dde_cluster_partition_dde_UNL_hist_elbowp1}
\end{figure}

\begin{figure}[!t]
\centering
\subfigure{
\centering
\includegraphics[width=0.46\textwidth]{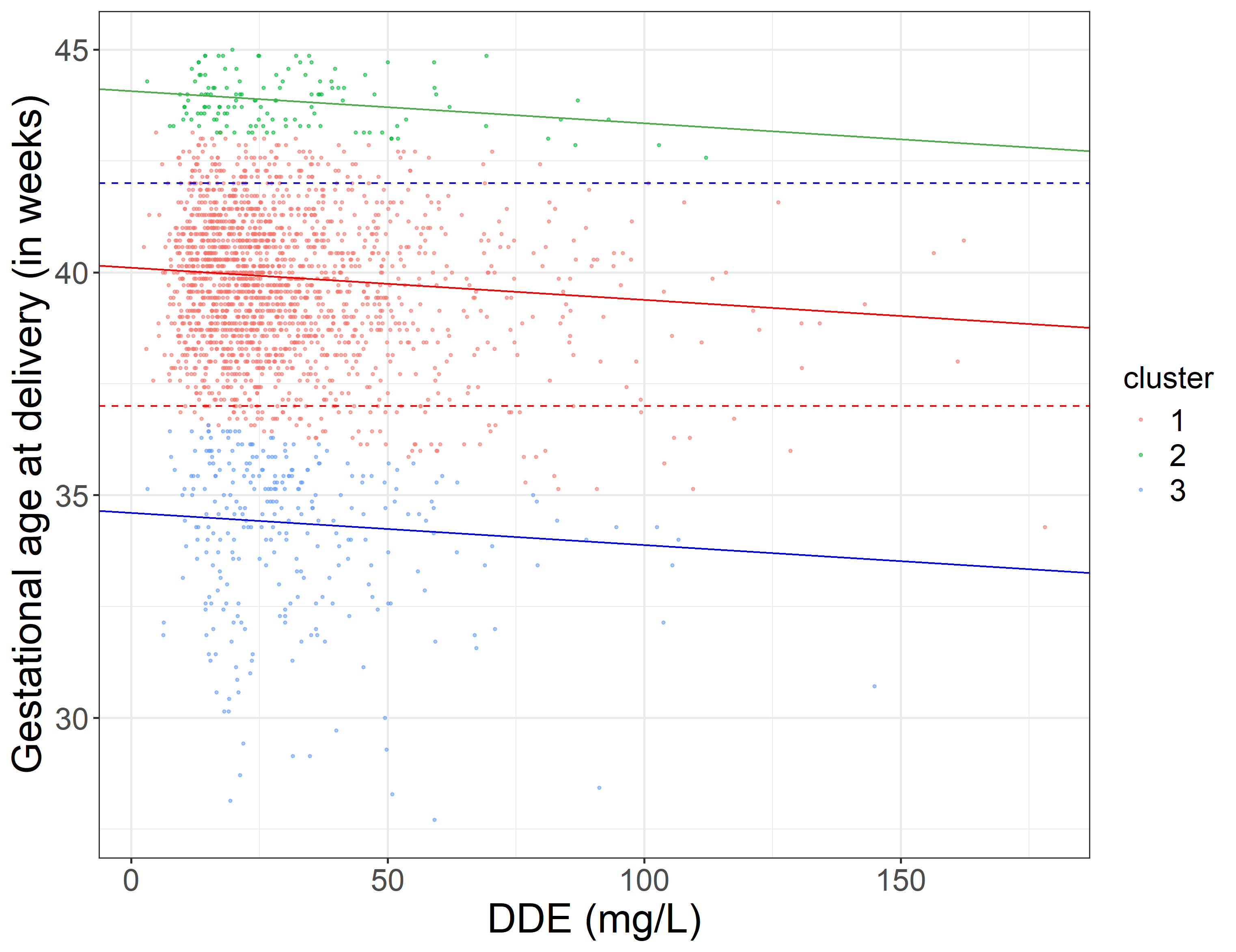}
}
\subfigure{
\centering
\includegraphics[width=0.46\textwidth]{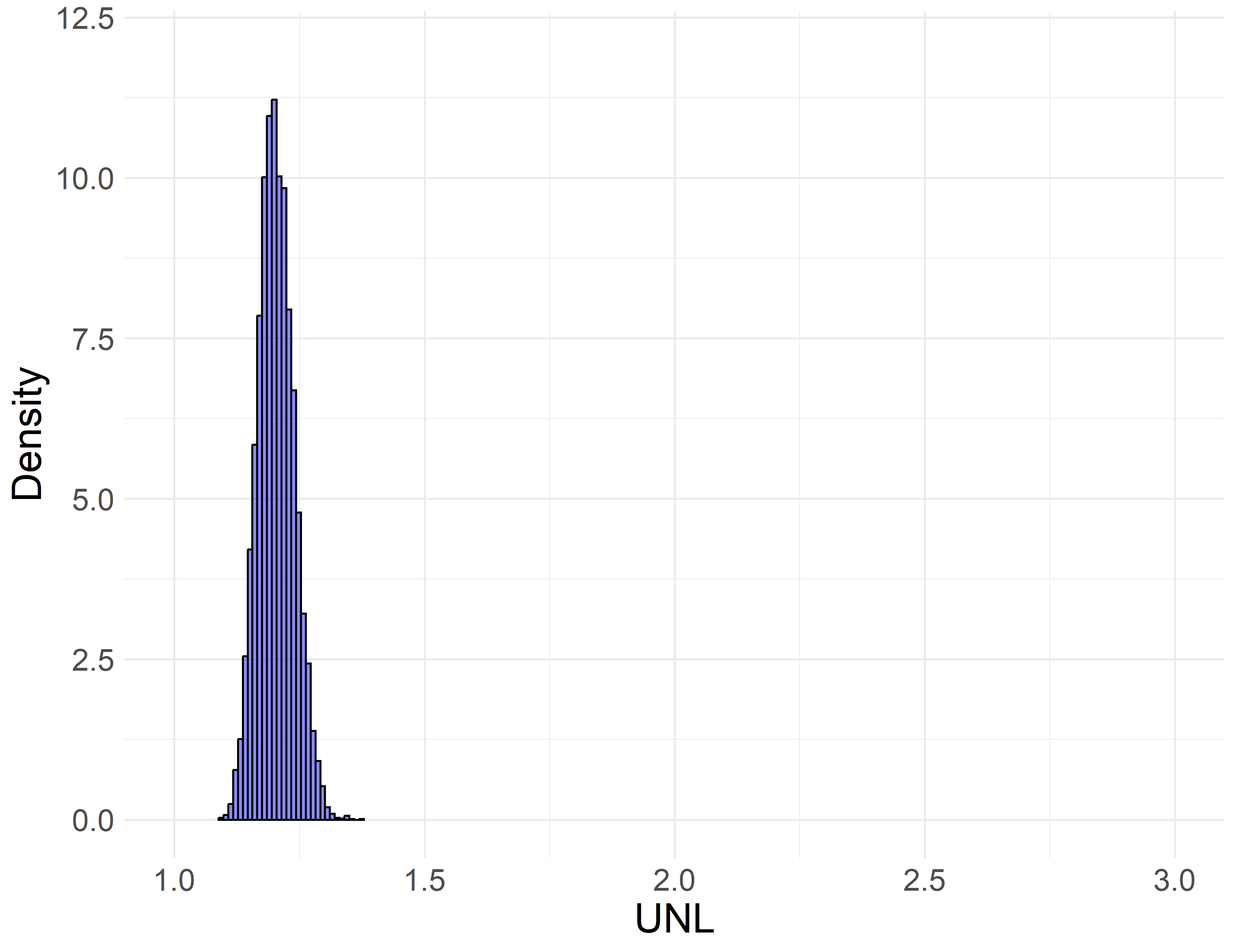}
}
\centering
\caption{Left panel: Partition structure of the second particle of the four-particle WASABI solution of the LDDP clustering of the pregnancy term toxicology analysis. The red line represents the 37-week preterm threshold used in clinical practice. The blue line represents the 42-week postterm threshold used in clinical practice. Right panel: The histogram of estimated underlap coefficient for DDE based on the partition structure of the second particle of the four-particle WASABI solution of LDDP.}
\label{dde_cluster_partition_dde_UNL_hist_elbowp2}
\end{figure}

\clearpage
\section{Some additional plots in Section 4}
\begin{figure}[htbp]
\centering
\subfigure{
\centering
\includegraphics[width=0.22\textwidth]{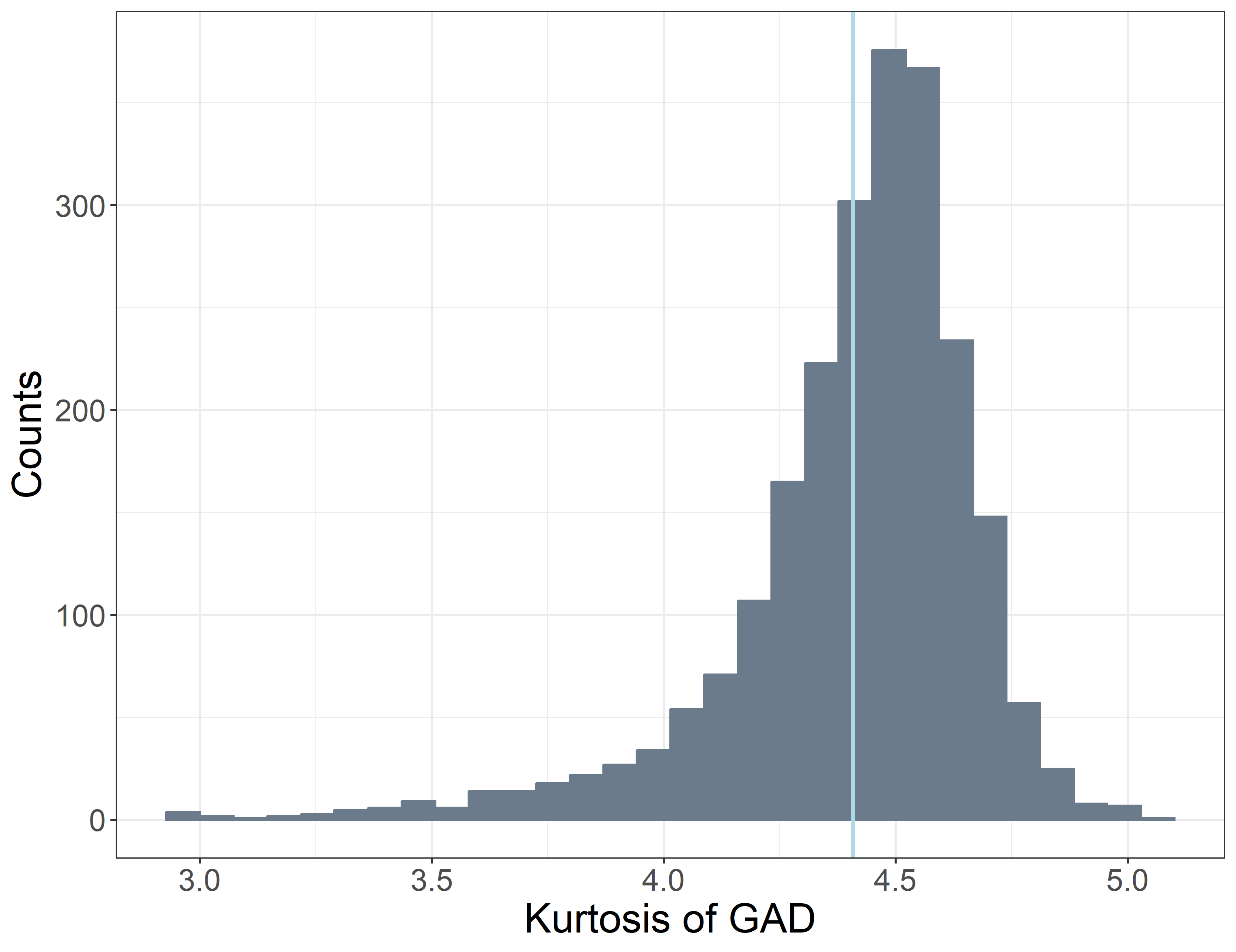}
}
\subfigure{
\centering
\includegraphics[width=0.22\textwidth]{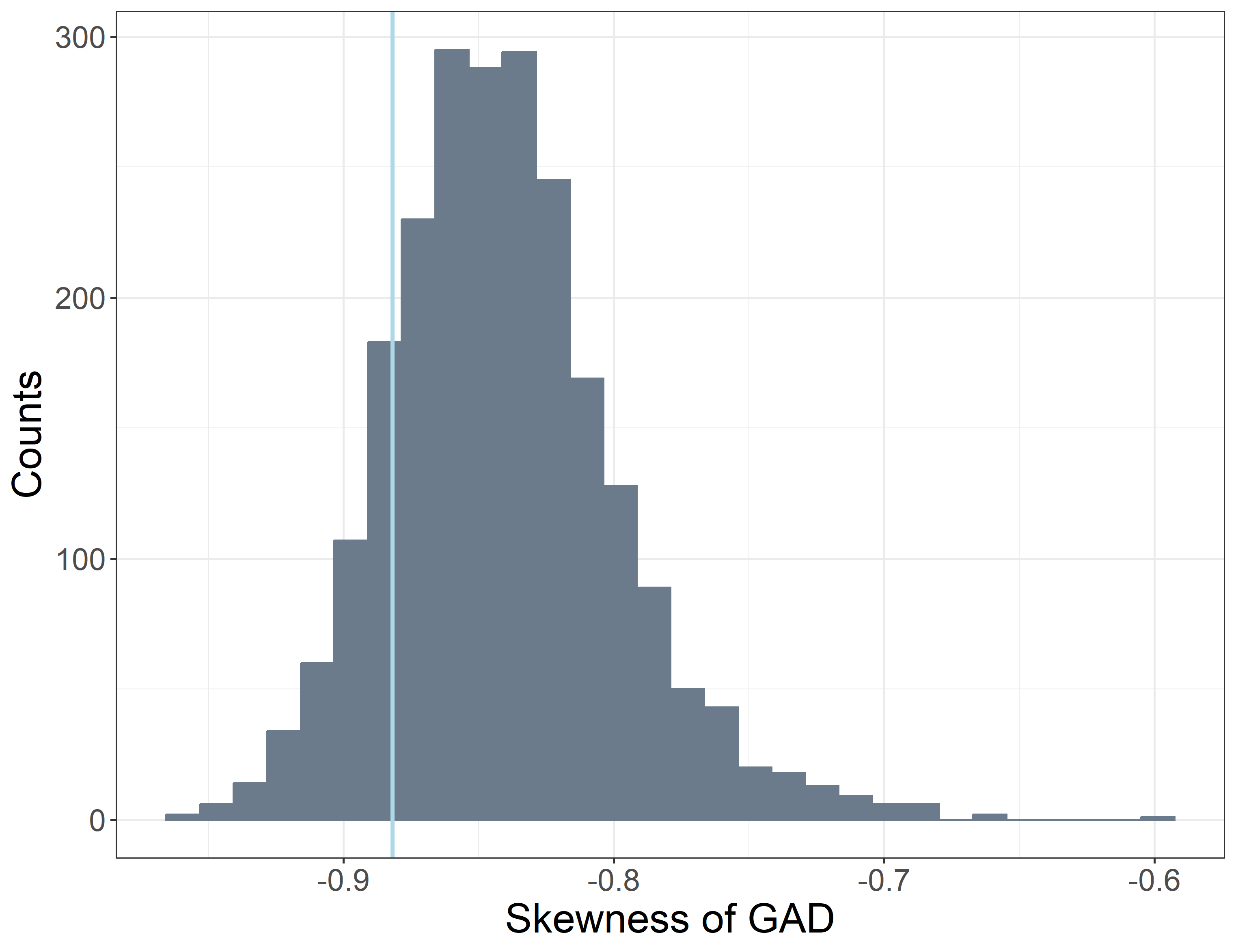}
}
\subfigure{
\centering
\includegraphics[width=0.22\textwidth]{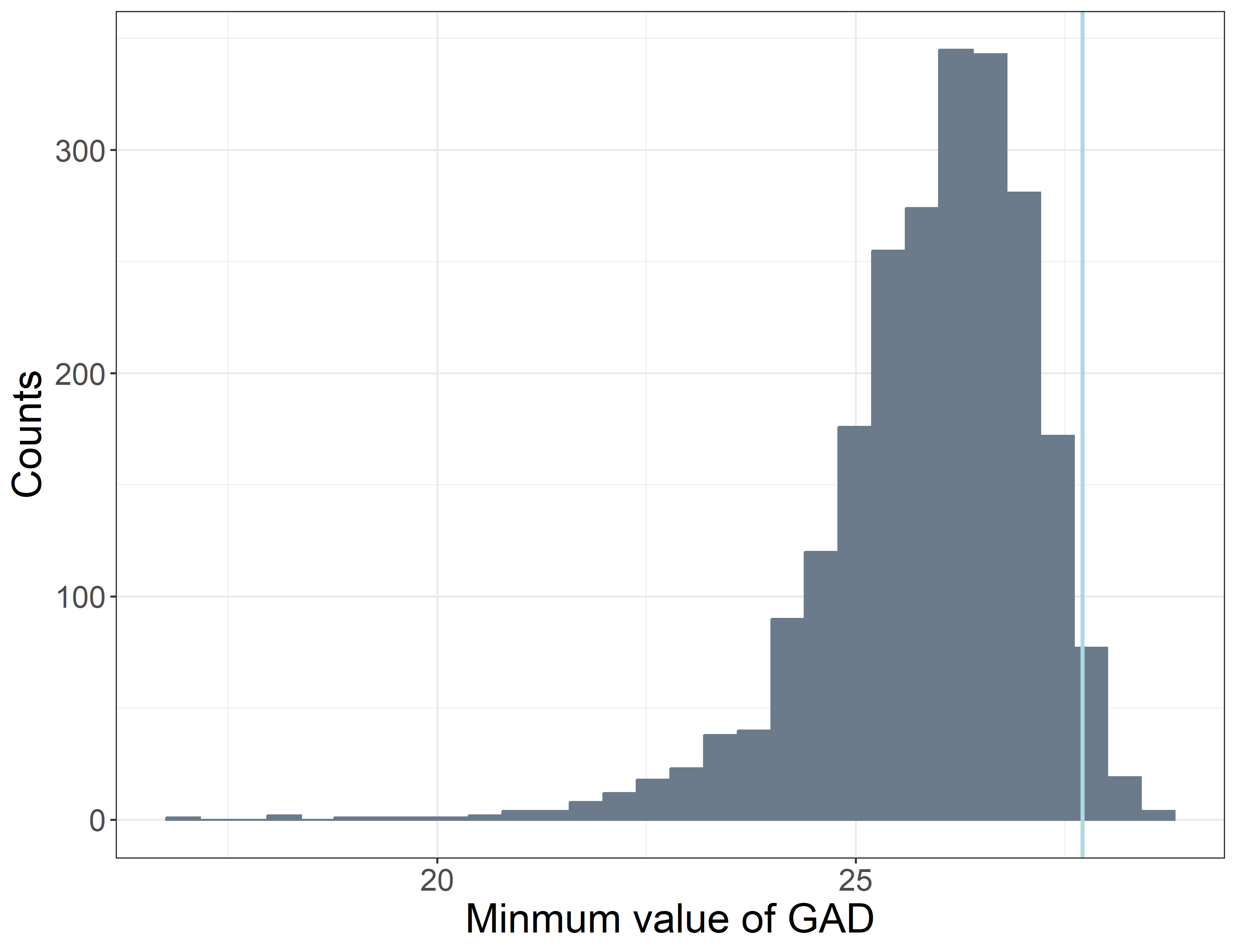}
}
\subfigure{
\centering
\includegraphics[width=0.22\textwidth]{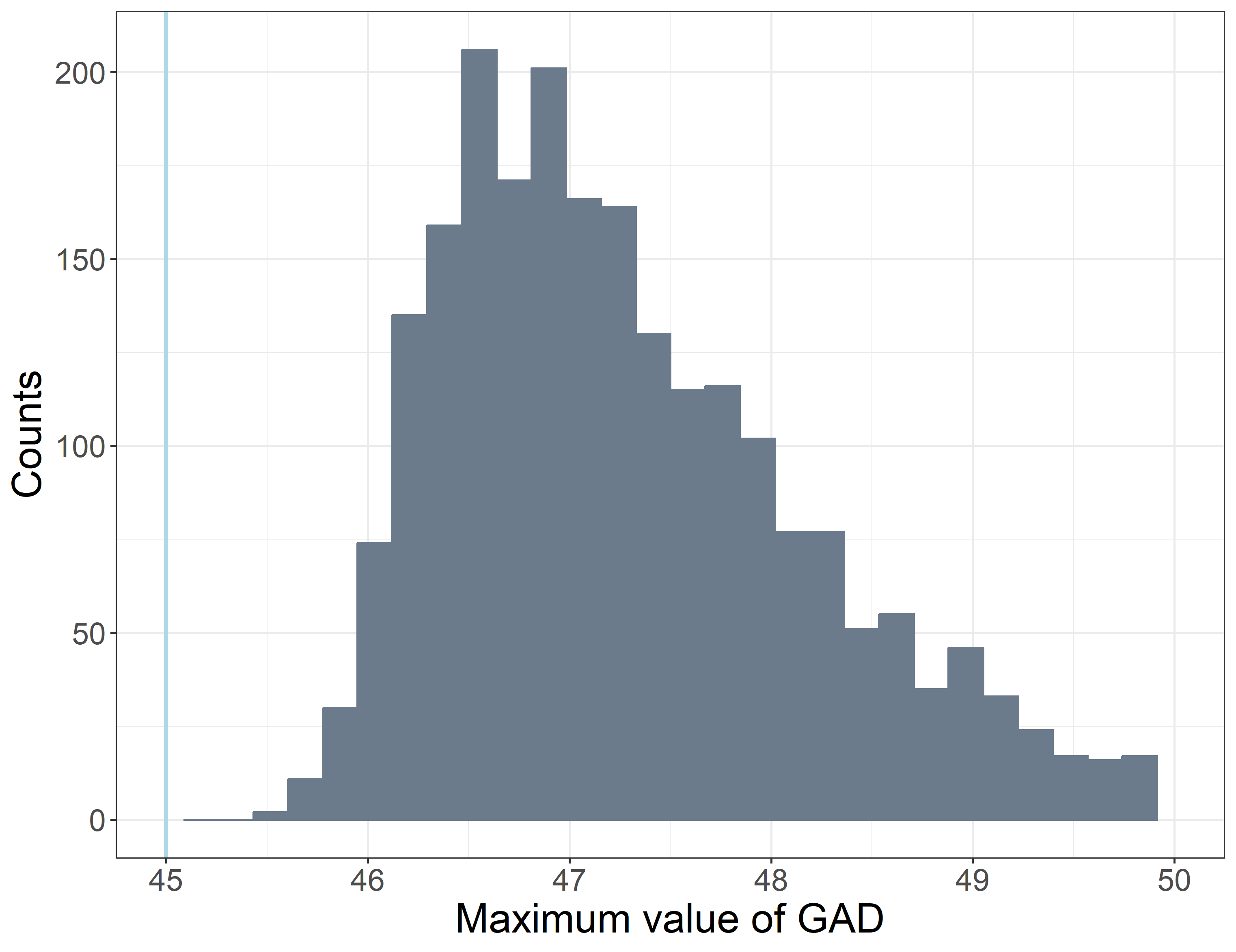}
}
\\
\subfigure{
\centering
\includegraphics[width=0.22\textwidth]{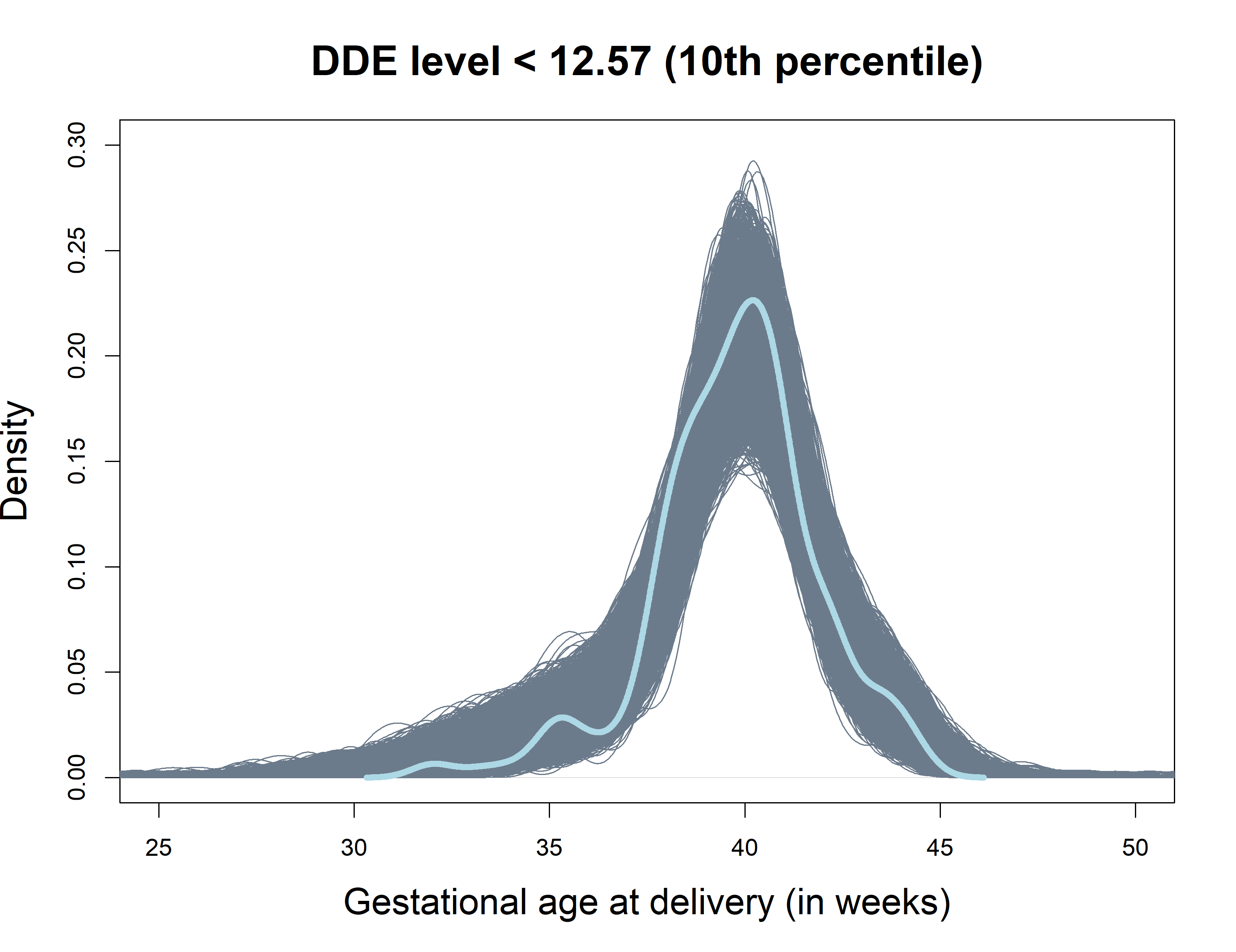}
}
\subfigure{
\centering
\includegraphics[width=0.22\textwidth]{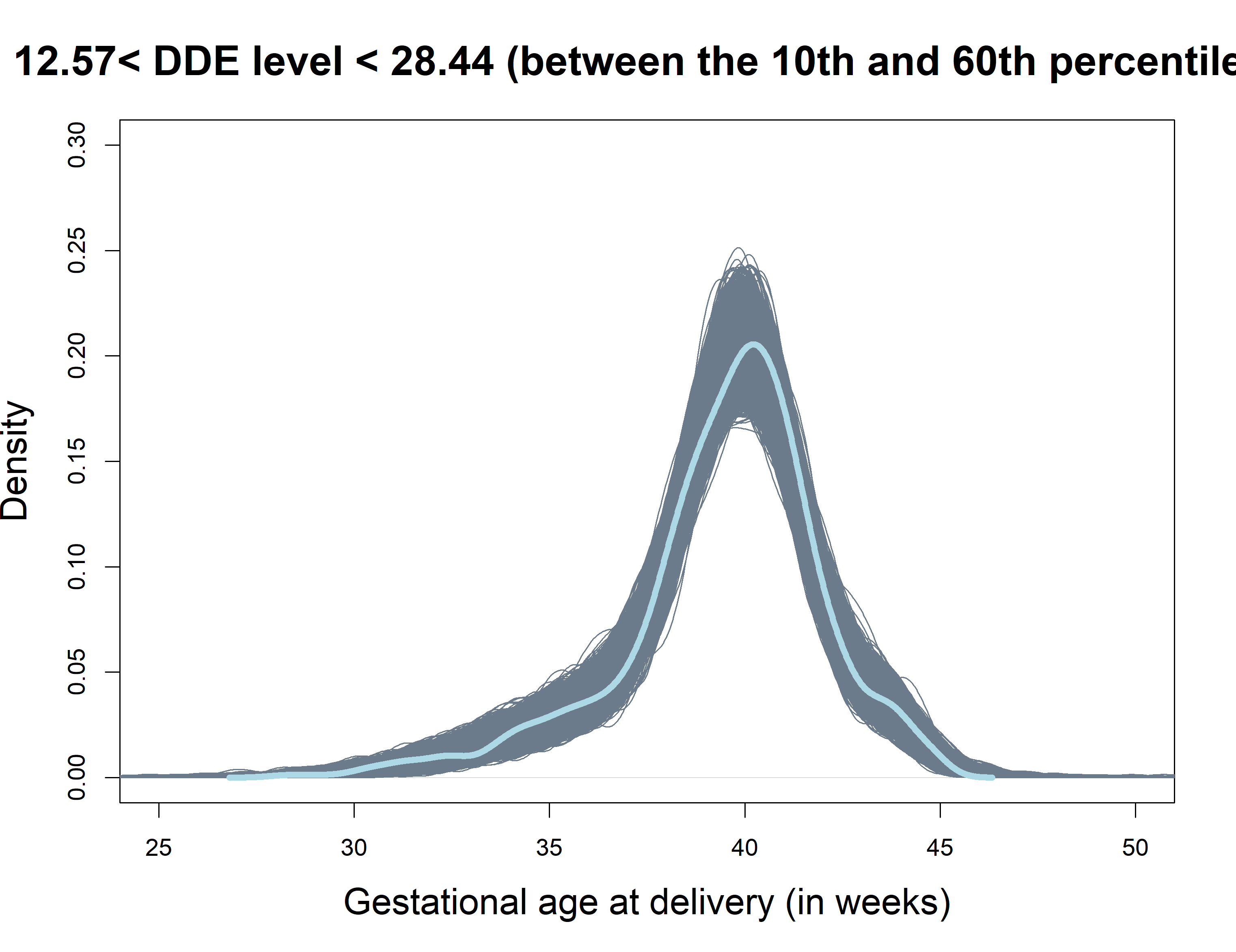}
}
\subfigure{
\centering
\includegraphics[width=0.22\textwidth]{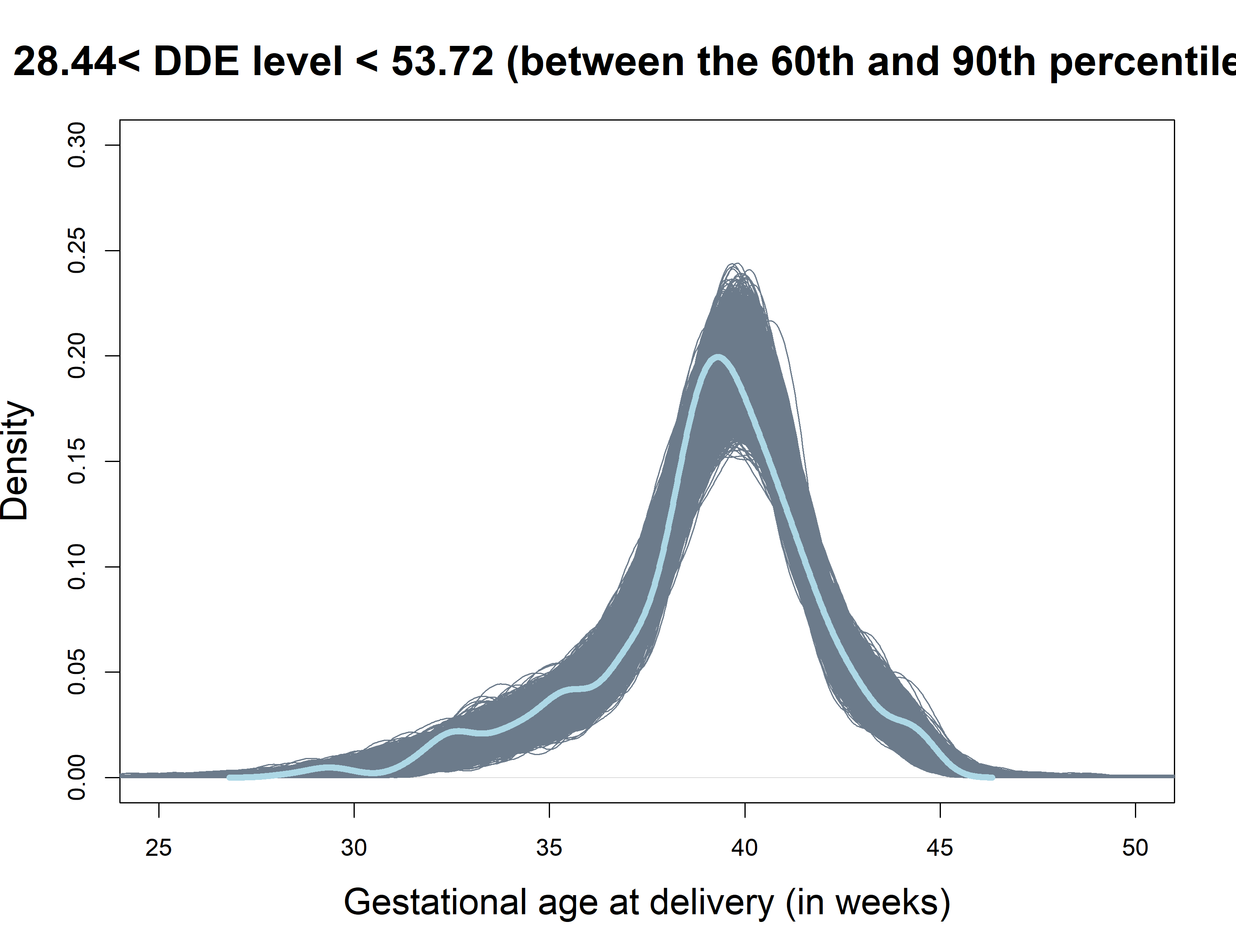}
}
\subfigure{
\centering
\includegraphics[width=0.22\textwidth]{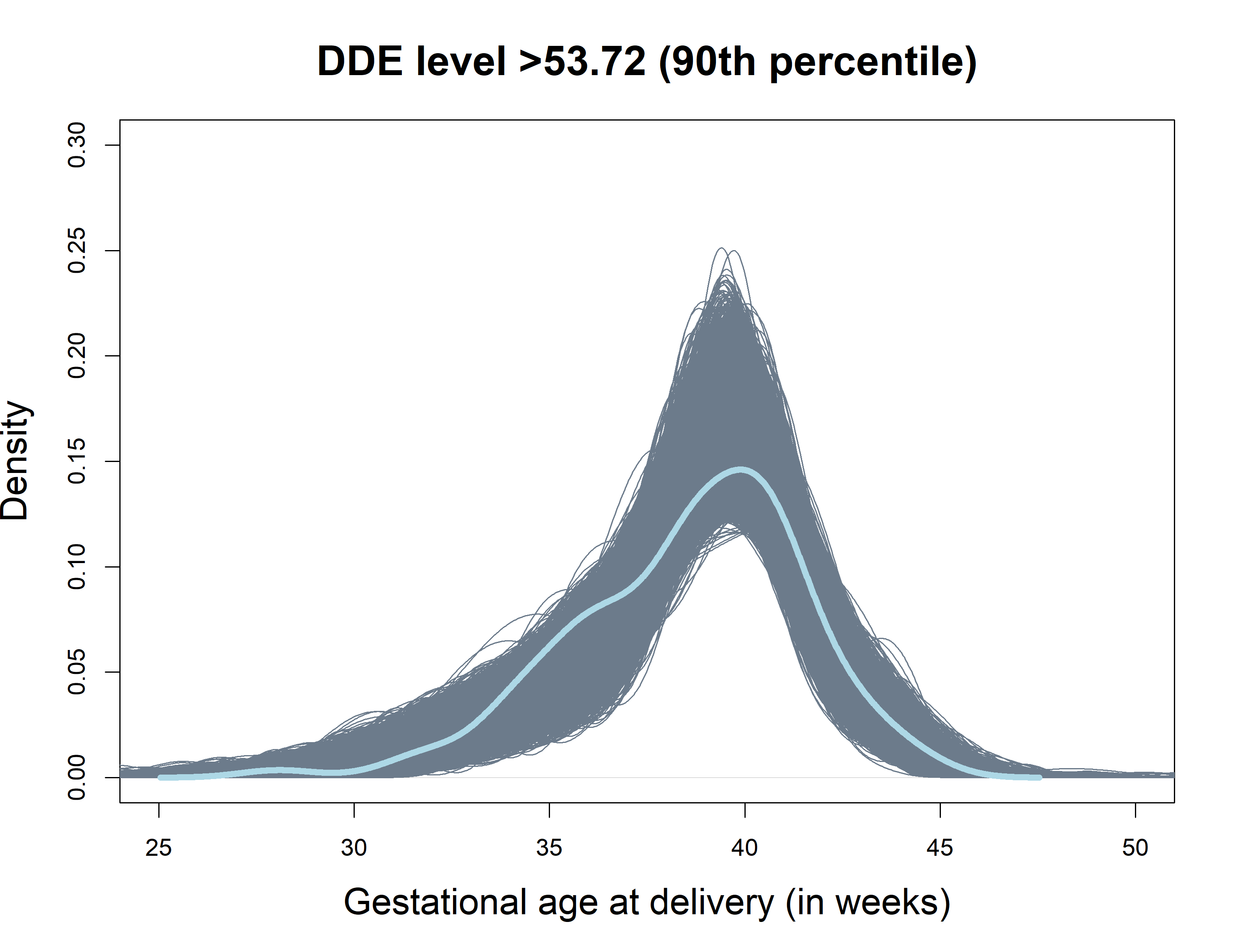}
}

\centering
\caption{Posterior predictive checks. Top row: Estimated kurtosis and skewness of the DDE with the observed minmum and maximum value (light blue), shown alongside estimates from the 5000 datasets drawn from the posterior predictive distribution (grey). Bottom row: Kernel density estimate of the observed DDE (light blue), shown alongside kernel density estimates from the 5000 datasets drawn from the posterior predictive distribution (grey).}
\label{dde_post_checks}
\end{figure}

\bibliography{reference}   
\bibliographystyle{chicago}      

\end{document}